\documentclass[10pt, reqno]{amsart}
\usepackage{og-amsart}

\numberwithin{equation}{section}\address{Department of Mathematics, Princeton University}\email{onyxg@princeton.edu}
\author{Onyx Gautam}
\date{\today}
\title{Semilinear wave equations on the Witten bubble spacetime}
\begin{document}

\maketitle
\begin{abstract}
We initiate the study of \emph{nonlinear} wave equations on the Witten bubble spacetime
(also known as the ``bubble of nothing''), which is an \(\mathrm{SO}(3,1)\times
\mathrm{U}(1)\)-symmetric solution to the Einstein vacuum equations in \((4 +
1)\) dimensions. This spacetime was introduced by Witten \cite{WITTEN1982481} to
model the semiclassical instability of the Kaluza--Klein spacetime
\(\R^{3+1}\times S^1\) (which is classically stable by work of
Huneau--Stingo--Wyatt \cite{kaluzakleinstability}). We prove a small-data global
existence result without symmetry assumptions and an improved decay result for
\(\mathrm{U}(1)\)-symmetric solutions to a class of semilinear equations
satisfying a version of the null condition. Key to the proof are novel estimates
for the \emph{linear} wave equation, which has been studied in the physics literature
by Bhawal and Viveshwara \cite{PhysRevD.42.1996} and in the mathematics
literature by Bachelot \cite{Bachelot_2016}.
\end{abstract}
\setcounter{tocdepth}{2}
\tableofcontents
\newpage
\section{Introduction}
The flat Kaluza--Klein spacetime \(\R^{3 + 1}\times S^1\), although stable as a
classical solution to the Einstein vacuum equations \cite{kaluzakleinstability},
is not a suitable ground state for gravity in \((4+1)\) dimensions. Indeed,
Witten \cite{WITTEN1982481} showed that a semiclassical decay process called
barrier penetration will cause a hole, or bubble, to spontaneously form in
space, smoothly pinching off the compact dimension. Although the bubble is
initially of microscopic size comparable to that of the compact \(S^1\)
dimension, it expands with uniform acceleration. The hole therefore approaches
the speed of light and destroys the Kaluza--Klein vacuum.

In fact, the bubble cannot expand all the way to infinity. The instability will
cause bubbles to form all over spacetime. These bubbles will expand until their
boundaries meet, at which point the evolution becomes intractable. In the model
scenario, only one bubble forms and indeed expands to infinity. This scenario is
described by the \emph{Witten bubble spacetime}, also known as the ``bubble of nothing.''
This spacetime is nonstationary, but it is highly symmetric, with symmetry group
\(\SO(3,1)\times \textnormal{U}(1)\).

We are interested in the classical stability of the Witten bubble spacetime as a
solution to the Einstein vacuum equations, which can be thought of as a system
of quasilinear wave equations. This is a difficult problem. As a first step, we
study the model problem of a semilinear wave equation on this spacetime. A basic
example is the equation
\begin{equation}\label{model-semilinear-equation}
\Box_g\varphi = f\cdot g^{\alpha{}\beta{}}\partial_\alpha{}\varphi\partial_\beta{}\varphi.
\end{equation}
Here \(g\) is the Witten bubble metric (see \zcref{intro-WB-metric}),
\(\Box_g\varphi{} \coloneqq{} g^{\mu{}\nu{}}\Grad _\mu{}\Grad _\nu{}\varphi{}\) is
the associated Laplace--Beltrami operator, and \(f\) is a spacetime function
which is bounded (together with its derivatives).

We prove that solutions to a class of equations satisfying a version of the null
condition (including \zcref{model-semilinear-equation}) that arise from small,
regular, and decaying initial data (without symmetry assumptions) exist globally
in time. Moreover, solutions that are \(\textnormal{U}(1)\)-symmetric (that is,
independent of the \(S^1\) dimension) enjoy additional decay properties in time.
The proof requires novel estimates for the linear wave equation
\(\Box{}_g\varphi{} = 0\). See
\zcref{main-theorem-linear-intro,main-theorem-semilinear-intro} for detailed
statements of our main results.

One difficulty in the proof is that the \((4 + 1)\)-dimensional Witten bubble
geometry interpolates between a near-bubble region that is \((2 +
1)\)-dimensional with an extra \(S^2\) dimension (where the geometry is far from
flat) and an asymptotically flat region that is \((3 + 1)\)-dimensional with an
extra \(S^1\) dimension. To prove linear estimates in both of these regions, we
decompose the scalar field into its \(\textnormal{U}(1)\)-symmetric and
non-\(\textnormal{U}(1)\)-symmetric parts. The \(\textnormal{U}(1)\)-symmetric
part behaves like a massless wave on Minkowski space and enjoys strong
decay-in-time estimates. On the other hand, the
non-\(\textnormal{U}(1)\)-symmetric part is non-Minkowskian. It acquires an
effective mass that grows at spatial infinity. The growth of this mass produces
an effective potential that is \emph{confining}. As a result, after accounting for the
overall expansion of the spacetime, the non-\(\textnormal{U}(1)\)-symmetric part
of the scalar field may stay trapped near the bubble and fail to decay in time.
Nevertheless, this part of the wave decays faster in space due to the growth of
its effective mass. To handle nonlinear interactions, we balance the better
decay in time of the \(\textnormal{U}(1)\)-symmetric part and better decay in
space of the non-\(\textnormal{U}(1)\)-symmetric part.
\subsection{Introduction to the Witten bubble geometry}
\label{intro-geometry} The Witten bubble spacetime is an
\(\textnormal{SO}(3,1)\times \textnormal{U}(1)\)-symmetric solution to the
Einstein vacuum equations in \((4 + 1)\) dimensions. See \zcref{fig:bubble-penrose,fig:bubble-R2} for
depictions of its geometry. It has smooth structure \(\mathcal{M} = \R\times
\R^2\times S^2\). The \(\R\) factor represents time. Away from its origin, the
\(\R^2\) factor is \(\R\times S^1\), and the size of the \(S^1\) factor tends to
a constant at infinity. The geometry of the \(\R^2\) factor is therefore
asymptotically cylindrical. When combined with the \(S^2\) factor, the \(\R\)
factor of the \(\R\times S^1\) reproduces, via spherical polar coordinates, the
familiar \(\R^3\) near infinity. Thus the Witten bubble spacetime approaches the
Kaluza--Klein spacetime \(\R^{3 + 1}\times S^1\) at infinity.
\begin{figure}
 \begin{subfigure}{0.49\textwidth}
     \def\svgwidth{\textwidth}
     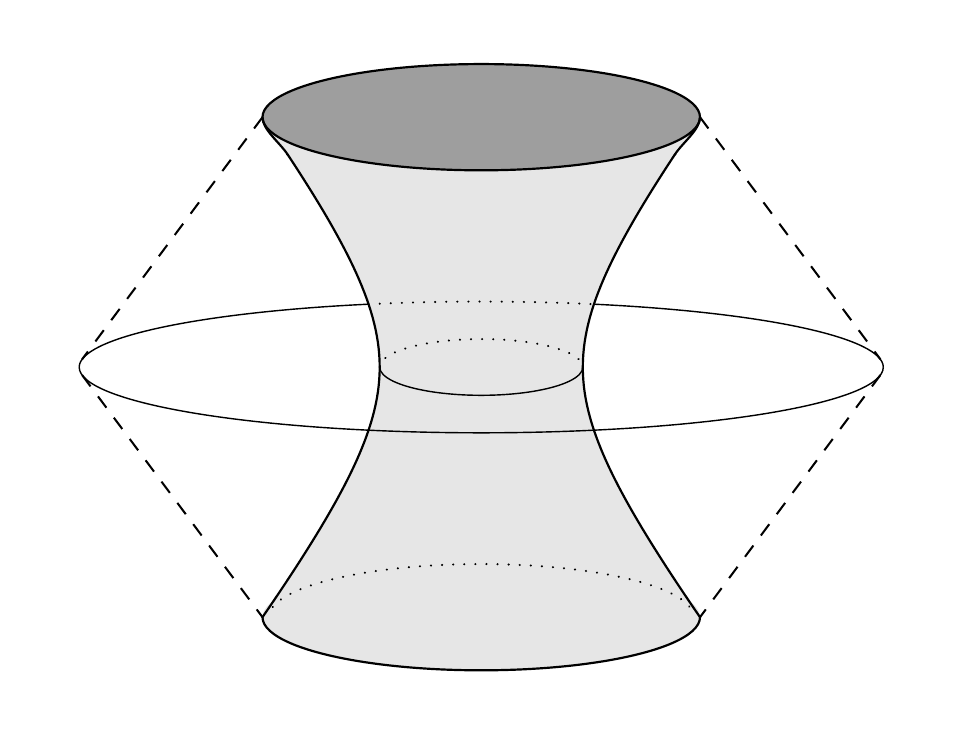
     \caption{\((2+1)\)-dimensional Penrose-type diagram of \(\mathcal{M}\)}
     \label{fig:bubble-3d}
 \end{subfigure}
 \hfill
 \begin{subfigure}{0.49\textwidth}
     \def\svgwidth{\textwidth}
\begingroup%
  \makeatletter%
  \@ifundefined{includegraphics}{%
    \PackageError{bubble-2d}{The graphicx package is required}{Load \string\usepackage{graphicx} in the preamble.}%
  }{}%
  \@ifundefined{rotatebox}{%
    \PackageError{bubble-2d}{The graphicx package is required for rotated labels}{Load \string\usepackage{graphicx} in the preamble.}%
  }{}%
  \@ifundefined{scalebox}{%
    \PackageError{bubble-2d}{The graphicx package is required for scaled labels}{Load \string\usepackage{graphicx} in the preamble.}%
  }{}%
  \providecommand\color[2][]{\renewcommand\color[2][]{}}%
  \providecommand\transparent[1]{}%
  \newcommand*\fsize{\dimexpr\f@size pt\relax}%
  \newcommand*\lineheight[1]{\fontsize{\fsize}{#1\fsize}\selectfont}%
  \ifx\svgwidth\undefined%
    \setlength{\unitlength}{450bp}%
    \ifx\svgscale\undefined%
      \relax%
    \else%
      \setlength{\unitlength}{\unitlength * \real{\svgscale}}%
    \fi%
  \else%
    \setlength{\unitlength}{\svgwidth}%
  \fi%
  \global\let\svgwidth\undefined%
  \global\let\svgscale\undefined%
  \makeatother%
  \begin{picture}(1,0.84615385)%
    \lineheight{1}%
    \setlength\tabcolsep{0pt}%
    \put(0,0){\includegraphics[width=\unitlength,page=1]{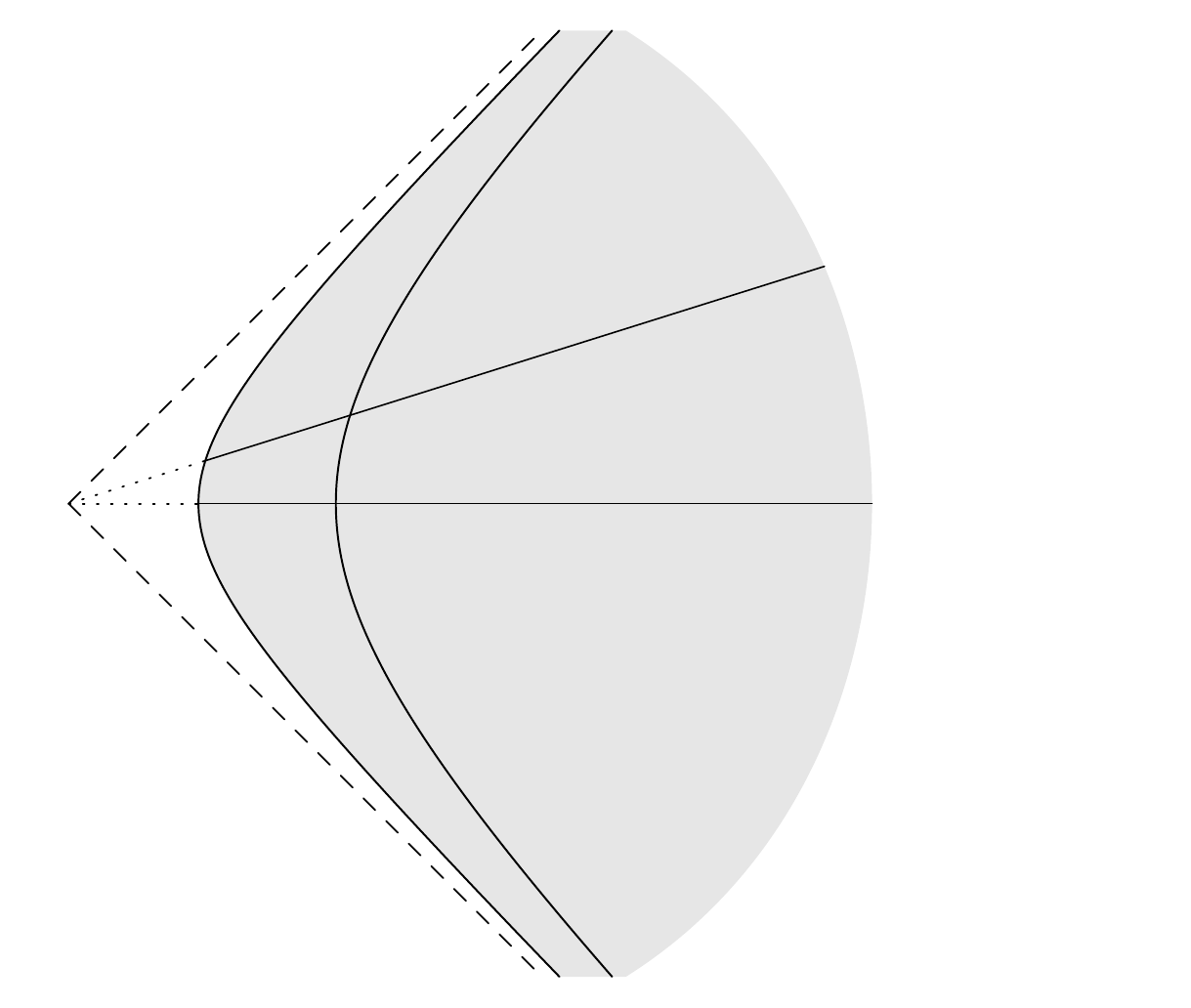}}%
    \put(0.18333333,0.38461538){\color[rgb]{0,0,0}\makebox(0,0)[lt]{\lineheight{1.22}\smash{\rotatebox{-56}{\scalebox{0.82}{\begin{tabular}[t]{c}$\{\rho = R\}$\end{tabular}}}}}}%
    \put(0.32435897,0.30128205){\color[rgb]{0,0,0}\makebox(0,0)[lt]{\lineheight{1.22}\smash{\rotatebox{-54}{\scalebox{0.82}{\begin{tabular}[t]{c}$\{\rho = \rho_0 > R\}$\end{tabular}}}}}}%
    \put(0.71794872,0.60256410){\color[rgb]{0,0,0}\makebox(0,0)[lt]{\lineheight{1.22}\smash{\begin{tabular}[t]{c}$\{\tau = \tau_0 > 0\}$\end{tabular}}}}%
    \put(0.75384615,0.41282051){\color[rgb]{0,0,0}\makebox(0,0)[lt]{\lineheight{1.22}\smash{\begin{tabular}[t]{c}$\{\tau = 0\}$\end{tabular}}}}%
  \end{picture}%
\endgroup%

     \caption{\((1+1)\)-dimensional representation of \(\mathcal{M}\)}
     \label{fig:bubble-2d}
 \end{subfigure}
 \caption{\zcref[S]{fig:bubble-3d} depicts the \(\rho{}\)-direction, the \(\tau{}\)-direction, and the equatorial direction of the \(S^2\)-factor of \(\mathcal{M}\). The space "inside" the shaded region is not part of \(\mathcal{M}\). The set \(\set{\rho = R}\), which is shaded in light gray, traces out \((2+1)\)-dimensional de Sitter space. In \zcref{fig:bubble-2d}, each point away from \(\set{\rho{}=R}\) represents an
\(S^2\times S^1\). We view this quotient of \(\mathcal{M}\) as the exterior of a
one-sheeted hyperboloid in Minkowski space. As \(R\to 0\), this construction
recovers the exterior of a double-light cone, which is drawn with dashed lines. The \((\tau{},\rho{})\) coordinates
are related to the Minkowskian \((t,r)\) coordinates by \(r = \rho{}\cosh
\tau{}\) and \(t = \rho{}\sinh \tau{}\). This figure also includes level sets of \(\tau{}\), which in Minkowski correspond to
integral curves of the scaling vector field \(t\partial_t + r\partial_r\).}
 \label{fig:bubble-penrose}
\end{figure}
The subset of \(\mathcal{M}\) defined by the origin of the \(\R^2\) factor is a
codimension-two timelike hypersurface with topology \(\R\times S^2\). We refer
to this hypersurface as \emph{the bubble}. It represents the boundary of the expanding
hole, where the compact dimension is pinched off. The Witten bubble metric away
from the bubble takes the form
\begin{equation}\label{intro-WB-metric}
g = -\rho^2\dd{}\tau^2 + (1-(R/\rho{})^2)^{-1}\dd{}\rho^2 + \rho^2\cosh^2 \tau{}g_{S^2} + R^2(1-(R/\rho{})^2)g_{S^1},
\end{equation}
where \(g_{S^2}\) is the metric on the unit round sphere. The domain of the
coordinates in \zcref{intro-WB-metric} is \(\R_{\tau{}}\times
(R,\infty)_\rho{}\times S^2\times S^1\), where \(S^1 = \R/2\pi{}\Z\) and \(R >
0\) is a real parameter representing the radius of both the compact dimension
near infinity and the minimal-area sphere of symmetry. The
\((R,\infty)_\rho{}\times S^1\) factor represents the complement of the origin
in the \(\R^2\) factor of the full Witten bubble spacetime. The apparent
degeneracy in the metric \zcref{intro-WB-metric} as \(\rho{}\to R\) is merely a
coordinate singularity. It can be resolved by a change of coordinates in the
\((\rho{},\theta{})\) domain, where \(\theta{}\) is the coordinate on \(S^1\).

The induced metric on the bubble is \(g|_{\set{\rho{}=R}} = R^2[-\dd{}\tau^2 + \cosh
^2\tau{}g_{S^2}]\). The bubble is therefore conformal to \((2 + 1)\)-dimensional
de Sitter space \(\textnormal{dS}_{2+1}\) with a constant conformal factor \(R^2\).
In fact, the level sets of \(\rho{}\) are foliated by submanifolds conformal to
\(\textnormal{dS}_{2 + 1}\) (with constant conformal factor), and the \(S^1\)
factor parametrizes the foliation.
\begin{figure}
 \begin{subfigure}{0.49\textwidth}
     \def\svgwidth{\textwidth}
     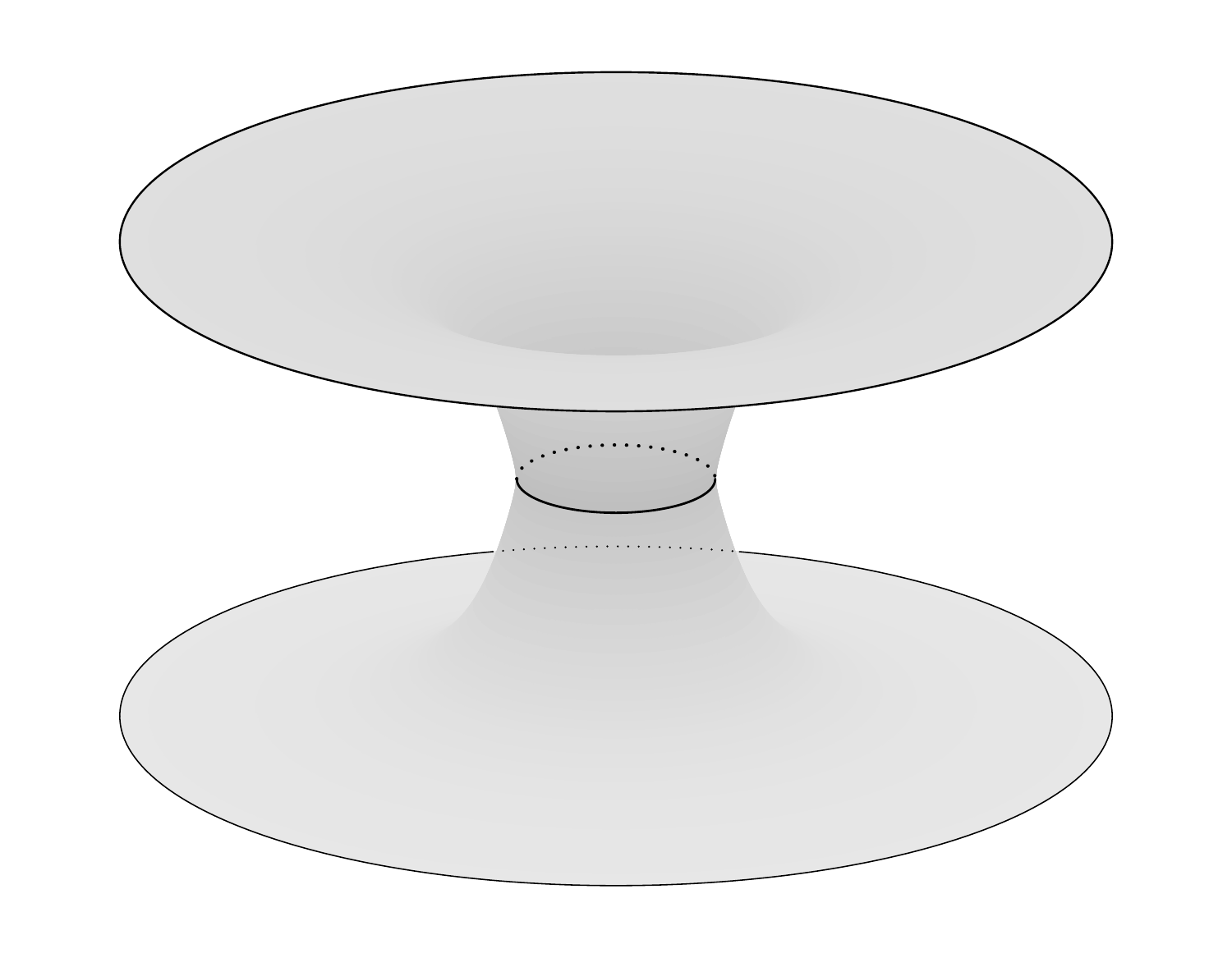
     \caption{Radial and \(S^2\)-equatorial directions of \(\mathcal{M}\)}
     \label{fig:bubble-throat}
 \end{subfigure}
 \hfill
 \begin{subfigure}{0.49\textwidth}
     \def\svgwidth{\textwidth}
     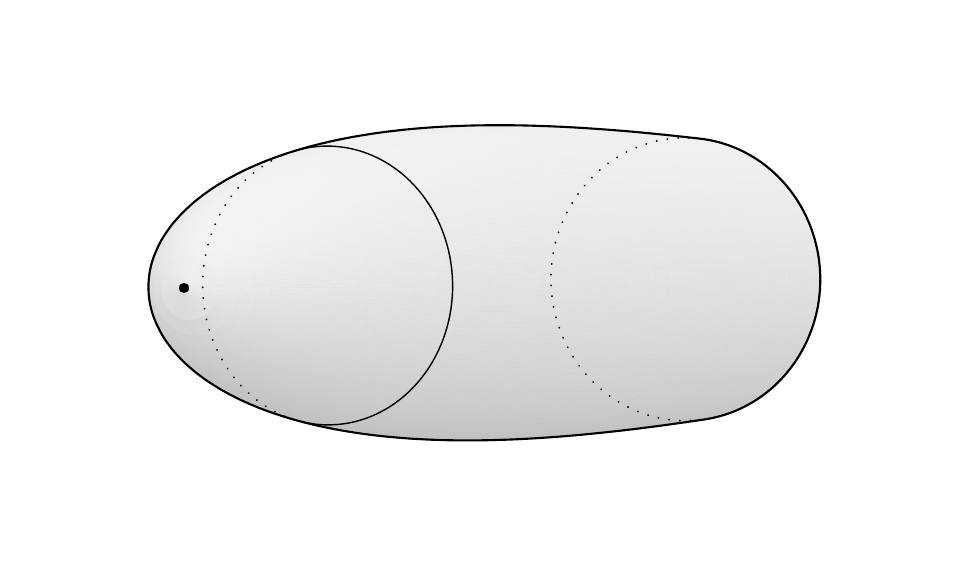
     \medskip
     \caption{Radial and \(S^1\)-directions of \(\mathcal{M}\)}
     \label{fig:bubble-cigar}
 \end{subfigure}
 \caption{\zcref[S]{fig:bubble-throat} depicts the radial \(\rho{}\)-direction, the equatorial direction of the \(S^2\)-factor together
with a pair of antipodal points in the \(S^1\). The circle around the "throat"
of this diagram is the bubble, where the antipodal points coincide and the
\(S^2\)-factor has minimal size. \zcref[S]{fig:bubble-cigar} depicts the full \(\R^2\)-factor of
\(\mathcal{M}\), which is shaped like a cylinder with a cap. Here \(\set{\rho{}=R}\) is a point (the origin). We have also drawn two other level sets of \(\rho\).}
\label{fig:bubble-R2}
\end{figure}
Formally setting \(R = 0\) in the metric \zcref{intro-WB-metric} removes the \(S^1\)
factor (which has radius \(R\)) and yields the Minkowski metric in so-called
Rindler coordinates:
\begin{equation}\label{intro-Mink-metric}
g_{\textnormal{mink}} = -\rho^2\dd{}\tau^2 + \dd{}\rho^2 + \rho^2\cosh ^2\tau{}g_{S^2}.
\end{equation}
The \((\tau{},\rho{})\) coordinates are related to the usual Minkowskian \((t,r)\)
coordinates by \(r = \rho{}\cosh \tau{}\) and \(t = \rho{}\sinh \tau{}\). Since
\(\rho{}\) originally ranged over \((R,\infty)\), the region in Minkowski space
obtained after setting \(R = 0\) is the exterior of the double light cone. For
\(R > 0\), one can think of the Witten bubble spacetime as living in the
exterior of a one-sheeted hyperboloid in Minkowski space, as depicted in
\zcref{fig:bubble-2d}.

We finish this section with a discussion of null geodesics (see
\zcref{geodesic-motion} for more details). All null geodesics on the Witten bubble
spacetime reach \(\set{r=\infty}\), as they must, in view of the expansion to
infinity of the bubble itself. The Witten bubble spacetime therefore does not
exhibit trapping in the traditional sense. However, null geodesics with non-zero
angular momentum in the \(S^1\) direction remain within a bounded distance of
the bubble.\footnote{More precisely, for each null geodesic with non-zero angular momentum in
the \(S^1\) direction, there exist \(\rho_1\) and \(\rho_2\) such that \(R <
\rho_1\le \rho_2<\infty\) and the geodesic remains in the region
\(\set{\rho_1\le \rho{}\le \rho_2}\) (see \zcref{geodesic-motion} for further
discussion). In more geometric terms, the ratio between the \(r\)-value of the
geodesic and the \(r\)-value of the bubble at time \(\tau{}\) remains bounded,
uniformly in \(\tau{}\). Here \(r\) is the area-radius function, a geometric
quantity. Because of the exponential expansion of the spacetime (in time
\(\tau{}\)), it is appropriate to consider the ratio of \(r\)-values, rather
than the difference.} For this reason, we call these geodesics \emph{weakly trapped}. The
source of the weak trapping is that the motion of geodesics with non-zero
\(S^1\)-angular momentum is determined by an effective radial potential that is
confining, in the sense that it grows towards spatial infinity. On the other
hand, null geodesics with no \(S^1\)-angular momentum are generically not
constrained in this way, and ``escape'' the bubble. The analytic consequence of
weak trapping is that, for a solution \(\varphi{}\) to \(\Box_g\varphi{} = 0\),
the non-\(\textnormal{U}(1)\)-symmetric part of the renormalized quantity
\(r\varphi{}\) can be periodic in \(\tau{}\) (see \zcref{time-periodic-intro} of
\zcref{main-theorem-linear-intro}). Here \(r = \rho{}\cosh \tau{}\) is the
area-radius function associated to the spheres of symmetry.
\subsection{Our main results}
We now state our main results on the linear wave equation and on semilinear wave
equations in \zcref{main-theorem-linear-intro,main-theorem-semilinear-intro},
respectively. We prove (and give a precise version of)
\zcref{main-theorem-linear-intro} in \zcref{main-proofs}. Part
\zcref{intro:semilinear-outside-symmetry} of \zcref{main-theorem-semilinear-intro} is
proven in \zcref{semilinear}, and part \zcref{intro:semilinear-symmetric} is proven in
\zcref{semilinear-axisymmetric}. We discuss the main ideas of these proofs in
\zcref{proof-ideas}.
\begin{theorem}[Main results on the linear wave equation]
Let \(\varphi{}\) be a solution to the linear wave equation \(\Box_{g}\varphi{} = 0\) on the
Witten bubble spacetime \((\mathcal{M},g)\) arising from sufficiently regular
and decaying initial data. Recalling the coordinate expression
\zcref{intro-WB-metric} for the Witten bubble metric, define
\(\psi{}\coloneqq{}r\varphi{}\), where \(r = \rho{}\cosh \tau{}\) is the
area-radius function corresponding to the orbits of spherical symmetry. Let
\(\psi_0\) be the \(\textnormal{U}(1)\)-symmetric part of \(\psi{}\) (namely the
average of \(\psi{}\) over the orbits of \(\textnormal{U}(1)\) symmetry), and
define \(\psi_{\ge 1}\coloneqq{}\psi{}-\psi{}_0\). Then:
\begin{enumerate}
\item \label{bounded-intro} \emph{(Quantitative boundedness and decay)} The quantity \(\psi{}_{\ge
   1}\) decays as \(\rho^{-1/2}\), while the quantity \(\psi_0\) is bounded.
Since \(r = \rho{}\cosh \tau{}\) grows in \(\tau{}\), the boundedness of
\(\psi{} = r\varphi{}\) can be interpreted as a \emph{time decay} result for
\(\varphi{}\). These bounds can be expressed in terms of energies of
(differentiated versions of) \(\psi{}\) associated to fluxes of appropriate
vector fields.
\item \label{radiation-field-intro} \emph{(Existence of a radiation field)} The quantity \(\psi{}\)
attains a limit towards null infinity (namely, along outgoing null geodesics)
that is \(\mathrm{U}(1)\)-symmetric and can be non-vanishing.
\item \label{bubble-vanishing-intro} \emph{(Vanishing property of solutions at the bubble)} If
\(\varphi{}\) has vanishing \(\mathrm{U}(1)\)-symmetric part, then it vanishes along
the bubble.
\item \label{U1-improved-decay-intro} \emph{(Improved time decay for
\(\mathrm{U}(1)\)-symmetric solutions)} If \(\varphi{}\) is
\(\mathrm{U}(1)\)-symmetric, then the energy of \(\psi{} = r\varphi{}\)
decays through a foliation whose leaves intersect null infinity
appropriately. In particular, \(\psi{}\) decays along the bubble and along
null infinity (at an inverse polynomial rate with respect to Minkowskian
coordinates).
\item \label{time-periodic-intro} \emph{(Lack of improved time decay outside of
\(\textnormal{U}(1)\)-symmetry)} There exist non-zero solutions \(\varphi{}\)
with vanishing \(\textnormal{U}(1)\)-symmetric part such that \(\psi{} =
   r\varphi{}\) is periodic in \(\tau{}\) along integral curves of
\(\partial_\tau{}\). In particular, \(\psi{}\) need not decay along such a
curve (unless this curve lies in the bubble itself, where \(\psi{}\) vanishes
by \zcref{bubble-vanishing-intro}).
\end{enumerate}
\label{main-theorem-linear-intro}
\end{theorem}
Next, we state our main theorem on semilinear wave equations. Since the Witten
bubble geometry converges to that of Minkowski (with an extra \(S^1\) factor)
near infinity, small-data global existence for semilinear equations is not
possible for arbitrary quadratic nonlinearities. A counterexample on Minkowski
space, due to John \cite{fritz-john-blowup}, is given by the equation
\(\Box_{\textnormal{m}}\varphi{} = (\partial_t\varphi{})^2\). It is well-known
that if the quadratic nonlinearity satisfies the \emph{null condition}, which in
particular rules out the example of John \cite{fritz-john-blowup}, then solutions
arising from sufficiently small data exist globally in time
\cite{christodoulou-null-condition,klainerman-null-condition}. In this work, we
formulate a global version of the null condition suitable for the Witten bubble
spacetime. Our null condition reduces to (a generalization of) the classical
null condition near null infinity, where the geometry is asymptotically
Kaluza--Klein.
\begin{theorem}[Main results on semilinear wave equations]
Consider a semilinear wave equation of the form
\begin{equation}\label{intro-semilinear}
\Box_g\varphi{} = F(\Grad \varphi{})
\end{equation}
for some nonlinearity \(F(\Grad \varphi{})\) that is quadratic in the
derivatives of \(\varphi{}\) and satisfies a version of the \emph{null condition}. In
particular, the nonlinearity can be \(F(\Grad \varphi{}) =
g^{\alpha{}\beta{}}\partial_\alpha{}\varphi{}\partial_\beta{}\varphi{}\) or
versions thereof with weights that grow in \(\tau{}\), where \(\tau{}\) is the
time function in \zcref{intro-WB-metric}. Write \(\psi{} = r\varphi{}\). Then:
\begin{enumerate}
\item \label{intro:semilinear-outside-symmetry} \emph{(Small-data global existence for
semilinear wave equations outside of symmetry)} Suppose that the
non-\(\textnormal{U}(1)\)-symmetric part of \(F(\Grad \varphi{})\) has
sufficient decay in space. Then solutions to \zcref{intro-semilinear} arising
from sufficiently small initial data posed on a \emph{constant-\(\tau{}\)
hypersurface} \(\Sigma{}\), which terminates at \emph{spacelike infinity}, exist
globally in \(\tau{}\)-time to the future of \(\Sigma{}\), and \(\psi{}\)
grows at most polynomially in \(\tau{}\) (and so at most logarithmically in
the Minkowskian time \(t=\rho{}\sinh \tau{}\)).
\item \label{intro:semilinear-symmetric} \emph{(Semilinear wave equations within
\(\textnormal{U}(1)\)-symmetry)} Suppose that \(F(\Grad \varphi{})\) is
\(\textnormal{U}(1)\)-symmetric. Then the results of part
\zcref{intro:semilinear-outside-symmetry} apply, and:
\begin{enumerate}
\item \label{intro:semilinear-symmetric-ge} \emph{(Small-data global existence)} Solutions
to \zcref{intro-semilinear} arising from sufficiently small
\(\mathrm{U}(1)\)-symmetric initial data posed on a \emph{hyperboloidal
hypersurface} \(\mathcal{H}\) that intersects \emph{null infinity} appropriately
exist globally to the future of \(\mathcal{H}\), and the quantity
\(\psi{}\) attains a finite limit at null infinity.
\item \emph{(Improved time decay results)} The global solution of part
\zcref{intro:semilinear-symmetric-ge} enjoys improved decay properties as in part
\zcref{U1-improved-decay-intro} of \zcref{main-theorem-linear-intro}. Moreover,
the admissible nonlinearities in \(\mathrm{U}(1)\)-symmetry can include
weights in \(\tau{}\) that grow faster than the weights admissible outside
of symmetry.
\end{enumerate}
\end{enumerate}
\label{main-theorem-semilinear-intro}
\end{theorem}
\begin{remark}[Dichotomy in the decay properties of \(\textnormal{U}(1)\)-symmetric and non-\(\textnormal{U}(1)\)-symmetric waves]
Parts \zcref{U1-improved-decay-intro,time-periodic-intro} of
\zcref{main-theorem-linear-intro} say that \(\psi_0\) decays faster in time than
\(\psi_{\ge 1}\), while part \zcref{bounded-intro} says that \(\psi_{\ge 1}\) decays
faster in space than \(\psi_0\). The balance between improved decay in time and
improved decay in space motivates the splitting of \(\psi{}\) into its
\(\textnormal{U}(1)\)-symmetric and non-\(\textnormal{U}(1)\)-symmetric parts
when considering nonlinear interactions in the proof of
\zcref{main-theorem-semilinear-intro}.
\end{remark}
\begin{remark}[Null geodesics and the behaviour of \(\psi\)]
Part \zcref{bubble-vanishing-intro} of \zcref{main-theorem-linear-intro} says that
\(\psi_{\ge 1}\) vanishes at the bubble. Geometrically, this corresponds to the
fact that null geodesics with non-zero angular momentum in the \(S^1\) direction
do not reach the bubble. These geodesics have \(\rho{}\)-value that is periodic
in the affine parameter. This corresponds to part \zcref{time-periodic-intro}, which
says that \(\psi_{\ge 1}\) may be periodic in \(\tau{}\). Finally, generic null
geodesics with zero \(S^1\)-angular momentum escape any finite-\(\rho{}\) region
given sufficient \(\tau{}\)-time, which correspond to the improved decay in part
\zcref{U1-improved-decay-intro}.
\end{remark}
\begin{remark}[The growth in time outside symmetry]
In part \zcref{intro:semilinear-outside-symmetry} of
\zcref{main-theorem-semilinear-intro}, we are only able to show that \(\psi{}\)
grows polynomially in \(\tau{}\). This should be contrasted with the classical
results on Minkowski space
\cite{klainerman-null-condition,christodoulou-null-condition}, where \(\psi{}\)
remains bounded. The reason for this loss is the possible periodicity in
\(\tau{}\) of \(\psi_{\ge 1}\), which is a non-Minkowskian phenomenon. See
\zcref{weak-trapping} for further discussion.
\end{remark}
\begin{remark}[Systems of semilinear wave equations]
Although we have stated \zcref{main-theorem-semilinear-intro} for scalar semilinear
wave equations, our proof extends to systems of semilinear wave equations
satisfying the null condition.
\end{remark}
\subsection{Previous literature}
\subsubsection{Mathematics literature}
The only mathematical study of the Witten bubble spacetime \((\mathcal{M},g)\)
is due to Bachelot \cite{Bachelot_2016}. The main difference between our work and
\cite{Bachelot_2016} is that we study nonlinear equations, while
\cite{Bachelot_2016} is restricted to linear equations. We also establish a novel
\(r^p\)-type estimate and provide quantitative pointwise estimates, both of
which are needed to study nonlinear equations. Moreover, \cite{Bachelot_2016}
uses spectral methods, as opposed to the physical space methods we use here.

The main result of \cite{Bachelot_2016} is a classical and quantum scattering
theory on the Witten bubble spacetime, which shows in particular that there is
no ``particle creation,'' despite the time dependence of the background. Although
Bachelot \cite{Bachelot_2016} produces an explicit formula for solutions to the
linear wave equation on \((\mathcal{M},g)\) (in terms of an infinite sum of
special functions), it is unclear how to deduce the pointwise estimates we prove
in \zcref{main-theorem-linear-intro}. We note finally that, unlike our work,
\cite{Bachelot_2016} also studies the Lorentzian version of the Riemannian
Hawking wormhole \cite{HAWKING1987337}, namely the submanifold of
\((\mathcal{M},g)\) consisting of a pair of antipodal points on \(S^1\). This
geometry was first considered in \cite{Culetu_2010}.

The Witten bubble spacetime represents a semiclassical instability of the
Kaluza--Klein spacetime \(\R^{3 + 1}\times S^1\), which is classically stable,
namely as a solution to the Einstein vacuum equations, by
\cite{kaluzakleinstability}. Classical stability is also known for the product of
high-dimensional Minkowski space with certain Ricci flat compact Riemannian
manifolds \cite{Andersson_2023}. These spacetimes are relevant in supergravity
and string theory.

The physical space methods in our work use techniques from the study of wave
equations on non-flat backgrounds, in particular the \(r^p\)-weighted energy
method of Dafermos--Rodnianski \cite{rp-method} and its extension to two space
dimensions in our work \cite{gautam-2d-waves}. The class of nonlinear equations
we study satisfies a version of the null condition of Klainerman
\cite{klainerman-null-condition}.
\subsubsection{Physics literature}
The Witten bubble spacetime was introduced in \cite{WITTEN1982481}. The first
investigation of (causal) geodesics in the Witten bubble spacetime is due to
Brill--Matlin \cite{PhysRevD.39.3151} (see also the subsequent works
\cite{matlin-thesis,Ofer_Aharony_2002}). In \cite{PhysRevD.42.1996},
Bhawal--Viveshwara studied linear waves on \((\mathcal{M},g)\).

The Witten bubble spacetime can be thought of as a double Wick rotation of the
\((4 + 1)\)-dimensional Schwarzschild black hole spacetime (see
\zcref{wick-rotation}). Analogously, one can consider a double Wick rotation of a
rotating Kerr black hole, as in \cite{Dowker_1995,Ofer_Aharony_2002}. This
produces a ``rotating bubble'' spacetime that is qualitatively different from the
Witten bubble spacetime. Rather than expanding indefinitely, the bubble reaches
a maximal radius and then shrinks again. Moreover, the compact \(S^1\) dimension
grows at infinity. Note that, although Schwarzschild is a member of the smooth
one-parameter family of Kerr black holes, the rotating bubble spacetimes do not
smoothly bifurcate off the Witten bubble spacetime, due to an ``intertwining''
between the \(S^1\) direction and one of the \(S^2\) directions. In particular,
the Kerr bubble spacetimes are not relevant for the problem of nonlinear
stability of the Witten bubble spacetime as a solution to the Einstein vacuum
equations. One can also consider higher-dimensional variants of the Witten
bubble spacetime and rotating bubble spacetimes \cite{Ofer_Aharony_2002}.
\subsection{Ideas of the proof}
\label{proof-ideas} Throughout this section, we will use the following notation:
\begin{center}
\begin{tabular}{ll}
Notation & Meaning\\
\hline
\(\varphi{}\) & solution to \(\Box_g\varphi{} = 0\)\\
\(\psi{}\) & \(\psi{} = r\varphi{}\)\\
\(\psi_0\) and \(\psi_{\ge 1}\) & \(\psi_0\) is the average of \(\psi{}\) over orbits of \(\textnormal{U}(1)\)-symmetry, and \(\psi_{\ge 1}\coloneqq{}\psi{}-\psi_0\)\\
\((\tau{},x)\) & a rescaling of the \((\tau{},\rho{})\) coordinates of \zcref{intro-WB-metric}, with \(\rho{} = R\cosh x\)\\
\((\partial_\tau{},\partial_x)\) & the coordinate derivatives in \((\tau{},x)\) coordinates\\
\end{tabular}
\end{center}
This section is structured as follows. In \zcref{brief-summary}, we give a brief
summary of the proof of
\zcref{main-theorem-linear-intro,main-theorem-semilinear-intro}. We formulate the
equation in terms of a twisted operator corresponding to \(\psi{} = r\varphi{}\)
(see \zcref{intro-coordinates}). In \zcref{intro-double-null}, we introduce the double
null coordinates on the Witten bubble spacetime. We use two energy estimates,
namely a \(\partial_\tau{}\)-energy estimate (see \zcref{intro-dt-estimate}) and a
version of the \(r^p\)-weighted energy estimates of Dafermos--Rodnianski
\cite{rp-method} (see \zcref{intro-rp-estimate-sec}). Using higher-order energy
estimates, we establish pointwise estimates via Sobolev embedding (see
\zcref{intro-pointwise}). Finally, we sketch the proof of
\zcref{main-theorem-semilinear-intro} in \zcref{intro-semilinear-sec}.
\subsubsection{A brief summary of the proof}
\label{brief-summary}
The proof of \zcref{main-theorem-linear-intro} is based on physical space energy
estimates for the linear wave equation. We derive these estimates by multiplying
the wave equation with appropriate \emph{vector field multipliers} and integrating by
parts. We also derive pointwise estimates, which use higher-order energies. The
basic strategy to prove \zcref{main-theorem-linear-intro} is a continuity argument
that uses the linear estimates (with inhomogeneity) as a priori estimates.

Throughout, we formulate energy estimates in terms of \(\psi{}\) rather than \(\varphi{}\),
since energy estimates for the latter do not hold on the Witten bubble spacetime
(see \zcref{intro-dt-estimate}). We derive two energy estimates: a basic estimate
associated to the multiplier \(\partial_\tau{}\), valid for all scalar fields,
and an \(r^p\)-type estimate associated to the multiplier \((\sinh
x)(\partial{}_\tau{} + \partial{}_x)\), valid only for \(\U(1)\)-symmetric scalar
fields. Since non-\(\U(1)\)-symmetric scalar fields solve a ``massive'' wave equation,
one cannot hope to obtain an \(r^p\)-type estimate for them. Indeed, \(\psi_{\ge
1}\) may be periodic in \(\tau{}\) (see \zcref{time-periodic-intro} of
\zcref{main-theorem-linear-intro}), and so the integrated energy associated to
\(\psi_{\ge 1}\) that an \(r^p\)-type estimate would control cannot be finite.

The pointwise estimates distinguish between \(\U(1)\)-symmetric and non-\(\U(1)\)-symmetric
modes:
\begin{equation}
\psi_0\lesssim 1, \qquad \partial_x\psi_0\lesssim 1,\qquad \psi_{\ge 1}\lesssim (\cosh x)^{-1/2},\qquad \partial_x\psi_{\ge 1}\lesssim (\cosh x)^{1/2}.
\end{equation}
The reason \(\psi_{\ge 1}\) decays faster in space than \(\psi_0\) is that \(\psi_{\ge 1}\)
solves a wave equation with ``effective mass'' that grows in space. The reason
\(\partial_x\psi_{\ge 1}\) decays worse in space than \(\partial_x\psi_0\) is
that \(\partial_x\) is analogous to the scaling vector field on Minkowski space,
which commutes well with the massless wave equation but not with the massive
wave equation.

We now turn to the nonlinear estimates. The nonlinearity in the semilinear
equation \(\Box_g\varphi{} = F(\Grad \varphi{})\) is quadratic in the
derivatives of \(\varphi{}\). Since we must work with energies of \(\psi{}\), we
rewrite this equation in terms of \(\psi{}\) as \(P\psi{} =
\tilde{F}(\psi{},\Grad \psi{})\), where \(P\) is a rescaled version of
\(\Box_g\). Since \(\psi{} = r\varphi{}\), the quadratic nonlinearity
\(\tilde{F}\) has an extra weight of \(r^{-1}\) and can now contain derivative
terms such as \((\partial_x\psi{})^2\), zeroth order terms such as \(\psi^2\),
and mixed terms such as \(\psi{}\partial_x\psi{}\). The main obstruction to
closing nonlinear estimates is decay in \(x\). The multiplier with the strongest
\(x\)-weight, namely the \(r^p\)-type multiplier \((\sinh x)(\partial{}_\tau{} +
\partial{}_x)\), is used only for the \(\U(1)\)-symmetric part of the scalar field. The
quantity with the weakest \(x\)-weight in the energy is \(\psi_0\). The quantity
with the weakest pointwise \(x\)-decay is \(\partial_x\psi_{\ge 1}\). In terms
of \(x\)-decay, the worst possible frequency interaction therefore arises from
the mixed term in \(P(\mathfrak{d}\psi{})_0\approx
\mathfrak{d}(\tilde{F}(\psi{},\Grad \psi{}))_0\), where \(\mathfrak{d}\) is a
top-order derivative. It takes the form
\((\mathfrak{d}\psi)_0\partial_x\psi_{\ge 1}\). This term can occur only when
the coefficient in front of the mixed term in the nonlinearity is
non-\(\U(1)\)-symmetric. More precisely, a quadratic term
\(\mathfrak{d}(f\psi{}\partial{}_x\psi{})_0\) contributes this most dangerous
term only in the form \(f_{\ge 1}(\mathfrak{d}\psi{})_{0}\partial_x\psi_{\ge
1}\). To control this term, we assume that \(f_{\ge 1}\) decays sufficiently
fast in space.

We discuss these ideas in more detail in
\zcref{intro-coordinates,intro-double-null,intro-dt-estimate,intro-rp-estimate-sec,intro-pointwise,intro-semilinear-sec}.
\subsubsection{A rescaled coordinate system and a twisted operator}
\label{intro-coordinates} It is convenient to introduce the coordinate \(x\) defined by \(\rho{} = R\cosh x\),
so that the bubble lies at \(\set{x=0}\) and the metric \zcref{intro-WB-metric}
takes the following form:
\begin{equation}\label{intro-x-metric}
g = R^2\bigl[-\cosh ^2x\dd{}\tau^2 + \cosh ^2x\dd{}x^2 + \cosh ^2x\cosh ^2\tau{}g_{S^2} + \tanh ^2xg_{S^1}\bigr].
\end{equation}
The area radius function then takes the form \(r = R\cosh x\cosh \tau{}\). As we will
discuss in \zcref{intro-dt-estimate}, to prove energy estimates, we must work not
with \(\varphi{}\) itself but with the twisted quantity \(\psi{} = r\varphi{}\),
which satisfies the equation \(P\psi{} = 0\), where
\begin{equation}\label{intro-P-def}
P\psi{}\coloneqq{}-\partial{}_\tau^2\psi{} + \partial{}_x^2\psi{} + \frac{1}{\cosh x\sinh x}\partial{}_x\psi{} - \frac{1}{\cosh^2x}\psi{} + \frac{1}{\cosh ^2\tau{}}\Lapl _{S^2}\psi{} +  \frac{\cosh^2x}{\tanh ^2x}\partial{}_\theta^2\psi{}.
\end{equation}
In terms of the wave operator \(\Box{}\), the operator \(P\) takes the form
\begin{equation}
P\psi{}=\rho^2r\Box{}(r^{-1}\psi{}).
\end{equation}
The choice to work with \((\tau{},x)\) coordinates and the operator \(P\) reveals the
\(-\partial_\tau^2 + \partial_x^2 + x^{-1}\partial_x\) structure near \(x=0\),
which motivates the use of techniques from the study of wave equations in two
space dimensions as in \cite{gautam-2d-waves}.
\subsubsection{The double null coordinates}
\label{intro-double-null} From now on, we use the following notation:
\begin{center}
\begin{tabular}{ll}
Notation & Meaning\\
\hline
\((u,v)\) & double null coordinates defined by \(u = \frac{1}{2}(\tau{}-x)\) and \(v = \frac{1}{2}(\tau{}+x)\)\\
\(\underline{L}\) and \(L\) & \(\underline{L} = \partial{}_u\) and \(L = \partial_v\) in \((u,v)\) coordinates\\
\((U,V)\) & Minkowskian double null coordinates, with \(U\sim -e^{-2u}\) and \(V\sim e^{2v}\)\\
\end{tabular}
\end{center}
The coordinates
\begin{equation}
u\coloneqq{} \frac{1}{2}(\tau{}-x),\qquad v\coloneqq{} \frac{1}{2}(\tau{}+x)
\end{equation}
are null with respect to the metric \zcref{intro-x-metric}. The associated
coordinate vector fields are
\begin{equation}\label{intro-L-def}
\underline{L}\coloneqq{}\partial{}_\tau{} - \partial{}_x, \qquad L\coloneqq{}\partial{}_\tau{} + \partial{}_x.
\end{equation}
Since \(g(\underline{L},L) = -2\cosh ^2x\), the \((u,v)\) coordinates are not
normalized as in Minkowski space (where
\(g_{\textnormal{m}}(\underline{L}_{\textnormal{m}},L_{\textnormal{m}}) = -2\)).
In fact, the \((u,v)\) coordinates are logarithmic with respect to the
Minkowskian double null coordinates \((U,V)\) defined by
\begin{equation}
U \coloneqq{} -\frac{R}{2}e^{-2u},\qquad V \coloneqq{} \frac{R}{2}e^{2v}.
\end{equation}
Note that \(U\le 0\) and \(V\ge \abs{U}\ge 0\) in the Witten bubble spacetime, and that
\(U\to 0\) as \(\tau{}\to \infty\) for fixed \(x\) (this is a hyperbolic
trajectory in Minkowski space). In \((U,V)\) coordinates, the metric takes the
form
\begin{equation}
g = -4(1+e^{-2x})^2\dd{}U\dd{}V + r^2g_{S^2} + R^2\tanh ^2xg_{S^1},
\end{equation}
from which we see that the Witten bubble metric approaches the Kaluza--Klein
metric (with circle of radius \(R\)) as \(x\to \infty\). In terms of the Minkowskian
double null coordinates, we have
\begin{equation}\label{LU-relation}
\underline{L}\sim \abs{U}\partial{}_U,\qquad  L \sim V\partial{}_V\sim r\partial{}_{V}.
\end{equation}
\subsubsection{A \texorpdfstring{\(\partial_\tau{}\)}{∂τ}-estimate for the twisted quantity \texorpdfstring{\(\psi\)}{ψ}}
\label{intro-dt-estimate} We introduce the following notation for use in this and
the following sections:
\begin{center}
\begin{tabular}{ll}
Notation & Meaning\\
\hline
\(\dd{}\mu{}\) & \(\dd{}\mu{}\coloneqq{}\dd{}x\dd{}\omega{}\dd{}\theta{}\), where \(\dd{}\omega{}\) is the volume form on the unit round sphere\\
\(\Sigma{}(\tau{})\) & a level set of \(\tau{}\)\\
\(s\) & a ``hyperboloidal'' time function satisfying \(s\sim u\)\\
\(\mathcal{H}(s)\) & a level set of \(s\)\\
\(E_T[\psi{}](\tau{})\) & a \(\partial_\tau{}\)-energy flux of \(\psi{}\) through \(\Sigma{}(\tau{})\)\\
\(\mathcal{E}_T[\psi{}](s)\) & a \(\partial_\tau{}\)-energy flux of \(\psi{}\) through \(\mathcal{H}(s)\)\\
\(\Psi{}_0\) & \(\Psi_0 = (\tanh x)^{1/2}\psi_0\) is used to define the \(r^p\)-type energies\\
\(\mathcal{E}_p[\psi{}](s)\) & an \(r^p\)-type energy flux of \(\psi{}\) through \(\mathcal{H}(s)\)\\
\(E_p[\psi{}](\tau{})\) & an \(r^p\)-type energy flux of \(\psi{}\) through \(\Sigma{}(\tau{})\)\\
\end{tabular}
\end{center}
We prove the following \(\partial_\tau{}\)-energy estimate by multiplying the wave equation
\zcref{intro-P-def} by \(\partial_\tau{}\psi{}\) and integrating by parts (with
respect to the volume form \(\tanh x\dd{}x\dd{}\omega{}\dd{}\theta{}\)):
\begin{equation}\label{intro-T-estimate}
E_T[\psi{}](\tau_2) + \int_{\tau_1}^{\tau_2} \int _{\Sigma{}(\tau{})}\frac{\tanh \tau{}}{\cosh ^2\tau{}}\abs{\Grad _{S^2}\psi{}}^2\dd{}\mu{}\dd{}\tau{}\lesssim E_T[\psi{}](\tau_1).
\end{equation}
Here \(0\le \tau_1\le \tau_2\), the hypersurface \(\Sigma{}(\tau{})\) is a level set of the time
function \(\tau{}\), and the \(\partial_\tau{}\)-energy flux through \(\Sigma{}(\tau{})\) is
\begin{equation}\label{intro-dt-energy}
E_T[\psi{}](\tau{}) \coloneqq{}\int _{\Sigma{}(\tau{})}\Bigl[(\partial{}_\tau{}\psi{})^2 + (\partial{}_x\psi{})^2 + \frac{1}{\cosh ^2x}\psi^2  + \frac{1}{\cosh ^2\tau{}}\abs{\Grad _{S^2}\psi{}}^2 + \frac{\cosh ^2x}{\tanh ^2x}(\partial{}_\theta{}\psi{})^2\Bigr]\tanh x\dd{}\mu{}.
\end{equation}
The vector field \(\partial_\tau{}\) is (uniformly) timelike, and so its associated flux
through spacelike hypersurfaces is coercive. We can also establish a
\(\partial_\tau{}\)-estimate associated not to the constant-\(\tau{}\) foliation
\(\Sigma{}(\tau{})\), but to a hyperboloidal foliation \(\mathcal{H}(s)\), whose
leaves reach null infinity and are indexed so that \(s\sim u\), for the null
coordinate \(u\) introduced in \zcref{intro-double-null}. The estimate in this
setting reads
\begin{equation}\label{intro-T-estimate-hyp}
\mathcal{E}_T[\psi{}](s_2) + \int_{s_1}^{s_2} \int _{\mathcal{H}(s)}\frac{\tanh \tau{}}{\cosh ^2\tau{}}\abs{\Grad _{S^2}\psi{}}^2\lesssim \mathcal{E}_T[\psi{}](s_1),
\end{equation}
where the \(\partial_\tau{}\)-energy flux through the hyperboloidal hypersurface
\(\mathcal{H}(s)\) is
\begin{equation}
\mathcal{E}_T[\psi{}](s)\coloneqq{} \int _{\mathcal{H}(s)} \Bigl[(L\psi{})^2 + \langle{}x\rangle^{-2}(\underline{L}\psi{})^2 + \frac{1}{\cosh ^2x}\psi^2 + \frac{1}{\cosh ^2\tau{}}\abs{\Grad _{S^2}\psi{}}^2 + \frac{\cosh ^2x}{\tanh ^2x}(\partial{}_\theta{}\psi{})^2\Bigr]\tanh x\dd{}\mu{}.
\end{equation}
Here \(L\) and \(\underline{L}\) are as in \zcref{intro-L-def}.

We formulate the above estimates in terms of \(\psi{}\), rather than the solution
\(\varphi{}\) to the wave equation itself.\footnote{In more geometric language, these estimates can be derived using the
twisted current formalism of \cite{Holzegel_2014}.} This is because on the Witten
bubble spacetime, a \(\partial_\tau{}\)-energy estimate formulated in terms of
\(\varphi{}\) instead of \(\psi{}\) \emph{cannot hold}. Indeed, for \(\varphi{}\)
arising from compactly supported data on \(\mathcal{H}(0)\) (so that all
relevant energy norms defined there are finite), such an estimate would imply
\begin{equation}\label{intro-Lphi-estimate}
\int _{\mathcal{H}(s)}r^2(L\varphi{})^2\tanh x\dd{}x\dd{}\omega{}\dd{}\theta{} < \infty
\end{equation}
for \(s\ge 0\), in analogy with the control
\begin{equation}\label{intro-Lpsi-estimate}
\int _{\mathcal{H}(s)}(L\psi{})^2\tanh x\dd{}x\dd{}\omega{}\dd{}\theta{} < \infty
\end{equation}
we get from \zcref{intro-T-estimate-hyp}. Note that the volume form in
\zcref{intro-Lphi-estimate} contains \(r\), while in \zcref{intro-Lpsi-estimate}, the
factors of \(r\) are included in \(\psi{}\). Since \(r = R\cosh x\cosh \tau{}\),
we have \(Lr\sim r\) when \(\tau{}\) and \(x\) are large. Since \(r\varphi{}\sim
(Lr)\varphi{} = L(r\varphi{}) - r(L\varphi{})\),
\zcref{intro-Lpsi-estimate,intro-Lphi-estimate} together would imply
\begin{equation}
\int _{\mathcal{H}(s)}\psi{}^2\tanh x\dd{}x\dd{}\omega{}\dd{}\theta{} < \infty.
\end{equation}
This cannot be, since \(\psi{}\) may attain a non-zero limit along \(\mathcal{H}(s)\)
as \(x\to \infty\). That is, the solution can have a non-trivial radiation field
along null infinity.
\subsubsection{An \texorpdfstring{\(r^p\)}{rᵖ}-type energy estimate and energy decay for the \texorpdfstring{\(\U(1)\)}{U(1)}-symmetric part \texorpdfstring{\(\psi_0\)}{ψ0}}
\label{intro-rp-estimate-sec} In addition to the \(\partial_\tau{}\)-energy, we use
a version of the \(r^p\)-weighted energy estimates of Dafermos--Rodnianski
\cite{rp-method}. Using a multiplier that for large \(x\) takes the form
\((\sinh x)L\) (where the outgoing null vector field \(L\) is as in
\zcref{intro-L-def}), we show that
\begin{equation}\label{intro-rp-estimate}
\mathcal{E}[\psi{}_0](s_2)  + \int_{s_1}^{s_2} \mathcal{E}[\psi{}_0](s)\dd{}s\lesssim \mathcal{E}[\psi{}_0](s_1),
\end{equation}
where
\begin{equation}
\mathcal{E}[\psi_0](s)\coloneqq{}\mathcal{E}_T[\psi_0](s) + \mathcal{E}_p[\psi_0](s),
\end{equation}
with
\begin{equation}
\mathcal{E}_p[\psi{}_0](s)\coloneqq{}\int _{\mathcal{H}(s)} \sinh x (L\Psi_0)^2 + \langle{}x\rangle^{-2}\frac{1}{\cosh x}\psi_0^2 + \langle{}x\rangle^{-2}\frac{\cosh x}{\cosh ^2\tau{}}\abs{\Grad _{S^2}\Psi_0}^2\dd{}\mu{}.
\end{equation}
Here \(\Psi{}_0\coloneqq{}(\tanh x)^{1/2}\psi{}_0\). The estimate \zcref{intro-rp-estimate} readily
implies exponential decay in \(s\) for \(\mathcal{E}[\psi{}_0](s)\). Since
\(s\sim u\) and \(u \sim \log \abs{U}\) for \(U\) the Minkowskian null
coordinate introduced in \zcref{intro-double-null}, this corresponds to polynomial
energy decay in \(\abs{U}\).
\begin{remark}[Relation to estimates in two space dimensions]
As in our work \cite{gautam-2d-waves} on the linear wave equation in two space
dimensions, we prove the estimate \zcref{intro-rp-estimate} using a \emph{global}
multiplier all the way to the bubble itself, where the spacetime geometry is not
Minkowskian. In particular, we do not use as input an integrated local energy
decay statement. The two space dimensions here correspond to the \(\R^2\) factor
of the Witten bubble spacetime (see \zcref{intro-geometry}).
\end{remark}
\begin{remark}[The requirement of \(\U(1)\)-symmetry]
The estimate \zcref{intro-rp-estimate} holds only for the \(\U(1)\)-symmetric
part \(\psi_0\) of the scalar field. This is because the \(r^p\)-type estimate
captures improved decay for outgoing null derivatives. As discussed in
\zcref{rp-axisymmetric-required}, we can think of \(\psi_{\ge 1}\) as satisfying a
massive wave equation. Since solutions to such equations do not enjoy improved
decay for outgoing null derivatives, we cannot expect an \(r^p\)-type estimate
to hold for \(\psi_{\ge 1}\). More directly, \(\psi_{\ge 1}\) may be periodic by
\zcref{time-periodic-intro} of \zcref{main-theorem-linear-intro}, and so the integrated
energy on the left-hand side of \zcref{intro-rp-estimate} cannot be uniformly
bounded.
\end{remark}
\begin{remark}[Comparison to the \(r^p\)-weighted estimates of Dafermos--Rodnianski]
From \zcref{intro-double-null}, the multiplier \((\sinh x)L\) behaves like
\(\rho{}r\partial_V\), where \(\partial_V\) is the Minkowskian outgoing null
vector field. The \(r^p\)-weighted estimates of Dafermos--Rodnianski \cite{rp-method} use a
multiplier \(r^p\partial{}_V\). Since \(\rho{}\sim \abs{U}^{1/2}V^{1/2}\) (where
\((U,V)\) are as in \zcref{intro-double-null}) and \(V\sim r\), we can think of
\zcref{intro-rp-estimate} as a version of the \(r^p\)-weighted estimates on
Minkowski with \(p = 3/2\), at least in a region where \(\abs{U}\gtrsim 1\).
Note that this region is well separated from the bubble at late times.
\end{remark}
\subsubsection{Pointwise estimates}
\label{intro-pointwise} Using Sobolev embedding and, near the bubble, elliptic
estimates for the spatial part of the wave operator, we obtain the following
pointwise estimates:
\begin{equation}\label{intro-pointwise-estimates}
\abs{\psi{}_0}^2\lesssim E, \qquad \abs{\psi{}_{\ge 1}}^2\lesssim (\cosh x)^{-1}E,\qquad \abs{\partial{}_x\psi{}_0}^2\lesssim E, \qquad \abs{\partial{}_x\psi{}_{\ge 1}}^2\lesssim (\cosh x)E.
\end{equation}
Here \(E\) stands for a higher-order energy of \(\psi{}\). Note that \(\psi_{\ge 1}\)
has improved decay in \(x\) relative to \(\psi_0\), which is only bounded. In
particular, \(\psi_{\ge 1}\) vanishes at null infinity (see part
\zcref{radiation-field-intro} of \zcref{main-theorem-linear-intro}). This is because the
energy for \(\psi_{\ge 1}\) (see \zcref{intro-dt-energy}) contains a term
\((\partial_\theta{}\psi_{\ge 1})^2\) with a weight \((\tanh x)^{-2}(\cosh
x)^2\), and \(\partial{}_\theta{}\psi{}_{\ge 1}\) controls \(\psi_{\ge 1}\) in
\(L^2\) by a Poincaré inequality on \(S^1\). It is this same term that causes
\(\partial_x\psi_{\ge 1}\) to decay slower in \(x\) than \(\psi_0\). The
\(\partial_\theta^2\)-term in the equation for \(\psi_{\ge 1}\) can be
interpreted as a ``mass term.'' Near infinity, the vector field \(\partial_x\)
corresponds to the scaling vector field, which does not commute well with
massive wave equations (such as the Klein--Gordon equation).

To obtain the pointwise estimates in \zcref{intro-pointwise-estimates}, we cannot
use only the energies associated to \(\psi{}\) introduced in
\zcref{intro-dt-estimate,intro-rp-estimate-sec}. Instead, we must use energies
associated to higher-order quantities \(\mathfrak{D}^{\mathbf{k}}\psi{}\), where
\(\mathbf{k}\) is a multi-index and \(\mathfrak{D}\) is a commutator vector
field. We use the commutator vector fields \(\Omega{}\) (the rotations
generating \(\SO(3)\)-symmetry), \(\partial_\tau{}\), and \(\partial_\theta{}\)
(the generator of \(\textnormal{U}(1)\)-symmetry), as well as two Cartesian
vector fields spanning the \(\R^2\)-factor of \(\mathcal{M}\) (see
\zcref{intro-geometry}) needed to control \(\partial_x\psi_{\ge 1}\) near the bubble.
\subsubsection{Small-data global existence for semilinear equations satisfying the null condition}
\label{intro-semilinear-sec} In this section, we sketch the proof of small-data
global existence for the equation
\begin{equation}
\Box{}_g\varphi{} = F(\Grad \varphi{}),
\end{equation}
where the nonlinearity \(F(\Grad \varphi{})\) is quadratic in the derivatives of \(\varphi{}\) and
satisfies an appropriate version of the null condition
\cite{klainerman-null-condition}. To prove global existence, we use a continuity
argument, or bootstrap argument. That is, we suppose that the solution to the
nonlinear equation exists and satisfies a set of estimates, which we call
``bootstrap assumptions,'' on a certain time interval \([0,\tau{}_f]\) for \(\tau_f >
0\). The bootstrap assumptions take the rough form
\begin{align}
E_T[\mathfrak{D}^{\mathbf{k}}\psi{}](\tau{})\le  \epsilon{}, \label{intro-boot-1} \\
E_p[(\mathfrak{D}^{\mathbf{k}}\psi{})_0](\tau{})\le  \epsilon{}  \label{intro-boot-2}
\end{align}
for \(\tau{}\in [0,\tau{}_f]\) and \(\abs{\mathbf{k}}\le k_{\textnormal{max}}\) (where we can take
\(k_{\textnormal{max}}\ge 8\)). Here \(\mathfrak{D}\) is a commutator vector
field (see \zcref{intro-pointwise}) and \(E_p\) is an \(r^p\)-type energy associated
to the constant-\(\tau{}\) foliation \(\Sigma{}(\tau{})\):
\begin{equation}
E_p[\psi_0](s)\coloneqq{}\int _{\Sigma{}(\tau{})} \sinh x (L\Psi_0)^2 + \frac{1}{\cosh x}\psi_0^2 + \frac{\cosh x}{\cosh ^2\tau{}}\abs{\Grad _{S^2}\Psi_0}^2\dd{}\mu{}.
\end{equation}
Here \(\Psi{}_0 = (\tanh x)^{1/2}\psi{}_0\). We then use the linear estimates as \emph{a priori estimates} for the nonlinear
equation, treating the nonlinearity as an inhomogeneity. The goal of the
continuity argument is to use the smallness present in the problem (for example,
that of the initial data) to improve the bootstrap assumptions, namely show that
estimates \zcref{intro-boot-1,intro-boot-2} hold with \(\frac{1}{2}\epsilon{}\) on
the right-hand side in place of \(\epsilon{}\). A local existence theory then
implies that the solution actually exists and satisfies the bootstrap
assumptions on a uniformly larger time interval \([0,\tau{}_f+\delta{}]\).

We now sketch the key ideas we use to improve the bootstrap assumptions. An
example of a possible term in \(F(\Grad \varphi{})\) is
\(f\partial_U\varphi{}\partial_V\varphi{}\), where \((U,V)\) are the Minkowskian
double null coordinates (see \zcref{intro-double-null}) and \(f\) is a bounded
function. Since \(\rho^2\sim \abs{U}V\), the equation for \(\psi{}\) becomes,
using \zcref{LU-relation},
\begin{equation}
P\psi{} = rf\abs{U}\partial{}_U\varphi{}V\partial{}_V\varphi{}\sim rf\underline{L}\varphi{}L\varphi{},
\end{equation}
Here \(\sim\) means equality up to bounded factors. Since we only estimate terms of
the form \(L(r\varphi{})\), and not \(rL\varphi{}\) (see \zcref{intro-dt-estimate}),
we must rewrite the right-hand side in terms of \(\psi{} = r\varphi{}\). The equation then becomes
\begin{equation}\label{intro-Ppsi-expr}
P\psi{} \sim  r^{-1}f(\underline{L}\psi{}L\psi{} + \psi^2 + \psi{}\underline{L}\psi{} + \psi{}L\psi{}).
\end{equation}
In particular, the nonlinearity includes zeroth-order terms.

To improve the bootstrap assumptions, we use our energy estimates for the linear
wave equation. The \(\partial_\tau{}\)-estimate applied to
\(\mathfrak{D}^{\mathbf{k}}\psi{}\) reads
\begin{equation}\label{intro-boot-dt}
E_T[\mathfrak{D}^{\mathbf{k}}\psi{}](\tau{}) \lesssim \textnormal{data} + \int_0^{\tau_f} \int _{\Sigma{}(\tau{})} \abs{\partial{}_\tau{}\mathfrak{D}^{\mathbf{k}}\psi{}}\abs{P\mathfrak{D}^{\mathbf{k}}\psi{}} \dd{}\mu{}\dd{}\tau{}.
\end{equation}
Ignoring terms coming from \([P,\mathfrak{D}^{\mathbf{k}}]\), we have
\(P\mathfrak{D}^{\mathbf{k}}\psi{}\sim \mathfrak{D}^{\mathbf{k}}P\psi{}\). The
worst terms in \(\mathfrak{D}^{\mathbf{k}}P\psi{}\) are ones where all the
derivatives fall on one of the terms in the quadratic nonlinearity. Examples of
such terms are
\begin{equation}\label{intro-ex-terms}
r^{-1}f\underline{L}\mathfrak{D}^{\mathbf{k}}\psi{}L\psi{}\quad \textnormal{or}\quad r^{-1}f\mathfrak{D}^{\mathbf{k}}\psi{}\underline{L}\psi{}.
\end{equation}
Since pointwise estimates are not available for the top-order terms, we must put
the term with fewer derivatives in \(L^2\) and put the top-order term in
\(L^2\).

The bootstrap assumptions \zcref{intro-boot-1,intro-boot-2} (and a Poincaré
inequality on \(S^1\)) imply in particular that
\begin{equation}\label{intro-boot-energy}
\begin{split}
&\int _{\Sigma{}(\tau{})}\Bigl[(\partial{}_\tau{}\mathfrak{D}^{\mathbf{k}}\psi{})^2 + (\partial{}_x\mathfrak{D}^{\mathbf{k}}\psi{})^2 + \frac{\cosh ^2x}{\tanh ^2x}(\mathfrak{D}^{\mathbf{k}}\psi{})_{\ge 1}^2\Bigr]\tanh x + \frac{1}{\cosh x}(\mathfrak{D}^{\mathbf{k}}\psi)_0^2\dd{}\mu{}\lesssim \epsilon{}.
\end{split}
\end{equation}
Combining the bootstrap assumptions with the pointwise estimates of
\zcref{intro-pointwise}, we have
\begin{equation}\label{intro-pointwise-estimates-boot}
\abs{\partial{}_\tau^{\le 1}\psi{}_0}^2\lesssim \epsilon{}, \qquad \abs{\partial{}_\tau^{\le 1}\psi{}_{\ge 1}}^2\lesssim \epsilon{}(\cosh x)^{-1},\qquad \abs{\partial{}_x\psi{}_0}^2\lesssim \epsilon{}, \qquad \abs{\partial{}_x\psi{}_{\ge 1}}^2\lesssim \epsilon{}\cosh x.
\end{equation}
The term with the weakest \(x\)-weight in \(L^2\) is \(\psi_0\), and \(\partial_\tau{}\) and
\(\partial_x\) have the same \(x\)-weight in \(L^2\). On the other hand, the
term with the weakest \(x\)-decay in \(L^\infty\) is \(\partial_x\psi_{\ge 1}\).
It follows that the most dangerous term is one where we must put
\(\partial_x\psi_{\ge 1}\) in \(L^\infty\) and \(\psi_0\) in \(L^2\). Expanding
\(L = \partial_\tau{} + \partial_x\) and \(\underline{L} = \partial_\tau{} -
\partial_x\), we see that the first term in \zcref{intro-ex-terms} cannot produce
this most dangerous term, while the second term can. From now on, we restrict
our attention to the term \(r^{-1}f\mathfrak{D}^\mathbf{k}\psi{}\partial{}_x\psi{}\) appearing in
\(P\mathfrak{D}^{\mathbf{k}}\psi{}\).

We now explain how to improve \zcref{intro-boot-1}. Recalling that \(r\sim
e^\tau{}\cosh x\), the terms in \zcref{intro-ex-terms} produce terms of the
following form arising from the right-hand side of \zcref{intro-boot-dt}:
\begin{equation}\label{intro-terms}
\textnormal{(I)}\coloneqq{} \int_0^{\tau_f} e^{-\tau{}}\int _{\Sigma{}(\tau{})}\frac{1}{\cosh x} \abs{\partial{}_\tau{}\mathfrak{D}^{\mathbf{k}}\psi{}}\abs{\mathfrak{D}^{\mathbf{k}}\psi{}}\abs{\partial{}_x\psi{}}\dd{}\mu{}\dd{}\tau{}.
\end{equation}
Using the pointwise estimate for \(\partial_x\psi{}\), we get
\begin{equation}
\begin{split}
\textnormal{(I)}&\lesssim \epsilon^{1/2}\int_0^{\tau_f} e^{-\tau{}}\int _{\Sigma{}(\tau{})}\frac{1}{(\cosh x)^{1/2}} \abs{\partial{}_\tau{}\mathfrak{D}^{\mathbf{k}}\psi{}}\abs{\mathfrak{D}^{\mathbf{k}}\psi{}}\dd{}\mu{}\dd{}\tau{}\\
&\lesssim \epsilon^{1/2}\int_0^{\tau_f} e^{-\tau{}}\Bigl(\int _{\Sigma{}(\tau{})} \abs{\partial{}_\tau{}\mathfrak{D}^{\mathbf{k}}\psi{}}^2\dd{}\mu{}\Bigr)^{1/2}\Bigl(\int _{\Sigma{}(\tau{})}\frac{1}{\cosh x}\abs{\mathfrak{D}^{\mathbf{k}}\psi{}}^2\Bigr)^{1/2}\dd{}\mu{}\dd{}\tau{}
\end{split}
\end{equation}
Splitting \(\mathfrak{D}^{\mathbf{k}}\psi{}\) into \((\mathfrak{D}^{\mathbf{k}}\psi)_0\) and
\((\mathfrak{D}^{\mathbf{k}}\psi)_{\ge 1}\), we have (from \zcref{intro-boot-energy})
\begin{equation}
\begin{split}
\int _{\Sigma{}(\tau{})}\frac{1}{\cosh x}\abs{\mathfrak{D}^{\mathbf{k}}\psi{}}^2\dd{}\mu{}&\lesssim \int _{\Sigma{}(\tau{})}\frac{1}{\cosh x}\abs{(\mathfrak{D}^{\mathbf{k}}\psi{})_0}^2 + \frac{1}{\cosh x}\abs{(\mathfrak{D}^{\mathbf{k}}\psi{})_{\ge 1}}^2\dd{}\mu{}\lesssim \epsilon{}.
\end{split}
\end{equation}
It follows that
\begin{equation}
\textnormal{(I)}\lesssim \epsilon^{1/2}\int_0^{\tau{}_f} e^{-\tau{}}\epsilon{}\dd{}\tau{}\lesssim \epsilon^{3/2}.
\end{equation}
In this way we show that
\begin{equation}\label{intro-T-boot-improvement}
E_T[\mathfrak{D}^{\mathbf{k}}\psi{}](\tau{})\lesssim \textnormal{data} + \epsilon^{3/2}.
\end{equation}
Taking \(\epsilon{}\) to be sufficiently small and taking the initial data to be
sufficiently small relative to \(\epsilon{}\), we can arrange for the right-hand
side of \zcref{intro-T-boot-improvement} to be less than \(\frac{1}{2}\epsilon{}\),
improving the bootstrap assumption \zcref{intro-boot-1}.

We now explain how to improve \zcref{intro-boot-2}, focusing again on the term
\(r^{-1}f\mathfrak{D}^{\mathbf{k}}\psi{}\partial_x\psi{}\) appearing in \(P\mathfrak{D}^{\mathbf{k}}\psi{}\). The \(r^p\)-type estimate applied to
\((\mathfrak{D}^{\mathbf{k}}\psi{})_0\) reads
\begin{equation}\label{intro-boot-rp}
E_p[(\mathfrak{D}^{\mathbf{k}}\psi{})_0](\tau{}) \lesssim \textnormal{data} + \int_0^{\tau_f} \int _{\Sigma{}(\tau{})} \sinh x\abs{(P\mathfrak{D}^{\mathbf{k}}\psi{})_0}^2\dd{}\mu{}\dd{}\tau{}.
\end{equation}
We therefore need to estimate
\begin{equation}
\textnormal{(II)}\coloneqq{}\int_0^{\tau_f} \int _{\Sigma{}(\tau{})} r^{-2}\sinh x\abs{(f\mathfrak{D}^{\mathbf{k}}\psi{}\partial{}_x\psi{})_0}^2\dd{}\mu{}\dd{}\tau{}
\end{equation}
Since \(r\sim e^\tau{}\cosh x\), we have the preliminary estimate
\begin{equation}
\textnormal{(II)}\lesssim \int_0^{\tau_f} e^{-2\tau{}}\int _{\Sigma{}(\tau{})} \frac{1}{\cosh x}\abs{(f\mathfrak{D}^{\mathbf{k}}\psi{}\partial{}_x\psi{})_0}^2\tanh x\dd{}\mu{}\dd{}\tau{}.
\end{equation}
To estimate \((f\mathfrak{D}^{\mathbf{k}}\psi{}\partial_x\psi{})_0\), we use the decomposition
\begin{equation}
(fgh)_0 = f_0g_0h_0 + f_0(g_{\ge 1}h_{\ge 1})_0 + (f_{\ge 1}g_{\ge 1}h_0)_0 + \boxed{(f_{\ge 1}g_0h_{\ge 1})_0} + (f_{\ge 1}(g_{\ge 1}h_{\ge 1})_{\ge 1})_0,
\end{equation}
valid for functions \(f,g,h:S^1\to \C\), applied with \(g = \mathfrak{D}^{\mathbf{k}}\psi{}\) and
\(h = \partial_x\psi{}\). The interaction between
\((\mathfrak{D}^{\mathbf{k}}\psi{})_0\) and \(\partial_x\psi_{\ge 1}\) appears
only in the boxed term. Ignoring the non-boxed terms, we therefore have
\begin{equation}
\textnormal{(II)}\lesssim \int_0^{\tau_f} e^{-2\tau{}}\int _{\Sigma{}(\tau{})} \frac{1}{\cosh x}\norm{f_{\ge 1}}_{L^\infty(\Sigma{}(\tau{}))}^2\abs{(\mathfrak{D}^{\mathbf{k}}\psi{})_0}^2\abs{\partial{}_x\psi{}_{\ge 1}}^2\tanh x\dd{}\mu{}\dd{}\tau{}.
\end{equation}
Using the pointwise estimate for \(\partial_x\psi_{\ge 1}\), we get
\begin{equation}
\textnormal{(II)}\lesssim \epsilon{}\int_0^{\tau_f} e^{-2\tau{}}\int _{\Sigma{}(\tau{})} \abs{f_{\ge 1}}^2\abs{(\mathfrak{D}^{\mathbf{k}}\psi{})_0}^2\tanh x\dd{}\mu{}\dd{}\tau{}.
\end{equation}
In order to control this term, we demand that \(\abs{f_{\ge 1}}\lesssim (\cosh
x)^{-1/2}\) (see \zcref{intro-boot-energy}). We then obtain
\begin{equation}
\textnormal{(II)}\lesssim \epsilon{}\int_0^{\tau_f} e^{-2\tau{}}\epsilon{}\dd{}\tau{}\lesssim \epsilon^2.
\end{equation}
We conclude from \zcref{intro-boot-rp} that
\begin{equation}
E_p[(\mathfrak{D}^{\mathbf{k}}\psi{})_0](\tau{}) \lesssim \textnormal{data} + \epsilon^2,
\end{equation}
which improves the bootstrap assumption \zcref{intro-boot-2} and completes the proof
of global existence, which is part \zcref{intro:semilinear-outside-symmetry} of
\zcref{main-theorem-semilinear-intro}.

To prove item \zcref{intro:semilinear-symmetric} of
\zcref{main-theorem-semilinear-intro}, which concerns improved results within
symmetry, we use the above strategy together with the observation that, for
\(\U(1)\)-symmetric solutions, we can use the \(r^p\)-type energy for the full
solution. We then prove decay as in \zcref{intro-rp-estimate-sec}.
\subsection*{AI usage statement}
The diagrams in \zcref{fig:bubble-penrose,fig:bubble-R2} were created with the assistance of ChatGPT 5.6.
AI systems were not used to generate any of the mathematical content or writing
in this manuscript.
\subsection*{Acknowledgements}
\addcontentsline{toc}{subsection}{Acknowledgements}
We thank Igor Rodnianski for suggesting the problem and for helpful discussions
and comments on the manuscript.
\section{The geometry of the Witten bubble spacetime}
\label{WB-geometry}
\subsection{Smooth structure, coordinate systems, and metric}
The Witten bubble spacetime manifold has smooth structure
\(\mathcal{M} \coloneqq{} \R\times \R^2\times S^2\). As we will see, the Witten
bubble metric makes the \(\R^2\) factor not asymptotically flat, but
asymptotically cylindrical.
\subsubsection{Schwarzschild coordinates and Cartesian coordinates}
\label{sec:coordinates} The Witten bubble metric was introduced in
\cite{WITTEN1982481} in what we will call ``Schwarzschild coordinates,'' which
cover the manifold \(\mathcal{M}^{\circ } = \R_\tau{}\times
(R,\infty)_\rho{}\times S^2_{\vartheta{},\phi{}}\times S^1_\theta{}\) (where the
\((R,\infty)_\rho{}\times S^1_\theta{}\) factor covers the complement of the
origin in the \(\R^2\) factor in \(\mathcal{M}\)), as
\begin{equation}\label{WB-schw-metric}
g = -\rho^2\dd{}\tau^2 +(1-(R/\rho{})^2)^{-1}\dd{}\rho^2 + \rho^2\cosh^2 \tau{}g_{S^2} + R^2(1-(R/\rho{})^2)\dd{}\theta^2.
\end{equation}
Here \(g_{S^2} = \dd{}\vartheta^2 + \sin^2 \vartheta{}\dd{}\phi^2\) is the round metric on the unit
sphere, \(R > 0\) describes the radius of the circle near infinity and radius of
the minimal-area sphere in the spacetime, and we think of \(S^1\) as the unit
circle \(S^1 = \R/2\pi{}\Z\).

It is convenient to introduce the rescaled coordinate \(x\) defined by \(\rho{}
= R\cosh x\) (as in \cite{PhysRevD.42.1996}), so that \(x\in (0,\infty)\) and the
dependence of the metric on \(R\) is isolated:
\begin{equation}
g = R^2\bigl[\cosh ^2x(-\dd{}\tau^2 + \dd{}x^2) + \cosh ^2x\cosh ^2\tau{}g_{S^2} + \tanh^2 x\dd{}\theta^2\bigr].
\end{equation}
The manifold \((\mathcal{M}^{\circ },g)\) is geodesically incomplete as \(x\to 0\) (or
\(\rho{}\to R\)). To extend the metric smoothly to \(\set{x=0}\), we interpret
\((x,\theta{})\) as ``polar coordinates'' on \(\R^2\). This is possible because
the metric on the \(\R^2\) factor looks like \(R^2[\dd{}x^2 +
x^2\dd{}\theta^2]\) to leading order as \(x\to 0\). We can therefore equip the
\(\R^2\) factor of \(\mathcal{M} = \R\times \R^2\times S^2\) with Cartesian
coordinates \((\hat{y},\hat{z})\), defined by
\begin{equation}\label{cart-def}
\hat{y} = \mathfrak{f}(x)\cos \theta{},\qquad \hat{z}=\mathfrak{f}(x)\sin \theta{},\qquad \mathfrak{f}(x)\coloneqq{}e^{\cosh x}\tanh (x/2).
\end{equation}
Then \(\mathfrak{f}'(x)/\mathfrak{f}(x) = \cosh (x)/\tanh (x)\), and so the Witten bubble metric on
\(\R_\tau{}\times (\R^2_{\hat{y},\hat{z}}\setminus \set{\hat{y}=\hat{z}=0})\times
S^2_{\vartheta{},\phi{}}\subset \mathcal{M}\) becomes
\begin{equation}\label{cart-metric}
g = R^2\bigl[-\cosh ^2x \dd{}\tau^2 + \mathfrak{F}(x)^2(\dd{}\hat{y}^2 + \dd{}\hat{z}^2) + \cosh ^2x \cosh ^2\tau{}g_{S^2}\bigr],\qquad \mathfrak{F}(x) \coloneqq{} \frac{\tanh x}{\mathfrak{f}(x)}.
\end{equation}
\begin{definition}[The bubble]
We refer to the set \(\mathcal{B}\coloneqq{}\R_\tau{}\times \set{\hat{y}=\hat{z}=0}\times
S^2_{\vartheta{},\phi{}}\subset \mathcal{M}\) as \emph{the bubble}. It is a timelike
codimension-two hypersurface in \(\mathcal{M}\). From \zcref{cart-metric}, we see
that the bubble is, up to the constant scale factor \(R^2\), isometric to the de
Sitter space \(\textnormal{dS}_{2+1}\) of \((2 + 1)\) dimensions (which has metric
\(-\dd{}\tau^2 + \cosh^2 \tau{}g_{S^2}\)). Since it is the fixed point set of
the isometric \(\textnormal{U}(1)\)-action, the bubble is a totally geodesic
submanifold; in particular, it is a minimal submanifold.
\label{bubble}
\end{definition}
Since the function \(x^2\) is smooth on \(\R^2_{\hat{y},\hat{z}}\) (where we
extend \(x\) to the origin by \(x|_{\hat{y}=0,\hat{z}=0} = 0\)), the metric in
\zcref{cart-metric} extends smoothly to \(\mathcal{B}\).
\begin{remark}[Conventions for the name of a coordinate system]
We will write \(\omega{} = (\vartheta{},\phi{})\) for the standard spherical coordinates on \(S^2\),
where \(\vartheta{}\in (0,\pi{})\) and \(\phi{}\in (0,2\pi{})\).

When discussing coordinates on \(\mathcal{M}^{\circ }\), we will always use the
decomposition \(\mathcal{M}^{\circ } = (\R^2\setminus \set{0})\times (S^2\times S^1)\) and use the standard
\((\omega{},\theta{})\) coordinates on the \(S^2\times S^1\) factor. We may therefore refer to a
coordinate system on \(\mathcal{M}^{\circ }\) by the corresponding coordinate system
on the \(\R\times \R_ +\) factor. For example, we will refer to the Schwarzschild
coordinates as \((\tau{},\rho{})\) coordinates, rather than the more verbose
\((\tau{},\rho{},\omega{},\theta{})\) coordinates.
\end{remark}
The Cartesian \((\tau{},\hat{y},\hat{z},\omega{})\) coordinates form a global chart for
\(\mathcal{M}\), modulo the usual degeneration of the spherical coordinates. It
will be useful to record for future reference that
\begin{equation}\label{dx-in-terms-of-dxhat}
\partial{}_x = \mathfrak{f}'(x)\cos \theta{}\partial{}_{\hat{y}} + \mathfrak{f}'(x)\sin \theta{}\partial{}_{\hat{z}},\qquad \partial{}_\theta{} = -\mathfrak{f}(x)\sin \theta{}\partial{}_{\hat{y}} + \mathfrak{f}(x)\cos \theta{}\partial{}_{\hat{z}} = -\hat{z}\partial{}_{\hat{y}} + \hat{y}\partial{}_{\hat{z}}
\end{equation}
and
\begin{equation}\label{hat-derivatives}
\partial{}_{\hat{y}} = \frac{1}{\mathfrak{f}'(x)}\cos \theta{}\partial{}_x - \frac{1}{\mathfrak{f}(x)}\sin \theta{}\partial{}_\theta{},\qquad \partial_{\hat{z}} = \frac{1}{\mathfrak{f}'(x)}\sin \theta{}\partial{}_x + \frac{1}{\mathfrak{f}(x)}\cos \theta{}\partial{}_\theta{}.
\end{equation}
\subsubsection{Time orientation}
We endow \(\mathcal{M}\) with a time orientation by demanding that the
global timelike vector field \(\partial_\tau{}\) is future directed.
\subsubsection{The area-radius function}
We introduce the area-radius function
\begin{equation}\label{r-def}
r\coloneqq{}\rho{}\cosh \tau{}=R\cosh x\cosh \tau{},
\end{equation}
which describes the radius of the round spheres in
\((\mathcal{M},g)\).
\subsubsection{Double null coordinates}
\label{double-null}
We introduce the double null coordinates
\begin{equation}
u \coloneqq{} \frac{1}{2}(\tau{}-x),\qquad v \coloneqq{} \frac{1}{2}(\tau{}+x),
\end{equation}
in which the metric takes the form
\begin{equation}
g = R^2\Bigl[-4\cosh ^2x\dd{}u\dd{}v + \cosh^2x\cosh^2\tau{}g_{S^2} + \tanh^2 xg_{S^1}\Bigr].
\end{equation}
The \((u,v)\) coordinates are related to the Minkowskian double null coordinates
\((U,V)\) by
\begin{equation}
U \coloneqq{} -\frac{R}{2}e^{-2u},\qquad V\coloneqq{} \frac{R}{2}e^{2v}.
\end{equation}
In \((U,V)\) coordinates, the metric takes the form
\begin{equation}
g = -4(1+e^{-2x})^2\dd{}U\dd{}V + r^2g_{S^2} + R^2\tanh ^2xg_{S^1}.
\end{equation}
In standard double null coordinates \(U_{\textnormal{m}} = \frac{1}{2}(t-r)\)
and \(V_{\textnormal{m}} = \frac{1}{2}(t + r)\) on Minkowski space, the
Kaluza--Klein metric with circle of radius \(R\) takes the form
\begin{equation}
g_{\textnormal{mink}} = -4\dd{}U_{\textnormal{m}}\dd{}V_{\textnormal{m}} + r^2g_{S^2} + R^2g_{S^1},
\end{equation}
and so the Witten bubble metric in \((U,V)\) coordinates approaches the
Kaluza--Klein metric as \(x\to \infty\), in particular near null infinity
(\(v\to \infty\) as \(u = \textnormal{const}\)) and spacelike infinity (\(x\to
\infty\) as \(\tau{}=\textnormal{const}\)).
\subsection{Killing vector fields and comparison to Minkowski space}
\label{killing}
There are 7 Killing vector fields associated to the \(\textnormal{SO}(3,1)\times
\textnormal{U}(1)\) symmetry of the Witten bubble spacetime, namely the
rotations
\begin{equation}\label{rotations}
\begin{split}
\Omega{}_1 = \partial{}_\phi{}, \qquad \Omega{}_2 = \cos \phi{}\partial{}_\vartheta{} - \cot \vartheta{} \sin \phi{} \partial{}_\phi{},\qquad \Omega{}_3 = -\sin \phi{}\partial{}_\vartheta{} - \cot \vartheta{} \cos \phi{} \partial{}_\phi{},
\end{split}
\end{equation}
the boosts
\begin{equation}\label{boosts}
\begin{split}
K_{1} &= \cos \vartheta{}\partial{}_\tau{} - \tanh \tau{} \sin \vartheta{}\partial{}_\vartheta{}, \\
K_{2} &= \sin \vartheta{}\sin \phi{}\partial{}_\tau{} + \tanh \tau{}\cos \vartheta{}\sin \phi{}\partial{}_\vartheta{} + \tanh \tau{} \cos \phi{} \frac{1}{\sin \vartheta{}}\partial{}_\phi{}, \\
K_{3} &= \sin \vartheta{}\cos \phi{} \partial{}_\tau{} + \tanh \tau{}\cos \vartheta{}\cos \phi{} \partial{}_\vartheta{} - \tanh \tau{} \sin \phi{} \frac{1}{\sin \vartheta{}}\partial{}_\phi{},
\end{split}
\end{equation}
and \(\partial{}_\theta{}\), the generator of translations in the \(S^1\) direction. These
vector fields satisfy the commutation relations
\begin{equation}
[\Omega{}_i,\Omega{}_j]=\epsilon{}_{ij}{}^k\Omega{}_k,\qquad [K_i,\Omega{}_j] = \epsilon{}_{ij}{}^kK_k,\qquad [K_i,K_j] = -\epsilon{}_{ij}{}^k\Omega{}_k,\qquad [\Omega{}_i,\partial{}_\theta{}] = [K_i,\partial{}_\theta{}] = 0,
\end{equation}
where \(i,j,k\in \set{1,2,3}\) and \(\epsilon\) is the Levi--Civita symbol.

Formally setting \(R = 0\) in the Witten bubble metric ``pops the bubble'' and
yields the Minkowski metric in \((3 + 1)\) dimensions. The Schwarzschild
coordinates are related to the usual coordinates on Minkowski space as follows:
\begin{equation}\label{mink-coords}
x^0 = \rho\sinh \tau{},  \quad
x^1 = \rho\cosh \tau{}\sin \vartheta{}\cos \phi{}, \quad
x^2 = \rho\cosh \tau{}\sin \vartheta{}\sin \phi{}, \quad
x^3 = \rho\cosh \tau{}\cos \vartheta{},
\end{equation}
where \(\rho{}\cosh \tau{} = r\) is the area radius function. These functions define a
Minkowskian Cartesian coordinate system \((x^0,x^1,x^2,x^3,\theta{})\) on
\(\mathcal{M}^{\circ }\), with respect to which the boosts and rotations take
their familiar form
\begin{equation}
K_{i} \coloneqq{} x^0\partial{}_{x^i} + x^i\partial{}_{x^0}\,\, (1\le i\le 3),\quad \Omega{}_1 = x^1\partial{}_{x^2} - x^2\partial{}_{x^1},\quad \Omega{}_2 = x^3\partial{}_{x^1} - x^1\partial{}_{x^3}, \quad \Omega{}_3 = x^2\partial{}_{x^3} - x^3\partial{}_{x^2}.
\end{equation}

With respect to Minkowskian \((t,r,\omega{},\theta{})\) coordinates, we have \(\partial_\tau{} = r\partial_t +
t\partial_r\), and the scaling vector field \(S = t\partial_t + r\partial_r\)
satisfies \(S = \rho{}\partial_\rho{} = (\tanh x)^{-1}\partial_x\).
\section{Preliminaries}
For \(x\in \R\), we write \(\langle{}x\rangle{}\coloneqq{}\sqrt{1 + x^2}\). We write \(A\lesssim B\) (or
\(B\gtrsim A\)) when there is a constant \(C\) depending only on the function
\(h\) used to construct the foliation \(\mathcal{H}(s)\) defined in
\zcref{hyperboloidal-foliation} such that \(A\le CB\). We write \(A\sim B\) if
\(A\lesssim B\) and \(A\gtrsim B\).
\subsection{Well-posedness and coordinate expressions for the wave equation}
We write \(\Box{}\) for \(\Box_g\), the wave operator associated to the Witten bubble
metric \(g\). It is shown in \cite[Prop.~II.1]{Bachelot_2016} that the Witten
bubble spacetime is globally hyperbolic. It follows that the wave equation on
\((\mathcal{M},g)\) is well-posed. More precisely, it follows from
\cite[Thm~12.19]{MR2527641} that, given a spacelike Cauchy hypersurface
\(\mathcal{S}\), the corresponding future-directed unit normal vector field
\(\nu{}\), and smooth functions \((\varphi{}_0,\varphi{}_1)\) on
\(\mathcal{S}\), there exists a unique solution to the following initial value
problem:
\begin{equation}
\Box{}\varphi{} = 0,\qquad
\varphi{}|_{\mathcal{S}} = \varphi{}_0, \qquad
\nu{}\varphi{}|_{\mathcal{S}} = \varphi{}_1.
\end{equation}
We now express the wave operator \(\Box{}\) in various coordinate systems.
\begin{lemma}[Coordinate expressions for the wave equation]
Let \(\varphi{}\in C^\infty(\mathcal{M})\). In Schwarzschild \((\tau{},\rho{})\) coordinates, we have
\begin{equation}\label{schwarzschild-box}
\begin{split}
\Box{}\varphi{} &= -\frac{1}{\rho^2}\partial{}_\tau^2\varphi{} - \frac{2\tanh \tau{}}{\rho^2}\partial{}_\tau{}\varphi{} + (1-(R/\rho{})^2)\partial{}_\rho^2\varphi{} + \frac{1}{\rho{}}(3 - (R/\rho{})^2)\partial{}_\rho{}\varphi{} \\
&\qquad + \frac{1}{\rho^2\cosh^2 \tau{}}\Lapl _{S^2}\varphi{} + \frac{1}{R^2(1-(R/\rho{})^2)}\partial{}_\theta^2\varphi{}.
\end{split}
\end{equation}
In \((\tau{},x)\) coordinates, we have
\begin{equation}\label{P-x-expression}
\begin{split}
P\psi{} &= -\partial{}_\tau^2\psi{} + \partial{}_x^2\psi{} + \frac{1}{\cosh x\sinh x}\partial{}_x\psi{} - \frac{1}{\cosh^2x}\psi{} + \frac{1}{\cosh ^2\tau{}}\Lapl _{S^2}\psi{} +  \frac{\cosh^2x}{\tanh ^2x}\partial{}_\theta^2\psi{},
\end{split}
\end{equation}
where \(\psi{}\coloneqq{}r\varphi{}\) and
\begin{equation}\label{P-def}
P\psi{}\coloneqq{}\rho^2r\Box{}(r^{-1}\psi{}).
\end{equation}
In the Cartesian \((\tau{},\hat{y},\hat{z},\omega{})\) coordinates, we have
\begin{equation}\label{P-cart-expression}
P\psi{} = -\partial{}_\tau^2\psi{} + \frac{1}{\cosh ^2\tau{}}\Lapl _{S^2}\psi{} - \frac{1}{\cosh ^2x}\psi{} + \frac{\cosh ^2x}{\mathfrak{F}(x)^2}(\partial{}_{\hat{y}}^2\psi{} + \partial{}_{\hat{z}}^2\psi{}) + 3\cosh x(\hat{y}\partial{}_{\hat{y}}\psi{} + \hat{z}\partial{}_{\hat{z}}\psi{}),
\end{equation}
where \(\mathfrak{F}(x)\) is defined in \zcref{cart-metric}. In terms of the
quantity \(\Psi{}\coloneqq{}(\tanh x)^{1/2}\psi{}\), we have
\begin{equation}\label{P-Psi-equation}
\begin{split}
(\tanh x)^{1/2}P\psi{} &= -\partial{}_\tau^2\Psi{} + \partial{}_x^2\Psi{} + \frac{1}{4\cosh^2 x\sinh^2 x}\Psi{} + \frac{1}{\cosh ^2\tau{}}\Lapl _{S^2}\Psi{} +  \frac{\cosh^2x}{\tanh ^2x}\partial{}_\theta^2\Psi{}.
\end{split}
\end{equation}
\label{wave-expressions}
\end{lemma}
\begin{remark}[Twisting the wave equation by \(r\) and by \((\tanh x)^{1/2}\)]
To prove energy estimates in \zcref{sec:energy-estimates}, we consider the equation
\zcref{P-x-expression} for \(\psi{}\), rather than the equation for \(\varphi{}\)
itself. See \zcref{twisted-necessary} for further discussion. Using the operator \(P\) and the \((\tau{},x)\) coordinates (as opposed to the
\((\tau{},\rho{})\) coordinates) also reveals the familiar \(-\partial_\tau^2 +
\partial_x^2\) term in the wave equation, whereas this structure is unclear from
\zcref{schwarzschild-box}. Since \(\cosh x\sinh x\sim x\) as \(x\to 0\), the wave
equation \zcref{P-x-expression} for \(\U(1)\)-symmetric scalar fields resembles the wave
equation on \((2 + 1)\)-dimensional Minkowski space (see \zcref{2d-waves-relation})
as \(x\to 0\), up to the angular \(\Lapl _{S^2}\) term.

We consider the quantity \(\Psi{} = (\tanh x)^{1/2}\psi{}\) because the equation it
satisfies (see \zcref{P-Psi-equation}) contains no first-order term, which is
helpful to prove the \(r^p\)-type estimate of \zcref{rp-type-estimate}.
\end{remark}
\begin{proof}
\step{Step 1: Proof of \zcref{schwarzschild-box}.} The expression \zcref{schwarzschild-box}
follows from the expression \zcref{WB-schw-metric} for the metric in \((\tau{},\rho{})\)
coordinates and the expression
\begin{equation}
\Box{}\varphi{}=\abs{\Det g}^{-1/2}\partial{}_\alpha{}(\abs{\Det g}^{1/2}g^{\alpha{}\beta{}}\partial{}_\beta{}\varphi{}).
\end{equation}

\step{Step 2: Proof of \zcref{P-x-expression}.} To prove \zcref{P-x-expression}, we commute
\zcref{schwarzschild-box} with \(r = \rho{}\cosh \tau{}\). We claim that
\begin{equation}\label{rboxphi}
\begin{split}
r\Box{}\varphi{} &= -\frac{1}{\rho^2}\partial{}_\tau^2\psi{} + (1-(R/\rho{})^2)\partial{}_\rho^2\psi{} + \frac{1}{\rho{}}(1 + (R/\rho{})^2)\partial{}_\rho{}\psi{} - \frac{1}{\rho^2}(R/\rho{})^2r\varphi{} \\
&\qquad + \frac{1}{r^2}\Lapl _{S^2}\psi{} + \frac{1}{R^2(1-(R/\rho{})^2)}\partial{}_\theta^2\psi{}.
\end{split}
\end{equation}
To obtain \zcref{P-x-expression} from \zcref{rboxphi}, multiply by \(\rho^2\) and use
\(\partial_\rho{} = R^{-1}(\sinh x)^{-1}\partial_x\). Since \(\partial{}_\tau{}r
= r\tanh \tau{}\), we have \(\partial_\tau^2r = r\), and so
\begin{equation}
\partial{}_\tau^2\psi{} - r\partial{}_\tau^2\varphi{} = [\partial{}_\tau^2,r]\varphi{} = \partial{}_\tau^2r\varphi{} + 2\partial{}_\tau{}r\partial{}_\tau{}\varphi{} = r\varphi{} + 2r\tanh \tau{}\partial{}_\tau{}\varphi{} = 2\tanh \tau{}\partial{}_\tau{}\psi{} + (1 - 2\tanh ^2\tau{})r\varphi{}
\end{equation}
It follows that
\begin{equation}
-\frac{1}{\rho^2}r\partial{}_\tau^2\varphi{} - \frac{2\tanh \tau{}}{\rho^2}r\partial{}_\tau{}\varphi{} = -\frac{1}{\rho^2}\partial{}_\tau^2\psi{} + \frac{1}{\rho^2}r\varphi{}
\end{equation}
Since \(\partial_\rho{}r = \rho^{-1}r = \cosh \tau{}\) and \(\partial_\rho^2r = 0\), we have
\begin{equation}
[\partial{}_\rho^2,r]\varphi{} = 2\rho{}^{-1}r\partial{}_\rho{}\varphi{} = 2\rho{}^{-1}\partial{}_\rho{}\psi{} - 2\rho{}^{-2}r\varphi{},
\end{equation}
and
\begin{equation}
(1-(R/\rho{})^2)r\partial{}_\rho^2\varphi{} = (1-(R/\rho{})^2)\partial{}_\rho^2\psi{} - \frac{2}{\rho{}}(1-(R/\rho{})^2)\partial{}_\rho{}(\rho{}\varphi{}) + \frac{2}{\rho^2}(1-(R/\rho{})^2)r\varphi{}
\end{equation}
and
\begin{equation}
\frac{1}{\rho{}}(2+(1-(R/\rho{})^2))r\partial{}_\rho{}\varphi{} = \frac{1}{\rho{}}(2+(1-(R/\rho{})^2))\partial{}_\rho{}\psi{} - \frac{1}{\rho^2}(2+(1-(R/\rho{})^2))r\varphi{}.
\end{equation}
Thus
\begin{equation}
(1-(R/\rho{})^2)r\partial{}_\rho^2\varphi{} + \frac{1}{\rho{}}(2+(1-(R/\rho{})^2))r\partial{}_\rho{}\varphi{} = (1-(R/\rho{})^2)\partial{}_\rho^2\psi{} + \frac{1}{\rho{}}(2-(1-(R/\rho{})^2))\partial{}_\rho{}\psi{} - \frac{1}{\rho^2}(2-(1-(R/\rho{})^2))r\varphi{}.
\end{equation}
We conclude \zcref{rboxphi}.

\step{Step 3: Proof of \zcref{P-cart-expression}.} Combine \zcref{P-x-expression} with
\zcref{dx-in-terms-of-dxhat}.

\step{Step 4: Proof of \zcref{P-Psi-equation}.} Commute \zcref{P-x-expression} with \((\tanh
x)^{1/2}\). Useful computations include
\begin{equation}
\begin{gathered}
\partial{}_x(\tanh x)^{1/2} = \frac{1}{2}(\tanh x)^{-1/2}(\cosh x)^{-2},  \\
\partial{}_x^2(\tanh x)^{1/2} = -\frac{1}{4}(\tanh x)^{-3/2}(\cosh x)^{-4} - (\tanh x)^{1/2}(\cosh x)^{-2}.
\end{gathered}
\end{equation}
\end{proof}
\subsection{Foliations and spacetime regions}
\subsubsection{The constant-\texorpdfstring{\(\tau\)}{τ} foliation}
We will write \(\Sigma{}(\tau_0)\coloneqq{}\set{\tau{}=\tau_0}\) for the constant-\(\tau{}\) foliation of
\(\mathcal{M}\), and \(\Sigma{}(\tau_0,v_0)\) for the cutoff surface \(\Sigma{}(\tau_0)\cap \set{v\le v_0}\).
\subsubsection{Outgoing cones}
Write \(C(u_0)\coloneqq{}\set{u=u_0}\) for the outgoing cone with \(u\)-value \(u_0\).
\subsubsection{The hyperboloidal foliation in the interior of a light cone}
\label{hyperboloidal-foliation} Let \(h : (0,\infty)\to (0,2)\) be a smooth
integrable function that extends smoothly to \([0,\infty)\) and satisfies the
following conditions:
\begin{enumerate}
\item \label{conditions-on-h-at-0} (constant-\(\tau{}\) near the origin) \(h(x) = 1\) for \(x\le 1\),
\item \label{h-hyperboloidal} (asymptotically null) we have \(\abs{h^{(k)}(x)}\lesssim _k \langle{}x\rangle{}^{-2}\) for \(k\ge 0\),
\end{enumerate}
Define a new coordinate
\begin{equation}\label{tau-def}
s \coloneqq{} u - \frac{1}{2}\int_x^\infty h(y)\dd{}y.
\end{equation}
Then \((s,x)\) are smooth coordinates on \(\mathcal{M}^{\circ }\). The condition
\zcref{conditions-on-h-at-0}, which enforces that \(s(\tau{},x) = \tau{} + C\) when
\(x\le 1\) for some constant \(C = C(h)\), ensures in particular that \(s\) and
\(h\circ x\) are smooth functions on \(\mathcal{M}\). Condition
\zcref{conditions-on-h-at-0} will also be technically convenient in the proof of
\zcref{dxpsi0-near}.

We introduce the following notation for level sets of \(s\):
\begin{equation}
\mathcal{H}(s_0) \coloneqq{} \set{s = s_0}, \qquad \mathcal{H}(s,\tau{}_0) = \mathcal{H}(s)\cap \set{\tau{}\le \tau{}_0},\qquad \mathcal{R}(s_0,s_1)\coloneqq{}\set{s_0\le s \le s_1}.
\end{equation}
By \zcref{h-hyperboloidal}, the hypersurfaces \(\mathcal{H}(s)\)
are asymptotically null.
\subsection{Vector fields}
\subsubsection{Coordinate derivatives}
We will always write \(\partial_\tau{} = \partial_\tau{}|_{(\tau{},x,\omega{},\theta{})}\) and \(\partial_x = \partial_x|_{(\tau{},x,\omega{},\theta{})}\)
for the coordinate derivatives in \((\tau{},x)\) coordinates. Define the vector
fields
\begin{equation}
T = \partial{}_\tau{}\qquad L = \partial{}_\tau{} + \partial{}_x,\qquad \underline{L} = \partial{}_\tau{} - \partial{}_x, \qquad X = \partial{}_x|_{(s,x,\omega{},\theta{})},
\end{equation}
where \(s\) is the hyperboloidal coordinate defined in
\zcref{hyperboloidal-foliation}. Then
\begin{equation}
\partial{}_s|_{(s,x,\omega{},\theta{})} = 2T,\qquad \underline{L} = (2-h(x))\partial{}_\tau{} - X,\qquad L = h(x)\partial{}_\tau{} + X\qquad \partial{}_x|_{(\tau{},x,\omega{},\theta{})} = X - (1-h(x))\partial{}_\tau{}.
\end{equation}
\subsubsection{Commutator vector fields}
\label{commutators} We use the commutator vector fields
\begin{equation}
\mathfrak{D}_{i} \coloneqq{} \Omega_i\quad (i\in \set{1,2,3}),\qquad \mathfrak{D}_{4} = \partial{}_\theta{},\qquad \mathfrak{D}_5 \coloneqq{}\partial{}_\tau{},\qquad \mathfrak{D}_6 = \partial{}_{\hat{y}},\qquad \mathfrak{D}_7 = \partial{}_{\hat{z}},
\end{equation}
where the \(\Omega_i\) are rotations on \(S^2\) (see \zcref{killing}). Recall from \zcref{killing} that
\(\mathfrak{D}_i\) is a Killing vector field for \(1\le i\le 4\).

We will write \(\Omega{}\) for a generic element of \(\set{\Omega_1,\Omega_2,\Omega_3}\). We write
\(\mathfrak{D}\) for a generic element of
\(\set{\mathfrak{D}_1,\ldots,\mathfrak{D}_7}\). We order tuples \(\mathbf{k} =
(k_1,\ldots,k_7)\) of order \(\abs{\mathbf{k}}\coloneqq{}k_1 + \cdots{} + k_7\)
lexicographically, and we write \(\mathfrak{D}^{\mathbf{k}} \coloneqq{}
\mathfrak{D}_1^{k_1}\cdots{}\mathfrak{D}_7^{k_7}\). We write \(\mathfrak{D}^{\le
k}\) for a generic expression \(\mathfrak{D}^{\mathbf{k}}\) with
\(\abs{\mathbf{k}}\le k\) (including \(\mathbf{k} = 0\)). We will also write
\(\mathring{\mathfrak{D}}\) for a generic element of the restricted set
\(\set{\mathfrak{D}_i : 1\le i\le 5}\), and use the notation
\(\mathring{\mathfrak{D}}^{\mathbf{k}}\) and \(\mathring{\mathfrak{D}}^{\le k}\)
analogously. The reason to introduce the notation \(\mathring{\mathfrak{D}}\) is
that \(\mathring{\mathfrak{D}}\)-vector fields preserve \(\U(1)\)-symmetry, in the
sense that \((\mathring{\mathfrak{D}}\psi{})_0 = \mathring{\mathfrak{D}}\psi_0\).
\begin{remark}[The purpose of the lexicographic ordering]
As shown in \zcref{frak-D-commutation}, the commutator \([P,\mathfrak{D}^{\mathbf{k}}]\)
(where \(P\) is the operator defined in \zcref{P-def}) consists of terms that are
lower-order with respect to the lexicographic ordering, but are not necessarily
lower-order with respect to the total number of commutator vector fields. That
is, \([P,\mathfrak{D}^{\mathbf{k}}]\) contains only terms involving
\(\mathfrak{D}^{\mathbf{k}'}\) for \(\mathbf{k}' < \mathbf{k}\), but we may have \(\abs{\mathbf{k}'} = \abs{\mathbf{k}}\).
\end{remark}
\begin{definition}[\(O_{\mathfrak{D}}\)-notation]
If \(f\) and \(g\) are two functions, we schematically write \(f =
O_{\mathfrak{D}}(g)\) to mean that for each multi-index \(\mathbf{k}\) we have
\(\abs{\mathfrak{D}^{\mathbf{k}}f}\lesssim _k g\).
\end{definition}
\begin{lemma}[Examples of quantities that are \(O_{\mathfrak{D}}(1)\)]
Suppose \(f,g : \R\to \R\) are such that \(f\) is smooth and \(\abs{f^{(k)}(\lambda{})}\lesssim_k g(\lambda{})\) for all
\(k\ge 0\). Then
\begin{enumerate}
\item \label{OD-1} \(f(\tau{})\) is \(O_{\mathfrak{D}}(g(\tau{}))\),
\item \label{OD-2} \(f(x^2)\) is \(O_{\mathfrak{D}}(g(x^2))\).
\end{enumerate}
\label{OD-check}
\end{lemma}
\begin{proof}
First, \zcref{OD-1} is immediate from the assumptions on \(f\) and the fact that
\(\mathfrak{D}\tau{}=0\) for \(\mathfrak{D}\neq{}\partial{}_\tau{}\). For
\zcref{OD-2}, note that \(\mathfrak{D}x = 0\) unless \(\mathfrak{D}\in
\set{\partial_{\hat{y}},\partial_{\hat{z}}}\). From \zcref{hat-derivatives}, we see
that
\(\abs{\partial{}_{\hat{y}}^{k_1}\partial{}_{\hat{z}}^{k_2}\partial{}_\theta^Nx}\lesssim
1\) in the region \(\set{x\ge 1}\) when \(k_1 + k_2\ge 1\) and \(N\ge 0\). In
the compact region \(\set{x\le 1}\), the function \(f(x^2)\) is smooth in
\((\hat{y},\hat{z})\). The result now follows from the chain rule.
\end{proof}
\begin{lemma}[Rearrangement lemma]
For multi-indices \(\mathbf{k}_1\) and \(\mathbf{k}_2\), there are constants \(C_{\mathbf{k}}\) such that
\begin{equation}
\mathfrak{D}^{\mathbf{k}_1}\mathfrak{D}^{\mathbf{k}_2} = \sum_{\mathbf{k}\le \mathbf{k}_1 + \mathbf{k}_2} C_{\mathbf{k}}\mathfrak{D}^{\mathbf{k}},
\end{equation}
where addition of multi-indices is defined component-by-component: the sum of
\(\mathbf{k} = (k_1,\ldots,k_7)\) and \(\mathbf{k}' = (k_1',\ldots,k_7')\) is
\((k_1 + k_1',\ldots,k_7 + k_7')\).
\label{rearrangement}
\end{lemma}
\begin{proof}
The rotations form an algebra (that is, their linear span is closed under
commutation), as do the vector fields
\(\set{\partial{}_\theta{},\partial{}_{\hat{y}},\partial{}_{\hat{z}}}\) and
\(\set{\partial{}_\tau{}}\), and these algebras each commute with each other.
The result follows.
\end{proof}
\begin{proposition}[Commutation formula for the vector fields \(\mathfrak{D}\)]
Let \(P\) be the operator defined in \zcref{P-def}, and let \(\mathfrak{D}\) be a
commutator vector field. We have
\begin{equation}\label{P-D-comm}
[P,\mathfrak{D}^{\mathbf{k}}]\psi{} = \textnormal{(I)} + \textnormal{(II)},
\end{equation}
where
\begin{equation}\label{P-D-comm-I}
\textnormal{(I)} = \frac{1}{\cosh^2\tau{}} \sum_{\mathbf{k}' < \mathbf{k}}f_{\mathbf{k}',\mathbf{k}}(\tau{})\Omega{}\mathfrak{D}^{\mathbf{k}'}
\end{equation}
for some bounded functions \(f_{\mathbf{k}',\mathbf{k}}(\tau{})\) and
\begin{equation}\label{P-D-comm-II}
\abs{\textnormal{(II)}}\lesssim \sum_{\mathbf{k}'<\mathbf{k}} \frac{1}{e^{\cosh x}}\Bigl(\abs{\mathfrak{D}^{\mathbf{k}'}P\psi{}} + \abs{\partial{}_\tau{}\mathfrak{D}^{\mathbf{k}'}\psi{}} + \abs{\partial{}_x\mathfrak{D}^{\mathbf{k}'}\psi{}} + \frac{1}{\cosh ^2\tau{}}\abs{\Omega{}\mathfrak{D}^{\mathbf{k}'}\psi{}} + \frac{\cosh x}{\tanh x}\abs{\partial{}_\theta{}\mathfrak{D}^{\mathbf{k}'}\psi{}} + \frac{1}{\cosh^2 x}\abs{\mathfrak{D}^{\mathbf{k}'}\psi{}}\Bigr).
\end{equation}
The statement also holds with \(\mathfrak{D}\) replaced with
\(\mathring{\mathfrak{D}}\), and in that case term \(\textnormal{(II)}\) is not
present in \zcref{P-D-comm}.
\label{frak-D-commutation}
\end{proposition}
\begin{proof}
The terms in \(\textnormal{(I)}\) arise from commuting with \(\partial_\tau{}\), and the
terms in \(\textnormal{(II)}\) arise from commuting with \(\partial_{\hat{y}}\) and
\(\partial_{\hat{z}}\). Since the vector fields \(\partial_{\hat{y}}\) and
\(\partial_{\hat{z}}\) are not present in the commutation set
\(\mathring{\mathfrak{D}}\), the final part of the statement of
\zcref{frak-D-commutation} follows from the proof below.

\step{Step 1: Commutation formula for the Killing vector fields \(\Omega{}\) and \(\partial_\theta{}\).} We
have \([P,\Omega_i] = [P,\partial_\theta{}] = 0\) from the definition of \(P\)
in \zcref{P-def}. It follows that \([P,\Omega{}_1^{k_1}\Omega{}_2^{k_2}\Omega{}_3^{k_3}\partial_\theta^{k_4}] = 0\).

\step{Step 2: Commutation formula for \(\partial_\tau{}\).} We compute
\begin{equation}
[P,\partial{}_\tau{}] = \frac{2\tanh \tau{}}{\cosh^2 \tau{}}\Lapl _{S^2}.
\end{equation}
It follows by induction that there are bounded functions \(f_{n,N}(\tau{})\) such
that
\begin{equation}\label{P-dt-comm}
[P,\partial_\tau^N] = \frac{1}{\cosh ^2\tau{}}\sum_{n=0}^{N-1}f_{n,N}(\tau{})\Lapl _{S^2}\partial_\tau^n.
\end{equation}

\step{Step 3: Commutation formula the Cartesian vector fields \(\partial_{\hat{y}}\) and
\(\partial_{\hat{z}}\).} We use \zcref{P-cart-expression} to rewrite
\(\partial{}_{\hat{y}}^2 + \partial{}_{\hat{z}}^2\) in terms of \(P\) and
compute
\begin{equation}\label{P-dy-comm-prep}
\begin{split}
[P,\partial{}_{\hat{y}}] &= -(\partial{}_{\hat{y}}x)\frac{2\tanh x}{\cosh^2 x} - (\partial{}_{\hat{y}}x)\Bigl(\frac{\cosh ^2x}{\mathfrak{F}(x)^2}\Bigr)'(\partial{}_{\hat{y}}^2 + \partial{}_{\hat{z}}^2) - 3(\partial{}_{\hat{y}}x)\sinh x(\hat{y}\partial{}_{\hat{y}} + \hat{z}\partial{}_{\hat{z}}) - 3\cosh x\partial{}_{\hat{y}} \\
&= -(\partial{}_{\hat{y}}x)\frac{2\tanh x}{\cosh^2 x} - (\partial{}_{\hat{y}}x)\Bigl(\log \frac{\cosh ^2x}{\mathfrak{F}(x)^2}\Bigr)' \Bigl(P + \partial{}_{\tau{}}^2 - \frac{1}{\cosh ^2\tau{}}\Lapl _{S^2} + \frac{1}{\cosh^2 x} - 3\cosh x(\hat{y}\partial{}_{\hat{y}} + \hat{z}\partial{}_{\hat{z}})\Bigr) \\
&\qquad - 3(\partial{}_{\hat{y}}x)\sinh x(\hat{y}\partial{}_{\hat{y}} + \hat{z}\partial{}_{\hat{z}}) - 3\cosh x\partial{}_{\hat{y}}.
\end{split}
\end{equation}
Here \(\mathfrak{F}(x)\) is as in \zcref{cart-metric}. We compute
\begin{equation}\label{OD-LHS}
(\partial{}_{\hat{y}}x)\Bigl(\log \frac{\cosh ^2x}{\mathfrak{F}(x)^2}\Bigr)' = \frac{1}{\mathfrak{f}(x)\mathfrak{f}'(x)}(\sinh x - \tanh (x/2) + 2\tanh x)\hat{y},
\end{equation}
where \(\mathfrak{f}(x)\) is as in \zcref{cart-def}. The factor multiplying
\(\hat{y}\) is \(O_{\mathfrak{D}}(\cosh x/\mathfrak{f}'(x))\) (by \zcref{OD-check}),
and \(\hat{y} = O_{\mathfrak{D}}(\abs{\hat{y}})\). It follows from the identity
\(\hat{y} = \mathfrak{f}(x)\cos \theta{}\) that the expression in \zcref{OD-LHS} is
\(O_{\mathfrak{D}}(\cosh x/\mathfrak{f}'(x))\). Substituting this computation
into \zcref{P-dy-comm-prep} and recalling \(\mathfrak{f}(x)/\mathfrak{f}'(x) = \tanh
x/\cosh x\), we obtain
\begin{equation}
\begin{split}
[P,\partial{}_{\hat{y}}] &= O_{\mathfrak{D}}\Bigl(\frac{\cosh x}{\mathfrak{f}'(x)}\Bigr)\Bigl(P + \partial{}_{\tau{}}^2 - \frac{1}{\cosh ^2\tau{}}\Lapl _{S^2} + \frac{1}{\cosh^2 x}\Bigr) + O_{\mathfrak{D}}(\cosh x)\widehat{\partial},
\end{split}
\end{equation}
where \(\widehat{\partial}\in \set{\partial_{\hat{y}},\partial_{\hat{z}}}\) and the expression on the
right should be interpreted after the \(O_{\mathfrak{D}}(\cosh
x/\mathfrak{f}'(x))\) term is distributed over the parentheses. An induction
argument now gives
\begin{equation}\label{P-dy-comm}
[P,\partial{}_{\hat{y}}^N] = \sum_{n=0}^{N-1}O_{\mathfrak{D}}\Bigl(\frac{\cosh x}{\mathfrak{f}'(x)}\Bigr)\Bigl(\partial{}_{\hat{y}}^nP + \partial{}_\tau{}\partial{}_\tau{}\partial{}_{\hat{y}}^n - \frac{1}{\cosh ^2\tau{}}\Lapl _{S^2}\partial{}_{\hat{y}}^n + \frac{1}{\cosh^2 x}\partial{}_{\hat{y}}^n\Bigr) + O_{\mathfrak{D}}(\cosh x)\widehat{\partial}\partial{}_{\hat{y}}^n.
\end{equation}
An analogous formula of course holds for \(\partial_{\hat{z}}\) in place of
\(\partial_{\hat{y}}\). It follows that
\begin{equation}\label{P-Dhat-comm}
\begin{split}
[P,\widehat{\partial}^{(N_1,N_2)}] &= \sum_{n_1=0}^{N_1-1}\sum_{n_2=0}^{N_2-1}\Bigl[O_{\mathfrak{D}}\Bigl(\frac{\cosh x}{\mathfrak{f}'(x)}\Bigr)\Bigl(\widehat{\partial}^{(n_1,n_2)}P + \partial{}_\tau{}\partial{}_\tau{}\widehat{\partial}^{(n_1,n_2)}P - \frac{1}{\cosh ^2\tau{}}\Lapl _{S^2} + \frac{1}{\cosh^2 x}\widehat{\partial}^{(n_1,n_2)}P\Bigr) \\
&\qquad + O_{\mathfrak{D}}(\cosh x)\widehat{\partial}\widehat{\partial}^{(n_1,n_2)}\Bigr],
\end{split}
\end{equation}
where we have written \(\widehat{\partial}^{(n_1,n_2)} =
\partial_{\hat{y}}^{n_1}\partial_{\hat{z}}^{n_2}\) and used the relation
\begin{equation}
\partial_{\hat{y}}^{n_1}P\partial_{\hat{z}}^{n_2} = \partial_{\hat{y}}^{n_1}\partial_{\hat{z}}^{n_2}P + \partial_{\hat{y}}^{n_1}[P,\partial{}_{\hat{z}}^{n_2}].
\end{equation}

\step{Step 4: Putting it all together; proof of \zcref{P-D-comm}.} We have
\begin{equation}
[P,\mathfrak{D}^{\mathbf{k}}] = [P,\Omega^{(k_1,k_2,k_3)}\partial{}_\theta^{k_4}\partial{}_\tau^{k_5}\widehat{\partial}^{(k_6,k_7)}],
\end{equation}
where we have written \(\Omega^{(k_1,k_2,k_3)}\coloneqq{}\Omega_1^{k_1}\Omega_2^{k_2}\Omega_3^{k_3}\) and
\(\widehat{\partial}^{(k_6,k_7)} =
\partial_{\hat{y}}^{k_6}\partial_{\hat{z}}^{k_7}\). By Step 1, we have
\begin{equation}\label{P-D-comm-prep}
[P,\mathfrak{D}^{\mathbf{k}}] =\Omega^{(k_1,k_2,k_3)}\partial{}_\theta^{k_4}[P,\partial{}_\tau^{k_5}]\widehat{\partial}^{(k_6,k_7)} +  \Omega^{(k_1,k_2,k_3)}\partial{}_\theta^{k_4}\partial{}_\tau^{k_5}[P,\widehat{\partial}^{(k_6,k_7)}].
\end{equation}
Since \(\Omega_i\) are lower-order commutators than \(\partial_\tau{}\) and both \(\Omega_i\) and
\(\partial_\tau{}\) are lower-order commutators compared to
\(\partial_{\hat{y}}\) and \(\partial_{\hat{z}}\), the result follows
from \zcref{P-dt-comm,P-Dhat-comm,P-D-comm-prep} together with the formulas \zcref{hat-derivatives} expressing
\(\widehat{\partial}\) in terms of \(\partial_x\) and \(\partial_\theta{}\)
(and the expression for \(\mathfrak{f}(x)\) in \zcref{cart-def}). In
particular, the terms in \(\textnormal{(I)}\) and \(\textnormal{(II)}\) arise
from the first and second terms, respectively, on the right-hand side of
\zcref{P-D-comm-prep}. Note that we have used the identity \(\Lapl _{S^2} = \sum_{i=1}^3 \Omega{}_i^2\)
to write \zcref{P-D-comm-I}.
\end{proof}
\subsection{Energy norms}
\label{sec:energies} Write \(\dd{}\mu{}\coloneqq{}\dd{}x\dd{}\omega{}\dd{}\theta{}\). We first define the energies we
will use to study problems outside of symmetry. They are associated to the
constant-\(\tau{}\) foliation \(\Sigma{}(\tau{})\). Define the \(\partial_\tau{}\)-energy
\begin{equation}
E_T[\psi{}](\tau{})\coloneqq{}\int _{\Sigma{}(\tau{})} \Bigl[(\partial{}_\tau{}\psi{})^2 + (\partial{}_x\psi{})^2 + \frac{1}{\cosh^2x}\psi^2 + \frac{1}{\cosh ^2\tau{}}\abs{\Grad _{S^2}\psi{}}^2 + \frac{\cosh ^2x}{\tanh ^2x}(\partial{}_\theta{}\psi{})^2\Bigr]\tanh x\dd{}\mu{}.
\end{equation}
We also define an \(r^p\)-type energy that we will use for the \(\U(1)\)-symmetric part
of the scalar field:
\begin{equation}
E_{p}[\psi{}_0](\tau{})\coloneqq{}\int _{\Sigma{}(\tau{})}\sinh x(L\Psi{}_0)^2 + \frac{1}{\cosh x}\psi_0^2 + \frac{\sinh x}{\cosh ^2\tau{}}\abs{\Grad _{S^2}\Psi{}_0}^2\dd{}\mu{}.
\end{equation}
Here we recall the notation \(\Psi{} \coloneqq{} (\tanh x)^{1/2}\psi{}\). We will also need the
\(r^p\)-flux through outgoing cones:
\begin{equation}
F_p[\psi_0](\tau{}_1,\tau{}_2)\coloneqq{}\sup_{u\in \R}\int _{C(u)\cap \set{\tau_1\le \tau{}\le \tau_2}} \sinh x(L\Psi{}_0)^2\dd{}v\dd{}\omega{}\dd{}\theta{}.
\end{equation}

We now define the energies that we will use for \(\U(1)\)-symmetric scalar fields. They
are associated to the hyperboloidal foliation \(\mathcal{H}(s)\). Define the
\(\partial_\tau{}\)-energy
\begin{equation}
\mathcal{E}_T[\psi{}](s,\tau)\coloneqq{}\int _{\mathcal{H}(s,\tau)} \Bigl[(L\psi{})^2 + \langle{}x\rangle^{-2}(\underline{L}\psi{})^2 + \frac{1}{\cosh^2x}\psi^2 + \frac{1}{\cosh ^2\tau{}}\abs{\Grad _{S^2}\psi{}}^2 + \frac{\cosh ^2x}{\tanh ^2x}(\partial{}_\theta{}\psi{})^2\Bigr]\tanh x\dd{}\mu{}
\end{equation}
associated to \(\mathcal{H}(s,\tau{})\) and the associated flux through \(\Sigma{}(\tau{})\):
\begin{equation}
\mathcal{F}_T[\psi{}](s_1,s_2,\tau{})\coloneqq{}\int _{\Sigma{}(\tau{})\cap \set{s_1\le s\le s_2}} \Bigl[(L\psi{})^2 + (\underline{L}\psi{})^2 + \frac{1}{\cosh^2x}\psi^2 + \frac{1}{\cosh ^2\tau{}}\abs{\Grad _{S^2}\psi{}}^2 + \frac{\cosh ^2x}{\tanh ^2x}(\partial{}_\theta{}\psi{})^2\Bigr]\tanh x\dd{}\mu{}.
\end{equation}
We also define the \(r^p\)-type energy
\begin{equation}
\mathcal{E}_p[\psi{}](s,\tau)\coloneqq{}\int _{\mathcal{H}(s,\tau)}\sinh x(L\Psi{})^2 + \langle{}x\rangle^{-2}\frac{1}{\cosh x}\psi^2 + \langle{}x\rangle^{-2} \frac{\sinh x}{\cosh ^2\tau{}}\abs{\Grad _{S^2}\Psi{}}^2\dd{}\mu{}.
\end{equation}
Finally, we define the master energy
\begin{equation}
\mathcal{E}[\psi{}](s,\tau)\coloneqq{}\mathcal{E}_T[\psi{}](s,\tau) + \mathcal{E}_p[\psi{}](s,\tau).
\end{equation}
We will write
\begin{equation}
\mathcal{E}[\psi{}](s)\coloneqq{}\sup_{\tau{}\ge 0} \mathcal{E}[\psi{}](s,\tau)
\end{equation}
and use similar notation for the other energies.
\begin{remark}[Guide to the notation for energies]
A subscript ``\(T\)'' indicates that the quantity is associated to the
\(\partial_\tau{}\)-energy, and a subscript ``\(p\)'' denotes an \(r^p\)-type energy.

We use the standard math font for energies of scalar fields outside of symmetry,
which are defined with respect to the constant-\(\tau{}\) foliation
\(\Sigma{}(\tau{})\). Energies in calligraphic font are used strictly for
\(\U(1)\)-symmetric scalar fields and are defined with respect to the hyperboloidal
foliation \(\mathcal{H}(s)\) defined in \zcref{hyperboloidal-foliation}.
\end{remark}
\begin{remark}[The character of the hypersurface used to truncate the leaves of the foliation]
We truncate the constant-\(\tau{}\) foliation \(\Sigma{}(\tau{})\) by constant-\(v\)
hypersurfaces, which are null, but we truncate the constant-\(s\) hyperboloidal
foliation \(\mathcal{H}(s)\) by constant-\(\tau{}\) hypersurfaces, which are
spacelike. The control of a flux term on the spacelike hypersurface
\(\Sigma{}(\tau{})\), to which \(\partial_x\) is tangent, is useful for the
pointwise estimate for the \(\U(1)\)-symmetric scalar field \(\partial_x\psi_0\)
established in \zcref{dx-far}. Roughly speaking, this is because \(\partial_x^2\)
appears in the wave equation for \(\psi{}\) but not \(\underline{L}^2\) (where
\(\underline{L}\) is tangent to constant-\(v\) hypersurfaces).
\label{cutoff-character}
\end{remark}
\begin{remark}[Control of the zeroth-order term near the bubble]
The \(r^p\)-type energies \(E_p\) and \(\mathcal{E}_p\) include better control
of the zeroth-order term near the bubble (where \(\set{x=0}\)) than the
\(\partial_\tau{}\)-energies \(E_T\) and \(\mathcal{E}_T\). We do not use this
fact in our work. We used an analogous fact in our work \cite{gautam-2d-waves} on
linear waves in two space dimensions. This work inspired the proof of the
\(r^p\)-type estimates established in \zcref{rp-type-estimate} (see
\zcref{2d-waves-relation}).
\end{remark}
\section{Energy estimates}
\label{sec:energy-estimates}
In this section, we use vector field multipliers to prove energy estimates. We
will multiply the wave equation for \(\psi{}\) (see \zcref{wave-expressions}) with a
suitable vector field and directly integrate by parts in coordinates adapted to
the foliation. One could also obtain these estimates in a more geometric way by
using the method of compatible currents.

In \zcref{sec:dt-estimate}, we prove a \(\partial_\tau{}\)-estimate, and in \zcref{sec:rp-estimate},
we prove an analogue of the \(r^p\)-weighted energy estimates of
Dafermos--Rodnianski \cite{rp-method}.
\subsection{The basic \texorpdfstring{\(\partial_\tau\)}{∂τ}-energy estimate}
\label{sec:dt-estimate}
In this section, we use multipliers based on \(\partial_\tau{}\) to derive basic energy
estimates valid outside of symmetry.
\begin{proposition}[The \(\partial_\tau{}\)-estimate]
Let \(\psi{}\in C^\infty(\mathcal{M})\). For \(0\le s_1\le s_2\) and \(\tau{}\ge 0\), we have
\begin{equation}\label{T-estimate}
\begin{split}
&\mathcal{E}_T[\psi{}](s_2,\tau{}) + \sup_{0\le \tau{}'\le \tau{}} \mathcal{F}_T[\psi{}](s_1,s_2,\tau{}') + \int_{s_1}^{s_2} \int _{\mathcal{H}(s,\tau)} \frac{\tanh \tau{}}{\cosh ^2\tau{}}\abs{\Grad _{S^2}\psi{}}^2\tanh x\dd{}\mu{}\dd{}s \\
&\lesssim \mathcal{E}_T[\psi{}](s_1,\tau{}) + \int_{s_1}^{s_2} \int _{\mathcal{H}(s,\tau)}\abs{\partial{}_\tau{}\psi{}}\abs{P\psi{}}\tanh x\dd{}\mu{}\dd{}s
\end{split}
\end{equation}
and
\begin{equation}\label{T-estimate-2}
\begin{split}
&\mathcal{E}_T[\psi{}](s_2,\tau{})  + \sup_{0\le \tau{}'\le \tau{}} \mathcal{F}_T[\psi{}](s_1,s_2,\tau{}') +\int_{s_1}^{s_2} \int _{\mathcal{H}(s,\tau)} \Bigl[\langle{}x\rangle{}^{-2}(\underline{L}\psi{})^2 + \frac{\tanh \tau{}}{\cosh ^2\tau{}}\abs{\Grad _{S^2}\psi{}}^2\Bigr]\tanh x\dd{}\mu{}\dd{}s \\
&\lesssim \mathcal{E}_T[\psi{}](s_1,\tau{}) +\int_{s_1}^{s_2} \int _{\mathcal{H}(s,\tau)} \langle{}x\rangle{}^{-2}(L\psi{})^2 \tanh x\dd{}\mu{}\dd{}s +  \int_{s_1}^{s_2} \int _{\mathcal{H}(s,\tau)}\abs{\partial{}_\tau{}\psi{}}\abs{P\psi{}}\tanh x\dd{}\mu{}\dd{}s.
\end{split}
\end{equation}
For \(0\le \tau_1\le \tau_2\), we have
\begin{equation}\label{T-estimate-3}
\begin{split}
&E_T[\psi{}](\tau{}_2) + \int_{\tau{}_1}^{\tau{}_2} \int _{\Sigma{}(\tau{})} \frac{\tanh \tau{}}{\cosh ^2\tau{}}\abs{\Grad _{S^2}\psi{}}^2\tanh x\dd{}\mu{}\dd{}\tau{}\lesssim E_T[\psi{}](\tau{}_1) + \int_{\tau{}_1}^{\tau{}_2} \int _{\Sigma{}(\tau{})}\abs{\partial{}_\tau{}\psi{}}\abs{P\psi{}}\tanh x\dd{}\mu{}\dd{}\tau{}.
\end{split}
\end{equation}
\label{T-estimate-prop}
\end{proposition}
\begin{remark}[Why we estimate \(\psi\) instead of \(\varphi\)]
On the Witten bubble spacetime, the analogue of \zcref{T-estimate-prop} formulated
for \(\varphi{}\) itself, where \(\varphi{}\) solves \(\Box{}\varphi{} = 0\),
cannot hold, and so we work with \(\psi{} = r\varphi{}\) instead. Accordingly, the
volume form in the energies used in \zcref{T-estimate-prop} (see \zcref{sec:energies})
is not the geometric volume form \(\dd{}\vol\), but \(r^{-2}\dd{}\vol\). This
reflects the two powers of \(r\) present in the energy density itself. In the
language of compatible currents, these estimates arise from currents that are
twisted by \(r\), rather than the usual currents. The framework of twisted
currents was introduced in \cite{Holzegel_2014} (see also
\cite[App.~B]{dafermos2024quasilinearwaveequationsasymptotically}).

On stationary backgrounds, when \(Tr = 0\) for \(T\) the Killing vector field
associated to stationarity and \(r\) (a version of) the area radius function,
the twisted and untwisted formulations of a \(T\)-estimate are largely
equivalent. By contrast, on the Witten bubble spacetime, the
\(\partial_\tau{}\)-estimate cannot hold for both \(\varphi{}\) and \(\psi{}\).
To see this, note first that, since \(r = R\cosh x\cosh \tau{}\), we have
\(Lr\sim r\) for \(L = \partial_\tau{} + \partial_x\), at least when \(\tau{}\)
and \(x\) are large. The energy \(\mathcal{E}_T[\psi{}]\) controls
\(L(r\varphi{})\) in the \(L^2\) norm along a hypersurface reaching null
infinity, whereas an estimate formulated for \(\varphi{}\) would control
\(r(L\varphi{})\). The two estimates together would control \(r\varphi{}\sim
(Lr)\varphi{} = L(r\varphi{}) - r(L\varphi{})\) in the \(L^2\) norm. Since
\(r\varphi{}\) can attain a non-zero limit along a hypersurface reaching null
infinity, such control is not possible.
\label{twisted-necessary}
\end{remark}
\begin{remark}[Control of the angular bulk term]
We will use the control over the angular term in the bulk in
\zcref{T-estimate,T-estimate-2,T-estimate-3} to absorb error terms in a
finite-\(\tau{}\) region in the proof of the \(r^p\)-type energy estimates of
\zcref{sec:rp-estimate}.
\end{remark}
\begin{proof}
Multiply the wave equation \zcref{P-x-expression} by \(-2\tanh x\partial_\tau{}\psi{}\) and use the Leibniz rule to get
\begin{equation}\label{T-estimate-1}
\begin{split}
&-2\partial{}_\tau{}\psi{}P\psi{}\tanh x \\
&= \partial{}_\tau{}\Bigl(\tanh x(\partial{}_\tau{}\psi{})^2 + \frac{\tanh x}{\cosh ^2x}\psi^2\Bigr) - 2\tanh x\partial{}_\tau{}\psi{}\partial{}_x^2\psi{} - \frac{2}{\cosh ^2x}\partial{}_\tau{}\psi{}\partial{}_x\psi{}  \\
&\qquad  - \tanh x \frac{1}{\cosh ^2\tau{}}\partial{}_\tau{}\psi{}\Lapl _{S^2}\psi{} - \tanh x \frac{\cosh ^2x}{\tanh^2 x}\partial{}_\tau{}\psi{}\partial{}_\theta^2\psi{} \\
&= \partial{}_\tau{}\Bigl(\Bigl((\partial{}_\tau{}\psi{})^2 + (\partial{}_x\psi{})^2 + \frac{1}{\cosh ^2x}\psi^2 + \frac{1}{\cosh ^2\tau{}}\abs{\Grad _{S^2}\psi{}}^2 + \frac{\cosh ^2x}{\tanh^2 x} (\partial{}_\theta{}\psi{})^2\Bigr)\tanh x\Bigr) \\
&\qquad  - \partial{}_x(2\partial{}_\tau{}\psi{}\partial{}_x\psi{}\tanh x) + \partial{}_\theta{}(\cdots{}) + \div_{S^2}(\cdots{}) + \frac{2\tanh \tau{}}{\cosh ^2\tau{}}\abs{\Grad _{S^2}\psi{}}^2\tanh x.  \\
\end{split}
\end{equation}
Multiply \zcref{T-estimate-1} by \(w = w(x)\) for a non-negative non-increasing
function \(w : \R_{\ge 0}\to \R_{\ge 0}\) to get
\begin{equation}\label{T-estimate-0}
\begin{split}
&-2w\partial{}_\tau{}\psi{}P\psi{}\tanh x \\
&= \partial{}_\tau{}\Bigl(\Bigl((\partial{}_\tau{}\psi{})^2 + (\partial{}_x\psi{})^2 + \frac{1}{\cosh ^2x}\psi^2 + \frac{1}{\cosh ^2\tau{}}\abs{\Grad _{S^2}\psi{}}^2 + \frac{\cosh ^2x}{\tanh^2 x} (\partial{}_\theta{}\psi{})^2\Bigr)w\tanh x\Bigr) \\
&\qquad  - \partial{}_x(2w\partial{}_\tau{}\psi{}\partial{}_x\psi{}\tanh x) + \partial{}_\theta{}(\cdots{}) + \div_{S^2}(\cdots{}) + w\frac{2\tanh \tau{}}{\cosh ^2\tau{}}\abs{\Grad _{S^2}\psi{}}^2\tanh x  \\
&\qquad + \frac{1}{2}(-w')(\underline{L}\psi{})^2\tanh x - \frac{1}{2}(-w')(L\psi{})^2\tanh x.
\end{split}
\end{equation}
We integrate \zcref{T-estimate-0} over the region \(\set{s_1\le s\le s_2}\cap \set{\tau{}\le \tau{}_0}\cap \set{x\ge x_0}\).
The boundary term at \(\set{x=x_0}\) vanishes as \(x_0\to 0\). The boundary
term at \(\set{\tau{}=\tau{}_0}\) is
\begin{equation}
\int _{\Sigma{}(\tau{})\cap \set{s_1\le s\le s_2}}\Bigl[(\underline{L}\psi{})^2 + (L\psi{})^2 + \frac{1}{\cosh ^2x}\psi^2 + \frac{1}{\cosh ^2\tau{}}\abs{\Grad _{S^2}\psi{}}^2 + \frac{\cosh ^2x}{\tanh^2 x} (\partial{}_\theta{}\psi{})^2\Bigr]w\tanh x\dd{}\mu{}.
\end{equation}
Take \(w \equiv 1\) to get \zcref{T-estimate}. Take \(w = 1 + (1 + x)^{-1}\) to get
\zcref{T-estimate-2}.

To obtain \zcref{T-estimate-3}, instead integrate \zcref{T-estimate-1} over the region
\(\set{\tau{}_1\le \tau{}\le \tau{}'}\cap \set{v\le v_0}\), where \(\tau{}'\in
[\tau{}_1,\tau{}_2]\). The boundary term at \(\set{v=v_0}\) has a good sign. After
dropping this term, we obtain \zcref{T-estimate-3}.
\end{proof}
\subsection{An \texorpdfstring{\(r^p\)}{rᵖ}-type energy estimate and energy decay for \texorpdfstring{\(\U(1)\)}{U(1)}-symmetric scalar fields}
\label{sec:rp-estimate} The main result of this section is \zcref{rp-type-estimate},
which establishes a version of the \(r^p\)-weighted energy estimates of
Dafermos--Rodnianski \cite{rp-method}. These estimates hold only for \(\U(1)\)-symmetric
scalar fields. We then use this estimate to deduce energy decay for \(\U(1)\)-symmetric
solutions to the wave equation, analogous to the decay that follows from an
\(r^p\)-weighted energy hierarchy.
\begin{proposition}[An \(r^p\)-type energy estimate]
Let \(\psi{}\in C^\infty(\mathcal{M})\) be \(\U(1)\)-symmetric. For \(0\le s_1\le s_2\) and \(\tau{}\ge 0\),
we have
\begin{equation}\label{rp-equation}
\begin{split}
\mathcal{E}[\psi{}](s_2,\tau{}) + \sup_{0\le \tau{}'\le \tau{}} \mathcal{F}_T[\psi{}](s_1,s_2,\tau{}') + \mathcal{B}[\psi{}](s_1,s_2,\tau{}) &\lesssim \mathcal{E}[\psi{}](s_1,\tau{}) + \int_{s_1}^{s_2} \int _{\mathcal{H}(s,\tau)}\sinh x\abs{P\psi{}}^2\dd{}\mu{}\dd{}s,
\end{split}
\end{equation}
where
\begin{equation}\label{bulk-term-def}
\mathcal{B}[\psi{}](s_1,s_2,\tau{})\coloneqq{}\int_{s_1}^{s_2} \int _{\mathcal{H}(s,\tau)}\langle{}x\rangle{}^{-2}\tanh x(\underline{L}\psi{})^2 + \sinh x(L\psi{})^2 + \frac{\sinh x}{\cosh ^2\tau{}} \abs{\Grad _{S^2}\psi{}}^2 + \frac{1}{\cosh x}\psi^2\dd{}\mu{}\dd{}s
\end{equation}
satisfies
\begin{equation}\label{rp-K-bound}
\mathcal{B}[\psi{}](s_1,s_2,\tau{})\gtrsim \int_{s_1}^{s_2} \mathcal{E}[\psi{}](s,\tau)\dd{}s.
\end{equation}
Moreover, for \(0\le \tau_1\le \tau_2\), we have
\begin{equation}\label{rp-equation-2}
E_T[\psi{}](\tau{}_2) + E_p[\psi{}](\tau{}_2) + F_p[\psi{}](\tau{}_1,\tau{}_2) + B[\psi{}](\tau{}_1,\tau{}_2)\lesssim E_T[\psi{}](\tau{}_1) + E_p[\psi{}](\tau{}_1) + \int_{\tau{}_1}^{\tau{}_2} \int _{\Sigma{}(\tau{})} \sinh x\abs{P\psi{}}^2\dd{}\mu{}\dd{}\tau{},
\end{equation}
where
\begin{equation}
B[\psi{}](\tau{}_1,\tau{}_2)\coloneqq{}\int_{\tau{}_1}^{\tau{}_2} \int _{\Sigma{}(\tau{})}\langle{}x\rangle{}^{-2}\tanh x(\underline{L}\psi{})^2 + \sinh x(L\psi{})^2 + \frac{\sinh x}{\cosh ^2\tau{}} \abs{\Grad _{S^2}\psi{}}^2 + \frac{1}{\cosh x}\psi^2\dd{}\mu{}\dd{}\tau{}.
\end{equation}
\label{rp-type-estimate}
\end{proposition}
\begin{remark}[Relation to the \(r^p\)-weighted energy estimates of Dafermos--Rodnianski]
On \((3 + 1)\)-dimensional Minkowski space, the \(r^p\)-weighted energy
estimates of Dafermos--Rodnianski \cite{rp-method} are derived using a multiplier
of the form \(r^p\partial_v\) in a large-\(r\) region, where \(p\in [0,2]\) and
\(\partial_v = \partial_t + \partial_r\) in standard \((t,r,\omega{})\) polar
coordinates. To derive \zcref{rp-type-estimate}, we use a multiplier of the form
\((\sinh x)L\). Since \(L = 2V\partial_V\) in Minkowskian double null
coordinates \((U,V)\) (see \zcref{double-null}) and \(\sinh x\sim
\abs{U}^{1/2}V^{1/2}\) when \(x\gg 1\), our multiplier is like
\(\abs{U}^{1/2}V^{3/2}\partial_V\). Since \(V\sim r\), our estimates have the
same scaling as the usual \(r^p\)-weighted energy estimates with \(p = 3/2\), at
least in the asymptotic region where \(\abs{U}\gtrsim 1\) and \(x\gg 1\). This region
is well-separated from the bubble.

In fact, the proof of \zcref{rp-type-estimate} generalizes to multipliers that for
large \(x\) take the form \((\sinh x)^qL\) for \(q\in (0,2)\), which correspond
to \(r^p\partial_v\) with \(p = 1 + q/2\), so that \(p\in (1,2)\). The
restriction \(q > 0\) corresponds to the fact that we must use a multiplier
\(f(x)L\) with \(f'(x)\sim f(x)\) in a large-\(x\) region to obtain exponential
decay, while the upper bound \(q<2\) is imposed by the positivity of the
zeroth-order bulk term (see \zcref{rp-identity}). However, our proof of global
existence for a semilinear equation outside symmetry closes only for the
estimate in \zcref{rp-type-estimate}, where \(q = 1\), which corresponds to \(p =
3/2\) (see \zcref{semilinear-rp-value}).
\label{rp-relation}
\end{remark}
\begin{remark}[The requirement of \(\U(1)\)-symmetry]
The estimates of \zcref{rp-type-estimate} can hold only for \(\U(1)\)-symmetric scalar
fields \(\psi{}\), in view of the possible periodicity of \(\psi{}_{\ge 1}\) in
\(\tau{}\) (see \zcref{time-periodic-intro} of \zcref{main-theorem-linear-intro}). In
particular, one cannot hope to control \(\psi_{\ge 1}\) in a spacetime cylinder
in \(L^2\), as the bulk term \(B(0,\infty)\) does.

This phenomenon is reflected in the structure of the wave equation, which
contains a \(\partial_\theta^2\)-term with coefficient \((\tanh x)^{-2}\cosh
^2x\). After decomposition into \(S^1\)-modes, each non-zero mode therefore
carries an effective (\(x\)-dependent) mass. In this sense, \(\psi_{\ge 1}\)
behaves more like a solution (or rather, a superposition of solutions) of the
massive Klein--Gordon equation than like a solution of the massless wave
equation.

Since the \(r^p\)-weighted estimates capture the improved decay of outgoing null
derivatives, and such decay is absent for solutions to the Klein--Gordon
equation on Minkowski space, one should not expect such estimates to hold for
\(\psi_{\ge 1}\). At the level of the multiplier identity, the obstruction to
establishing these estimates appears through a bulk term of a bad sign involving
\((\partial_{\theta{}}\psi{})^2\).
\label{rp-axisymmetric-required}
\end{remark}
\begin{remark}[No need for a Morawetz estimate]
The usual proof of \(r^p\)-type estimates is done in a ``far'' region (which is
asymptotically flat and hence close to Minkowski space) and uses a Morawetz, or
integrated local energy decay, estimate to control error terms in a ``near''
region. \Cref{rp-type-estimate} establishes a \emph{global} estimate that does not
distinguish between a ``near'' region and a ``far'' region, and does not need as
input a Morawetz estimate. In particular, we can obtain a good bulk term using
the same vector field multiplier for all regions of spacetime, including the
region near the bubble itself (which is not asymptotically flat).
\end{remark}
\begin{remark}[Relation to estimates for the linear wave equation in two space dimensions]
The proof of \zcref{rp-type-estimate} (in particular the expansion
\zcref{expanded-LPsi} and subsequent differentiation by parts) is inspired by our
work \cite{gautam-2d-waves}, which extends the \(r^p\)-weighted energy
estimates of Dafermos--Rodnianski \cite{rp-method} to two space dimensions. The
two space dimensions in question here comprise the \(\R^2\) factor of the Witten
bubble spacetime \(\mathcal{M}\), which has smooth structure \(\R\times \R^2\times S^2\).
\label{2d-waves-relation}
\end{remark}
\begin{proof}
The proof has five steps. In Steps 1--4, we prove the estimate \zcref{rp-equation}
associated to the hyperboloidal foliation. In Step 1, we derive a general
identity for multipliers of the form \(f(x)L\). In Step 2, we provide
computations showing that \(f(x) = \sinh x + x/(1+x)\) produces an identity with
good bulk and boundary terms, at least in a large-\(\tau{}\) region. In Step 3,
we integrate the multiplier identity to produce an energy estimate with errors
in a finite-\(\tau{}\) region that we absorb using the
\(\partial_\tau{}\)-energy estimate of \zcref{T-estimate-prop} and Grönwall's
inequality. This proves \zcref{rp-equation,rp-equation-2}. In Step 4, we show that
the bulk term present in this energy estimate controls the integral of the flux
terms through hyperboloidal hypersurfaces (see \zcref{rp-K-bound}). We will use this
estimate in \zcref{energy-decay} to deduce exponential decay in \(s\) for solutions
to the wave equation. Finally, in Step 5, we sketch the modifications to the
previous steps needed to obtain the estimate \zcref{rp-equation-2} associated to the
constant-\(\tau{}\) foliation.

\step{Step 1: Deriving a multiplier identity.} We begin by deriving the general
identity \zcref{rp-identity} for multipliers of the form \(f(x)L\). Write \(V =
-1/(4\cosh^2 x\sinh^2x)\), and let \(f = f(x)\) be a \(C^2\) function. Multiply
\zcref{P-Psi-equation} by \(-2f(x)L\Psi{}\) and use the Leibniz rule to get
\begin{equation}
\begin{split}
&-2(\tanh x)^{1/2}fL\Psi{}P\psi{} = f\underline{L}((L\Psi{})^2) + fVL(\Psi^2) - f \frac{1}{\cosh ^2\tau{}}2\Lapl _{S^2}\Psi{}L\Psi{}.
\end{split}
\end{equation}
Differentiating by parts, we obtain
\begin{equation}\label{rp-prep}
\begin{split}
&-2(\tanh x)^{1/2}fL\Psi{}P\psi{} \\
&= \underline{L}(f(L\Psi{})^2) + L\Bigl(\frac{1}{\cosh ^2\tau{}}\abs{\Grad _{S^2}\Psi{}}^2 + fV\Psi^2\Bigr) + \div_{S^2}(\cdots{}) + f'(L\Psi{})^2 - L\Bigl(\frac{f}{\cosh ^2\tau{}}\Bigr)\abs{\Grad _{S^2}\Psi{}}^2 - (fV)'\Psi^2.
\end{split}
\end{equation}
Expand \(\Psi{} = (\tanh x)^{1/2}\psi{}\) to get
\begin{equation}\label{expanded-LPsi}
(L\Psi{})^2 = \tanh x(L\psi{})^2 + \frac{1}{4}(\tanh x)^{-1}(\cosh x)^{-4}\psi^2 + \frac{1}{2}(\cosh x)^{-2}L(\psi^2).
\end{equation}
Substitute \zcref{expanded-LPsi} into \zcref{rp-prep} and differentiate the \(L(\psi^2)\)
term by parts to obtain
\begin{equation}\label{rp-identity}
\begin{split}
&-2(\tanh x)^{1/2}fL\Psi{}P\psi{}\\
&= \underline{L}(f(L\Psi{})^2) + L\Bigl(\frac{f}{\cosh ^2\tau{}}\abs{\Grad _{S^2}\Psi{}}^2 + \Bigl(\frac{1}{2}f'(\cosh x)^{-2} + f\tanh xV\Bigr)\psi{}^2\Bigr) + \div_{S^2}(\cdots{}) \\
&\qquad + f'\tanh x(L\psi{})^2 - L\Bigl(\frac{f}{\cosh ^2\tau{}}\Bigr)\abs{\Grad _{S^2}\Psi{}}^2 + \frac{1}{4}\Bigl[\frac{f'}{\tanh x\cosh ^4x} - 2\Bigl(\frac{f'}{\cosh^2 x}\Bigr)' - 4\tanh x(fV)'\Bigr]\psi{}^2.
\end{split}
\end{equation}

\step{Step 2: Constructing a good multiplier.} In this step, we show that, in a
large-\(\tau{}\) region, taking \(f(x) = \sinh x\) in the result \zcref{rp-identity}
of Step 1 produces a good multiplier for large \(x\) (see Step 2a) and that
\(x/(1 + x)\) is good for small \(x\) (see Step 2b). The remaining
finite-\(\tau{}\) region will be controlled in Step 3 using the
\(\partial_\tau{}\)-energy estimate of \zcref{T-estimate-prop} and a Grönwall
inequality.

\step{Step 2a: Analysis of \(f(x) = \sinh x\)}.
When \(f(x) = \sinh x\) in \zcref{rp-identity}, we have
\begin{equation}\label{rp-sinh}
\begin{split}
&-2(\tanh x)^{1/2}\sinh xL\Psi{}P\psi{}\\
&= \underline{L}(\sinh x (L\Psi{})^2) + L\Bigl(\frac{\sinh x}{\cosh ^2\tau{}}\abs{\Grad _{S^2}\Psi{}}^2 + \Bigl(\frac{1}{2\cosh x} - \frac{1}{4\cosh^3 x}\Bigr)\psi{}^2\Bigr) + \div_{S^2}(\cdots{}) \\
&\qquad + \sinh x(L\psi{})^2 - L\Bigl(\frac{\sinh x}{\cosh ^2\tau{}}\Bigr)\abs{\Grad _{S^2}\Psi{}}^2  + \frac{1}{2} \frac{\tanh^3 x}{\cosh x}\psi^2.
\end{split}
\end{equation}
We have
\begin{equation}\label{rp-sinh-L}
-L\Bigl(\frac{\sinh x}{\cosh ^2\tau{}}\Bigr) = \frac{\cosh x}{\cosh ^2\tau{}}(2\tanh x \tanh \tau{} - 1)\gtrsim \frac{\cosh x}{\cosh ^2\tau{}}\quad \textnormal{in }\set{x\gg 1}\cap \set{\tau{}\gg 1}
\end{equation}
and
\begin{equation}\label{rp-sinh-L-2}
\abs[\Big]{L\Bigl(\frac{\sinh x}{\cosh ^2\tau{}}\Bigr)}\lesssim \frac{\cosh x}{\cosh ^2\tau{}}.
\end{equation}

\step{Step 2b: Analysis of \(f(x) = x/(1 + x)\).} When \(f(x) = x/(1 + x)\) in \zcref{rp-identity}, we have
\begin{equation}\label{rp-x}
\begin{split}
&-2(\tanh x)^{1/2}x(1+x)^{-1}L\Psi{}P\psi{}\\
&= \underline{L}\Bigl(\frac{x}{1+x}(L\Psi{})^2\Bigr) + L\Bigl(\frac{x}{(1+x)\cosh ^2\tau{}}\abs{\Grad _{S^2}\Psi{}}^2 + \textnormal{(II)}\cdot \psi{}^2\Bigr) + \div_{S^2}(\cdots{}) \\
&\qquad + \frac{\tanh x}{(1+x)^2}(L\psi{})^2 - L\Bigl(\frac{x}{(1+x)} \frac{1}{\cosh ^2\tau{}}\Bigr)\abs{\Grad _{S^2}\Psi{}}^2 + \textnormal{(III)}\cdot \psi^2
\end{split}
\end{equation}
for
\begin{equation}
\textnormal{(II)}\coloneqq{}\frac{1}{(1+x)\cosh ^2x}\Bigl(\frac{1}{2} \frac{1}{(1+x)} - \frac{1}{4}\frac{x\tanh x}{\sinh^2 x}\Bigr)
\end{equation}
and
\begin{equation}
\textnormal{(III)}\coloneqq{}\frac{1}{2(1+x)^3\cosh ^2x}\Bigl[2+\tanh x(1+x)+(1+x) \frac{x}{\sinh x} \frac{1}{\sinh x}\Bigl(\cosh x \frac{\sinh x}{x}+\frac{1+x}{\cosh ^2x}-2(1+x)\Bigr)\Bigr].
\end{equation}
We can compute
\begin{equation}\label{rp-x-III}
\textnormal{(III)} > 0,\qquad \textnormal{(III)}|_{x=0} = \frac{1}{2}.
\end{equation}
We have
\begin{equation}\label{rp-x-L}
-L\Bigl(\frac{x}{(1+x)} \frac{1}{\cosh ^2\tau{}}\Bigr) = \frac{1}{(1+x)\cosh^2\tau{}}\Bigl(2x\tanh \tau{}-\frac{1}{1+x}\Bigr)\gtrsim \frac{1}{\cosh ^2\tau{}}\quad \textnormal{in }\set{x\gg 1}\cap \set{\tau{}\gg 1}
\end{equation}
and
\begin{equation}\label{rp-x-L-2}
\abs[\Big]{L\Bigl(\frac{x}{(1+x)} \frac{1}{\cosh ^2\tau{}}\Bigr)}\lesssim \frac{1}{\cosh ^2\tau{}}.
\end{equation}

\step{Step 3: Producing an energy estimate; proof of \zcref{rp-equation,rp-equation-2}.} We
now integrate the multiplier identities of Step 2 over suitable spacetime
regions to prove \zcref{rp-equation,rp-equation-2}. We will give a detailed proof of
\zcref{rp-equation}; the proof of \zcref{rp-equation-2} is similar.

Adding \zcref{rp-sinh,rp-x} produces an identity
\begin{equation}\label{rp-identity-prep}
\begin{split}
&O((\tanh x)^{1/2}\sinh x)L\Psi{}P\psi{}\\
&= \underline{L}(\alpha{}_1(L\Psi{})^2) + L(\alpha{}_2\abs{\Grad _{S^2}\Psi{}}^2 + \alpha{}_3\psi^2) + \div_{S^2}(\cdots{}) + \alpha{}_4(L\psi{})^2 + \alpha{}_5\abs{\Grad _{S^2}\Psi{}}^2 + \alpha{}_6\psi^2
\end{split}
\end{equation}
with coefficients \(\alpha_i\) (\(1\le i\le 6\)). In Step 3a, we analyze the coercivity
properties of the \(\alpha_i\), and in Step 3b, we complete the proof of
\zcref{rp-equation,rp-equation-2}.

\step{Step 3a: Positivity of the bulk and flux terms and vanishing of the boundary
term at the bubble.} We now study the properties of the coefficients \(\alpha_i\). By
construction we have
\begin{equation}\label{rp-alpha-1}
\alpha_1=\sinh x + x/(1+x)\ge \sinh x\ge 0.
\end{equation}
From \zcref{rp-sinh,rp-x}, we have
\begin{equation}\label{rp-alpha-2}
\alpha_2 = \frac{1}{\cosh ^2\tau{}}\Bigl(\sinh x + \frac{x}{1+x}\Bigr)\ge 0
\end{equation}
and
\begin{equation}\label{rp-alpha-3}
\alpha{}_3 = \Bigl(\frac{1}{2\cosh x} - \frac{1}{4\cosh^3 x}\Bigr) + \frac{1}{(1+x)\cosh^2 x}\Bigl(\frac{1}{2} \frac{1}{1+x} - \frac{1}{4} \frac{x\tanh x}{\sinh^2 x}\Bigr)\ge 0.
\end{equation}
Next, from \zcref{rp-sinh,rp-x} it is clear that \(\alpha_4\ge \sinh x\). From \zcref{rp-sinh-L,rp-sinh-L-2,rp-x-L,rp-x-L-2}, we have
\begin{equation}\label{large-x-large-tau}
\alpha{}_5\gtrsim \frac{\cosh x}{\cosh ^2\tau{}}\quad \textnormal{in }\set{x\gg 1}\cap \set{\tau{}\gg 1},\qquad \abs{\alpha{}_5}\lesssim \frac{\cosh x}{\cosh ^2\tau{}}.
\end{equation}
Finally, by applying \zcref{rp-sinh} in \(\set{x\ge 1}\) and \zcref{rp-x-III} in \(\set{x\le
1}\), we have
\begin{equation}
\alpha_6\gtrsim \frac{1}{\cosh x}.
\end{equation}

We now integrate \zcref{rp-identity-prep} over the region \(\set{s_1\le s \le s_2}\cap
\set{v\le v_0}\cap \set{x\ge x_0}\) for some \(v_0\ge 0\) and \(x_0>0\). The
boundary term at \(\set{x=x_0}\) has integrand
\begin{equation}\label{vanishing-coeffs}
\begin{split}
&\alpha{}_1(L\Psi{})^2 - \alpha{}_2\abs{\Grad _{S^2}\Psi{}}^2 - \alpha{}_3\psi^2  \\
&= \alpha{}_1\tanh x(L\psi{})^2 + \Bigl(\frac{1}{4}(\tanh x)^{-1}(\cosh x)^{-4}\alpha{}_1 - \alpha{}_3\Bigr)\psi^2 + \alpha{}_1(\cosh x)^{-2}\psi{}L\psi{} - \alpha{}_2\tanh x\abs{\Grad _{S^2}\psi{}}^2.
\end{split}
\end{equation}
This boundary term vanishes as \(x_0\to 0\), in view of the boundedness of \(\psi{}\),
\(L\psi{}\), \(\Grad _{S^2}\psi{}\) at the bubble. Indeed, each coefficient in
\zcref{vanishing-coeffs} vanishes linearly as \(x\to 0\), because of the expansions
\begin{equation}
\alpha{}_1 = 2x + O(x^2),\quad \alpha{}_2 = O(1),\quad \alpha{}_3 = \frac{1}{2} + O(x)
\end{equation}
as \(x\to 0\), uniformly in \(\tau{}\). These expansions can be obtained from the
explicit expressions \zcref{rp-alpha-1,rp-alpha-2,rp-alpha-3}. Combining the above,
we obtain an estimate
\begin{equation}\label{rp-estimate-prep}
\begin{split}
&\mathcal{E}_p[\psi{}](s_2) + \int_{s_1}^{s_2} \int _{\mathcal{H}(s,\tau)}\sinh x(L\psi{})^2 + \frac{\sinh x}{\cosh ^2\tau{}} \abs{\Grad _{S^2}\psi{}}^2 + \frac{1}{\cosh x}\psi^2\dd{}\mu{}\dd{}s \\
&\lesssim \mathcal{E}_p[\psi{}](s_1) + \int_{s_1}^{s_2} \int _{\mathcal{H}(s,\tau)\cap (\set{x\lesssim 1}\cup \set{\tau{}\lesssim 1})} \frac{\cosh x}{\cosh ^2\tau{}}\abs{\Grad _{S^2}\psi{}}^2\tanh x\dd{}\mu{}\dd{}s\\
&\qquad +\int_{s_1}^{s_2} \int _{\mathcal{H}(s,\tau)}(\tanh x)^{1/2}\sinh x\abs{L\Psi{}}\abs{P\psi{}}\dd{}\mu{}\dd{}s.
\end{split}
\end{equation}
The angular bulk term supported in \(\set{x\lesssim 1}\cup \set{\tau{}\lesssim 1}\)
on the right-hand side of \zcref{rp-estimate-prep} arises from
\zcref{large-x-large-tau}.

\step{Step 3b: Controlling the error terms and producing a closed energy estimate.}
Adding to \zcref{rp-estimate-prep} a large multiple of the \(\partial{}_\tau{}\)-energy estimate \zcref{T-estimate} to control the
angular bulk term on the right side in the region \(\set{x\lesssim 1}\cap
\set{\tau{}\gg 1}\), we find that there is a universal \(\tau_0\ge 0\) such that
\begin{equation}\label{rp-estimate-prep-2}
\begin{split}
&\mathcal{E}_p[\psi{}](s_2)+ \mathcal{E}_T[\psi{}](s_2) + \sup_{0\le \tau{}'\le \tau{}}\mathcal{F}_T[\psi{}](s_1,s_2,\tau{}') \\
&\qquad + \int_{s_1}^{s_2} \int _{\mathcal{H}(s,\tau)}\sinh x(L\psi{})^2 + \frac{\sinh x}{\cosh ^2\tau{}} \abs{\Grad _{S^2}\psi{}}^2 + \frac{1}{\cosh x}\psi^2\dd{}\mu{}\dd{}s \\
&\lesssim \mathcal{E}_p[\psi{}](s_1) + \mathcal{E}_T[\psi{}](s_1) +  \int_{s_1}^{s_2} \int _{\mathcal{H}(s,\tau)\cap \set{\tau{}\le \tau{}_0}} \frac{\cosh x}{\cosh ^2\tau{}}\abs{\Grad _{S^2}\psi{}}^2\tanh x\dd{}\mu{}\dd{}s\\
&\qquad +\int_{s_1}^{s_2} \int _{\mathcal{H}(s,\tau)}(\tanh x\abs{\partial{}_\tau{}\psi{}} + (\tanh x)^{1/2}\sinh x\abs{L\Psi{}})\abs{P\psi{}}\dd{}\mu{}\dd{}s.
\end{split}
\end{equation}
Since \(\mathcal{H}(s,\tau)\cap \set{\tau{}\le \tau{}_0}\subset \set{x\lesssim 1}\), we conclude, after
adding a small multiple of \zcref{T-estimate-2} to \zcref{rp-estimate-prep-2}, that
\begin{equation}\label{rp-equation-prep-2}
\begin{split}
&\mathcal{E}_p[\psi{}](s_2) + \mathcal{E}_T[\psi{}](s_2) + \int_{s_1}^{s_2} \int _{\mathcal{H}(s,\tau)}\langle{}x\rangle{}^{-2}\tanh x(\underline{L}\psi{})^2 + \sinh x(L\psi{})^2 + \frac{\sinh x}{\cosh ^2\tau{}} \abs{\Grad _{S^2}\psi{}}^2 + \frac{1}{\cosh x}\psi^2\dd{}\mu{}\dd{}s \\
&\lesssim \mathcal{E}_p[\psi{}](s_1) + \mathcal{E}_T[\psi{}](s_1) + \int_{s_1}^{s_2} \int _{\mathcal{H}(s,\tau)\cap \set{\tau{}\le \tau{}_0}} \frac{1}{\cosh ^2\tau{}}\abs{\Grad _{S^2}\psi{}}^2\tanh x\dd{}\mu{}\dd{}s\\
&\qquad +\int_{s_1}^{s_2} \int _{\mathcal{H}(s,\tau)}(\tanh x\abs{\partial{}_\tau{}\psi{}} + (\tanh x)^{1/2}\sinh x\abs{L\Psi{}})\abs{P\psi{}}\dd{}\mu{}\dd{}s.
\end{split}
\end{equation}
We now control the final bulk term on the right-hand side of
\zcref{rp-equation-prep-2}. Since
\begin{equation}\label{LPsi-expansion}
\abs{L\Psi{}}\lesssim (\tanh x)^{1/2}\abs{L\psi{}} + (\tanh x)^{-1/2}(\cosh x)^{-2}\abs{\psi{}},
\end{equation}
we have
\begin{equation}
\begin{split}
(\tanh x\abs{\partial{}_\tau{}\psi{}} + (\tanh x)^{1/2}\sinh x\abs{L\Psi{}})\abs{P\psi{}} &\lesssim (\tanh x\abs{\underline{L}\psi{}} + \sinh x\abs{L\psi{}} + \tanh x (\cosh x)^{-1}\abs{\psi{}})\abs{P\psi{}}.
\end{split}
\end{equation}
After using Young's inequality and recalling the definition of the bulk term
\(\mathcal{B}\) from \zcref{bulk-term-def}, we obtain
\begin{equation}\label{rp-prep-3}
\begin{split}
&\mathcal{E}_p[\psi{}](s_2) + \mathcal{E}_T[\psi{}](s_2) + \mathcal{B}[\psi{}](s_1,s_2,\tau{}) \\
&\lesssim \mathcal{E}_p[\psi{}](s_1) + \mathcal{E}_T[\psi{}](s_1) + \int_{s_1}^{s_2} \int _{\mathcal{H}(s,\tau)\cap \set{\tau{}\le \tau{}_0}} \frac{1}{\cosh ^2\tau{}}\abs{\Grad _{S^2}\psi{}}^2\tanh x\dd{}\mu{}\dd{}s  +\int_{s_1}^{s_2} \int _{\mathcal{H}(s,\tau)}\sinh x\abs{P\psi{}}^2\dd{}\mu{}\dd{}s.
\end{split}
\end{equation}
It remains to control the bulk integral in a finite-\(\tau{}\) region on the
right-hand side of \zcref{rp-prep-3}. This is possible because a finite-\(\tau{}\)
region is contained in a finite-\(s\) region. That is, given \(\tau_0\ge 0\),
there is \(s_0 \ge 0\) such that \(\mathcal{H}(s,v_0)\cap \set{\tau{}\le
\tau{}_0}\subset \set{v\le v_0}\cap \set{s\le s_0}\) for all \(s\ge 0\) and
\(v_0\ge 0\). It follows that
\begin{equation}\label{rp-prep-4}
\begin{split}
\int_{s_1}^{s_2} \int _{\mathcal{H}(s,\tau)\cap \set{\tau{}\le \tau{}_0}} \frac{1}{\cosh ^2\tau{}}\abs{\Grad _{S^2}\psi{}}^2\tanh x\dd{}\mu{}\dd{}s &\lesssim \int_{\min  (s_1,s_0)}^{s_0}  \int _{\mathcal{H}(s,\tau)} \frac{1}{\cosh ^2\tau{}}\abs{\Grad _{S^2}\psi{}}^2\tanh x\dd{}\mu{}\dd{}s \\
&\lesssim \int_{\min  (s_1,s_0)}^{s_0} \mathcal{E}_T[\psi{}](s)\dd{}s.
\end{split}
\end{equation}
Substituting \zcref{rp-prep-4} into \zcref{rp-prep-3} and applying Grönwall's inequality
(on the bounded interval \([\min (s_1,s_0),s_0]\)), we obtain \zcref{rp-equation}.

\step{Step 4: Showing that the bulk term controls the integrated energy; proof of
\zcref{rp-K-bound}.} Recall that, when \(\psi{}\in C^\infty(\mathcal{M})\) is \(\U(1)\)-symmetric, we have
\begin{equation}
\begin{split}
\mathcal{E}[\psi{}](s,\tau{}) &= \int _{\mathcal{H}(s,\tau)} \Bigl[(L\psi{})^2 + \langle{}x\rangle^{-2}(\underline{L}\psi{})^2 + \frac{1}{\cosh^2x}\psi^2 + \frac{1}{\cosh ^2\tau{}}\abs{\Grad _{S^2}\psi{}}^2\Bigr]\tanh x \dd{}\mu{} \\
&\qquad  + \int _{\mathcal{H}(s,\tau)}\sinh x(L\Psi{})^2 + \langle{}x\rangle^{-2}\frac{1}{\cosh x}\psi^2 + \langle{}x\rangle^{-2}\frac{\sinh x}{\cosh ^2\tau{}}\abs{\Grad _{S^2}\Psi{}}^2\dd{}\mu{} \\
&\lesssim \int _{\mathcal{H}(s,\tau)} \sinh x(L\psi{})^2 + \langle{}x\rangle^{-2}(\underline{L}\psi{})^2\tanh x + \frac{1}{\cosh x}\psi^2 + \frac{\sinh x}{\cosh ^2\tau{}}\abs{\Grad _{S^2}\psi{}}^2 \dd{}\mu{}.
\end{split}
\end{equation}
To pass to the final line, we used the estimate \zcref{LPsi-expansion} for \(L\Psi{}\) in
terms of \(L\psi{}\) and \(\psi{}\). Now \zcref{rp-K-bound} follows from the definition of the
bulk term \(\mathcal{B}\) (see \zcref{bulk-term-def}).

\step{Step 5: Sketch of the proof of \zcref{rp-equation-2}.} The proof of \zcref{rp-equation-2}
is largely the same as that of \zcref{rp-equation}. Steps 1 and 2 are the same. In
Step 3a, we integrate instead over the region \(\set{\tau_1\le \tau{}\le \tau_2}\cap \set{u\le u_0}\cap
\set{x\ge x_0}\) for some \(u_0\in \R\) and \(x_0 > 0\). The future boundary term on
\(\set{u=u_0}\) is
\begin{equation}
\int _{C(u_0)\cap \set{\tau{}_1\le \tau{}\le \tau{}_2}\cap \set{x\ge x_0}} \sinh x(L\Psi{}_0)^2\dd{}v\dd{}\omega{}\dd{}\theta{}.
\end{equation}
Using the positivity analysis in Step 3a and taking \(x_0\to 0\) and taking a
supremum over \(u_0\in \R\), we obtain an estimate
\begin{equation}
\begin{split}
&E_p[\psi{}](\tau{}_2) + F_p[\psi{}](\tau{}_1,\tau{}_2) + \int_{\tau{}_1}^{\tau{}_2} \int _{\Sigma{}(\tau{})}\sinh x(L\psi{})^2 + \frac{\sinh x}{\cosh ^2\tau{}} \abs{\Grad _{S^2}\psi{}}^2 + \frac{1}{\cosh x}\psi^2\dd{}\mu{}\dd{}\tau{} \\
&\lesssim E_p[\psi{}](\tau{}_1) + \int_{\tau{}_1}^{\tau{}_2} \int _{\Sigma{}(\tau{})\cap (\set{x\lesssim 1}\cup \set{\tau{}\lesssim 1})} \frac{\cosh x}{\cosh ^2\tau{}}\abs{\Grad _{S^2}\psi{}}^2\tanh x\dd{}\mu{}\dd{}\tau{}\\
&\qquad +\int_{\tau{}_1}^{\tau{}_2} \int _{\Sigma{}(\tau{})}(\tanh x)^{1/2}\sinh x\abs{L\Psi{}}\abs{P\psi{}}\dd{}\mu{}\dd{}\tau{}.
\end{split}
\end{equation}
in place of \zcref{rp-estimate-prep}. From here, Step 3b proceeds in the same way,
and we obtain \zcref{rp-equation-2}.
\end{proof}
\begin{lemma}[Exponential decay lemma]
Let \(T\ge 0\) and let \(f:[0,T]\to \R\) be a locally integrable function. Suppose
that for \(0\le t_1\le t_2\le T\) we have
\begin{equation}\label{exponential-decay-assumption}
f(t_2) + \int_{t_1}^{t_2} f(t)\dd{}t\le Bf(t_1).
\end{equation}
for some constant \(B \ge 1\). Then there are constants \(C,c>0\) depending only
on \(B\) such that for \(t\in [0,T]\) we have
\begin{equation}
f(t)\le Cf(0)e^{-ct}.
\end{equation}
\label{exponential-decay}
\end{lemma}
\begin{proof}
Extending \(f\) by zero for \(t\ge T\), we may assume that \(T = +\infty\). Let \(A >
0\). We claim that for \(n\ge 0\), we have
\begin{equation}\label{exponential-decay-induction}
f(nA)\le (B^2/A)^nf(0).
\end{equation}
The case \(n = 0\) is trivial. From \zcref{exponential-decay-assumption} applied on
the time interval \([nA,(n+1)A]\), we find (using the pigeonhole principle) a
time \(t\in [nA,(n+1)A]\) such that \(f(t)\le (B/A)f(nA)\), so that another
application of \zcref{exponential-decay-assumption}, now on the time interval
\([t,(n+1)A]\), gives \(f((n+1)A)\le Bf(t)\le (B^2/A)f(nA)\). By induction on
\(n\), we have established \zcref{exponential-decay-induction}.

Now fix \(t\ge 0\) and let \(n\ge 0\) be such that \(nA\le t<(n + 1)A\). Then
\zcref{exponential-decay-induction} together with \zcref{exponential-decay-assumption}
applied on the time interval \([nA,t]\) shows that \(f(t)\le B(B^2/A)^nf(0)\).
Taking \(A = 2B^2\), we obtain the desired conclusion:
\begin{equation}
f(t)\le B\cdot 2^{-n}f(0)\le 2B\cdot 2^{-(n+1)}f(0)\le 2B\cdot 2^{-t/(2B^2)}f(0).
\end{equation}
\end{proof}
\begin{corollary}[Energy decay for \(\U(1)\)-symmetric solutions to the wave equation]
There exists \(c > 0\) such that, if \(\psi{}\in C^\infty(\mathcal{M})\) is \(\U(1)\)-symmetric and
solves \(P\psi{} = 0\), then for \(0\le s_1\le s_2\) and \(\tau{}\ge 0\), we have
\begin{equation}
\mathcal{E}[\psi{}](s_2,\tau{})\lesssim e^{-c(s_2-s_1)}\mathcal{E}[\psi{}](s_1,\tau{}).
\end{equation}
\label{energy-decay}
\end{corollary}
\begin{proof}
In view of \zcref{rp-equation,rp-K-bound}, this follows from \zcref{exponential-decay}
(with \(\mathcal{E}[\psi{}](\cdot ,\tau{})\) playing the role of \(f\)).
\end{proof}
\section{Pointwise estimates}
The goal of this section is to prove the following pointwise estimates for \(\psi{}\)
and its derivatives. We first state estimates formulated in terms of the
hyperboloidal foliation \(\mathcal{H}(s)\).
\begin{proposition}[Global pointwise estimates in terms of a hyperboloidal foliation]
Let \(s\ge 0\) and let \(\tau{}\ge 0\). There is a universal constant \(C > 0\)
(depending only on the function \(h\) defining the hyperboloidal foliation
introduced in \zcref{hyperboloidal-foliation}) such that if \(\tau{}-2s \ge C\),
then the following pointwise estimates using energies associated to the
hyperboloidal foliation hold for \(\psi{}\in C^\infty(\mathcal{M})\):
\begin{align}
\norm{\psi{}_0}_{L^\infty(\mathcal{H}(s,\tau{}))}^2 &\lesssim \mathcal{E}[(\mathring{\mathfrak{D}}^{\le 3}\psi{})_0](s,\tau{}) + \int _{\mathcal{H}(s,\tau{})\cap \set{x\lesssim 1}} (P\mathring{\mathfrak{D}}^{\le 2}\psi{})_0^2\tanh x \dd{}\mu{}, \label{Loo-1-asym} \\
\norm{(\cosh x)^{1/2}\psi{}_{\ge 1}}_{L^\infty(\mathcal{H}(s,\tau{}))}^2 &\lesssim \mathcal{E}_T[\mathring{\mathfrak{D}}^{\le 3}\psi{}](s,\tau{}) +  \int _{\mathcal{H}(s,\tau{})\cap \set{x\lesssim 1}} (P\mathring{\mathfrak{D}}^{\le 2}\psi{})^2\tanh x \dd{}\mu{}, \label{Loo-2-asym}
\end{align}
and
\begin{equation}\label{Loo-3-asym}
\begin{split}
\norm{\partial{}_x\psi{}_0}_{L^\infty(\mathcal{H}(s,\tau{}))}^2 &\lesssim \mathcal{E}_T[(\mathring{\mathfrak{D}}^{\le 3}\psi{})_0](s,\tau{}) + \mathcal{F}_T[(\mathring{\mathfrak{D}}^{\le 3}\psi{})_0](\max  (0,s-1),s,\tau{}) + \mathbf{1}_{s<1}\norm{\partial{}_x\psi{}_0}_{L^\infty(\mathcal{H}(0))}^2  \\
&\qquad + \sup_{\tau{}'\le \tau{}} \int _{\Sigma{}(\tau{}')\cap \mathcal{R}(\max (0,s-1),s)}(P\mathring{\mathfrak{D}}^{\le 2}\psi{})_0^2\tanh x\dd{}\mu{} +  \int _{\mathcal{H}(s,\tau{})\cap \set{x\lesssim 1}} (P\mathring{\mathfrak{D}}^{\le 2}\psi{})_0^2\tanh x \dd{}\mu{}, \\
\end{split}
\end{equation}
and
\begin{equation}
\begin{split}
&\norm{(\cosh x)^{-1/2}\partial{}_x\psi{}}_{L^\infty(\mathcal{H}(s,\tau{}))}^2 \\
&\lesssim \mathcal{E}_T[\mathfrak{D}^{\le 4}\psi{}](s,\tau{}) + \mathcal{F}_T[\mathfrak{D}^{\le 4}\psi{}](\max (0,s-1),s,\tau{}) + \mathbf{1}_{s<1}\norm{(\cosh x)^{-1/2}\partial{}_x\psi{}}_{L^\infty(\mathcal{H}(0))}^2 \\
&\qquad + \sup_{\tau{}'\le \tau{}} \int _{\Sigma{}(\tau{}')\cap \mathcal{R}(\max (0,s-1),s)} \frac{1}{\cosh x}(P\mathfrak{D}^{\le 3}\psi{})^2\tanh x\dd{}\mu{} + \int _{\mathcal{H}(s,\tau{})\cap \set{x\lesssim 1}} (P\mathfrak{D}^{\le 3}\psi{})^2\tanh x \dd{}\mu{},
\end{split}
\end{equation}
and
\begin{equation}
\norm{(\tanh x)^{-1}(\cosh x)^{1/2}\partial_\theta{}\psi}_{L^\infty(\mathcal{H}(s,\tau{}))}^2\lesssim \mathcal{E}_T[\mathfrak{D}^{\le 4}\psi{}](s,\tau{}) +  \int _{\mathcal{H}(s,\tau{})\cap \set{x\lesssim 1}} (P\mathfrak{D}^{\le 3}\psi{})^2\tanh x \dd{}\mu{}.
\end{equation}
Here the implicit constants in the regions of integration depend on the fixed
constant \(C\), and \(\mathring{\mathfrak{D}}\in
\set{\Omega{},\partial_\tau{},\partial_\theta{}}\) and \(\mathfrak{D}\in
\set{\Omega{},\partial{}_\tau{},\partial{}_{\theta{}},\partial{}_{\hat{y}},\partial{}_{\hat{z}}}\)
are commutator vector fields (see \zcref{commutators}).
\label{global-pointwise-s}
\end{proposition}
We now state estimates formulated in terms of the constant-\(\tau{}\) foliation
\(\Sigma{}(\tau{})\).
\begin{proposition}[Global pointwise estimates in terms of the constant-\(\tau\) foliation]
Let \(\tau{}\ge 0\). We have the following pointwise estimates in terms of energies
associated to the constant-\(\tau{}\) foliation:
\begin{equation}\label{Loo-0}
\begin{split}
\norm{\psi_0}_{L^\infty(\Sigma{}(\tau{}))}^2 &\lesssim E_T[(\mathring{\mathfrak{D}}^{\le 3}\psi)_0](\tau{}) + F_p[(\mathring{\mathfrak{D}}^{\le 3}\psi{})_0](0,\tau{}) + \sup_{x\ge 1} \int _{S^2\times S^1} \abs{(\mathring{\mathfrak{D}}^{\le 3}\psi{})_0}^2(\tau{}=0,x)\dd{}\omega{}\dd{}\theta{} \\
&\qquad  + \sup_{\tau{}'\in [0,\tau{}]}\int _{\Sigma{}(\tau{}')} (P\mathring{\mathfrak{D}}^{\le 2}\psi{})_0^2\tanh x\dd{}\mu{},
\end{split}
\end{equation}
and
\begin{equation}\label{Loo-1}
\norm{(\cosh x)^{1/2}\psi{}_{\ge 1}}_{L^\infty(\Sigma{}(\tau{}))}^2 \lesssim E_T[\mathfrak{D}^{\le 3}\psi{}](\tau{}) +  \int _{\Sigma{}(\tau{})\cap \set{x\le 1}}(P\mathfrak{D}^{\le 2}\psi{})^2\tanh x\dd{}\mu{},
\end{equation}
and
\begin{equation}\label{Loo-2a}
\norm{(\cosh x)^{-1/2}\partial{}_x\psi{}}_{L^\infty(\Sigma{}(\tau{}))}^2\lesssim E_T[\mathfrak{D}^{\le 4}\psi{}](\tau{}) + \int _{\Sigma{}(\tau{})}(P\mathfrak{D}^{\le 3}\psi{})^2\tanh x\dd{}\mu{},
\end{equation}
and
\begin{equation}\label{Loo-2b}
\norm{\partial{}_x\psi{}_0}_{L^\infty(\Sigma{}(\tau{}))}^2 \lesssim E_T[(\mathring{\mathfrak{D}}^{\le 3}\psi{})_0](\tau{}) + \int _{\Sigma{}(\tau{})}(P\mathfrak{D}^{\le 2}\psi{})_{0}^2\tanh x\dd{}\mu{}.
\end{equation}
Here \(\mathring{\mathfrak{D}}\in
\set{\Omega{},\partial_\tau{},\partial_\theta{}}\) and \(\mathfrak{D}\in
\set{\Omega{},\partial{}_\tau{},\partial{}_{\theta{}},\partial{}_{\hat{y}},\partial{}_{\hat{z}}}\)
are commutator vector fields (see \zcref{commutators}).
\label{global-pointwise-tau}
\end{proposition}
\begin{remark}[Comparing the decay rates of \(\U(1)\)-symmetric and non-\(\U(1)\)-symmetric quantities]
Note that \(\psi_{\ge 1}\) has improved decay in \(x\) relative to \(\psi_0\), which is
only bounded. It follows that \(\partial_\theta{}\psi{}\) has improved decay in
\(x\) relative to \(\psi\), since \(\partial_\theta{}\psi{} =
\partial_\theta{}\psi_{\ge 1}\) is supported on non-zero \(S^1\) modes. The
reason for the improved decay is that \(\psi_{\ge 1}\) solves an equation with
\((\tanh x)^{-2}(\cosh x)^2\partial_\theta^2\psi_{\ge 1}\), which serves as a
``mass term'' with strong \(x\)-weights, since \(\partial_\theta{}\psi_{\ge 1}\)
controls \(\psi_{\ge 1}\) in \(L^2\) by a Poincaré inequality on \(S^1\). On the
other hand, the equation for \(\psi_0\) does not contain this term.

However, the presence of a \(\partial_\theta^2\)-derivative in the equation for \(\psi_{\ge 1}\)
is also the reason that \(\partial_x\psi_{\ge 1}\) has worse decay in \(x\) than
\(\partial_x\psi_0\), which is bounded. One can interpret this as a
manifestation of the fact that \(\partial_x\) does not commute well with the
wave equation near infinity, due to the ``mass term.'' Indeed, \(\tanh
x\partial_x\) corresponds to the scaling vector field on Minkowski space (that
is, in the Minkowskian coordinates introduced in \zcref{mink-coords}). It is well
known that the scaling vector field does not commute well with the massive
Klein--Gordon operator. Although one can derive pointwise estimates for
\(\partial_x^k\psi_0\) for \(k\ge 2\) (near infinity), we do not need to commute
with scaling to close the proof of global existence for semilinear wave
equations in \zcref{sec:semilinear}.
\label{pointwise-decay-comparison}
\end{remark}
\begin{proof}
This follows from
\zcref{psi0-far,psi0-far-T,psige1-far,dx-far,elliptic-near,elliptic-near-bubble-prep-10,dxpsi0-near}.
\end{proof}
\subsection{Pointwise estimates away from the bubble}
In this section, we prove pointwise estimates away from the bubble, valid in the
region \(\set{x\ge 1}\).
\begin{proposition}[Estimate for \(\psi_0\)]
Let \(\psi{}\in C^\infty(\mathcal{M})\). For \(s\ge 0\) and \(\tau{}\ge 0\), we have
\begin{equation}\label{psi0-far-0}
\norm{\psi_0}_{L^\infty(\mathcal{H}(s,\tau{})\cap \set{x\ge 1})}^2\lesssim \mathcal{E}[\Omega^{\le 2}\psi{}_0](s,\tau{}).
\end{equation}
Here \(\Omega{}\) is a rotation vector field (see \zcref{commutators}).
\label{psi0-far}
\end{proposition}
\begin{proof}
Let \(\psi{}\) be \(\U(1)\)-symmetric, so that in the proof we can write \(\psi{}\) in place
of \(\psi_0\). Let \(x_{\textnormal{min}}\in [1/2,1]\). Let \((s,x_0)\in
\mathcal{H}(s,\tau_0)\) be such that \(x_0\ge 1\). The fundamental theorem of
calculus in the \(X\)-direction, the integrability of \(x^{-2}\) on
\([x_{\textnormal{min}},x_0]\), and Cauchy--Schwarz imply
\begin{equation}
\abs{\Psi{}}^2(s,x_0)\lesssim \abs{\Psi{}}^2(s,x_{\textnormal{min}})+ \int_{x_{\textnormal{min}}}^{x_0} x^{2}(X\Psi{})^2(s,x)\dd{}x.
\end{equation}
Use \(X = L - hT\) and \(\Psi{} = (\tanh x)^{1/2}\psi{}\) and \(h(x)\lesssim \langle{}x\rangle^{-2}\) to get
\begin{equation}
\abs{\psi{}}^2(s,x_0)\lesssim \abs{\psi{}}^2(s,x_{\textnormal{min}})+ \int_{x_{\textnormal{min}}}^{x_0} x^{2}(L\Psi{})^2(s,x) + h(x)(T\psi{})^2(s,x)\dd{}x.
\end{equation}
Integrate over \(S^2\times S^1\), control \(x^{2}\lesssim \sinh x\) (on
\([1,x_0]\)), and use the fact that, for fixed \(s\), the function \(\tau{}\) is
monotone increasing in \(x\) to get
\begin{equation}
\int_{S^2\times S^1}\abs{\psi{}}^2(s,x_0) \lesssim  \int _{S^2\times S^1} \abs{\psi{}}^2(s,x_{\textnormal{min}})\dd{}\omega{}\dd{}\theta{} + \mathcal{E}_T[\psi{}](s,\tau{}_0) + \mathcal{E}_p[\psi{}](s,\tau{}_0).
\end{equation}
After averaging over \(x_{\textnormal{min}}\in [1/2,1]\) and using the
\(L^2\)-control of \(\psi{}\) in the \(\mathcal{E}_T\)-energy, we obtain
\begin{equation}
\int_{S^2\times S^1}\abs{\psi{}}^2(s,x_0) \lesssim \mathcal{E}[\psi{}](s,\tau{}_0).
\end{equation}
In view of the \(\U(1)\)-symmetry of \(\psi{}\), Sobolev embedding on \(S^2\) completes the
proof.
\end{proof}
\begin{proposition}[Estimate for \(\psi_0\) in terms of the constant-\(\tau\) foliation]
Let \(\psi{}\in C^\infty(\mathcal{M})\). For \(\tau{}_0\ge 0\), we have
\begin{equation}
\begin{split}
\norm{\psi_0}_{L^\infty(\Sigma{}(\tau{}_0)\cap \set{x\ge 1})}^2 &\lesssim \sup_{x\ge 1}\int _{S^2\times S^1}\abs{\mathring{\mathfrak{D}}^{\le 2}\psi{}}^2(\tau{}=0,x)\dd{}\omega{}\dd{}\theta{} + \sup_{\tau{}\in [0,\tau{}_0]}\int _{S^2\times S^1}\abs{\mathring{\mathfrak{D}}^{\le 2}\psi{}}^2(\tau{},x=1)\dd{}\omega{}\dd{}\theta{} \\
&\qquad + F_p[\mathring{\mathfrak{D}}^{\le 2}\psi{}](0,\tau{}_0).
\end{split}
\end{equation}
Here
\(\mathring{\mathfrak{D}}\) is a commutator vector field (see \zcref{commutators}).
\label{psi0-far-T}
\end{proposition}
\begin{proof}
Let \(\psi{}\) be \(\U(1)\)-symmetric, so that in the proof we can write \(\psi{}\) in place of
\(\psi_0\) and that we can suppress the \(\theta{}\) variable. Fix \((\tau{}_0,x_0,\omega{}_0)\), with \(u\)-value \(u_0\), such that \(x_0\ge
1\). We integrate backwards in the \(L\)-direction along \(C(u_0)\), to the past
boundary of \(\set{0\le \tau{}\le \tau_0}\cap \set{x\ge 1}\), namely \(\mathfrak{B}\). This gives
\begin{equation}
\abs{\Psi{}}(\tau{}_0,x_0,\omega{}_0)\le \abs{\Psi{}}(\tau{}_p,x_p,\omega{}_0)+ \int _{C(u_0)\cap \set{(\omega{},\theta{}) = (\omega{}_0,\theta{}_0)}\cap \set{0\le \tau{}\le \tau{}_0}\cap \set{x\ge 1}} \abs{L\Psi{}}\dd{}v
\end{equation}
for some \((\tau_p,x_p)\in \mathfrak{B}\coloneqq{} (\Sigma{}(0)\cap \set{x\ge 1})\cup (\set{x=1}\cap \set{0\le \tau{}\le \tau_0})\). Control the integral using Cauchy--Schwarz
and the integrability of \(x^{-2}\) on \(C(u_0)\cap \set{x\ge 1}\)
and then square both sides to get
\begin{equation}
\abs{\Psi{}}^2(\tau{}_0,x_0,\omega{}_0)\lesssim \abs{\Psi{}}^2(\tau{}_p,x_p,\omega{}_0)+ \int _{C(u_0)\cap \set{(\omega{},\theta{}) = (\omega{}_0,\theta{}_0)}\cap \set{0\le \tau{}\le \tau{}_0}\cap \set{x\ge 1}} x^2\abs{L\Psi{}}^2\dd{}v.
\end{equation}
Integrate over \(S^2\times S^1\) to get
\begin{equation}
\int _{S^2\times S^1}\abs{\Psi{}}^2(\tau{}_0,x_0)\dd{}\omega{}\dd{}\theta{}\lesssim \int _{S^2\times S^1}\abs{\Psi{}}^2(\tau{}_p,x_p)\dd{}\omega{}\dd{}\theta{} + \int _{C(u_0)\cap \set{0\le \tau{}\le \tau{}_0}\cap \set{x\ge 1}} x^2\abs{L\Psi{}}^2\dd{}v\dd{}\omega{}\dd{}\theta{}.
\end{equation}
The final term is controlled by \(F_p[\psi{}](0,\tau{})\), since \(x^2\lesssim \sinh x\)
in the region of interest. Using also \(\tanh x\sim 1\) in the region of
interest and the definition of \(\mathfrak{B}\), we get
\begin{equation}\label{psi0-far-T-prep}
\begin{split}
\sup_{x\ge 1}\int _{S^2\times S^1}\abs{\psi{}}^2(\tau{}_0,x)\dd{}\omega{}\dd{}\theta{}&\lesssim \sup_{x\ge 1}\int _{S^2\times S^1}\abs{\psi{}}^2(\tau{}=0,x)\dd{}\omega{}\dd{}\theta{} + \sup_{\tau{}\in [0,\tau{}_0]}\int _{S^2\times S^1}\abs{\psi{}}^2(\tau{},x=1)\dd{}\omega{}\dd{}\theta{}\\
&\qquad  + F_p[\psi{}](0,\tau{}_0).
\end{split}
\end{equation}
Now Sobolev embedding on \(S^2\) and the \(\U(1)\)-symmetry of \(\psi{}\) complete the
proof.
\end{proof}
\begin{proposition}[Improved decay in space for \(\psi_{\ge 1}\)]
Let \(\psi{}\in C^\infty(\mathcal{M})\). For \(s\ge 0\) and \(\tau{}\ge 0\), we have
\begin{equation}
\norm{(\sinh x)^{1/2}\psi{}_{\ge 1}}_{L^\infty(\mathcal{H}(s,\tau{}))}^2\lesssim \mathcal{E}_T[\mathring{\mathfrak{D}}^{\le 3}\psi{}_{\ge 1}](s,\tau{}).
\end{equation}
Here \(\mathring{\mathfrak{D}}\in \set{\Omega{},\partial{}_{\theta{}}}\)
is a commutator vector field (see \zcref{commutators}). The same estimate holds with
\(\mathcal{H}(s,\tau{})\) replaced by \(\Sigma{}(\tau{})\) and \(\mathcal{E}_T[\cdot ](s,\tau{})\)
replaced with \(E_T[\cdot ](\tau{})\).
\label{psige1-far}
\end{proposition}
\begin{proof}
We give the proof for the hypersurface \(\mathcal{H}(s,\tau{})\), since the proof for
\(\Sigma{}(\tau{})\) is the same (after replacing \(X\) with \(\partial_x\)). By
the Poincaré inequality
\begin{equation}
\int _{S^1}(\partial{}_\theta{}\psi{}_{\ge 1})^2\dd{}\theta{}\ge \int _{S^1}\psi{}_{\ge 1}^2\dd{}\theta{},
\end{equation}
it follows from \(X = L - hT\) that
\begin{equation}\label{psi-ge1-energy}
\int _{\mathcal{H}(s,\tau{})} \Bigl[(X\psi{}_{\ge 1})^2 + \frac{\cosh ^2x}{\tanh ^2x}\psi{}_{\ge 1}^2\Bigr]\tanh x\dd{}\mu{}\lesssim \mathcal{E}_T[\psi{}_{\ge 1}](s,\tau{}).
\end{equation}

Next, compute
\begin{equation}
\abs{X(\sinh  x\psi{}_{\ge 1}^2)}\lesssim \cosh x \psi{}_{\ge 1}^2 + \sinh x \abs{\psi{}_{\ge 1}}\abs{X\psi{}_{\ge 1}} \lesssim  ((X\psi{}_{\ge 1})^2 + \cosh ^2x\psi{}_{\ge 1}^2)\tanh x
\end{equation}
Integrating in the \(X\)-direction to \(\set{x=0}\) (where \(\sinh x \psi_{\ge 1}^2\)
vanishes since \(\psi_{\ge 1}\) is bounded), we obtain
\begin{equation}\label{psi-ge1-pointwise}
\sinh x \psi{}_{\ge 1}^2(s,x_0)\lesssim \int_0^x  ((X\psi{}_{\ge 1})^2 + \cosh ^2x\psi{}_{\ge 1}^2)\tanh x\dd{}x
\end{equation}
To complete the proof, integrate \zcref{psi-ge1-pointwise} over \(S^2\times S^1\) and use
\zcref{psi-ge1-energy}, Sobolev embedding on \(S^2\times S^1\), and the fact that
\(\tau{}\) is monotone increasing in \(x\) for fixed \(s\).
\end{proof}
\begin{proposition}[Estimates for \(\partial_x\psi\)]
Let \(\psi{}\in C^\infty(\mathcal{M})\). For \(s\ge 0\) and \(\tau{}\ge 0\), we have
\begin{equation}\label{dx-bound-general}
\begin{split}
\norm{(\cosh x)^{-1/2}\partial{}_x\psi{}}_{L^\infty(\mathcal{H}(s,\tau{})\cap \set{x\ge 1})}^2&\lesssim \mathcal{F}_T[\mathring{\mathfrak{D}}^{\le 4}\psi{}](\max (0,s-1),s,\tau{}) + \mathbf{1}_{s<1}\norm{(\cosh x)^{-1/2}\partial{}_x\psi{}}_{L^\infty(\mathcal{H}(0)\cap \set{x\ge 1})}^2 \\
&\qquad + \sup_{\tau{}'\le \tau{}} \int _{\Sigma{}(\tau{}')\cap \mathcal{R}(\max (0,s-1),s)} \frac{1}{\cosh x}(P\mathring{\mathfrak{D}}^{\le 3}\psi{})^2\tanh x\dd{}\mu{}
\end{split}
\end{equation}
and
\begin{equation}\label{dx-bound-axisymmetric}
\begin{split}
\norm{\partial{}_x\psi{}_0}_{L^\infty(\mathcal{H}(s,\tau{})\cap \set{x\ge 1})}^2&\lesssim \mathcal{F}_T[(\mathring{\mathfrak{D}}^{\le 3}\psi{})_0](\max (0,s-1),s,\tau{}) + \mathbf{1}_{s<1}\norm{\partial{}_x\psi{}_0}_{L^\infty(\mathcal{H}(0)\cap \set{x\ge 1})}^2 \\
&\qquad + \sup_{\tau{}'\le \tau{}} \int _{\Sigma{}(\tau{}')\cap \mathcal{R}(\max (0,s-1),s)} (P\mathring{\mathfrak{D}}^{\le 2}\psi{})_0^2\tanh x\dd{}\mu{}.
\end{split}
\end{equation}
Moreover, we have
\begin{equation}\label{dx-bound-tau}
\norm{(\cosh x)^{-1/2}\partial{}_x\psi{}}_{L^\infty(\Sigma{}(\tau{})\cap \set{x\ge 1})}^2\lesssim E_T[\mathring{\mathfrak{D}}^{\le 4}\psi{}](\tau{}) + \int _{\Sigma{}(\tau{})} \frac{1}{\cosh x}(P\mathring{\mathfrak{D}}^{\le 3}\psi{})^2\tanh x\dd{}\mu{}
\end{equation}
and
\begin{equation}\label{dx-bound-tau-2}
\norm{\partial{}_x\psi{}_0}_{L^\infty(\Sigma{}(\tau{})\cap \set{x\ge 1})}^2\lesssim E_T[(\mathring{\mathfrak{D}}^{\le 3}\psi{})_0](\tau{}) + \int _{\Sigma{}(\tau{})} (P\mathring{\mathfrak{D}}^{\le 2}\psi{})_0^2\tanh x\dd{}\mu{}.
\end{equation}
Here \(\mathring{\mathfrak{D}}\in \set{\Omega{},\partial{}_\tau{},\partial{}_{\theta{}}}\) is a commutator vector field (see
\zcref{commutators}).
\label{dx-far}
\end{proposition}
\begin{remark}
Observe that \(\partial_x\psi_{\ge 1}\) has worse decay in space than \(\partial_x\psi_0\) (see
\zcref{pointwise-decay-comparison}).
\end{remark}
\begin{proof}
We will prove \zcref{dx-bound-general,dx-bound-axisymmetric}, since the proof of
\zcref{dx-bound-tau,dx-bound-tau-2} is similar (in fact, simpler). Rearrange
\zcref{P-x-expression} to obtain
\begin{equation}
\partial{}_x^2\psi{} = P\psi{} + \partial{}_\tau^2\psi{} - \frac{1}{\cosh x\sinh x}\partial{}_x\psi{} + \frac{1}{\cosh^2x}\psi{} - \frac{1}{\cosh ^2\tau{}}\Lapl _{S^2}\psi{} -  \frac{\cosh^2x}{\tanh ^2x}\partial{}_\theta^2\psi{},
\end{equation}
It follows that, in the region \(\set{x\ge 1}\), we have
\begin{equation}\label{dx-dx-psi-estimate}
\begin{split}
&\abs[\Big]{\partial{}_x\Bigl(\frac{1}{\cosh x}(\partial{}_x\psi{})^2\Bigr)} \lesssim \frac{1}{\cosh x}(\partial{}_x\psi{})^2 + \frac{1}{\cosh x}\abs{\partial{}_x\psi{}}\abs{\partial{}_x^2\psi{}} \\
&\lesssim (\partial{}_x\psi{})^2 + \frac{1}{\cosh x}(P\psi{})^2 + (\partial{}_\tau \partial{}_\tau{}\psi{})^2 + \frac{1}{\cosh ^2x}\psi^2 + \frac{1}{\cosh ^2\tau{}}(\Lapl _{S^2}\psi{})^2 + \frac{\cosh ^2x}{\tanh^4 x}(\partial{}_\theta^2\psi{})^2.
\end{split}
\end{equation}
Note that we have used the weighted Young's inequality
\begin{equation}
\frac{1}{\cosh x}\abs{\partial{}_x\psi{}} \cdot \frac{\cosh ^2x}{\tanh ^2x}\abs{\partial{}_\theta^2\psi{}}\lesssim (\partial{}_x\psi{})^2 + \frac{\cosh ^2x}{\tanh ^4x}(\partial{}_\theta^2\psi{})^2.
\end{equation}
Now fix \(s_0\ge 1\) and \(x_0\ge 1\). Let \(\tau_0\) be the \(\tau{}\)-value associated to
the points in \(\set{s=s_0}\cap \set{x=x_0}\). Let \(x_1\) be the \(x\)-value
associated to the points in \(\set{\tau{}=\tau_0}\cap \set{s=s_0-1}\). By
\zcref{tau-def} and the positivity of \(h\), we have
\begin{equation}\label{x1-x0-separation}
x_1 - x_0 = 2 + \int_{x_0}^{x_1} h(y)\dd{}y\ge 2.
\end{equation}
Let \(\chi{}(x)\) be a bump function such that \(\chi{}(x_0) = 1\) and \(\chi{}(x_1) = 0\). By
\zcref{x1-x0-separation}, \(\chi{}\) can be chosen so that \(\abs{\chi{}'}\lesssim
1\). It now follows from \zcref{dx-dx-psi-estimate} that
\begin{equation}
\begin{split}
&\int _{S^2\times S^1} \frac{1}{\cosh x}(\partial{}_x\psi{})^2(\tau{}_0,x_0,\omega{},\theta{})\dd{}\omega{}\dd{}\theta{}\\
&\le \int _{\Sigma{}(\tau{}_0)\cap \set{x_0\le x\le x_1}} \abs[\Big]{\partial{}_x\Bigl(\chi{}\frac{1}{\cosh x}(\partial{}_x\psi{})^2(\tau{}_0,x,\omega{},\theta{})\Bigr)}\dd{}\mu{} \\
&\le \int _{\Sigma{}(\tau{}_0)\cap \set{s_0-1\le s\le s_0}} \side{RHS}{dx-dx-psi-estimate}\dd{}\mu{} \\
&\lesssim \mathcal{F}_T[\mathfrak{D}^{\le 1}\psi{}](s_0-1,s_0,\tau{}_0) + \int _{\Sigma{}(\tau{}_0)\cap \set{s_0-1\le s\le s_0}} \frac{1}{\cosh x}(P\psi{})^2\dd{}\mu{}.
\end{split}
\end{equation}
Sobolev embedding on \(S^2\times S^1\) completes the proof of \zcref{dx-bound-general} in
the case \(s\ge 1\). If \(0\le s < 1\), then instead of introducing the cutoff
function \(\chi{}\), we simply integrate in the \(\partial_x\)-direction to
\(\mathcal{H}(0)\), which produces a boundary term at \(\mathcal{H}(0)\). If
\(\psi{}\) is \(\U(1)\)-symmetric, then \(\partial_\theta{}\psi{} = 0\), so the above
proof works without the weight \((\cosh x)^{-1}\). In this case, only Sobolev
embedding on \(S^2\) is required.
\end{proof}
\subsection{Pointwise estimates near the bubble}
\subsubsection{Estimates for general scalar fields}
In this section, we prove pointwise estimates in a finite-\(x\) region, in
particular near the bubble.
\begin{proposition}[Estimates for \(\psi\) and its derivatives near the bubble]
Let \(\psi{}\in C^\infty(\mathcal{M})\), and let \(s\ge 0\) and \(\tau{}\ge 0\). There is a universal
constant \(C > 0\) (depending only on the function \(h\) defining the
hyperboloidal foliation introduced in \zcref{hyperboloidal-foliation}) such that if
\(\tau{}-2s \ge C\), then for \(1\le x_0\le \frac{1}{8}(\tau{}-2s)\), then we have
\begin{equation}\label{elliptic-near-1}
\norm{\psi{}}_{L^\infty(\mathcal{H}(s,\tau{})\cap \set{x\le x_0})}^2 \lesssim_{x_0} \mathcal{E}_T[\mathring{\mathfrak{D}}^{\le 3}\psi{}](s,\tau{}) + \int _{\mathcal{H}(s,\tau{})\cap \set{x\le 2x_0}} (P\mathring{\mathfrak{D}}^{\le 2}\psi{})^2\tanh x \dd{}\mu{},
\end{equation}
where the dependence on \(x_0\) is double-exponential, and
\(\mathring{\mathfrak{D}}\in
\set{\Omega{},\partial{}_\tau{},\partial{}_{\theta{}}}\) is a commutator vector
field (see \zcref{commutators}). We also have
\begin{equation}\label{elliptic-near-bubble-equation-dx-dtheta}
\begin{split}
&\norm{\partial{}_x\psi{}}_{L^\infty(\mathcal{H}(s,\tau{})\cap \set{x\le x_0})}^2 + \norm{(\tanh x)^{-1}\partial{}_\theta{}\psi{}}_{L^\infty(\mathcal{H}(s,\tau{})\cap \set{x\le x_0})}^2 \\
&\lesssim_{x_0} \mathcal{E}_T[\mathfrak{D}^{\le 4}\psi{}](s,\tau{}) + \int _{\mathcal{H}(s,\tau{})\cap \set{x\le 2x_0}} (P\mathfrak{D}^{\le 3}\psi{})^2\tanh x \dd{}\mu{},
\end{split}
\end{equation}
where now \(\mathfrak{D}\in \set{\Omega{},\partial_\tau{},\partial{}_{\theta{}},\partial{}_{\hat{y}},\partial{}_{\hat{z}}}\).

Moreover, the same estimates hold with \(\mathcal{H}(s,\tau{})\) replaced by \(\Sigma{}(\tau{})\)
and \(\mathcal{E}_T[\cdot ](s,\tau{})\) replaced by \(E_T[\cdot ](\tau{})\) (now
without any restriction on the relative size of \(\tau{}\) and \(s\), since the
role of \(s\) has been removed).
\label{elliptic-near}
\end{proposition}
\begin{remark}
The role of the lower bound on \(\tau{}-2s\) is only to ensure that \(\mathcal{H}(s)\cap \set{x\le 4x_0}\subset \mathcal{H}(s,\tau{})\).
\end{remark}
\begin{proof}
We will only prove the estimates for the hyperboloidal foliation
\(\mathcal{H}(s,\tau{})\), since the proof is similar for the foliation \(\Sigma{}(\tau{})\).

\step{Step 0: Choice of \(x_0\).} Let \(x_{\textnormal{max}}\) be the largest possible
\(x\)-value in \(\mathcal{H}(s,\tau{})\). By \zcref{tau-def} and the positivity of \(h\),
we have
\begin{equation}
x_{\textnormal{max}} = \tau{}-2s - \int_{x_{\textnormal{max}}}^\infty h(y)\dd{}y\ge \tau{} - 2s - C_0,
\end{equation}
where \(C_0 \coloneqq{} \int_{0}^\infty h(y)\dd{}y\). It follows that if \(\tau{}-2s\ge 2C_0\), then we
have \(x_{\textnormal{max}}\ge \frac{1}{2}(\tau{}-2s)\). In particular, if
\(x_0\le \frac{1}{8}(\tau{}-2s)\), then \(\mathcal{H}(s)\cap \set{x\le 4x_0}\subset
\mathcal{H}(s,\tau{})\).

\step{Step 1: Proof of \zcref{elliptic-near-1}.} In this step, we prove
\zcref{elliptic-near-1}. The proof proceeds by separating out the
\(\partial_\tau{}\)-derivatives in the wave equation and using the ellipticity
properties of the remaining operator, which consists of spatial derivatives.

\step{Step 1a: An identity for the spatial part of the wave operator.}
Recalling the expression \zcref{P-cart-expression} for the wave equation in
Cartesian coordinates, define the operator
\begin{equation}\label{calP-def}
\mathcal{P}\psi{}\coloneqq{} \frac{1}{\cosh ^2\tau{}}\Lapl _{S^2}\psi{} + \mathfrak{h}(x)(\partial{}_{\hat{y}}^2\psi{} + \partial{}_{\hat{z}}^2\psi{}),\qquad \mathfrak{h}(x) \coloneqq{} \frac{\cosh ^2x}{\mathfrak{F}(x)^2} =  e^{2\cosh x} \cosh ^2x\frac{\tanh^2 (x/2)}{\tanh ^2x},
\end{equation}
where \(\mathfrak{F}(x)\) was defined in \zcref{cart-metric}. Then
\begin{equation}\label{elliptic-calP-expression}
\mathcal{P}\psi{} = P\psi{} + \partial_\tau^2\psi{} + \frac{1}{\cosh ^2x}\psi{} + O(1) \frac{\mathfrak{h}(x)}{\mathfrak{f}'(x)} \widehat{\partial{}}\psi{},
\end{equation}
where \(\mathfrak{f}\) was defined in \zcref{cart-def} and we have schematically
written \(\widehat{\partial{}}\in \set{\partial{}_{\hat{y}},\partial{}_{\hat{z}}}\).

Let \(\chi{} = \chi{}(x)\) be a bump function such that \(\chi{}\equiv 1\) on \(\set{x\le x_0/2}\) and
\(\Supp \chi{}\subset \set{x\le x_0}\). Then
\begin{equation}
\begin{split}
\chi^2(\mathcal{P}\psi{})^2 &= \textnormal{(I)} + \textnormal{(II)} + \textnormal{(III)},
\end{split}
\end{equation}
where
\begin{equation}
\begin{split}
\textnormal{(I)} &\coloneqq{}\chi^2 \frac{1}{\cosh ^4\tau{}}(\Lapl _{S^2}\psi{})^2 + \chi^2\mathfrak{h}^2(\partial{}_{\hat{y}}^2\psi{})^2 + \chi^2H^2(\partial{}_{\hat{z}}^2\psi{})^2, \\
\textnormal{(II)} &\coloneqq{} \chi^2 \frac{2}{\cosh ^2\tau{}}\Lapl _{S^2}\psi{}\cdot \mathfrak{h}\partial{}_{\hat{y}}^2\psi{} + \chi^2 \frac{2}{\cosh ^2\tau{}}\Lapl _{S^2}\psi{}\cdot \mathfrak{h}H\partial{}_{\hat{z}}^2\psi{}, \\
\textnormal{(III)} &\coloneqq{} 2\chi^2\mathfrak{h}^2\partial{}_{\hat{y}}^2\psi{}\partial{}_{\hat{z}}^2\psi{}. \\
\end{split}
\end{equation}
Term \(\textnormal{(I)}\) is positive definite, so we leave it as is. We
differentiate term \(\textnormal{(II)}\) by parts twice:
\begin{equation}
\begin{split}
\textnormal{(II)} &= \mathfrak{h}\chi^2 \frac{2}{\cosh ^2\tau{}}(\abs{\Grad _{S^2}\partial{}_{\hat{y}}\psi{}}^2 + \abs{\Grad _{S^2}\partial{}_{\hat{z}}^2\psi{}}^2) + \widehat{\partial{}}(\cdots{}) + \div_{S^2}(\cdots{}) + O(1) \widehat{\partial{}}(\chi{}\mathfrak{h})\widehat{\partial{}}\psi{}\cdot \chi{} \frac{1}{\cosh ^2\tau{}}\Lapl _{S^2}\psi{}.
\end{split}
\end{equation}
We also differentiate term \(\textnormal{(III)}\) by parts twice:
\begin{equation}
\textnormal{(III)} = 2\chi^2H^2(\partial{}_{\hat{y}}\partial{}_{\hat{z}}\psi{})^2 + \widehat{\partial{}}(\cdots{}) + O(1)\widehat{\partial{}}(\chi{}\mathfrak{h})\widehat{\partial{}}\psi{}\cdot \chi{}\mathfrak{h}\widehat{\partial{}}^2\psi{}
\end{equation}
After combining the expressions for terms
\(\textnormal{(I)}\)--\(\textnormal{(III)}\) and using Young's inequality, we
conclude that
\begin{equation}\label{elliptic-prep-0}
\chi^2 \frac{1}{\cosh ^4\tau{}}(\Lapl _{S^2}\psi{})^2 + \chi^2\mathfrak{h}^2[(\partial{}_{\hat{y}}^2\psi{})^2 + (\partial{}_{\hat{z}}^2\psi{})^2 + (\partial{}_{\hat{y}}\partial{}_{\hat{z}}\psi{})^2]\le \widehat{\partial{}}(\cdots{}) + \div_{S^2}(\cdots{}) + C[\chi^2(\mathcal{P}\psi{})^2 + \abs{\widehat{\partial{}}(\chi{}\mathfrak{h})}^2 (\widehat{\partial{}}\psi{})^2]
\end{equation}
for some constant \(C > 0\).

\step{Step 1b: The pointwise estimate.} We now estimate the term
\(\widehat{\partial{}}(\chi{}\mathfrak{h})\) appearing on the right side of
\zcref{elliptic-prep-0}. From \zcref{hat-derivatives}, we have \(\widehat{\partial{}}x =
O(1)\frac{1}{\mathfrak{f}'(x)}\), and so
\begin{equation}\label{elliptic-dhat-chi-H-calc}
\abs{\widehat{\partial{}}(\chi{}\mathfrak{h})}\lesssim \frac{1}{\mathfrak{f}'(x)}(\chi{}\abs{\mathfrak{h}'} + \mathfrak{h}\abs{\chi{}'})\lesssim \mathbf{1}_{\Supp \chi{}} \frac{\mathfrak{h}(x)}{\mathfrak{f}'(x)}\cosh x,
\end{equation}
where \('\) refers to \(\partial_x\), and in passing to the last estimate we used the
calculation \(\abs{\mathfrak{h}'}\lesssim \mathfrak{h}\cosh x\) (which follows from \zcref{calP-def}). We now
compute
\begin{equation}\label{elliptic-f-prime-calc}
\mathfrak{f}'(x) = \sinh xe^{\cosh x}\tanh (x/2) + \frac{1}{2} \frac{1}{\cosh ^2(x/2)}e^{\cosh x}\sim e^{\cosh x}\cosh x.
\end{equation}
Using \zcref{elliptic-dhat-chi-H-calc,elliptic-f-prime-calc}, we obtain
\begin{equation}\label{elliptic-dx-H-calc-final}
\abs{\widehat{\partial{}}(\chi{}\mathfrak{h})}\lesssim \mathbf{1}_{\Supp \chi{}}e^{\cosh x} \cosh ^2x.
\end{equation}
By combining
\zcref{elliptic-calP-expression,elliptic-prep-0,elliptic-dx-H-calc-final}
we conclude that
\begin{equation}\label{elliptic-pre-estimate}
\begin{split}
&\int _{\mathcal{H}(s,\tau{})} \chi^2\mathfrak{h}^2[\psi^2 + (\partial{}_{\hat{y}}^2\psi{})^2 + (\partial{}_{\hat{z}}^2\psi{})^2 + (\partial{}_{\hat{y}}\partial{}_{\hat{z}}\psi{})^2]\dd{}\hat{y}\dd{}\hat{z}\dd{}\omega{}  \\
&\qquad \lesssim \int _{\mathcal{H}(s,\tau{})\cap \set{x\le x_0}} (P\psi{})^2 + (\partial{}_\tau \partial{}_\tau{}\psi{})^2 + \frac{1}{\cosh ^4x}\psi^2 + e^{\cosh x}\cosh ^2x(\widehat{\partial{}}\psi{})^2\dd{}\hat{y}\dd{}\hat{z}\dd{}\omega{}.
\end{split}
\end{equation}
We now estimate the right side of \zcref{elliptic-pre-estimate} in terms of energies
of \(\psi{}\). We first claim that
\begin{equation}\label{elliptic-dhat-psi-estimate}
(\widehat{\partial{}}\psi{})^2\dd{}\hat{y}\dd{}\hat{z} = \Bigl[\frac{\mathfrak{f}(x)}{\mathfrak{f}'(x)}(\partial{}_x\psi{})^2 + \frac{\mathfrak{f}'(x)}{\mathfrak{f}(x)}(\partial{}_\theta{}\psi{})^2\Bigr]\dd{}x\dd{}\theta{}\sim \Bigl[\frac{\tanh x}{\cosh x}(\partial_x\psi{})^2 + \frac{\cosh x}{\tanh x}(\partial{}_\theta{}\psi{})^2\Bigr]\dd{}x\dd{}\theta{}.
\end{equation}
Indeed, this follows from the identity
\begin{equation}
\dd{}\hat{y}\dd{}\hat{z} = \mathfrak{f}(x)\mathfrak{f}'(x)\dd{}x\dd{}\theta{},
\end{equation}
which follows from \zcref{cart-def}, and the following computation using \zcref{hat-derivatives}:
\begin{equation}\label{dhat-squared-pointwise-bound}
(\widehat{\partial{}}\psi{})^2\lesssim \frac{1}{\mathfrak{f}'(x)^2}(\partial{}_x\psi{})^2 + \frac{1}{\mathfrak{f}(x)^2}(\partial{}_\theta{}\psi{})^2.
\end{equation}
It now follows from \zcref{elliptic-pre-estimate,elliptic-dhat-psi-estimate} that
\begin{equation}\label{dhat-squared-bound}
\begin{split}
&\int _{\mathcal{H}(s,\tau{})} \chi^2\mathfrak{h}^2[\psi^2 + (\partial{}_{\hat{y}}^2\psi{})^2 + (\partial{}_{\hat{z}}^2\psi{})^2 + (\partial{}_{\hat{y}}\partial{}_{\hat{z}}\psi{})^2]\dd{}\hat{y}\dd{}\hat{z}\dd{}\omega{}  \\
&\lesssim e^{3\cosh x_0}\cosh x_0\Bigl(\mathcal{E}_T[\partial{}_\tau^{\le 1}\psi{}](s,\tau{}) + \int _{\mathcal{H}(s,\tau{})} (P\psi{})^2\tanh x\dd{}\mu{}\Bigr).
\end{split}
\end{equation}
We conclude by Sobolev embedding on \(\R^2\) (in \((\hat{y},\hat{z})\)
coordinates) that
\begin{equation}\label{elliptic-near-bubble-prep-10}
\sup_{(x,\theta{})\in \mathcal{H}(s,\tau{})\cap \set{x\le x_0/2}}\int _{S^2}\abs{\psi{}}^2(\tau{},x,\omega{},\theta{})\dd{}\omega{}\lesssim_{x_0} \mathcal{E}_T[\partial{}_\tau^{\le 1}\psi{}](s) + \int _{\mathcal{H}(s,\tau{})} (P\psi{})^2\tanh x \dd{}\mu{},
\end{equation}
To obtain \zcref{elliptic-near-1} from
\zcref{elliptic-near-bubble-prep-10}, use Sobolev embedding on \(S^2\).

\step{Step 2: Proof of
\zcref{elliptic-near-bubble-equation-dx-dtheta}.} Apply
the result of Step 1 to \(\widehat{\partial}\psi{}\) in place of \(\psi{}\)
(where \(\widehat{\partial}\in
\set{\partial{}_{\hat{y}},\partial{}_{\hat{z}}}\)) and use the formula
\zcref{dx-in-terms-of-dxhat} expressing \(\partial_x\) and \(\partial_\theta{}\) in
terms of \(\widehat{\partial}\), noting in particular that \(\mathfrak{f}(x)\)
and \(\mathfrak{f}'(x)\) are bounded in a finite-\(x\) region and
\(\mathfrak{f}(x)\sim \tanh x\) near \(\set{x=0}\).
\end{proof}
\subsubsection{Estimates for \texorpdfstring{\(\U(1)\)}{U(1)}--symmetric scalar fields}
In this section, we prove an estimate for \(\partial_x\psi_0\) near the bubble. This
estimate will be important for our proof of global existence for semilinear
equations in \(\U(1)\)-symmetry in \zcref{sec:semilinear-axisymmetric}, because it only uses
the vector fields \(\mathring{\mathfrak{D}}\), and not the more general set
\(\mathfrak{D}\) as in \zcref{elliptic-near}.
\begin{proposition}[Estimates for \(\partial_x\psi_0\) near the bubble]
Let \(\psi{}\in C^\infty(\mathcal{M})\). We have
\begin{equation}
\norm{\partial_x\psi_0}_{L^\infty(\mathcal{H}(s)\cap \set{x\le 1})}^2\lesssim \mathcal{E}_T[(\mathring{\mathfrak{D}}^{\le 3}\psi{})_0](s)  +\int _{\mathcal{H}(s)\cap \set{x\le 1}}\abs{(P\mathring{\mathfrak{D}}^{\le 2}\psi{})_0}^2 \tanh x\dd{}\mu{}
\end{equation}
where \(\mathring{\mathfrak{D}}\in \set{\partial_\tau{},\Omega{}}\) is a commutator vector field (see
\zcref{commutators}). Moreover, the same estimate holds with
\(\mathcal{E}_T[\cdot ](s)\) replaced by \(E_T[\cdot ](\tau{})\) and
\(\mathcal{H}(s)\) replaced with \(\Sigma{}(\tau{})\).
\label{dxpsi0-near}
\end{proposition}
\begin{proof}
First, note that \(\mathcal{H}(s)\) and \(\Sigma{}(\tau{})\) are both constant-\(\tau{}\)
hypersurfaces in \(\set{x\le 1}\), so the second statement follows from the
proof of the first. Fix \(s_0\ge 1\) and \(x_0\le 1\). Let \(\tau_0\) be the
\(\tau{}\)-value associated to the points in \(\set{s=s_0}\cap \set{x=x_0}\).
Since \(x_0\le 1\), the set \(\set{0\le x\le x_0}\cap \set{\tau{}=\tau_0}\) is
contained in \(\mathcal{H}(s_0)\) (see \zcref{conditions-on-h-at-0} of
\zcref{hyperboloidal-foliation}).

Rewrite the wave equation \zcref{P-x-expression} for a \(\U(1)\)-symmetric scalar field as
\begin{equation}
\partial{}_x(\tanh x\partial{}_x\psi{}_0) = \Bigl[(P\psi{})_0 + \partial{}_\tau^2\psi_0 + \frac{1}{\cosh ^2x}\psi{}_0 - \frac{1}{\cosh ^2\tau{}}\Lapl _{S^2}\psi{}_0\Bigr]\tanh x.
\end{equation}
Since \(\partial_x\psi_0\) and each term in the brackets on the right-hand side are
bounded as \(x\to 0\), we can integrate in the \(\partial_x\)-direction from \(\set{x=x_0}\) to \(\set{x=0}\),
where \(\tanh x\) vanishes, without boundary term and then use Cauchy--Schwarz
to obtain
\begin{equation}
\begin{split}
&\tanh x_0\abs{\partial{}_x\psi{}_0}(\tau{}_0,x=x_0)\\
&\le \Bigl(\int_0^{x_0} \tanh x\dd{}x\Bigr)^{1/2}\Bigl(\int_0^{x_0}\Bigl[\abs{(P\psi{})_0}^2 + \abs{\partial{}_\tau^2\psi{}_0}^2 + \frac{1}{\cosh ^2x}\abs{\psi{}_0}^2 + \frac{1}{\cosh ^2\tau{}}\abs{\Lapl _{S^2}\psi{}_0}^2\Bigr](\tau{},x)\tanh x\dd{}x\Bigr)^{1/2}.
\end{split}
\end{equation}
After squaring both sides and dividing by
\(\tanh^2 x_0\) (noting that \(x_0\le 1\), so that the integral of \(\tanh x\)
on \([0,x_0]\) is like \(\tanh^2 x_0\)), we obtain
\begin{equation}
\abs{\partial{}_x\psi{}_0}^2(\tau{}_0,x=x_0)\lesssim \int_0^{x_0}\Bigl[\abs{(P\psi{})_0}^2 + \abs{\partial{}_\tau^2\psi{}_0}^2 + \frac{1}{\cosh ^2x}\abs{\psi{}_0}^2 + \frac{1}{\cosh ^2\tau{}}\abs{\Lapl _{S^2}\psi{}_0}^2\Bigr](\tau{}_0,x)\tanh x\dd{}x.
\end{equation}
Integrate over \(S^2\times S^1\) and recall the definition of \(\tau_0\) to obtain
\begin{equation}
\begin{split}
&\int _{S^2\times S^1} \abs{\partial{}_x\psi{}_0}^2(\tau{}=\tau{}_0,x=x_0,\omega{},\theta{})\dd{}\omega{}\dd{}\theta{} \\
&\lesssim  \int _{\mathcal{H}(s_0)\cap \set{x\le 1}} \abs{(P\psi{})_0}^2\tanh x\dd{}\mu{} + \mathcal{E}_T[\psi{}](s_0) + \mathcal{E}_T[\partial{}_\tau{}\psi{}](s_0) + \mathcal{E}_T[\Omega{}\psi{}](s_0).
\end{split}
\end{equation}
Sobolev embedding on \(S^2\) completes the proof.
\end{proof}
\section{Proof of the main results on the linear wave equation}
\label{main-proofs} We now prove each part of \cref{main-theorem-linear-intro}, which
concerns solutions to the linear wave equation. Let \(\varphi{}\in
C^\infty(\mathcal{M})\) solve \(\Box_g\varphi{} = 0\), and let \(\psi{} =
r\varphi{}\). Then \(P\psi{} = 0\), where \(P\) is defined in \zcref{P-def}.
\subsection{Quantitative boundedness and decay, and improved decay for \texorpdfstring{\(\U(1)\)}{U(1)}--symmetric solutions}
In this section, we prove parts \zcref{bounded-intro,U1-improved-decay-intro} of
\zcref{main-theorem-linear-intro}. Recall from \zcref{commutators} that
\(\mathring{\mathfrak{D}}\in \set{\Omega{},\partial_\theta{},\partial_\tau{}}\) is a commutator vector field.

We will show that \(\psi_{\ge 1}\) is bounded in terms of the
\(\partial_\tau{}\)-energy of \(\mathring{\mathfrak{D}}^{\le 3}\psi{}_{\ge
1}\) as follows:
\begin{equation}\label{main-proof-psige1}
\norm{(\cosh x)^{1/2}\psi_{\ge 1}}_{L^\infty(\Sigma{}(\tau{}))}^2\lesssim E_T[\mathring{\mathfrak{D}}^{\le 3}\psi{}_{\ge 1}](0).
\end{equation}
We prove this estimate in \zcref{sec:psige1-bounded}. We
will also prove that, given control of the \(\partial_\tau{}\)-energy of
\(\psi_0\) as well as the \(r^p\)-type energy introduced in \zcref{sec:rp-estimate},
the quantity \(\psi_0\) is not only bounded but enjoys improved decay:
\begin{equation}\label{main-rp-decay}
\norm{\psi{}_0}_{L^\infty(\mathcal{H}(s))}^2\lesssim e^{-cs}\mathcal{E}[\mathring{\mathfrak{D}}^{\le 3}\psi{}_0](0).
\end{equation}
This establishes part \zcref{U1-improved-decay-intro} of \zcref{main-theorem-linear-intro}.
We give the proof in \zcref{main-improved-decay}. The estimate \zcref{main-rp-decay}
applies only in the future of the hyperboloidal hypersurface \(\mathcal{H}(0)\).
In \zcref{main-psi0-bounded}, we extend the boundedness of \(\psi_0\) to the future
of the constant-\(\tau{}\) hypersurface \(\Sigma{}(0)\). Combining this statement with
\zcref{main-proof-psige1}, we obtain part \zcref{bounded-intro} of \zcref{main-theorem-linear-intro}.
\subsubsection{Boundedness of \texorpdfstring{\(\psi{}_{\ge 1}\)}{ψ≥1} in the future of a constant-\texorpdfstring{\(\tau\)}{τ} hypersurface\texorpdfstring{; proof of \zcref{main-proof-psige1}}{}}
\label{sec:psige1-bounded}
To prove \zcref{main-proof-psige1}, we combine \zcref{psige1-far} in the region \(\set{x\ge 1}\)
with \zcref{elliptic-near} and the commutation formula of \zcref{frak-D-commutation} in
the region \(\set{x\le 1}\). Indeed, \zcref{psige1-far} says
\begin{equation}
\norm{(\cosh x)^{1/2}\psi{}_{\ge 1}}_{L^\infty(\Sigma{}(\tau{})\cap \set{x\ge 1})}^2\lesssim E_T[\mathring{\mathfrak{D}}^{\le 3}\psi{}_{\ge 1}](\tau{}),
\end{equation}
while \zcref{elliptic-near-1} of \zcref{elliptic-near} implies
\begin{equation}
\norm{\psi{}_{\ge 1}}_{L^\infty(\Sigma{}(\tau{})\cap \set{x\le 1})}^2\lesssim E_T[\mathring{\mathfrak{D}}^{\le 3}\psi{}](\tau{}) + \int _{\Sigma{}(\tau{})\cap \set{x\le 2}}(P\mathring{\mathfrak{D}}^{\le 2}\psi{})^2\tanh x\dd{}\mu{}\lesssim E_T[\mathring{\mathfrak{D}}^{\le 3}\psi{}](\tau{}),
\end{equation}
where we have used the commutation formula of \zcref{frak-D-commutation} (or, more
precisely, Steps 1 and 2 of the proof there) to estimate
\begin{equation}\label{main-proof-prep}
\int _{\Sigma{}(\tau{})} (P\mathring{\mathfrak{D}}^{\le 2}\psi{}_{\ge 1})^2\tanh x\dd{}\mu{}\lesssim E_T[\mathring{\mathfrak{D}}^{\le 1}\psi{}_{\ge 1}](\tau{}).
\end{equation}
In this way, we obtain the estimate
\begin{equation}
\norm{(\cosh x)^{1/2}\psi_{\ge 1}}_{L^\infty(\Sigma{}(\tau{}))}^2\lesssim E_T[\mathring{\mathfrak{D}}^{\le 3}\psi{}_{\ge 1}](\tau{}).
\end{equation}
By the \(\partial_\tau{}\)-estimate of \zcref{T-estimate-prop}, we conclude that
\zcref{main-proof-psige1}.
\subsubsection{Improved decay for \texorpdfstring{\(\psi_0\)}{ψ0} in the future of a hyperboloidal hypersurface\texorpdfstring{; proof of \zcref{main-rp-decay}}{}}
\label{main-improved-decay} We prove \zcref{main-rp-decay} by combining \zcref{psi0-far} in
the region \(\set{x\ge 1}\) with \zcref{elliptic-near} in the region \(\set{x\le 1}\).
First, \zcref{psi0-far} gives
\begin{equation}
\begin{split}
\norm{\psi{}_0}_{L^\infty(\mathcal{H}(s)\cap \set{x\ge 1})}^2&\lesssim \mathcal{E}[\mathring{\mathfrak{D}}^{\le 2}\psi{}_0](s) + \int _{S^2\times S^1}\abs{\mathring{\mathfrak{D}}^{\le 2}\psi{}_0}^2(s,x=1).
\end{split}
\end{equation}
Here \(\mathcal{H}(s)\) is a leaf of the hyperboloidal foliation introduced in
\zcref{hyperboloidal-foliation}. Estimating the second term on the right-hand side
using \zcref{elliptic-near-bubble-prep-10} applied to
\(\mathring{\mathfrak{D}}^{\le 2}\psi_0\), we obtain
\begin{equation}\label{main-proof-p-1}
\begin{split}
\norm{\psi{}_0}_{L^\infty(\mathcal{H}(s)\cap \set{x\ge 1})}^2&\lesssim \mathcal{E}[\mathring{\mathfrak{D}}^{\le 2}\psi{}_0](s) + \int _{\mathcal{H}(s)}(P\mathring{\mathfrak{D}}^{\le 2}\psi{}_0)^2\tanh x\dd{}\mu{}.
\end{split}
\end{equation}
On the other hand, \zcref{elliptic-near} gives
\begin{equation}\label{main-proof-p-2}
\norm{\psi{}_0}_{L^\infty(\mathcal{H}(s)\cap \set{x\le 1})}^2\lesssim \mathcal{E}[\mathring{\mathfrak{D}}^{\le 3}\psi{}_0](s) + \int _{\mathcal{H}(s)}(P\mathring{\mathfrak{D}}^{\le 2}\psi{}_0)^2\tanh x\dd{}\mu{}.
\end{equation}
As in \zcref{main-proof-prep}, the (proof of the) commutation formula of
\zcref{frak-D-commutation} gives
\begin{equation}\label{main-proof-p-3}
\int _{\mathcal{H}(s)} (P\mathring{\mathfrak{D}}^{\le 2}\psi{}_{0})^2\tanh x\dd{}\mu{}\lesssim \mathcal{E}_T[\mathring{\mathfrak{D}}^{\le 1}\psi{}_{0}](s).
\end{equation}
Combining \zcref{main-proof-p-1,main-proof-p-2,main-proof-p-3}, we obtain
\begin{equation}\label{main-proof-rp-prep-pointwise}
\norm{\psi{}_0}_{L^\infty(\mathcal{H}(s))}^2\lesssim \mathcal{E}[\mathring{\mathfrak{D}}^{\le 3}\psi{}_0](s).
\end{equation}
We now introduce the following key claim.
\begin{lemma}
For \(s_1\le s_2\), we have
\begin{equation}
\mathcal{E}[\mathring{\mathfrak{D}}^{\le k}\psi{}_0](s_2) + \mathcal{B}[\mathring{\mathfrak{D}}^{\le k}\psi{}_0](s_1,s_2) \lesssim \mathcal{E}[\mathring{\mathfrak{D}}^{\le k}\psi{}_0](s_1).
\end{equation}
\label{main-proof-rp-goal}
\end{lemma}
Given \zcref{main-proof-rp-goal}, we can conclude the improved decay for \(\psi_0\).
\begin{proof}[Proof of \cref{main-rp-decay} given \cref{main-proof-rp-goal}]
The bulk term controls the integrated energy by \zcref{rp-K-bound}, so we obtain
\begin{equation}
\mathcal{E}[\mathring{\mathfrak{D}}^{\le k}\psi{}_0](s_2) + \int_{s_1}^{s_2} \mathcal{E}[\mathring{\mathfrak{D}}^{\le k}\psi{}_0](s)\dd{}s \lesssim \mathcal{E}[\mathring{\mathfrak{D}}^{\le k}\psi{}_0](s_1),
\end{equation}
and so it follows from \zcref{exponential-decay} (with
\(\mathcal{E}[\mathring{\mathfrak{D}}^{\le k}\psi{}_0](\cdot )\) playing the role of
\(f\)) that there exists \(c > 0\) such that
\begin{equation}\label{main-proof-rp-prep-2}
\mathcal{E}[\mathring{\mathfrak{D}}^{\le k}\psi{}_0](s)\lesssim e^{-cs}\mathcal{E}[\mathring{\mathfrak{D}}^{\le k}\psi{}_0](0).
\end{equation}
Substituting \zcref{main-proof-rp-prep-2} with \(k=3\) into
\zcref{main-proof-rp-prep-pointwise}, we conclude the decay statement
\zcref{main-rp-decay}.
\end{proof}
It remains to prove \zcref{main-proof-rp-goal}.
\begin{proof}[Proof of \cref{main-proof-rp-goal}]
Let \(k\ge 0\) and \(n\ge 0\). The
\(r^p\)-type estimate of \zcref{rp-type-estimate} applied to
\(\mathring{\mathfrak{D}}^{\le k}\psi_0\) (which is \(\U(1)\)-symmetric since
\(\psi_0\) is and \(\mathring{\mathfrak{D}}\) preserves \(\U(1)\)-symmetry) gives
\begin{equation}\label{main-proof-rp-estimate}
\mathcal{E}[\Omega^n\mathring{\mathfrak{D}}^{\le k}\psi{}_0](s) + \mathcal{B}[\Omega^n\mathring{\mathfrak{D}}^{\le k}\psi{}_0](0,s) \lesssim \mathcal{E}[\Omega^n\mathring{\mathfrak{D}}^{\le k}\psi{}_0](0) + \int_0^s \int _{\mathcal{H}(s')}\sinh x\abs{P\Omega^n\mathring{\mathfrak{D}}^{\le k}\psi{}_0}^2\dd{}\mu{}\dd{}s.
\end{equation}
Using Steps 1 and 2 of the proof of the commutation formula of
\zcref{frak-D-commutation} and then the estimate \zcref{main-proof-rp-estimate}, for
\(k\ge 1\) one has
\begin{equation}
\begin{split}
&\int_0^s \int _{\mathcal{H}(s')}\sinh x\abs{P\Omega^n\mathring{\mathfrak{D}}^{\le k}\psi{}_0}^2\dd{}\mu{}\dd{}s \lesssim \int_0^s \int _{\mathcal{H}(s')}\frac{\sinh x}{\cosh^2 \tau{}}\abs{\Grad _{S^2}\Omega^{n+1}\mathring{\mathfrak{D}}^{\le k-1}}\dd{}\mu{}\dd{}s \\
&\lesssim \mathcal{B}[\Omega^{n+1}\mathring{\mathfrak{D}}^{\le k-1}](0,s) \lesssim \mathcal{E}[\mathring{\mathfrak{D}}^{\le k+n}\psi{}_0](0) + \int_0^s \int _{\mathcal{H}(s')}\sinh x\abs{P\Omega^{n+1}\mathring{\mathfrak{D}}^{\le k-1}\psi{}_0}^2\dd{}\mu{}\dd{}s.
\end{split}
\end{equation}
By induction on \(k\), we conclude that
\begin{equation}
\int_0^s \int _{\mathcal{H}(s')}\sinh x\abs{P\mathring{\mathfrak{D}}^{\le k}\psi{}_0}^2\dd{}\mu{}\dd{}s \lesssim \mathcal{E}[\mathring{\mathfrak{D}}^{\le k}\psi{}_0](0) + \int_0^s \int _{\mathcal{H}(s')}\sinh x\abs{P\Omega^{k}\psi{}_0}^2\dd{}\mu{}\dd{}s.
\end{equation}
Since \(P\) commutes with \(\Omega{}\) and \(P\psi_0 = (P\psi{})_0 = 0\), we conclude that
\begin{equation}\label{main-proof-rp-estimate-1}
\int_0^s \int _{\mathcal{H}(s')}\sinh x\abs{P\mathring{\mathfrak{D}}^{\le k}\psi{}_0}^2\dd{}\mu{}\dd{}s \lesssim \mathcal{E}[\mathring{\mathfrak{D}}^{\le k}\psi{}_0](0).
\end{equation}
for \(k\ge 0\). Applying \zcref{main-proof-rp-estimate} with \(n=0\) and substituting
\zcref{main-proof-rp-estimate-1}, we obtain \zcref{main-proof-rp-goal}.
\end{proof}
\subsubsection{Boundedness of \texorpdfstring{\(\psi_0\)}{ψ0} in the future of a constant-\texorpdfstring{\(\tau\)}{τ} hypersurface}
\label{main-psi0-bounded}
Strictly speaking, we have only proven the boundedness of \(\psi_0\) in the future
of a hyperboloidal hypersurface \(\mathcal{H}(0)\). To obtain the estimate in
the future of the constant-\(\tau{}\) hypersurface \(\Sigma{}(0)\), it remains
to estimate \(\psi_0\) in a region outside a light cone. We will only sketch the
proof. One can repeat the argument of \zcref{T-estimate-prop} with the multiplier
\(w(-u)\partial_\tau{}\) with \(w(q)\) a smooth increasing function that equals
\(1\) when \(q < 0\) and equals, say \(q^2\) for \(q\ge 1\) (the important thing
is that \(1/w(q)\) is integrable for \(q\) large). This provides control of a
weighted \(\partial_\tau{}\)-energy \(E_{T,w}\) whose integrand is that of the
\(\partial_\tau{}\)-energy \(E_T\) multiplied by \(w(-u)\) (as well as weighted
fluxes \(F_{T,w}\) along ingoing null cones analogous to \(F_T\)). Then the
argument used to prove \zcref{dx-far} can be used to estimate pointwise the quantity
\(w(-u)\partial{}_x\psi{}_0\) by the \(w\)-weighted \(\partial_\tau{}\) energy
(and associated fluxes along ingoing null cones) of
\(\mathring{\mathfrak{D}}^{\le 4}\psi{}_0\). Since \(u =
\frac{1}{2}(\tau{}-x)\), on a constant-\(\tau{}\) hypersurface
\(\Sigma{}(\tau{})\) this provides an estimate for \(\partial_x\psi_0\) that one
can integrate in \(x\) to \(\tau{}\sim x\), where \(\Sigma{}(\tau{})\) meets a
hyperboloidal hypersurface, to obtain an estimate for \(\psi_0\).
\subsection{Lack of improved decay outside \texorpdfstring{\(\U(1)\)}{U(1)}-symmetry}
\label{no-decay} In this section, we prove part \zcref{time-periodic-intro} of
\zcref{main-theorem-linear-intro}.
\begin{proposition}
There exists a spherically symmetric solution \(\psi{}\in C^\infty(\mathcal{M})\) to \(P\psi{} =
0\) with vanishing \(\U(1)\)-symmetric part that is periodic in \(\tau{}\) and
satisfies \(E_T[\psi{}](0) < \infty\).
\label{non-decaying}
\end{proposition}
\begin{proof}
Let \(n\in \Z_{\ge 0}\) and consider the Schrödinger
operator \(\mathcal{A}_n\) acting on \(L^2(0,\infty)\) by
\begin{equation}
\mathcal{A}_nf(x) = -f''(x) + V_nf(x), \qquad V_n(x) = -\frac{1}{4\cosh ^2x\sinh ^2x} + n^2\frac{\cosh ^2x}{\tanh^2 x}.
\end{equation}
Recall the expression \zcref{P-def} for the operator \(P\) and observe that, if
\(\Psi_{n,\lambda{}}\) solves the eigenvalue problem
\begin{equation}\label{eigenvalue-problem}
(\mathcal{A}_n - \lambda^2)\Psi{}_{n,\lambda{}} = 0\textnormal{ for }\lambda{}\in \R,
\end{equation}
then \(\psi{}_{n,\lambda{}}\) defined by
\begin{equation}
\psi{}_{n,\lambda{}}(\tau{},x,\omega{},\theta{}) = e^{-i\lambda{}\tau{}}(\tanh x)^{-1/2}\Psi{}_{n,\lambda{}}(x)e^{-in\theta{}}
\end{equation}
solves \(P\psi_{n,\lambda{}} = 0\). It is clear that \(\psi_{n,\lambda{}}\) is periodic in \(\tau{}\), is
spherically symmetric, and has vanishing \(\U(1)\)-symmetric part when \(n\ge 1\). It
remains to construct \(\Psi_{n,\lambda{}}\) solving \zcref{eigenvalue-problem} and
show that the associated quantity \(\psi_{n,\lambda{}}\) is smooth and has finite
\(E_T\)-energy.

The key claim is that, when \(n\ge 1\), the operator \(\mathcal{A}_n\) has purely
discrete spectrum contained in \(\R_{\ge 0}\), and its eigenfunctions form an
orthonormal basis. Moreover, the eigenfunctions are in the space
\begin{equation}\label{X-space}
X_n = C^\infty((0,\infty))\cap \dot{H}^1([0,\infty),\dd{}x)\cap L^2([0,\infty),V_n(x)\dd{}x).
\end{equation}
In particular, there exists a solution to \zcref{eigenvalue-problem} for some
\(\lambda{}\in \R\). Since \(V_n\in C^\infty((0,\infty))\) and \(V_n(x)\sim
cx^{-2}\) as \(x\to 0\) for \(c=n^2-1/4\ge 3/4\) and \(V_n(x)\to \infty\) as
\(x\to \infty\), this follows from classical results in Sturm--Liouville theory
(see for example \cite[Sec.~II]{titchmarsh1946eigenfunction}). Note that
\zcref{eigenvalue-problem} is a singular Sturm--Liouville problem posed on
\([0,\infty)\) for which both endpoints are in the so-called limit-point case.

We now show that \(\psi_{n,\lambda{}}\) is smooth (as a function on \(\mathcal{M}\)) and
that \(\psi_{n,\lambda{}}\) has finite \(E_T\)-energy. By the method of Frobenius
applied to the regular singular point of \zcref{eigenvalue-problem} at \(x = 0\)
(see \cite[Sec.~5.4]{Olver1997-en} for a reference), the quantity
\(x^{-(2n+1)/2}\Psi{}_{n,\lambda{}}(x)\) has a power series expansion in \(x\)
near \(x=0\). This power series contains only even powers of \(x\), since the
same is true for the power series of \(x^2V_n(x)\) near \(x = 0\). In particular,
\(\Psi_{n,\lambda{}}\) defines an element of \(x^{n +
1/2}C^\infty(\mathcal{M})\), and so \(\psi{}_{n,\lambda{}}\) defines an element
of \(x^ne^{-in\theta{}}C^\infty(\mathcal{M})\subset C^\infty(\mathcal{M})\).
That is, \(\psi_{n,\lambda{}}\) is smooth on \(\mathcal{M}\). Next, we show that
\(\psi_{n,\lambda{}}\) has finite \(E_T\)-energy. We compute
\begin{equation}
\begin{gathered}
\abs{\partial{}_\tau{}\psi{}_{n,\lambda{}}}\le  \lambda{}(\tanh x)^{-1/2}\abs{\Psi{}_{n,\lambda{}}},\qquad \abs{\partial{}_x\psi{}_{n,\lambda{}}}\lesssim (\tanh x)^{-3/2}\abs{\Psi{}_{n,\lambda{}}} + (\tanh x)^{-1/2}\abs{\partial{}_x\Psi{}_{n,\lambda{}}} \\
\abs{\partial{}_\theta{}\psi{}_{n,\lambda{}}}\le  n(\tanh x)^{-1/2}\abs{\Psi{}_{n,\lambda{}}}
\end{gathered}
\end{equation}
It follows that
\begin{equation}
\begin{split}
E_T[\psi_{n,\lambda{}}](0) &= \int _{\Sigma{}(0)} \Bigl[\abs{\partial{}_\tau{}\psi{}_{n,\lambda{}}}^2 + \abs{\partial{}_x\psi{}_{n,\lambda{}}}^2 + \frac{1}{\cosh ^2x}\abs{\psi{}_{n,\lambda{}}}^2 + \frac{\cosh ^2x}{\tanh ^2x}\abs{\partial{}_\theta{}\psi{}_{n,\lambda{}}}^2\Bigr]\tanh x\dd{}\mu{} \\
&\lesssim_{n,\lambda{}} \int _{\Sigma{}(0)} (\cosh x)^2(\tanh x)^{-3}\Psi{}_{n,\lambda{}}^2 + (\tanh x)^{-1}(\partial{}_x\Psi{}_{n,\lambda{}})^2\dd{}\mu{} \\
\end{split}
\end{equation}
As discussed above, the Frobenius method shows that \(\Psi_{n,\lambda{}}\sim x^{n+1/2}\) and
\(\partial_x\Psi_{n,\lambda{}}\sim x^{n-1/2}\) near \(x = 0\). In particular, the
above integrand is integrable near \(x = 0\). The integrability near \(x =
\infty\) follows from the fact that \(\Psi_{n,\lambda{}}\) belongs to the space
\(X_n\) defined in \zcref{X-space}. We conclude that the \(E_T\)-energy of
\(\psi_{n,\lambda{}}\) is finite.
\end{proof}
\subsection{Existence of the radiation field}
\label{linear-radiation-field}
In this section, we prove part \zcref{radiation-field-intro} of
\zcref{main-theorem-linear-intro}. By the estimate \zcref{main-proof-psige1} and the
fact that \(x\to \infty\) at null infinity, we know that \(\psi_{\ge 1}\)
vanishes at null infinity. Thus the limit of \(\psi{}\) at null infinity, if it
exists, is equal to that of \(\psi_0\). The existence of the limit \(\lim_{x\to
\infty}\psi{}_0(s,x,\omega{},\theta{})\) for each fixed
\((s,\omega{},\theta{})\) (where \(s\) is the ``hyperboloidal time'' indexing the
foliation \(\mathcal{H}(s)\) introduced in \zcref{hyperboloidal-foliation}), follows
from the finiteness of \(\mathcal{E}[\mathring{\mathfrak{D}}^{\le
3}\psi{}_0](s)\), using Sobolev embedding and the integrability of \((\sinh
x)(L\psi{})^2\) on \(\mathcal{H}(s)\cap \set{x\ge 1}\) as in the proof of
\zcref{psi0-far}. By \zcref{main-proof-rp-goal}, the desired finiteness of
\(\mathcal{E}[\mathring{\mathfrak{D}}^{\le 3}\psi{}_0](s)\) follows from the
finiteness of \(\mathcal{E}[\mathring{\mathfrak{D}}^{\le 3}\psi{}_0](0)\).

To see that the radiation field can be non-vanishing, we consider a solution
\(\psi{}\) arising from non-zero spherically symmetric and \(\U(1)\)-symmetric
data. We integrate the identity \zcref{T-estimate-1} applied to \(\psi\) over the
region \(\set{0\le s\le s_0}\cap \set{v\le v_0}\) (noting as in the surrounding
discussion there that the boundary terms at the bubble vanish). This yields an
identity
\begin{equation}
\mathcal{E}_T[\psi{}](s_0;v_0) + \mathbf{F}_T[\psi{}](s_0;v_0) = \mathcal{E}_T[\psi{}](0;v_0).
\end{equation}
Here \(\mathcal{E}_T[\psi{}](s;v_0)\) measures the \(\partial_\tau{}\)-energy flux through
\(\mathcal{H}(s)\cap \set{v\le v_0}\) and \(\mathbf{F}_T\) measures the
\(\partial_\tau{}\)-energy flux through \(\set{0\le s \le s_0}\cap
\set{v=v_0}\). The latter is coercive since \(\set{v=v_0}\) is null. We note
that an angular bulk term is not present because \(\psi{}\) is spherically
symmetric. Taking \(v_0\to \infty\), we obtain
\begin{equation}\label{linear-wave-eq-1}
\mathcal{E}_T[\psi{}](s_0) + \mathbf{F}_{T,\mathcal{I}^+}[\psi{}](s_0) = \mathcal{E}_T[\psi{}](0)
\end{equation}
where now \(\mathbf{F}_{T,\mathcal{I}^+}[\psi{}](s_0)\) measures the flux through the portion of
null infinity meeting \(\set{0\le s \le s_0}\) of the \(\partial_\tau{}\)-energy of the
radiation field associated to \(\psi{}\). By \zcref{main-proof-rp-prep-2}, the first term
on the left-hand side of \zcref{linear-wave-eq-1} vanishes as \(s_0\to \infty\). We
therefore obtain an identity
\begin{equation}
\mathbf{F}_{T,\mathcal{I}^+}[\psi{}] = \mathcal{E}_T[\psi{}](0),
\end{equation}
where \(\mathbf{F}_{T,\mathcal{I}^ + }[\psi{}]\) measures the full flux through null infinity
of the \(\partial_\tau{}\)-energy of the radiation field associated to
\(\psi{}\). In particular, we have
\begin{equation}
\mathbf{F}_{T,\mathcal{I}^+}[\psi{}] = c\int_0^\infty \abs{\underline{L}\psi{}_{\mathcal{I}^+}}^2\dd{}u
\end{equation}
for some positive normalization constant \(c > 0\), and where
\(\psi{}_{\mathcal{I}^+}\) is the radiation field associated to \(\psi{}\). We
conclude that if the initial data is non-zero, then the radiation field cannot
identically vanish.
\subsection{Vanishing at the bubble for solutions with vanishing \texorpdfstring{\(\U(1)\)}{U(1)}-symmetric part}
In this section, we prove part \zcref{bubble-vanishing-intro} of
\zcref{main-theorem-linear-intro}.
\begin{lemma}
Let \(\psi{}\in C^\infty(\mathcal{M})\) solve \(P\psi{} = 0\). Suppose that \(E_T[\psi_{\ge 1}](0) < \infty\). Then \(\psi_{\ge 1}\) vanishes at \(\set{x=0}\).
\end{lemma}
\begin{proof}
Fix \(\tau{}\ge 0\). By \zcref{T-estimate-prop}, we have \(E_T[\psi_{\ge 1}](\tau{})<\infty\). The
\(\partial_\tau{}\)-energy of \(\psi_{\ge 1}\) controls \(\psi_{\ge 1}\) in
\(L^2\) with a weight that is singular as \(x\to 0\):
\begin{equation}\label{pointwise-vanishing-bubble}
\int _{\Sigma{}(\tau{})} \frac{\cosh ^2x}{\tanh x}\psi{}_{\ge 1}^2\dd{}\mu{}\le \int _{\Sigma{}(\tau{})} \frac{\cosh ^2x}{\tanh ^2x}(\partial_\theta{}\psi{}_{\ge 1})^2\tanh x\dd{}\mu{}\le E_T[\psi_{\ge 1}](\tau{}).
\end{equation}
It follows that \(\psi_{\ge 1}(\tau{},x=0) = 0\) if \(E_T[\psi_{\ge 1}](\tau{})\) is finite. Indeed,
otherwise \(\psi_{\ge 1}^2\) would be bounded below away from zero near
\(\set{x=0}\), and the integral on the left-hand side of
\zcref{pointwise-vanishing-bubble} would fail to be finite.
\end{proof}
\section{Proof of the main results on semilinear wave equations}
In this section, we prove \zcref{main-theorem-semilinear-intro}, which establishes
small-data global existence results for a class of semilinear equations. The
main results of this section are \zcref{semilinear,semilinear-axisymmetric}.

In \zcref{sec:nonlinearities}, we introduce the class of nonlinearities under
consideration. In particular, we formulate a suitable version of the null
condition (see \zcref{null-condition}) that includes perturbations of
\(g^{\alpha{}\beta{}}\partial_\alpha{}\varphi{}\partial_\beta{}\varphi{}\) with
growing weights in time. In \zcref{sec:semilinear}, we consider the global existence
problem for semilinear equations satisfying the null condition, without symmetry
assumptions on the initial data. We assume here that the nonlinearity is almost
\(\U(1)\)-symmetric, in the sense that its non-\(\U(1)\)-symmetric part enjoys
sufficiently fast decay in space (see \zcref{axisymmetric-nonlinearity}). When the
nonlinearity and the initial data are exactly \(\U(1)\)-symmetric, we can obtain
improved decay in time results compared to the non-\(\U(1)\)-symmetric setting.
We consider this case in \zcref{sec:semilinear-axisymmetric}.
\subsection{The class of nonlinearities under consideration}
\label{sec:nonlinearities}
\subsubsection{The null condition and (almost) \texorpdfstring{\(\U(1)\)}{U(1)}-symmetric nonlinearities}
\begin{definition}[Null condition]
Let \(\varphi{}\in C^\infty(\mathcal{M})\), and fix a function \(\ell{} : \R\to (0,\infty)\). Let \((U,V)\) be the
Minkowskian double null coordinates (defined in \zcref{double-null}) with
corresponding coordinate vector fields \((\partial{}_U,\partial{}_V)\). We say a
quadratic nonlinearity \(F(\Grad \varphi{})\) satisfies the \emph{null condition with
time weight \(\ell{}\)} if in the region \(\set{x\ge 1}\) we have
\begin{equation}\label{null-condition-far}
\begin{split}
\ell{}(\tau{})^{-1}F(\Grad \varphi{}) &= f_{UV}\abs{U}^{-1/2}V^{1/2}\partial{}_U\varphi{}\partial{}_{V}\varphi{} + f_{\Omega{}\Omega{}}\abs{U}^{-3/4}V^{5/4}\abs{\slashed{\Grad{}}\varphi{}}^2 + f_{\theta{}\theta{}}\abs{U}^{-1/2}V^{1/2}(\partial{}_\theta{}\varphi{})^2 \\
&\qquad + f_{UU}\abs{U}^{1/2}V^{-1/2}(\partial{}_U\varphi{})^2 + f_{VV}\abs{U}^{-3/2}V^{3/2}(\partial{}_V\varphi{})^2 \\
&\qquad + f_{U\Omega{}}\abs{U}^{-1/4}V^{-1/4}\partial{}_U\varphi{}\slashed{\Grad{}}\varphi{} + f_{V\Omega{}}\abs{U}^{-5/4}V^{3/4}\partial{}_V\varphi{}\slashed{\Grad{}}\varphi{} \\
&\qquad +  f_{\theta{}\Omega{}}\abs{U}^{-1/2}V^{1/2}\partial{}_\theta{}\varphi{}\slashed{\Grad{}}\varphi{} + f_{\theta{}U}\abs{U}^{-1/4}V^{-1/4}\partial{}_\theta{}\varphi{}\partial{}_U\varphi{} + f_{\theta{}V}\abs{U}^{-5/4}V^{3/4}\partial{}_\theta{}\varphi{}\partial{}_V\varphi{} \\
\end{split}
\end{equation}
and in the region \(\set{x\le 1}\) we have
\begin{equation}\label{null-condition-near}
\begin{split}
\ell{}(\tau{})^{-1}F(\Grad \varphi{}) &= \tilde{f}_{\tau{}\tau{}}\abs{U}^{-1}(\partial{}_\tau{}\varphi{})^2 + \tilde{f}_{\Omega{}\Omega{}}\abs{U}^{-1/2}V^{3/2}\abs{\slashed{\Grad{}}\varphi{}}^2 + f_{\tau{}\Omega{}}\abs{U}^{-1/2}V^{1/2}\partial{}_\tau{}\varphi{}\slashed{\Grad{}}\varphi{} \\
&\qquad + \sum_{a,b\in \set{\hat{y},\hat{z}}}\tilde{f}_{a,b}\abs{U}^{-1}\partial{}_a\varphi{}\partial{}_b\varphi{} + \sum_{a\in \set{\hat{y},\hat{z}}}\tilde{f}_{\tau{},a}\abs{U}^{-1}\partial{}_\tau{}\varphi{}\partial{}_a\varphi{} + \sum_{a\in \set{\hat{y},\hat{z}}}\tilde{f}_{\Omega{}a}\abs{U}^{-1/2}V^{1/2}\slashed{\Grad{}}\varphi{}\partial{}_a\varphi{}
\end{split}
\end{equation}
for some coefficients \(f_\bullet{}\) and \(\tilde{f}_\bullet{}\) that are bounded together with
all their \(\mathfrak{D}\)-derivatives, where \(\mathfrak{D}\) is a commutator
vector field (see \zcref{commutators}). Here we have schematically written
\(\slashed{\Grad{}}\varphi{}\) for \(r^{-1}\Omega_i\varphi{}\), where \(1\le
i\le 3\), and \(\abs{\slashed{\Grad{}}\varphi{}}^2
\coloneqq{}\sum_{i=1}^3(r^{-1}\Omega{}_i\varphi{})^2\).
\label{null-condition}
\end{definition}
\begin{remark}[Geometric meaning of \(\abs{U}\) and \(V\)]
Recall from \zcref{double-null} that \(\abs{U}\sim e^{-(\tau{}-x)}\) and \(V\sim e^{\tau{} + x}\sim r\).
Since \(\abs{U}\to 0\) as \(\tau{}\to \infty\) for fixed \(x\), weights in \(\abs{U}\)
with positive powers are to be interpreted as weaker near the
bubble than those with negative powers. However, near spatial infinity,
\(\abs{U}\) and \(V\) are comparable, and both weights are stronger with
positive powers than with negative powers.
\end{remark}
\begin{remark}[Interpreting the weights in \(\abs{U}\) and \(V\)]
The weights in \(\abs{U}\) and \(V\) in \zcref{null-condition} are chosen so that each
term in \(F(\Grad \varphi{})\) that appears in the canonical nonlinearity
\(g^{\alpha{}\beta{}}\partial_\alpha{}\varphi{}\partial_\beta{}\varphi{}\)
behaves as it does there. Moreover, terms which do not appear in
\(g^{\alpha{}\beta{}}\partial_\alpha{}\varphi{}\partial_\beta{}\varphi{}\), such
as \((\partial_U\varphi{})^2\), are allowed with weights such that, when
\(F(\Grad \varphi{})\) is written in terms of \(\psi{}\) and expressed in
\((\tau{},x)\) coordinates, all terms produced are no worse than those produced
by \(g^{\alpha{}\beta{}}\partial_\alpha{}\varphi{}\partial_\beta{}\varphi{}\)
when the same procedure is done. For examples, see \zcref{ref-0,ref-1,ref-2}. In
this sense, \zcref{null-condition} recovers the classical null condition in the
asymptotically flat region of the Witten bubble spacetime.
\end{remark}
\begin{remark}[The role of the time weight \(\ell(\tau)\)]
It is easier to prove global existence for nonlinearities with a faster decaying
time weight. We include a time weight \(\ell{}(\tau{})\) in our definition of
the nonlinearity to highlight the fact that there is an exponential amount of
room in \(\tau{}\) when studying the wave equation \(\Box_g\varphi{} =
f\cdot g^{\alpha{}\beta{}}\partial_\alpha{}\varphi{}\partial_\beta{}\varphi{}\) for
\(g\) the Witten bubble metric (see \zcref{canonical-double-null}) and \(f\) a
spacetime function with bounded \(\mathfrak{D}\)-derivatives. This room is
crucial for closing our global existence argument outside of symmetry, because
we allow the solution to grow polynomially in \(\tau{}\) (see \zcref{semilinear}).

By comparing the statements of \zcref{semilinear,semilinear-axisymmetric}, we see
also that there is more room in \(\tau{}\) within \(\U(1)\)-symmetry (namely when
studying a \(\U(1)\)-symmetric nonlinearity, such as
\(g^{\alpha{}\beta{}}\partial_\alpha{}\varphi{}\partial_\beta{}\varphi{}\), with
\(\U(1)\)-symmetric initial data) than outside of symmetry.
\end{remark}
\begin{remark}[Comparing the variable \(\tau\) with the Minkowskian time coordinate \(t\)]
Since the Minkowskian time coordinate \(t\) takes the form \(t = R\cosh x\sinh
\tau{}\), an exponential decay rate in \(\tau{}\) corresponds to a polynomial
decay rate in \(t\), at least in a finite-\(x\) region. Similarly, a polynomial
rate in \(\tau{}\) corresponds to a logarithmic rate in \(t\).
\label{tau-t-comparison}
\end{remark}
\begin{remark}[Cubic and higher-order terms in the nonlinearity]
For the sake of simplicity, we have not included cubic and higher-order terms in
the nonlinearity (such as \(\varphi{}\partial_U\varphi{}\partial_V\varphi{}\) or
\((\partial_U\varphi{})^3\)). However, the addition of such terms would not
require any significant changes in the arguments to follow.
\end{remark}
\begin{definition}[\(\U(1)\)-Symmetric and almost \(\U(1)\)-symmetric nonlinearities]
We say a nonlinearity \(F(\Grad \varphi{})\) is \emph{\(\U(1)\)-symmetric} if its coefficients in \((U,V)\)
coordinates are \(\U(1)\)-symmetric, namely if the expression \zcref{null-condition-far}
holds for all \(x\) and some \(\U(1)\)-symmetric coefficients \(f_\bullet{}\).

A nonlinearity \(F(\Grad \varphi{})\) is \emph{almost \(\U(1)\)-symmetric} if the non-\(\U(1)\)-symmetric
part of its coefficients in \((U,V)\) coordinates decay sufficiently
fast in space. That is, we require that \zcref{null-condition-far} holds in the
region \(\set{x\ge 1}\) for coefficients \(f_\bullet{}\) such that
\(\abs{(\mathfrak{D}^{\mathbf{k}}f_\bullet{})_{\ge 1}}\lesssim (\cosh x)^{-1/2}\) for
all multi-indices \(\mathbf{k}\).
\label{axisymmetric-nonlinearity}
\end{definition}
\subsubsection{Estimates associated to nonlinearities satisfying the null condition}
\begin{lemma}[A nonlinearity that is \(\U(1)\)-symmetric and satisfies the null condition]
Let \(h\) be a \(\U(1)\)-symmetric function that is bounded together with its
\(\mathfrak{D}\)-derivatives. The quadratic nonlinearity \(F(\Grad \varphi{}) =
h\cdot g^{\alpha{}\beta{}}\partial_\alpha{}\varphi{}\partial_\beta{}\varphi{}\)
given by the Witten bubble metric \(g\) is \(\U(1)\)-symmetric (see
\zcref{axisymmetric-nonlinearity}) and satisfies the null condition with time weight
\(e^{-\tau{}}\) (see \zcref{null-condition}).
\label{canonical-double-null}
\end{lemma}
\begin{proof}
By the product rule and the fact that \(h\) is bounded together with its
\(\mathfrak{D}\)-derivatives, it is enough to consider the case \(h = 1\).

Let \(\tilde{F}(\Grad \varphi{}) = g^{\alpha{}\beta{}}\partial_\alpha{}\varphi{}\partial_\beta{}\varphi{}\). In \((U,V,\omega{},\theta{})\) coordinates, we have
\begin{equation}\label{canonical-nonlinearity-1}
\tilde{F}(\Grad \varphi{}) = -(1 + e^{-2x})^{-2}\partial_U\varphi{}\partial_V\varphi{} + \abs{\slashed{\Grad{}}\varphi{}}^2 + R^{-2}(\tanh x)^{-2}(\partial{}_\theta{}\varphi{})^2.
\end{equation}
It is clear that \(F(\Grad \varphi{})\) is \(\U(1)\)-symmetric, since all the coefficients and
derivatives present commute with \(\partial_\theta{}\). We now make the useful
observations (see \zcref{double-null}) that
\begin{equation}\label{gmunu-obs}
e^{-\tau{}} = \abs{U}^{1/2}V^{-1/2},\qquad e^x = 2R^{-1}\abs{U}^{1/2}V^{1/2}.
\end{equation}

From the definitions of \(U\) and \(V\) (see \zcref{double-null}), we have
\(e^{-\tau{}} = \abs{U}^{1/2}V^{-1/2}\). Using \zcref{gmunu-obs}, we find that
\(e^\tau{}\tilde{F}(\Grad \varphi{})\) is in the form \zcref{null-condition-far} with
non-vanishing coefficients
\begin{equation}
f_{UV} = -(1+e^{-2x})^{-2},\qquad f_{\Omega{}\Omega{}} = e^{-\tau{}},\qquad f_{\theta{}\theta{}} = R^{-2}(\tanh x)^{-2}.
\end{equation}
It is straightforward to check that these are bounded together with their
\(\mathfrak{D}\)-derivatives in the region \(\set{x\ge 1}\), recalling the
definitions of the \(\mathfrak{D}\)-derivatives from \zcref{commutators}.

In \((\tau{},\hat{y},\hat{z},\omega{})\) coordinates, we have (from \zcref{cart-metric})
\begin{equation}
\tilde{F}(\Grad \varphi{}) = -\rho^{-2}(\partial{}_\tau{}\varphi{})^2 + R^{-2}\mathfrak{F}(x)^{-2}((\partial{}_{\hat{y}}\varphi{})^2 + (\partial{}_{\hat{z}}\varphi{})^2) + \abs{\slashed{\Grad{}}\varphi{}}^2.
\end{equation}
Using \zcref{gmunu-obs}, we find that \(e^\tau{}F(\Grad \varphi{})\) is in the form
\zcref{null-condition-near} with non-vanishing coefficients
\begin{equation}
\tilde{f}_{\tau{}\tau{}} = 2R^{-1}e^{-x}(1 - 1/(1+e^{2x}))^2,\qquad \tilde{f}_{\hat{y}\hat{y}} = \tilde{f}_{\hat{z}\hat{z}} = \frac{1}{2}R^{-1}e^xe^{2\cosh x} \frac{\tanh(x/2)^2}{\tanh(x)^2},\qquad \tilde{f}_{\Omega{}\Omega{}} = 2R^{-1}e^{-(\tau{}+x)}.
\end{equation}
It is straightforward to check that these are bounded together with their
\(\mathfrak{D}\)-derivatives in the region \(\set{x\le 1}\).
\end{proof}
\begin{lemma}[Converting the equation solved by \(\varphi\) to an estimate for \(\psi\)]
Suppose \(\varphi{}\in C^\infty(\mathcal{M})\) solves \(\Box{}\varphi{} = F(\Grad \varphi{})\) for a quadratic
nonlinearity \(F(\Grad \varphi{})\) satisfying the null condition with time
weight \(\ell{}(\tau{})\) (see \zcref{null-condition}). Then \(\psi{} = r\varphi{}\)
satisfies the following estimate:
\begin{equation}\label{DkPpsi-semilinear}
\begin{split}
\abs{\mathfrak{D}^{\mathbf{k}}P\psi{}} &\lesssim \Bigl(\sum_{k=0}^{\abs{\mathbf{k}}}\abs{\partial{}_\tau^k\ell{}(\tau{})}\Bigr)\sum_{\mathbf{k}_1 + \mathbf{k}_2 \le \mathbf{k}} \Bigl[\frac{1}{\cosh x}(\abs{\partial{}_\tau{}\mathfrak{D}^{\mathbf{k}_1}\psi{}} + \abs{\partial{}_x\mathfrak{D}^{\mathbf{k}_1}\psi{}} + \abs{\mathfrak{D}^{\mathbf{k}_1}\psi{}})(\abs{\partial{}_\tau{}\mathfrak{D}^{\mathbf{k}_2}\psi{}} + \abs{\partial{}_x\mathfrak{D}^{\mathbf{k}_2}\psi{}} + \abs{\mathfrak{D}^{\mathbf{k}_2}\psi{}})  \\
&\qquad + \frac{1}{(\cosh x)^{1/2}\cosh \tau{}}\abs{\Grad _{S^2}\mathfrak{D}^{\mathbf{k}_1}\psi{}}\Bigl(\frac{1}{\cosh x}(\abs{\partial{}_\tau{}\mathfrak{D}^{\mathbf{k}_2}\psi{}} +\abs{\partial{}_x\mathfrak{D}^{\mathbf{k}_2}\psi{}} + \abs{\mathfrak{D}^{\mathbf{k}_2}\psi{}}) + \abs{\Grad _{S^2}\mathfrak{D}^{\mathbf{k}_2}\psi{}}\Bigr) \\
&\qquad + (\cosh x)^{1/2}\abs{\partial{}_\theta{}\mathfrak{D}^{\mathbf{k}_1}\psi{}}\Bigl((\cosh x)^{1/2}\abs{\partial{}_\theta{}\mathfrak{D}^{\mathbf{k}_2}\psi{}} + \frac{1}{\cosh x}(\abs{\partial{}_\tau{}\mathfrak{D}^{\mathbf{k}_2}\psi{}} + \abs{\partial{}_x\mathfrak{D}^{\mathbf{k}_2}\psi{}} + \abs{\mathfrak{D}^{\mathbf{k}_2}\psi{}})\\
&\qquad + \frac{1}{(\cosh x)^{1/2}\cosh \tau{}}\abs{\Grad _{S^2}\mathfrak{D}^{\mathbf{k}_2}\psi{}}\Bigr)\Bigr],
\end{split}
\end{equation}
Here \(P\) is the operator defined in \zcref{P-def} and \(\mathfrak{D}\) is a
commutator vector field (see \zcref{commutators}). The statement also holds with
\(\mathfrak{D}\) replaced with \(\mathring{\mathfrak{D}}\).
\label{Dk-P-psi-estimate}
\end{lemma}
\begin{proof}
To establish \zcref{Dk-P-psi-estimate}, we express the nonlinearity \(F(\Grad
\varphi{})\) in terms of \(\psi{}\) to derive an expression for \(P\psi{}\). In
particular, this produces zeroth-order terms in \(\psi{}\), even though
\(F(\Grad \varphi{})\) does not contain zeroth-order terms in \(\varphi{}\). For
example, we have the expansion \(2V\partial{}_V\varphi{} = L\varphi{} =
r^{-1}(L\psi{} - r^{-1}Lr\cdot \psi{})\), where \(L = \partial_\tau{} +
\partial_x\). We then apply \(\mathfrak{D}^{\mathbf{k}}\) and estimate the
resulting terms using the triangle inequality. Throughout the proof, we will
implicitly use \zcref{OD-check} to verify that various quantities are
\(O_{\mathfrak{D}}(1)\).

We now turn to the proof. We start by computing \(\underline{L} =
2\abs{U}\partial_U\) and \(L =2V\partial{}_V\). Using \zcref{gmunu-obs}, we now
compute
\begin{equation}\label{ref-0}
\rho^2r\cdot \abs{U}^{-1/2}V^{1/2}\partial{}_U\varphi{}\partial{}_V\varphi{} = 2R^{-1}(e^{-x}\cosh x)^2(e^{-\tau{}}\cosh \tau{})^{-1}\cdot  \frac{1}{\cosh x}r^2\underline{L}\varphi{}L\varphi{} = O_{\mathfrak{D}}(1)\frac{1}{\cosh x}r^2\underline{L}\varphi{}L\varphi{}.
\end{equation}
Here we write \(O_{\mathfrak{D}}(1)\) to mean a quantity that is bounded
together with all its \(\mathfrak{D}\)-derivatives \emph{in the region \(\set{x\ge
1}\)}. Since \(\underline{L}r/r\) and \(Lr/r\) are \(O_{\mathfrak{D}}(1)\) (being
linear combinations of \(\tanh \tau{}\) and \(\tanh x\)), we find that
\(\rho^2r\cdot \abs{U}^{-1/2}V^{1/2}\partial{}_U\varphi{}\partial{}_V\varphi{}\)
is an \(O_{\mathfrak{D}}(1)\)-linear combination of the following quantities:
\begin{equation}\label{ref-1}
\frac{1}{\cosh x}\underline{L}\psi{}L\psi{},\quad \frac{1}{\cosh x}\psi{}\underline{L}\psi{},\quad \frac{1}{\cosh x}\psi{}L\psi{},\quad \frac{1}{\cosh x}\psi^2.
\end{equation}
To give another example, we compute
\begin{equation}\label{ref-2}
\rho^2r\cdot \abs{U}^{1/2}V^{-1/2}(\partial{}_U\varphi{})^2 = O_{\mathfrak{D}}(1) \frac{1}{\cosh x}r^2\underline{L}\varphi{}\underline{L}\varphi{}.
\end{equation}
In this way we find that if \zcref{null-condition-far} holds, then in the region
\(\set{x\ge 1}\), the quadratic nonlinearity \(\rho^2r\cdot
\ell{}(\tau{})^{-1}F(\Grad \varphi{})\) is a \(O_{\mathfrak{D}}(1)\)-linear
combination of the following quantities:
\begin{equation}\label{semilinearity-decomposition}
\begin{gathered}
\frac{1}{\cosh x}\partial{}_\tau{}\psi{}\partial{}_\tau{}\psi{},\quad \frac{1}{\cosh x}\partial{}_\tau{}\psi{}\partial{}_x\psi{},\quad \frac{1}{\cosh x} \partial{}_x\psi{}\partial{}_x\psi{},\\
\frac{1}{(\cosh x)^{3/2}\cosh \tau{}}\partial{}_\tau\psi{}\Omega{}\psi{},\quad \frac{1}{(\cosh x)^{3/2}\cosh \tau{}}\partial{}_x\psi{}\Omega{}\psi{},\quad  \frac{1}{(\cosh x)^{1/2}\cosh \tau{}}\Omega{}\psi{}\Omega{}\psi{},\\
(\cosh x)(\partial{}_\theta{}\psi{})^2,\quad \frac{1}{(\cosh x)^{1/2}}\partial{}_\theta{}\psi{}\partial{}_\tau\psi{},\quad \frac{1}{(\cosh x)^{1/2}}\partial{}_\theta{}\psi{}\partial{}_x\psi{}, \quad \frac{\cosh x}{\cosh \tau{}}\partial{}_\theta{}\psi{}\Omega{}\psi{}, \\
\frac{1}{\cosh x}\psi{}\partial{}_\tau\psi{},\quad \frac{1}{\cosh x}\psi{}\partial{}_x\psi{},\quad \frac{1}{(\cosh x)^{3/2}\cosh \tau{}}\psi{}\Omega{}\psi{},\quad \frac{1}{(\cosh x)^{1/2}}\psi{}\partial{}_\theta{}\psi{}, \quad \frac{1}{\cosh x}\psi^2.
\end{gathered}
\end{equation}
Since (in the region \(\set{x\ge 1}\))
\begin{equation}
[\widehat{\partial},\partial{}_\tau{}] = 0,\qquad  [\widehat{\partial},\partial{}_x] = \frac{\mathfrak{f}''(x)}{\mathfrak{f}'(x)^2} O_{\mathfrak{D}}(1)\partial{}_x + \frac{\mathfrak{f}'(x)}{\mathfrak{f}(x)^2}O_{\mathfrak{D}}(1)\partial{}_\theta{},
\end{equation}
from the form of \(\mathfrak{f}\) (see \zcref{cart-def}) it follows that
\begin{equation}
\abs{[\partial{}_\tau{},\mathfrak{D}^{\mathbf{k}}]\psi{}} + \abs{[\partial{}_x,\mathfrak{D}^{\mathbf{k}}]\psi{}} \lesssim \sum_{\abs{\mathbf{k}'}<\abs{\mathbf{k}}}e^{-\frac{1}{2}\cosh x}(\abs{\partial{}_\tau{}\mathfrak{D}^{\mathbf{k}'}\psi{}} + \abs{\partial{}_x\mathfrak{D}^{\mathbf{k}'}\psi{}} + \abs{\partial{}_\theta{}\mathfrak{D}^{\mathbf{k}'}\psi{}}).
\end{equation}
From \zcref{hat-derivatives}, it follows that
\begin{equation}
\abs{[\partial{}_\theta{},\mathfrak{D}^{\mathbf{k}}]\psi{}}\lesssim \sum_{\abs{\mathbf{k}'}<\abs{\mathbf{k}}}e^{-\frac{1}{2}\cosh x}(\abs{\partial{}_\tau{}\mathfrak{D}^{\mathbf{k}'}\psi{}} + \abs{\partial{}_x\mathfrak{D}^{\mathbf{k}'}\psi{}} + \abs{\partial{}_\theta{}\mathfrak{D}^{\mathbf{k}'}\psi{}}).
\end{equation}
Since \(P\psi{} = \rho^2rF(\Grad \varphi{})\) when \(\Box{}\varphi{} = F(\Grad \varphi{})\), the result in the region
\(\set{x\ge 1}\) follows. The result in the region \(\set{x\le 1}\) is proved
similarly. We note only that, in this region, decaying \(x\)-weights are
irrelevant.

Since \(\partial_\tau{}\) and \(\partial_x\) commute with the restricted set of commutators
\(\mathring{\mathfrak{D}}\), the final statement of the proposition also holds.
\end{proof}
\subsection{Global existence for solutions to a semilinear equation outside of symmetry}
\label{sec:semilinear} The main result of this section is \zcref{semilinear}, which is a
small-data global existence result for solutions to a class of semilinear
equations including perturbations of \(\Box{}\varphi{} =
g^{\alpha{}\beta{}}\partial_\alpha{}\varphi{}\partial_\beta{}\varphi{}\). We do
not make any symmetry assumptions on the initial data (in contrast to the
results of \zcref{sec:semilinear-axisymmetric}).
\subsubsection{A refined estimate for almost \texorpdfstring{\(\U(1)\)}{U(1)}-symmetric nonlinearities}
We first prove a refined version of \zcref{Dk-P-psi-estimate} (which holds for
arbitrary nonlinearities satisfying the null condition) that applies to almost
\(\U(1)\)-symmetric nonlinearities (see \zcref{axisymmetric-nonlinearity}). The
non-\(\U(1)\)-symmetric part of such a nonlinearity has extra decay in space, and so a
certain ``bad frequency interaction'' appears only with a decaying weight in \(x\).
\begin{lemma}[Converting the equation solved by \(\varphi\) to an estimate for \(\psi\); almost \(\U(1)\)-symmetric case]
Suppose \(\varphi{}\in C^\infty(\mathcal{M})\) solves \(\Box{}\varphi{} = F(\Grad \varphi{})\) for a quadratic
nonlinearity \(F(\Grad \varphi{})\) satisfying the null condition with time
weight \(\ell{}(\tau{})\) (see \zcref{null-condition}). Suppose that \(F(\Grad
\varphi{})\) is almost \(\U(1)\)-symmetric (see \zcref{axisymmetric-nonlinearity}). Then,
in the region \(\set{x\ge 1}\), the quantity \(\psi{} = r\varphi{}\) satisfies the estimate
\begin{equation}\label{Dk-P-zero-mode}
\begin{split}
&\int _{S^1}\abs{(\mathfrak{D}^{\mathbf{k}}P\psi{})_0}^2\dd{}\theta{} \lesssim  \Bigl(\sum_{k=0}^{\abs{\mathbf{k}}}\abs{\partial{}_\tau^k\ell{}(\tau{})}^2\Bigr)\sum_{\mathbf{k}_1 + \mathbf{k}_2\le \mathbf{k}}\int _{S^1}(\textnormal{I}) + \textnormal{(II)}\dd{}\theta{},
\end{split}
\end{equation}
where
\begin{equation}\label{Dk-P-psi-estimate-I}
\begin{split}
\textnormal{(I)} &\coloneqq{}\frac{1}{\cosh^2 x}(\abs{\partial{}_\tau{}\mathfrak{D}^{\mathbf{k}_1}\psi{}}^2 + \abs{\partial{}_x\mathfrak{D}^{\mathbf{k}_1}\psi{}}^2)(\abs{\partial{}_\tau{}\mathfrak{D}^{\mathbf{k}_2}\psi{}}^2 + \abs{\partial{}_x\mathfrak{D}^{\mathbf{k}_2}\psi{}}^2) \\
&\qquad + \frac{1}{\cosh x\cosh^2 \tau{}}\abs{\Grad _{S^2}\mathfrak{D}^{\mathbf{k}_1}\psi{}}^2\Bigl(\frac{1}{\cosh^2 x}(\abs{\partial{}_\tau{}\mathfrak{D}^{\mathbf{k}_2}\psi{}}^2 + \abs{\partial{}_x\mathfrak{D}^{\mathbf{k}_2}\psi{}}^2 + \abs{\mathfrak{D}^{\mathbf{k}_2}\psi{}}^2)  + \abs{\Grad _{S^2}\mathfrak{D}^{\mathbf{k}_2}\psi{}}^2\Bigr) \\
&\qquad + \cosh x\abs{\partial{}_\theta{}\mathfrak{D}^{\mathbf{k}_1}\psi{}}^2\Bigl(\cosh x\abs{\partial{}_\theta{}\mathfrak{D}^{\mathbf{k}_2}\psi{}}^2 + \frac{1}{\cosh^2 x}(\abs{\partial{}_\tau{}\mathfrak{D}^{\mathbf{k}_2}\psi{}}^2 + \abs{\partial{}_x\mathfrak{D}^{\mathbf{k}_2}\psi{}}^2) \\
&\qquad +  \frac{1}{\cosh x\cosh^2 \tau{}}\abs{\Grad _{S^2}\mathfrak{D}^{\mathbf{k}_2}\psi{}}^2\Bigr) \\
\end{split}
\end{equation}
and
\begin{equation}\label{Dk-P-psi-estimate-II}
\begin{split}
\textnormal{(II)} &\coloneqq{} \frac{1}{\cosh^2 x}\bigl(\abs{\partial{}_\tau{}(\mathfrak{D}^{\mathbf{k}_1}\psi{})_0}^2 + \abs{\partial{}_x(\mathfrak{D}^{\mathbf{k}_1}\psi{})_0}^2 + \abs{(\mathfrak{D}^{\mathbf{k}_2}\psi{})_0}^2\bigr)\abs{(\mathfrak{D}^{\mathbf{k}_2}\psi{})_0}^2 \\
&\qquad +  \frac{1}{\cosh x}\bigl(\abs{\partial{}_\tau{}(\mathfrak{D}^{\mathbf{k}_1}\psi{})_{\ge 1}}^2  + (\cosh x)^{-1}\abs{\partial{}_x(\mathfrak{D}^{\mathbf{k}_1}\psi{})_{\ge 1}}^2 \\
&\qquad + \cosh x\abs{\partial{}_\theta{}(\mathfrak{D}^{\mathbf{k}_1}\psi{})_{\ge 1}}^2 + \abs{(\mathfrak{D}^{\mathbf{k}_2}\psi{})_{\ge 1}}^2\bigr)\abs{(\mathfrak{D}^{\mathbf{k}_2}\psi{})_{\ge 1}}^2  \\
&\qquad +   \frac{1}{\cosh^2 x}\bigl(\abs{\partial{}_\tau{}\mathfrak{D}^{\mathbf{k}_1}\psi{}}^2 + (\cosh x)^{-1}\abs{\partial{}_x\mathfrak{D}^{\mathbf{k}_1}\psi{}}^2 + \cosh x\abs{\partial{}_\theta{}\mathfrak{D}^{\mathbf{k}_1}\psi{}}^2 + \abs{\mathfrak{D}^{\mathbf{k}_2}\psi{}}^2\bigr)\abs{\mathfrak{D}^{\mathbf{k}_2}\psi{}}^2.
\end{split}
\end{equation}
Here \(P\) is the operator defined in \zcref{P-def} and \(\mathfrak{D}\) is a
commutator vector field (see \zcref{commutators}).
\label{Dk-P-psi-estimate-almost-axisymmetric}
\end{lemma}
\begin{proof}
We begin with an outline of the proof. We will project each term in the
expression for \(\mathfrak{D}^{\mathbf{k}}P\psi{}\) derived in the proof of
\zcref{Dk-P-psi-estimate} (see \zcref{semilinearity-decomposition}) to the zero
\(S^1\)-mode. Some of these terms, namely those that either contain an
\(\Omega{}\)-derivative or do not contain a zeroth-order term, are easier to
estimate in the proof of global existence (see \zcref{semilinear}). For these terms
we simply estimate the (\(L^2\) integral over \(S^1\) of the) zero mode by the
function itself. This is how the terms in \(\textnormal{(I)}\) are produced. The
terms in \(\textnormal{(II)}\) arise from the terms in the expression for
\(\mathfrak{D}^{\mathbf{k}}P\psi{}\) that contain a zeroth-order term and do not
contain an \(\Omega{}\)-derivative. To estimate these terms, we must distinguish
between the zero \(S^1\)-mode and the higher modes, and exploit the extra decay
enforced by the definition of almost \(\U(1)\)-symmetry.

We illustrate this with an example. One term that could appear is
\((f_{\bullet{}}\mathfrak{D}^{\mathbf{k}}\psi{}\partial{}_x\psi{})_0\), where
\(\mathfrak{D}^{\mathbf{k}}\) is a top-order derivative and \(f_\bullet{}\) is
one of the coefficients of the nonlinearity \(F(\Grad \varphi{})\) (see
\zcref{null-condition}). We claim that the worst term, with respect to decay in
\(x\), is the one where we must estimate \((\mathfrak{D}^{\mathbf{k}}\psi)_0\)
and \(\partial{}_x\psi{}_{\ge 1}\). Indeed, if the top-order term were
\((\mathfrak{D}^{\mathbf{k}}\psi)_{\ge 1}\), we could control it in \(L^2\) by
\(\partial_\theta{}\mathfrak{D}^{\mathbf{k}}\psi\) (using a Poincaré inequality
in \(S^1\)), which appears in the \(\partial_\tau{}\)-energy with a very strong
weight in \(x\). On the other hand, to control
\(\mathfrak{D}^{\mathbf{k}}\psi_0\) in \(L^2\), we must use the \(r^p\)-type
estimate of \zcref{rp-type-estimate}, which comes with a weaker weight in \(x\). The
most dangerous term is therefore one in which
\(\mathfrak{D}^{\mathbf{k}}\psi_0\) is multiplied by a quantity with the worst
possible decay in \(x\) that we must put in \(L^\infty\). In view of the
pointwise estimates of \zcref{global-pointwise-s,global-pointwise-tau}, this term is \(\partial_x\psi_{\ge
1}\). That is, the worst term is
\(\mathfrak{D}^{\mathbf{k}}\psi_0\partial_x\psi_{\ge 1}\). If \(f_\bullet{}\) is
\(\U(1)\)-symmetric, namely if \(f_\bullet{} = (f_\bullet{})_0\), then this term cannot
appear in \((f_\bullet{}\mathfrak{D}^\mathbf{k}\psi{}\partial_x\psi{})_0\). When
\(f_\bullet{}\) is not \(\U(1)\)-symmetric, there will be a term that looks like
\((f_\bullet{})_{\ge 1}\mathfrak{D}^{\mathbf{k}}\psi_0\partial_x\psi_{\ge 1}\).
The definition of almost \(\U(1)\)-symmetry is designed exactly so that
\((f_\bullet{})_{\ge 1}\) has enough \(x\)-decay that we can control this term.

We now turn to the proof of the estimate \zcref{Dk-P-zero-mode}. By the argument
used to prove \zcref{Dk-P-psi-estimate} (see in particular
\zcref{semilinearity-decomposition}), \((\mathfrak{D}^{\mathbf{k}}P\psi{})_0\) is a
sum of projections to the zero \(S^1\)-mode of terms that are products of
\(\mathfrak{D}^{\mathbf{k}_1}f_\bullet{}\) and one of the terms appearing in
\zcref{semilinearity-decomposition}, with the two occurrences of \(\psi{}\) there
replaced by \(\mathfrak{D}^{\mathbf{k}_2}\psi{}\) and
\(\mathfrak{D}^{\mathbf{k}_3}\psi{}\), for \(\mathbf{k}_1 + \mathbf{k}_2 +
\mathbf{k}_3\le \mathbf{k}\), potentially also multiplied by a \(\U(1)\)-symmetric
\(O_{\mathfrak{D}}(1)\) function (which will not affect the arguments below). An
example of such a term is
\begin{equation}
\textnormal{term} = \frac{1}{\cosh x}\mathfrak{D}^{\mathbf{k}_1}f_\bullet{}\partial{}_\tau{}\mathfrak{D}^{\mathbf{k}_2}\psi{}\partial{}_x\mathfrak{D}^{\mathbf{k}_3}\psi{}.
\end{equation}
If either one of the derivatives is an \(\Omega{}\)-derivative or there is no
zeroth-order term in this expression, we control
\begin{equation}
\int _{S^1} \abs{(\textnormal{term})_0}^2\dd{}\theta{}\le  \int _{S^1} \abs{\textnormal{term}}^2\dd{}\theta{}\lesssim \sum_{\mathbf{k}_1 + \mathbf{k}_2\le \mathbf{k}} \int _{S^1} \frac{1}{\cosh^2 x}\abs{\partial{}_\tau{}\mathfrak{D}^{\mathbf{k}_1}\psi{}}^2\abs{\mathfrak{D}^{\mathbf{k}_2}\psi{}}^2\dd{}\theta{}.
\end{equation}
This is how the terms in \(\textnormal{(I)}\) in \zcref{Dk-P-zero-mode} are
produced. The terms in \(\textnormal{(II)}\) arise when none of the terms in the
product is an \(\Omega{}\)-derivative, and (at least) one of the terms in the
product is a zeroth-order term. An example of such a term is
\begin{equation}
\textnormal{term}' = \frac{1}{\cosh x}\mathfrak{D}^{\mathbf{k}_1}f_\bullet{}\mathfrak{D}^{\mathbf{k}_2}\psi{}\partial{}_x\mathfrak{D}^{\mathbf{k}_3}\psi{}.
\end{equation}
For this we use the decomposition
\begin{equation}
(fgh)_0 = f_0g_0h_0 + f_0(g_{\ge 1}h_{\ge 1})_0 + (f_{\ge 1}(gh)_{\ge 1})_0,
\end{equation}
valid for functions \(f,g,h:S^1\to \C\), and estimate using Parseval's identity
\begin{equation}
\begin{split}
\int _{S^1}\abs{(fgh)_0}^2\dd{}\theta{} &\lesssim \int _{S^1}\abs{f_0}^2\abs{g_0}^2\abs{h_0}^2\dd{}\theta{}  +  \int _{S^1}\abs{f_0}^2\abs{(g_{\ge 1}h_{\ge 1})_0}^2\dd{}\theta{}  + \int _{S^1} \abs{f_{\ge 1}}^2\abs{(gh)_{\ge 1}}^2\dd{}\theta{} \\
&\lesssim \norm{f}_{L^\infty(S^1)}^2\int _{S^1}\abs{g_0}^2\abs{h_0}^2\dd{}\theta{} + \norm{f}_{L^\infty(S^1)}^2\int _{S^1}\abs{g_{\ge 1}}^2\abs{h_{\ge 1}}^2\dd{}\theta{} + \norm{f_{\ge 1}}_{L^\infty(S^1)}^2\int _{S^1}\abs{g}^2\abs{h}^2\dd{}\theta{}.
\end{split}
\end{equation}
Using now the assumption that \(F(\Grad \varphi{})\) is almost \(\U(1)\)-symmetric (see \zcref{axisymmetric-nonlinearity}), so that
\(\abs{(\mathfrak{D}^{\mathbf{k}}f_\bullet{})_{\ge 1}}^2\lesssim 1/(\cosh x)\) (in addition
to simply \(\abs{\mathfrak{D}^{\mathbf{k}}f_\bullet{}}\lesssim 1\)), we obtain
\begin{equation}
\begin{split}
\int _{S^1}\abs{(\textnormal{term}')_0}^2\dd{}\theta{} &\lesssim \sum_{\mathbf{k}_1+ \mathbf{k}_2\le \mathbf{k}} \int _{S^1} \frac{1}{\cosh^2 x}\abs{\partial{}_x(\mathfrak{D}^{\mathbf{k}_1}\psi{})_0}^2\abs{(\mathfrak{D}^{\mathbf{k}_2}\psi{})_0}^2\dd{}\theta{} \\
&\qquad  +  \int _{S^1} \frac{1}{\cosh^2 x}\abs{\partial{}_x(\mathfrak{D}^{\mathbf{k}_1}\psi{})_{\ge 1}}^2\abs{(\mathfrak{D}^{\mathbf{k}_2}\psi{})_{\ge 1}}^2\dd{}\theta{} \\
&\qquad  +  \int _{S^1} \frac{1}{\cosh^3 x}\abs{\partial{}_x\mathfrak{D}^{\mathbf{k}_1}\psi{}}^2\abs{\mathfrak{D}^{\mathbf{k}_2}\psi{}}^2\dd{}\theta{}.
\end{split}
\end{equation}
Repeating this procedure for all terms in \zcref{semilinearity-decomposition}, we get
the terms in \(\textnormal{(II)}\).
\end{proof}
\subsubsection{The global existence result}
\begin{proposition}[Global existence for solutions to a semilinear wave equation outside of symmetry]
Let \(F(\Grad \varphi{})\) be a quadratic nonlinearity satisfying the null condition with
time weight \(\ell{}(\tau{})\) (see \zcref{null-condition}). Suppose also that
\(F(\Grad \varphi{})\) is almost \(\U(1)\)-symmetric in the sense of
\zcref{axisymmetric-nonlinearity}.

Let \(k_{\textnormal{max}}\ge 8\). There is a large constant \(C_\ast{} =
C_\ast{}(k_{\textnormal{max}})\) and a small constant \(\epsilon{}_0 > 0\)
(independent of \(k_{\textnormal{max}}\)) such that whenever the time weight
\(\ell{}(\tau{})\) associated to the nonlinearity \(F(\Grad \varphi{})\) decays
at a sufficiently fast polynomial rate described by \(C_\ast{}\), in the sense
that \(\abs{\partial_\tau^k\ell{}(\tau{})}\lesssim (1+\tau{})^{-C_\ast{}}\) for
\(0\le k\le k_{\textnormal{max}}\), the following holds. Given
sufficiently small smooth initial data on \(\Sigma{}(0)\), of size at most \(\epsilon_0\), the unique solution of
\begin{equation}\label{box-semilinear-equation}
\Box{}\varphi{}=F(\Grad \varphi{})
\end{equation}
achieving this data is smooth and exists globally to the future of \(\Sigma{}(0)\).
More precisely, we assume that the initial data satisfies
\begin{equation}\label{data-small}
\sum_{\abs{\mathbf{k}}\le k_{\textnormal{max}}} E_T[\mathfrak{D}^{\mathbf{k}}\psi{}](0) + \sum_{\abs{\mathbf{k}}\le k_{\textnormal{max}}} E_p[(\mathfrak{D}^{\mathbf{k}}\psi{})_0](0) + \sum_{\abs{\mathbf{k}}\le k_{\textnormal{max}}-1}\sup_{x\ge 1}\int _{S^2\times S^1}\abs{(\mathfrak{D}^{\mathbf{k}}\psi{})_0}^2(\tau{}=0,x)\dd{}\omega{}\dd{}\theta{}\le \epsilon{}
\end{equation}
for some \(\epsilon{}\in [0,\epsilon{}_0]\).

Moreover, the solution grows at most polynomially in \(\tau{}\) in the norm of \zcref{data-small}, in the sense that
there exists a constant \(B = B(k_{\textnormal{max}}) > 0\) such that for \(\tau{}\ge
0\), we have
\begin{equation}\label{energy-growth}
\begin{split}
&\sum_{\abs{\mathbf{k}}\le k_{\textnormal{max}}} E_T[\mathfrak{D}^{\mathbf{k}}\psi{}](\tau{}) + \sum_{\abs{\mathbf{k}}\le k_{\textnormal{max}}} E_p[(\mathfrak{D}^{\mathbf{k}}\psi{})_{0}](\tau{}) +  \sum_{\abs{\mathbf{k}}\le k_{\textnormal{max}}-1}\sup_{x\ge 1}\int _{S^2\times S^1}\abs{(\mathfrak{D}^{\mathbf{k}}\psi{})_0}^2(\tau{},x)\dd{}\omega{}\dd{}\theta{}
\lesssim \epsilon{}(1 + \tau{})^B.
\end{split}
\end{equation}
We also establish the pointwise estimates
\begin{equation}\label{glob-exist-pointwise}
\begin{split}
&\sum_{\abs{\mathbf{k}}\le k_{\textnormal{max}}-3}\norm{\mathfrak{D}^{\mathbf{k}}\psi{}}_{L^\infty(\Sigma{}(\tau{}))}^2 + \sum_{\abs{\mathbf{k}}\le k_{\textnormal{max}}-4}\norm{\partial{}_x(\mathfrak{D}^{\mathbf{k}}\psi{})_0}_{L^\infty(\Sigma{}(\tau{}))}^2 + \sum_{\abs{\mathbf{k}}\le k_{\textnormal{max}}-4}\norm{(\cosh x)^{-1/2}\partial{}_x\mathfrak{D}^{\mathbf{k}}\psi{}}_{L^\infty(\Sigma{}(\tau{}))}^2 \\
&\qquad + \sum_{\abs{\mathbf{k}}\le k_{\textnormal{max}}-4}\norm{(\cosh x)^{1/2}\partial{}_\theta{}\mathfrak{D}^{\mathbf{k}}\psi{}}_{L^\infty(\Sigma{}(\tau{}))}^2 \\
&\lesssim \epsilon{}(1+\tau{})^{B}.
\end{split}
\end{equation}
Here the constants implicit in \(\lesssim\) depend on \(k_{\textnormal{max}}\).
\label{semilinear}
\end{proposition}
\begin{remark}[Global existence for the canonical nonlinearity]
\Cref{semilinear} can in particular be applied to the nonlinearity \(F(\Grad \varphi{}) =
f\cdot g^{\alpha{}\beta{}}\partial_\alpha{}\varphi{}\partial_\beta{}\varphi{}\) for \(f\) a \(\U(1)\)-symmetric function with bounded \(\mathfrak{D}\)-derivatives.
Indeed, by \zcref{canonical-double-null}, this nonlinearity is \(\U(1)\)-symmetric (hence
almost \(\U(1)\)-symmetric in the sense of \zcref{axisymmetric-nonlinearity}) and satisfies
the null condition as in \zcref{null-condition} with the exponentially decaying time
weight \(e^{-\tau{}}\).
\end{remark}
\begin{remark}[Why the solution is allowed to grow in time]
The polynomial growth in \(\tau{}\) of \(\psi{}\) in \zcref{semilinear} corresponds to a
logarithmic growth in the Minkowskian time coordinate \(t = \rho{}\sinh
\tau{}\). In particular, our proof does not show that the solution \(\varphi{}\)
has a radiation field, namely a finite limit for \(\psi{}=r\varphi{}\) as
\(\tau{}\to \infty\) along curves where \(\tau{}-x\) is constant.

We now explain the source of this growth in our proof. We must control
\(\partial_x\psi_{\ge 1}\) pointwise, all the way up to the bubble. In order to do so, we
use the Cartesian commutator vector fields \(\partial_{\hat{y}}\) and
\(\partial_{\hat{z}}\) (see \zcref{commutators}). This is because \(\partial_x\)
is a bounded linear combination of \(\partial_{\hat{y}}\) and
\(\partial{}_{\hat{z}}\). Commuting with these vector fields produces a
(non-\(\U(1)\)-symmetric) bulk error term that is quadratic in \(\psi_{\ge 1}\) and its
derivatives. Since \(\psi_{\ge 1}\) may be periodic in time (see
\zcref{time-periodic-intro} of \zcref{main-theorem-linear-intro}), we must treat this
error term as the integral in \(\tau{}\) of the energy flux through
\(\Sigma{}(\tau{})\). In particular, this error term grows in time.
\label{weak-trapping}
\end{remark}
\begin{remark}[The unique value of \(p\) in the \(r^p\)-type estimates that we can use to close the proof]
In \zcref{rp-relation}, we explained that the \(r^p\)-type estimate of
\zcref{rp-type-estimate}, which is crucial in our proof of \zcref{semilinear},
corresponds to the value \(p = 3/2\) in the usual \(r^p\)-weighted energy
estimates of Dafermos--Rodnianski \cite{rp-method} on Minkowski space \(\R^{3 +
1}\). By replacing \(\sinh x\) in the estimate for \zcref{rp-type-estimate} with
\((\sinh x)^q\) for \(q\in (0,2)\), we can extend our estimates to the range
analogous to \(p\in (1,2)\).

However, we must use the estimate corresponding to \(p = 3/2\), or \(q = 1\), to
close our proof of \zcref{semilinear}. The restriction \(q\ge 1\) comes from the
need to control a spacetime integral of a zeroth-order term \((\cosh
x)^{-1}\abs{\psi{}_0}^2\) in a large-\(x\) region (see
\zcref{Dk-Ppsi-estimate-pointwise}), while the \((\sinh x)^q\)-estimate would
control a term \((\sinh x)^q(\cosh x)^{-2}\abs{\psi{}_0}^2\). On the other hand,
we must have \(q\le 1\) because the \((\sinh x)^q\)-estimate would produce a
spacetime error term of the form \((\sinh x)^q(\cosh
x)^{-1}(\partial{}_\tau{}\psi{}_0)^2\) (see \zcref{semilinear-hardest-term-rp-I}),
while the best \(x\)-weight we can control in the bulk corresponds to
\((\partial_\tau{}\psi_0)^2\) (by treating the bulk integral as an integral of the
\(\partial_\tau{}\)-energy fluxes).
\label{semilinear-rp-value}
\end{remark}
\begin{proof}
The uniqueness aspect of \zcref{semilinear} follows from a suitable local
well-posedness theory in the norm of \zcref{data-small}. Given the a priori
estimates and pointwise estimates we have established (with respect to this
norm), the development of such a theory is standard (see for example
\cite[Chap.~9]{MR2527641} for details). The smoothness of the solution given smooth
initial data and the statement that \(\epsilon_0\) can be chosen independent of
\(k_{\textnormal{max}}\) follows from a persistence of regularity statement that
is part of the well-posedness theory. For the rest of the proof, we treat
\(k_{\textnormal{max}}\) as fixed.

To prove the global existence part of \zcref{semilinear}, we will use a continuity,
or bootstrap, argument. That is, we will show that the so-called ``bootstrap set''
of times \(\tau_{\textnormal{boot}}\in (0,\infty)\) for which the solution
exists and satisfies certain ``bootstrap assumptions'' (see
\zcref{b-boot-1,b-boot-2,b-boot-3}) in the region \(\set{0\le \tau{}\le
\tau{}_{\textnormal{boot}}}\) is non-empty, open, and closed in \((0,\infty)\).
Since the interval \((0,\infty)\) is connected, this shows that the solution
exists globally to the future of \(\Sigma{}(0)\) and satisfies the desired
estimates in this region. The non-emptiness of the bootstrap set follows from
local existence, and the closedness follows from a suitable Cauchy stability
statement that is part of the local well-posedness theory. We will focus on the
most difficult part, which is to show that the bootstrap set is open. This
follows from a standard argument using local existence together with the
``improvement of the bootstrap assumptions,'' namely the statement that the
bootstrap assumptions \zcref{b-boot-1,b-boot-2,b-boot-3} imply the same estimates
with a constant on the right-hand side that is improved by a factor \(1/2\).
This is done in Steps 2--4.

In Step 0, we introduce three bootstrap assumptions. In Step 1, we derive
pointwise estimates as consequences of the bootstrap assumptions. In Step 2--4,
we improve the bootstrap assumptions, and in Step 5, we complete the proof.

\step{Step 0: Setup and bootstrap assumptions.} For \(\abs{\mathbf{k}}\le
k_{\textnormal{max}}\), define constants \(A_{\mathbf{k}}\gg 1\) such that
\(A_{\mathbf{k}'}\ll A_{\mathbf{k}}\) whenever \(\mathbf{k}'<\mathbf{k}\). Write
\(A_{<\mathbf{k}} \coloneqq{} \max_{\mathbf{k}'<\mathbf{k}}A_{\mathbf{k}'}\) and
\(A \coloneqq{}\max _{\abs{\mathbf{k}}\le k_{\textnormal{max}}}A_{\mathbf{k}}\).
We assume \(\epsilon{}_0\) is chosen so that \(A\epsilon{}_0\ll 1\). Let
\(\mathbf{k}_{\textnormal{max}}\) be the largest multi-index with magnitude at
most \(k_{\textnormal{max}}\). For a multi-index \(\mathbf{k}\), let
\(\#\mathbf{k}\) denote the number of multi-indices \(\mathbf{k}'\) (including
the zero multi-index) such that \(\mathbf{k}' < \mathbf{k}\). Define
\(C_\ast{}\coloneqq{}2 + 3\#\mathbf{k}_{\textnormal{max}}\). For
\(\abs{\mathbf{k}}\le k_{\textnormal{max}}\), we suppose that the following
bootstrap assumptions hold for \(\tau{}\in{}[0,\tau{}_{\textnormal{boot}}]\) and
some \(\tau_{\textnormal{boot}} > 0\):
\begin{align}
E_T[\mathfrak{D}^{\mathbf{k}}\psi{}](\tau{})&\le A_{\mathbf{k}}\epsilon{}(1+\tau{})^{2\#\mathbf{k}}, \label{b-boot-1}\\
\int _{\Sigma{}(\tau{})}(P\mathfrak{D}^{\mathbf{k}}\psi{})^2\tanh x\dd{}\mu{}&\le A_{\mathbf{k}}\epsilon{}(1+\tau{})^{2\#\mathbf{k}},\label{b-boot-2} \\
E_p[(\mathfrak{D}^{\mathbf{k}}\psi{})_0](\tau{}) + F_p[(\mathfrak{D}^{\mathbf{k}}\psi{})_0](0,\tau{}) &\le A_{\mathbf{k}}\epsilon{}(1+\tau{})^{2\#\mathbf{k}}. \label{b-boot-3}
\end{align}
Here \(P\) is the operator defined in \zcref{P-def}.

\step{Step 1: Pointwise estimates.} We first derive some pointwise estimates that are
consequences of the bootstrap assumptions. By the bootstrap assumptions
\zcref{b-boot-1,b-boot-2}, and the estimates \zcref{Loo-0,Loo-1,Loo-2a,Loo-2b} of
\zcref{global-pointwise-tau}, and the smallness assumption \zcref{data-small} on the initial
data, the following pointwise estimates hold for \(\tau{}\in
[0,\tau{}_{\textnormal{boot}}]\):
\begin{equation}\label{pointwise-bootstrap}
\begin{split}
&\sum_{\abs{\mathbf{k}}\le k_{\textnormal{max}}-3}\norm{\mathfrak{D}^{\mathbf{k}}\psi{}}_{L^\infty(\Sigma{}(\tau{}))}^2 + \sum_{\abs{\mathbf{k}}\le k_{\textnormal{max}}-4}\norm{\partial{}_\tau{}\mathfrak{D}^{\mathbf{k}}\psi{}}_{L^\infty(\Sigma{}(\tau{}))}^2  +
 \sum_{\abs{\mathbf{k}}\le k_{\textnormal{max}}-4}\norm{(\cosh x)^{1/2}\partial{}_\theta{}\mathfrak{D}^{\mathbf{k}}\psi{}}_{L^\infty(\Sigma{}(\tau{}))}^2\\
&\qquad + \sum_{\abs{\mathbf{k}}\le k_{\textnormal{max}}-4}\norm{\partial{}_x(\mathfrak{D}^{\mathbf{k}}\psi{})_0}_{L^\infty(\Sigma{}(\tau{})\cap \set{x\ge 1})}^2 + \sum_{\abs{\mathbf{k}}\le k_{\textnormal{max}}-4}\norm{(\cosh x)^{-1/2}\partial{}_x\mathfrak{D}^{\mathbf{k}}\psi{}}_{L^\infty(\Sigma{}(\tau{}))}^2 \\
&\lesssim A\epsilon{}(1+\tau{})^{2\#\mathbf{k}_{\textnormal{max}}}.
\end{split}
\end{equation}
To control the terms on the first line, we have used the fact that
\(\partial_\tau{}\) and \(\partial_{\theta{}}\) are among the commutator vector
fields \(\mathfrak{D}\) and that
\(\partial_\theta{}\mathfrak{D}^{\mathbf{k}}\psi{} =
(\partial_\theta{}\mathfrak{D}^{\mathbf{k}}\psi{})_{\ge 1} =
(\mathfrak{D}^{\mathbf{k}'}\psi{})_{\ge 1}\) for \(\abs{\mathbf{k}'}\le
\abs{\mathbf{k}} + 1\), and then used \zcref{Loo-0,Loo-1}. The terms on the second
line are controlled using \zcref{Loo-2a,Loo-2b}.

\step{Step 2: Improving the bootstrap assumption \zcref{b-boot-1}.} In this step, we improve
the bootstrap assumption \zcref{b-boot-1} on the \(\partial_\tau{}\)-energy (and corresponding
flux on ingoing null cones) associated to \(\mathfrak{D}^{\mathbf{k}}\psi{}\). The key step is to prove the following
control for the bulk term that appears on the right-hand side of the
\(\partial_\tau{}\)-energy estimate for \(\mathfrak{D}^{\mathbf{k}}\psi{}\):
\begin{equation}\label{semilinear-T-bulk}
\int_0^{\tau{}_{\textnormal{boot}}} \int _{\Sigma{}(\tau{})}\abs{\partial{}_\tau{}\mathfrak{D}^{\mathbf{k}}\psi{}}\abs{P\mathfrak{D}^{\mathbf{k}}\psi{}}\tanh x\dd{}\mu{}\dd{}\tau{} \lesssim A_{\mathbf{k}}\epsilon{}\cdot A^{1/2}\epsilon^{1/2} + A_{\mathbf{k}}^{1/2}A_{<\mathbf{k}}^{1/2}\epsilon{}(1+\tau{})^{2\#\mathbf{k}}, \quad \abs{\mathbf{k}}\le k_{\textnormal{max}}.
\end{equation}
Indeed, once we have \zcref{semilinear-T-bulk}, the \(\partial_\tau{}\)-energy estimate of
\zcref{T-estimate-prop} (applied on the interval \([0,\tau{}]\) for \(\tau{}\in
[0,\tau{}_{\textnormal{boot}}]\)) and the smallness assumption \zcref{data-small}
imply
\begin{equation}\label{b-boot-1-improvement}
E_T[\mathfrak{D}^{\mathbf{k}}\psi{}](\tau{}) + F_T[\mathfrak{D}^{\mathbf{k}}\psi{}](\tau{})\lesssim \epsilon{} + A_{\mathbf{k}}\epsilon{}\cdot A^{1/2}\epsilon^{1/2} + A_{\mathbf{k}}^{1/2}A_{<\mathbf{k}}^{1/2}\epsilon{}(1+\tau{})^{2\#\mathbf{k}}.
\end{equation}
Since \(A_0\gg 1\) and \(A\epsilon{}\ll 1\) and \(A_{<\mathbf{k}}\ll A_{\mathbf{k}}\), the right side of
\zcref{b-boot-1-improvement} can be made less than
\(\frac{1}{2}A_{\mathbf{k}}\epsilon{}(1+\tau{})^{2\#\mathbf{k}}\). This improves
the bootstrap assumption \zcref{b-boot-1}. To establish \zcref{semilinear-T-bulk}, we
substitute \zcref{semilinear-prep-0} (proved in Step 2a) and \zcref{semilinear-DkPpsi}
(proved in Step 2b) into the basic estimate
\begin{equation}\label{basic-estimate-commutation}
\abs{P\mathfrak{D}^{\mathbf{k}}\psi{}}\le \abs{\mathfrak{D}^{\mathbf{k}}P\psi{}} + \abs{[P,\mathfrak{D}^{\mathbf{k}}]\psi{}}.
\end{equation}

\step{Step 2a: Estimate for the quadratic terms on the right-hand side of \zcref{basic-estimate-commutation}.} We first estimate the contribution
coming from the quadratic terms in \(\mathfrak{D}^{\mathbf{k}}P\psi{}\) (as
opposed to the linear terms in \([P,\mathfrak{D}^{\mathbf{k}}]\psi{}\), which we
consider in Step 2b). Use \zcref{Dk-P-psi-estimate} and the assumptions on the time
weight \(\ell{}\) to estimate
\begin{equation}
\begin{split}
\abs{\mathfrak{D}^{\mathbf{k}}P\psi{}} &\lesssim (1+\tau{})^{-C_\ast{}}\sum_{\mathbf{k}_1 + \mathbf{k}_2 \le  \mathbf{k}}[\textnormal{(I)} + \textnormal{(II)} + \textnormal{(III)}],
\end{split}
\end{equation}
where
\begin{equation}\label{semilinear-I--III}
\begin{split}
\textnormal{(I)}&\coloneqq{}\sum_{\mathbf{k}_1 + \mathbf{k}_2 \le  \mathbf{k}}\Bigl[\frac{1}{\cosh x}(\abs{\partial{}_\tau{}\mathfrak{D}^{\mathbf{k}_1}\psi{}} + \abs{\partial{}_x\mathfrak{D}^{\mathbf{k}_1}\psi{}} + \abs{\mathfrak{D}^{\mathbf{k}_1}\psi{}})(\abs{\partial{}_\tau{}\mathfrak{D}^{\mathbf{k}_2}\psi{}} + \abs{\partial{}_x\mathfrak{D}^{\mathbf{k}_2}\psi{}} + \abs{\mathfrak{D}^{\mathbf{k}_2}\psi{}})\Bigr] \\
\textnormal{(II)} &\coloneqq{} \sum_{\mathbf{k}_1 + \mathbf{k}_2 \le  \mathbf{k}}\frac{1}{\cosh \tau{}}\abs{\Grad _{S^2}\mathfrak{D}^\mathbf{k_1}\psi{}}\Bigl(\abs{\Grad _{S^2}\mathfrak{D}^\mathbf{k_2}\psi{}} + \frac{1}{\cosh x}(\abs{\partial{}_\tau{}\mathfrak{D}^\mathbf{k_2}\psi{}} + \abs{\partial{}_x\mathfrak{D}^\mathbf{k_2}\psi{}} + \abs{\mathfrak{D}^\mathbf{k_2}\psi{}}) \\
&\qquad + (\cosh x)^{1/2}\abs{\partial{}_\theta{}\mathfrak{D}^{\mathbf{k}_2}\psi{}}\Bigr),  \\
\textnormal{(III)} &\coloneqq{}\sum_{\mathbf{k}_1 + \mathbf{k}_2 \le  \mathbf{k}}(\cosh x)^{1/2}\abs{\partial{}_\theta{}\mathfrak{D}^{\mathbf{k}_1}\psi{}}\Bigl((\cosh x)^{1/2}\abs{\partial{}_\theta{}\mathfrak{D}^{\mathbf{k}_{2}}\psi{}} + \frac{1}{\cosh x}(\abs{\partial{}_\tau{}\mathfrak{D}^{\mathbf{k}_2}\psi{}} + \abs{\partial{}_x\mathfrak{D}^{\mathbf{k}_2}\psi{}} + \abs{\mathfrak{D}^{\mathbf{k}_2}\psi{}}) \\
&\qquad + \frac{1}{\cosh \tau{}}\abs{\Grad _{S^2}\mathfrak{D}^{\mathbf{k}_2}\psi{}}\Bigr).
\end{split}
\end{equation}
We can estimate terms \(\textnormal{(I)}\)--\(\textnormal{(III)}\) by
considering the case when \(\abs{\mathbf{k}_1}\le k_{\textnormal{max}}/2\) and the case when
\(\abs{\mathbf{k}_2}\le k_{\textnormal{max}}/2\) and bounding the term with
fewer \(\mathfrak{D}\)-derivatives using \zcref{pointwise-bootstrap}; since
\(k_{\textnormal{max}}\ge 8\), these terms will indeed have at most
\(k_{\textnormal{max}} - 4\) many \(\mathfrak{D}\)-derivatives. In this way
(using also that the rotations \(\Omega{}_i\) are commutator vector fields to
put terms of the form \(\abs{\Grad _{S^2}\mathfrak{D}^{\mathbf{k}'}\psi{}}
\lesssim \abs{\Omega{}\mathfrak{D}^{\mathbf{k}'}\psi{}}\) in \(L^\infty\)), we obtain
\begin{equation}\label{Dk-Ppsi-estimate-pointwise}
\begin{split}
\abs{\mathfrak{D}^{\mathbf{k}}P\psi{}} &\lesssim A^{1/2}\epsilon^{1/2}(1+\tau{})^{-C_\ast{}+\#\mathbf{k}_{\textnormal{max}}}\sum_{\mathbf{k}'\le \mathbf{k}} \Bigl[\frac{1}{(\cosh x)^{1/2}}(\abs{\partial{}_\tau{}\mathfrak{D}^{\mathbf{k}'}\psi{}} + \abs{\partial{}_x\mathfrak{D}^{\mathbf{k}'}\psi{}}) \\
&\qquad + \frac{1}{\cosh \tau{}}\abs{\Grad _{S^2}\mathfrak{D}^{\mathbf{k}'}\psi{}} + (\cosh x)^{1/2}\abs{\partial{}_\theta{}\mathfrak{D}^{\mathbf{k}'}\psi{}} \Bigr]\\
&\qquad +A^{1/2}\epsilon^{1/2}(1+\tau{})^{-C_\ast{}+\#\mathbf{k}_{\textnormal{max}}}\sum_{\mathbf{k}'\le \mathbf{k}}\frac{1}{(\cosh x)^{1/2}}\abs{\mathfrak{D}^{\mathbf{k}'}\psi{}}.
\end{split}
\end{equation}
Note that the zeroth-order term on the right-hand side of
\zcref{Dk-Ppsi-estimate-pointwise} comes from term \(\textnormal{(I)}\), after
estimating \(\abs{\partial_x\mathfrak{D}^{\mathbf{k}'}\psi{}}\lesssim (1 +
\tau{})^{\#\mathbf{k}_{\textnormal{max}}}(\cosh x)^{1/2}\) using \zcref{pointwise-bootstrap}. We now multiply
by \(\abs{\partial_\tau{}\mathfrak{D}^{\mathbf{k}}\psi{}}\), integrate in
spacetime, apply Cauchy--Schwarz in space, and use
\zcref{Dk-Ppsi-estimate-pointwise} to obtain
\begin{equation}\label{semilinear-Dk-P-psi-prep}
\begin{split}
&\int_0^{\tau{}_{\textnormal{boot}}} \int _{\Sigma{}(\tau{})}\abs{\partial{}_\tau{}\mathfrak{D}^{\mathbf{k}}\psi{}}\abs{\mathfrak{D}^{\mathbf{k}}P\psi{}}\tanh x\dd{}\mu{}\dd{}\tau{} \\
&\int_0^{\tau{}_{\text{boot}}} \Bigl(\int _{\Sigma{}(\tau{})}\abs{\partial{}_\tau{}\mathfrak{D}^{\mathbf{k}}\psi{}}^2\tanh x\dd{}\mu{}\Bigr)^{1/2}\Bigl(\int _{\Sigma{}(\tau{})}\abs{\mathfrak{D}^{\mathbf{k}}P\psi{}}^2\tanh x\Bigr)^{1/2}\dd{}\tau{} \\
&\lesssim A^{1/2}\epsilon^{1/2}\int_0^{\tau{}_{\text{boot}}} (1+\tau{})^{-C_\ast{}+\#\mathbf{k}_{\textnormal{max}}}\mathfrak{E}_1(\tau{})^{1/2}(\mathfrak{E}_2(\tau{}) + \mathfrak{E}_3(\tau{}))^{1/2}\dd{}\tau{}.
\end{split}
\end{equation}
Here we have written
\begin{equation}\label{semilinear-frak}
\begin{split}
\mathfrak{E}_1(\tau{}) &\coloneqq{}\int _{\Sigma{}(\tau{})}(\partial{}_\tau{}\mathfrak{D}^{\mathbf{k}}\psi{})^2 \tanh x\dd{}\mu{},\\
\mathfrak{E}_2(\tau{}) &\coloneqq{}\sum_{\mathbf{k}'\le \mathbf{k}}\int _{\Sigma{}(\tau{})}\Bigl[(\partial{}_\tau{}\mathfrak{D}^{\mathbf{k}}\psi{})^2 +  (\partial{}_x\mathfrak{D}^{\mathbf{k}'}\psi{})^2 + \frac{1}{\cosh ^2\tau{}}\abs{\Grad _{S^2}\mathfrak{D}^{\mathbf{k}'}\psi{}}^2 + (\cosh x)(\partial{}_\theta{}\mathfrak{D}^{\mathbf{k}'}\psi{})^2\Bigr]\tanh x\dd{}\mu{}, \\
\mathfrak{E}_3(\tau{}) &\coloneqq{} \sum_{\mathbf{k}'\le \mathbf{k}}\int _{\Sigma{}(\tau{})}\frac{1}{\cosh x}\abs{\mathfrak{D}^{\mathbf{k}'}\psi{}}^2\dd{}\mu{}.
\end{split}
\end{equation}
Terms \(\mathfrak{E}_1\) and \(\mathfrak{E}_2\) are controlled by the
\(\partial_\tau{}\)-energy of \(\mathfrak{D}^{\mathbf{k}}\psi{}\) and by the sum
of the \(\partial_\tau{}\)-energies of \(\mathfrak{D}^{\mathbf{k}'}\psi{}\) for
\(\mathbf{k}'\le \mathbf{k}\), respectively. To control \(\mathfrak{E}_3\), we
will consider separately the zero \(S^1\)-mode and the higher \(S^1\)-modes,
splitting
\begin{equation}\label{semilinear-split}
\frac{1}{\cosh x}\abs{\mathfrak{D}^{\mathbf{k}'}\psi{}}^2\le \frac{1}{\cosh x}\abs{(\mathfrak{D}^{\mathbf{k}'}\psi{})_{0}}^2 + \frac{1}{\cosh x}\abs{(\mathfrak{D}^{\mathbf{k}'}\psi{})_{\ge 1}}^2.
\end{equation}
The term involving the zero mode is controlled by the
\(r^p\)-type energy of \((\mathfrak{D}^{\mathbf{k}}\psi{})_0\). The term involving the \(\ge 1\) modes is controlled by the
\(\partial_\theta{}\)-term in the \(\partial_\tau{}\)-energy of \(\mathfrak{D}^{\mathbf{k}}\psi{}\), using a Poincaré
inequality on \(S^1\). By the bootstrap assumptions \zcref{b-boot-1,b-boot-3}, we
have
\begin{equation}\label{semilinear-E-estimate}
\mathfrak{E}_1(\tau{}) + \mathfrak{E}_2(\tau{}) + \mathfrak{E}_3(\tau{})\lesssim A_{\mathbf{k}}(1+\tau{})^{2\#\mathbf{k}}\epsilon{}\lesssim A_{\mathbf{k}}(1+\tau{})^{2\#\mathbf{k}_{\textnormal{max}}}\epsilon{}.
\end{equation}
Substituting \zcref{semilinear-E-estimate} into \zcref{semilinear-Dk-P-psi-prep} and
using the integrability in \(\tau{}\) of \((1+\tau{})^{-C_\ast{} +
3\#\mathbf{k}_{\textnormal{max}}}\) (by the definition of \(C_\ast{}\)), we
conclude that
\begin{equation}\label{semilinear-DkPpsi}
\begin{split}
&\int_0^{\tau{}_{\text{boot}}} \int _{\Sigma{}(\tau{})}\abs{\partial{}_\tau{}\mathfrak{D}^{\mathbf{k}}\psi{}}\abs{\mathfrak{D}^{\mathbf{k}}P\psi{}}\tanh x\dd{}\mu{}\dd{}\tau{} \lesssim A_{\mathbf{k}}\epsilon{}\cdot A^{1/2}\epsilon^{1/2}.
\end{split}
\end{equation}

\step{Step 2b: Estimate for the linear terms on the right-hand side of \zcref{basic-estimate-commutation}.} We now estimate the contribution of the
linear terms in \([P,\mathfrak{D}^{\mathbf{k}}]\psi{}\). By the commutation
formula of \zcref{frak-D-commutation}, we see that
\(\abs{[P,\mathfrak{D}^{\mathbf{k}}]\psi{}}^2\tanh x\) is controlled by the
integrand of the \(\partial_\tau{}\)-energy and
\(\sum_{\mathbf{k}'<\mathbf{k}}\abs{\mathfrak{D}^{\mathbf{k}'}P\psi{}}\). After
applying Cauchy--Schwarz in space, we therefore obtain
\begin{equation}\label{semilinear-P-Dk-comm}
\begin{split}
&\int_0^{\tau{}_{\text{boot}}} \int _{\Sigma{}(\tau{})}\abs{\partial{}_\tau{}\mathfrak{D}^{\mathbf{k}}\psi{}}\abs{[P,\mathfrak{D}^{\mathbf{k}}]\psi{}}\tanh x\dd{}\mu{}\dd{}\tau{} \\
&\lesssim \sum_{\mathbf{k}'<\mathbf{k}} \int_0^{\tau{}_{\text{boot}}} \int _{\Sigma{}(\tau{})}\abs{\partial{}_\tau{}\mathfrak{D}^{\mathbf{k}}\psi{}}\abs{\mathfrak{D}^{\mathbf{k}'}P\psi{}}\tanh x\dd{}\mu{}\dd{}\tau{} + \sum_{\mathbf{k}'<\mathbf{k}}\int_0^{\tau{}_{\text{boot}}} E_T[\mathfrak{D}^{\mathbf{k}}\psi{}](\tau{})^{1/2}E_T[\mathfrak{D}^{\mathbf{k}'}\psi{}](\tau{})^{1/2}\dd{}\tau{}.
\end{split}
\end{equation}
The first term on the right-hand side of \zcref{semilinear-P-Dk-comm} is controlled
by \zcref{semilinear-DkPpsi}. Since \(\#\mathbf{k}'\le \#\mathbf{k}-1\) whenever
\(\mathbf{k}' < \mathbf{k}\), we can use the bootstrap assumption \zcref{b-boot-1}
to estimate the second term on the right-hand side of \zcref{semilinear-P-Dk-comm} as
follows:
\begin{equation}\label{semilinear-prep-0}
 \sum_{\mathbf{k}'<\mathbf{k}}\int_0^{\tau{}_{\text{boot}}} E_T[\mathfrak{D}^{\mathbf{k}}\psi{}](\tau{})E_T[\mathfrak{D}^{\mathbf{k}'}\psi{}](\tau{})\dd{}\tau{}\lesssim A_{\mathbf{k}}^{1/2}A_{<\mathbf{k}}^{1/2}\epsilon{}\int_0^{\tau{}_{\text{boot}}} (1+\tau{})^{\#\mathbf{k}}(1+\tau{})^{\#\mathbf{k}-1} \dd{}\tau{}\lesssim  A_{\mathbf{k}}^{1/2}A_{<\mathbf{k}}^{1/2}\epsilon{}(1+\tau{})^{2\#\mathbf{k}}.
\end{equation}

\step{Step 3: Improving the bootstrap assumption \zcref{b-boot-2}.} In this step we improve
the bootstrap assumption \zcref{b-boot-2}. By the basic estimate
\zcref{basic-estimate-commutation} and the commutation formula of
\zcref{frak-D-commutation} we have
\begin{equation}
\begin{split}
\int _{\Sigma{}(\tau{})}\abs{P\mathfrak{D}^{\mathbf{k}}\psi{}}^2\tanh x\dd{}\mu{} &\lesssim \sum_{\mathbf{k}'\le \mathbf{k}}\int _{\Sigma{}(\tau{})}\abs{\mathfrak{D}^{\mathbf{k}'}P\psi{}}^2\tanh x\dd{}\mu{} + \sum_{\mathbf{k}'<\mathbf{k}}E_T[\mathfrak{D}^{\mathbf{k}'}\psi{}](\tau{}).
\end{split}
\end{equation}
Estimating the first term on the right-hand side as in Step 2 (see
\zcref{semilinear-Dk-P-psi-prep,semilinear-frak} and the following discussion), we
obtain
\begin{equation}
\int _{\Sigma{}(\tau{})}\abs{P\mathfrak{D}^{\mathbf{k}}\psi{}}^2\tanh x\dd{}\mu{} \lesssim  A\epsilon{}(1+\tau{})^{2(-C_\ast{}+\#\mathbf{k}_{\textnormal{max}})}\sum_{\mathbf{k}'\le \mathbf{k}}(E_T[\mathfrak{D}^{\mathbf{k}'}\psi{}](\tau{}) + E_p[(\mathfrak{D}^{\mathbf{k}'}\psi{})_0](\tau{}) ) + \sum_{\mathbf{k}'<\mathbf{k}}E_T[\mathfrak{D}^{\mathbf{k}'}\psi{}](\tau{}).
\end{equation}
By the definition of \(C_\ast{}\), the bootstrap assumption \zcref{b-boot-1}, and the
fact that \(\#\mathbf{k}'\le \#\mathbf{k}-1\) when \(\mathbf{k}'<\mathbf{k}\),
we conclude that
\begin{equation}\label{PDk-Psi-on-Sigma-tau}
\int _{\Sigma{}(\tau{})}\abs{P\mathfrak{D}^{\mathbf{k}}\psi{}}^2\tanh x\dd{}\mu{} \lesssim (A\epsilon{}\cdot A_{\mathbf{k}}\epsilon{} + A_{<\mathbf{k}}\epsilon{})(1+\tau{})^{2\#\mathbf{k}-2}
\end{equation}
Since \(A\epsilon{}\ll 1\) and \(A_{<\mathbf{k}}\ll A_{\mathbf{k}}\), this improves the
bootstrap assumption \zcref{b-boot-2}.

\step{Step 4: Improving the bootstrap assumption \zcref{b-boot-3}.} In this step, we
improve the bootstrap assumption \zcref{b-boot-3} on the \(r^p\)-type energy of
\((\mathfrak{D}^{\mathbf{k}}\psi{})_0\). The key step is to prove the following
control for the bulk term that appears on the right-hand side of the
\(r^p\)-type energy estimate for \((\mathfrak{D}^{\mathbf{k}}\psi{})_0\):
\begin{equation}\label{semilinear-3-goal}
\int_0^{\tau{}_{\textnormal{boot}}} \int _{\Sigma{}(\tau{})} \sinh x \abs{P(\mathfrak{D}^{\mathbf{k}}\psi{})_0}^2\dd{}\mu{}\dd{}\tau{}\lesssim (A\epsilon{}\cdot A_{\mathbf{k}}\epsilon{} + A_{<\mathbf{k}}\epsilon{})(1+\tau{})^{2\#\mathbf{k}}.
\end{equation}
Indeed, once we obtain \zcref{semilinear-3-goal}, we can apply the \(r^p\)-type
energy estimate of \zcref{rp-type-estimate} to the \(\U(1)\)-symmetric scalar field
\((\mathfrak{D}^{\mathbf{k}}\psi{})_0\) on the time interval \([0,\tau{}]\)
(where \(\tau{}\in [0,\tau{}_{\textnormal{boot}}]\)) and use the smallness
assumption \zcref{data-small} on the initial data to obtain
\begin{equation}\label{b-boot-3-improvement}
\begin{split}
&E_T[(\mathfrak{D}^{\mathbf{k}}\psi{})_0](\tau{}) + E_p[(\mathfrak{D}^{\mathbf{k}}\psi{})_0](\tau{}) + F_p[(\mathfrak{D}^{\mathbf{k}}\psi{})_0](0,\tau{}) + B[(\mathfrak{D}^{\mathbf{k}}\psi{})_0](0,\tau{}) \\
&\lesssim E_T[(\mathfrak{D}^{\mathbf{k}}\psi{})_0](0) + E_p[(\mathfrak{D}^{\mathbf{k}}\psi{})_0](0) + \int_{0}^{\tau{}} \int _{\Sigma{}(\tau{}')} \sinh x\abs{P(\mathfrak{D}^{\mathbf{k}}\psi{})_0}^2\dd{}\mu{}\dd{}\tau{}' \\
&\lesssim \epsilon{} + (A\epsilon{}\cdot A_{\mathbf{k}}\epsilon{} + A_{<\mathbf{k}}\epsilon{})(1+\tau{})^{2\#\mathbf{k}}.
\end{split}
\end{equation}
Since \(A_0\gg 1\) and \(A\epsilon{}\ll 1\) and \(A_{<\mathbf{k}}\ll A_{\mathbf{k}}\), the
right side of \zcref{b-boot-3-improvement} can be made less than
\(\frac{1}{2}A_{\mathbf{k}}\epsilon{}(1+\tau{})^{2\#\mathbf{k}}\). This improves
the bootstrap assumption \zcref{b-boot-3}.

We now establish \zcref{semilinear-3-goal}. Since \(P\) is \(\U(1)\)-symmetric, we have
\begin{equation}\label{semilinear-step-3-goal-prep}
\begin{split}
&\int_0^{\tau{}_{\textnormal{boot}}} \int _{\Sigma{}(\tau{})} \sinh x \abs{P(\mathfrak{D}^{\mathbf{k}}\psi{})_0}^2\dd{}\mu{}\dd{}\tau{} = \int_0^{\tau{}_{\textnormal{boot}}} \int _{\Sigma{}(\tau{})} \sinh x \abs{(P\mathfrak{D}^{\mathbf{k}}\psi{})_0}^2\dd{}\mu{}\dd{}\tau{} \\
&\lesssim \int_0^{\tau{}_{\textnormal{boot}}} \int _{\Sigma{}(\tau{})\cap \set{x\le 1}} \abs{P\mathfrak{D}^{\mathbf{k}}\psi{}}^2\tanh x\dd{}\mu{}\dd{}\tau{} \\
&\qquad + \int_0^{\tau{}_{\textnormal{boot}}} \int _{\Sigma{}(\tau{})\cap \set{x\ge 1}} \sinh x \abs{([P,\mathfrak{D}^{\mathbf{k}}]\psi{})_0}^2\dd{}\mu{}\dd{}\tau{}   +  \int_0^{\tau{}_{\textnormal{boot}}} \int _{\Sigma{}(\tau{})\cap \set{x\ge 1}} \sinh x \abs{(\mathfrak{D}^{\mathbf{k}}P\psi{})_0}^2\dd{}\mu{}\dd{}\tau{} .
\end{split}
\end{equation}
It remains to estimate the three terms on the right-hand side of
\zcref{semilinear-step-3-goal-prep}. We obtain \zcref{semilinear-3-goal}
by substituting
\zcref{semilinear-P-Dk-prep-0,semilinear-step-4-prep-g-conc,semilinear-P-Dk-worst-prep-2}
(the results of Steps 4a--c) into \zcref{semilinear-step-3-goal-prep}.

\step{Step 4a: Estimate for the first term on the right-hand side of
\zcref{semilinear-step-3-goal-prep}.} We have already estimated this term, in
\zcref{PDk-Psi-on-Sigma-tau}:
\begin{equation}\label{semilinear-P-Dk-prep-0}
\int_0^{\tau{}_{\textnormal{boot}}} \int _{\Sigma{}(\tau{})\cap \set{x\le 1}} \abs{P\mathfrak{D}^{\mathbf{k}}\psi{}}^2\tanh x\dd{}\mu{}\dd{}\tau{}\lesssim (A\epsilon{}\cdot A_{\mathbf{k}}\epsilon{} + A_{<\mathbf{k}}\epsilon{})(1+\tau{})^{2\#\mathbf{k}}.
\end{equation}

\step{Step 4b: Estimate for the second term on the right-hand side of
\zcref{semilinear-step-3-goal-prep}.} Let \(\textnormal{(I)}\) and
\(\textnormal{(II)}\) be as in the commutation formula of
\zcref{frak-D-commutation}, so that
\begin{equation}\label{semilinear-step-4-prep-g}
\int_0^{\tau{}_{\textnormal{boot}}} \int _{\Sigma{}(\tau{})} \sinh x \abs{([P,\mathfrak{D}^{\mathbf{k}}]\psi{})_0}^2\dd{}\mu{}\dd{}\tau{}\le \int_0^{\tau{}_{\textnormal{boot}}} \int _{\Sigma{}(\tau{})} \sinh x \abs{\textnormal{(I)}_0}^2\dd{}\mu{}\dd{}\tau{} + \int_0^{\tau{}_{\textnormal{boot}}} \int _{\Sigma{}(\tau{})} \sinh x \abs{\textnormal{(II)}}^2\dd{}\mu{}\dd{}\tau{}.
\end{equation}
Note that we have used Parseval's identity on \(S^1\) for term
\(\textnormal{(II)}\). Recalling the definition of \(\textnormal{(I)}\) in
\zcref{P-D-comm-I} and noting that the integral is supported in \(\set{x\ge 1}\), we
have
\begin{equation}\label{semilinear-step-4-prep-g2}
\begin{split}
\int_0^{\tau{}_{\textnormal{boot}}} \int _{\Sigma{}(\tau{})} \sinh x \abs{\textnormal{(I)}_0}^2\dd{}\mu{}\dd{}\tau{} &\lesssim \sum_{\mathbf{k}'<\mathbf{k}}\int_0^{\tau{}_{\textnormal{boot}}} \int _{\Sigma{}(\tau{})} \frac{\sinh x}{\cosh ^4\tau{}} \abs{\Grad _{S^2}(\mathfrak{D}^{\mathbf{k}'}\psi{})_0}^2 \dd{}\mu{}\dd{}\tau{} \\
&\lesssim \sum_{\mathbf{k}'<\mathbf{k}} \int_0^{\tau{}_{\textnormal{boot}}} \frac{1}{\cosh^2 \tau{}} E_p[(\mathfrak{D}^{\mathbf{k}'}\psi{})_0](\tau{})\dd{}\tau{}\\
&\lesssim A_{<\mathbf{k}}\epsilon{}.
\end{split}
\end{equation}
To pass to the final line, we used the bootstrap assumption \zcref{b-boot-3} and
used the integrability of \((\cosh \tau{})^{-1}\). Using
the fact that \(\#\mathbf{k}'\le \#\mathbf{k}-1\) whenever
\(\mathbf{k}'<\mathbf{k}\) and recalling the definition of \(\textnormal{(II)}\)
in \zcref{P-D-comm-I}, we see that the second term on the right-hand side of
\zcref{semilinear-step-4-prep-g} can be controlled as follows:
\begin{equation}\label{semilinear-step-4-prep-g3}
\begin{split}
\int_0^{\tau{}_{\textnormal{boot}}} \int _{\Sigma{}(\tau{})} \sinh x \abs{\textnormal{(II)}}^2\dd{}\mu{}\dd{}\tau{}&\lesssim \sum_{\mathbf{k}'<\mathbf{k}}\int_0^{\tau{}_{\textnormal{boot}}} E_T[\mathfrak{D}^{\mathbf{k}'}\psi{}](\tau{})\dd{}\tau{} \\
&\lesssim A_{<\mathbf{k}}\epsilon{}\int_0^{\tau{}_{\textnormal{boot}}} (1+\tau{})^{2\#\mathbf{k}-2}\dd{}\tau{}\lesssim A_{<\mathbf{k}}\epsilon{}(1+\tau{})^{2\#\mathbf{k}-1}.
\end{split}
\end{equation}
Combining
\zcref{semilinear-step-4-prep-g,semilinear-step-4-prep-g2,semilinear-step-4-prep-g3},
we conclude that
\begin{equation}\label{semilinear-step-4-prep-g-conc}
\int_0^{\tau{}_{\textnormal{boot}}} \int _{\Sigma{}(\tau{})} \sinh x \abs{([P,\mathfrak{D}^{\mathbf{k}}]\psi{})_0}^2\dd{}\mu{}\dd{}\tau{}\lesssim  A_{<\mathbf{k}}\epsilon{}(1+\tau{})^{2\#\mathbf{k}-1}.
\end{equation}

\step{Step 4c: Estimate for the third term on the right-hand side of
\zcref{semilinear-step-3-goal-prep}.} The hardest term to control on the right-hand
side of \zcref{semilinear-step-3-goal-prep} is the final one. Since \(F(\Grad
\varphi{})\) is almost \(\U(1)\)-symmetric, we can apply
\zcref{Dk-P-psi-estimate-almost-axisymmetric}. Let \(\textnormal{(III)}\) and
\(\textnormal{(IV)}\) equal \(\textnormal{(I)}\) in \zcref{Dk-P-psi-estimate-I} and
\(\textnormal{(II)}\) in \zcref{Dk-P-psi-estimate-II} of
\zcref{Dk-P-psi-estimate-almost-axisymmetric}, respectively, so that, by the
assumptions on the time weight \(\ell{}\), we have
\begin{equation}\label{semilinear-hardest-term-rp-sinh}
\int _{\Sigma{}(\tau{})\cap \set{x\ge 1}} \sinh x \abs{(\mathfrak{D}^{\mathbf{k}}P\psi{})_0}^2\dd{}\mu{}\lesssim (1+\tau{})^{-2C_\ast{}}\int _{\Sigma{}(\tau{})\cap \set{x\ge 1}} \sinh x[(\textnormal{III}) + \textnormal{(IV)}]\dd{}\mu{}.
\end{equation}
Controlling the terms on the right-hand side with fewer
\(\mathfrak{D}\)-derivatives using the pointwise estimates of
\zcref{pointwise-bootstrap} as in Step 2, we obtain
\begin{equation}\label{semilinear-hardest-term-rp-I}
\textnormal{(III)}\lesssim A\epsilon{}(1+\tau{})^{2+2\#\mathbf{k}_{\textnormal{max}}}\sum_{\mathbf{k}'\le \mathbf{k}}\frac{1}{\cosh x}\Bigl[(\partial{}_\tau{}\mathfrak{D}^{\mathbf{k}'}\psi{})^2 + (\partial{}_x\mathfrak{D}^{\mathbf{k}'}\psi{})^2 + \frac{1}{\cosh \tau{}}\abs{\Grad _{S^2}\mathfrak{D}^{\mathbf{k}'}\psi{}}^2 + \cosh ^2x(\partial{}_\theta{}\mathfrak{D}^{\mathbf{k}'}\psi{})^2\Bigr]
\end{equation}
and
\begin{equation}\label{semilinear-hardest-term-rp-II}
\begin{split}
\textnormal{(IV)}&\lesssim \side{RHS}{semilinear-hardest-term-rp-I} + A\epsilon{}(1+\tau{})^{2\#\mathbf{k}_{\textnormal{max}}}\sum_{\mathbf{k}'\le \mathbf{k}} \frac{1}{\cosh x}\abs{(\mathfrak{D}^{\mathbf{k}'}\psi{})_{\ge 1}}^2 +  \frac{1}{\cosh ^2x}\abs{(\mathfrak{D}^{\mathbf{k}'}\psi{})_0}^2.
\end{split}
\end{equation}
Substituting \zcref{semilinear-hardest-term-rp-I,semilinear-hardest-term-rp-II} into
\zcref{semilinear-hardest-term-rp-sinh}, we obtain
\begin{equation}\label{semilinear-P-Dk-worst-prep}
\begin{split}
&\int _{\Sigma{}(\tau{})\cap \set{x\ge 1}}\sinh x \abs{(\mathfrak{D}^{\mathbf{k}}P\psi{})_0}^2\dd{}\mu{}\\
&\lesssim A\epsilon{}(1+\tau{})^{2(-C_\ast{} + \#\mathbf{k}_{\textnormal{max}})}\sum_{\mathbf{k}'\le \mathbf{k}}\int _{\Sigma{}(\tau{})} \Bigl[\abs{\partial{}_\tau{}\mathfrak{D}^{\mathbf{k}'}\psi{}}^2 + \abs{\partial{}_x\mathfrak{D}^{\mathbf{k}'}\psi{}}^2 \\
&\qquad + \frac{1}{\cosh^2 \tau{}}\abs{\Grad _{S^2}\mathfrak{D}^{\mathbf{k}'}\psi{}}^2 + \cosh^2 x\abs{\partial{}_\theta{}\mathfrak{D}^{\mathbf{k}'}\psi{}}^2\Bigr]\tanh x\dd{}\mu{} \\
&\qquad  + A\epsilon{}(1+\tau{})^{2(-C_\ast{} + \#\mathbf{k}_{\textnormal{max}})}\sum_{\mathbf{k}'\le \mathbf{k}}\int _{\Sigma{}(\tau{})} \abs{(\mathfrak{D}^{\mathbf{k}'}\psi{})_{\ge 1}}^2\tanh x + \frac{1}{\cosh x}\abs{(\mathfrak{D}^{\mathbf{k}'}\psi{})_0}^2\dd{}\mu{}
\end{split}
\end{equation}
The first integral on the right-hand side of \zcref{semilinear-P-Dk-worst-prep} is
bounded by the \(\partial_{\tau{}}\)-energy. The \(\partial_\tau{}\)-energy also
controls, by a Poincaré inequality on \(S^1\), the first zeroth-order term in
the second integral on the right-hand side of \zcref{semilinear-P-Dk-worst-prep}
(since this term is supported on non-zero \(S^1\)-modes), while the second term
is controlled by the \(r^p\)-type energy of
\((\mathfrak{D}^{\mathbf{k}'}\psi{})_0\). After integrating in \(\tau{}\) (and
using the largeness of \(C_\ast{}\)), we obtain from \zcref{semilinear-P-Dk-worst-prep}
\begin{equation}\label{semilinear-P-Dk-worst-prep-2}
\begin{split}
&\int_0^{\tau{}_{\textnormal{boot}}} \int _{\Sigma{}(\tau{})\cap \set{x\ge 1}}\sinh x \abs{(\mathfrak{D}^{\mathbf{k}}P\psi{})_0}^2\dd{}\mu{}\dd{}\tau{} \\
&\lesssim A\epsilon{}\sum_{\mathbf{k}'\le \mathbf{k}}\int_0^{\tau{}_{\textnormal{boot}}} (1+\tau{})^{2(-C_\ast{} + \#\mathbf{k}_{\textnormal{max}})}(E_T[\mathfrak{D}^{\mathbf{k}'}\psi{}](\tau{}) + E_p[(\mathfrak{D}^{\mathbf{k}'}\psi{})_0](\tau{}))\dd{}\tau{} \\
&\lesssim A\epsilon{}\cdot A_{\mathbf{k}}\epsilon{}(1+\tau{})^{2\#\mathbf{k}}.
\end{split}
\end{equation}
To pass to the final line, we used the bootstrap assumptions
\zcref{b-boot-1,b-boot-3}.

\step{Step 5: Completing the proof of \zcref{energy-growth}.} In view of the (now improved)
bootstrap assumptions \zcref{b-boot-1,b-boot-3}, the pointwise estimates of
\zcref{pointwise-bootstrap}, it only remains to show that
\begin{equation}\label{final-pointwise-goal}
\sum_{\abs{\mathbf{k}}\le k_{\textnormal{max}}}\sup_{x\ge 1}\int _{S^2\times S^1}\abs{(\mathfrak{D}^{\mathbf{k}}\psi{})_0}^2(\tau{},x)\dd{}\omega{}\dd{}\theta{}\lesssim \epsilon{}(1 + \tau{})^B
\end{equation}
for some \(B = B(k_{\textnormal{max}})\). We will show that we can take \(B
=2\#\mathbf{k}_{\textnormal{max}}\). By \zcref{psi0-far-T-prep}, the smallness assumption
\zcref{data-small}, and the bootstrap assumption \zcref{b-boot-3} we have
\begin{equation}\label{final-pointwise-1}
\begin{split}
&\sum_{\abs{\mathbf{k}}\le k_{\textnormal{max}}-1}\int _{S^2\times S^1}\abs{(\mathfrak{D}^{\mathbf{k}}\psi{})_0}^2(\tau{},x)\dd{}\omega{}\dd{}\theta{} \\
&\lesssim \sum_{\abs{\mathbf{k}}\le k_{\textnormal{max}}-1}\Bigl[\sup_{x\ge 1}\int _{S^2\times S^1}\abs{(\mathfrak{D}^{\mathbf{k}}\psi{})_0}^2(\tau{}=0,x)\dd{}\omega{}\dd{}\theta{} \\
&\qquad + \sup_{\tau{}\in [0,\tau{}_0]}\int _{S^2\times S^1}\abs{(\mathfrak{D}^{\mathbf{k}}\psi{})_0}^2(\tau{},x=1)\dd{}\omega{}\dd{}\theta{} + F_p[(\mathfrak{D}^{\mathbf{k}}\psi{})_0](0,\tau{}_0)\Bigr] \\
&\lesssim A_{\mathbf{k}}\epsilon{}(1+\tau{})^{2\#\mathbf{k}_{\textnormal{max}}-2} + \sum_{\abs{\mathbf{k}}\le k_{\textnormal{max}}-1}\sup_{\tau{}\in [0,\tau{}_0]}\int _{S^2\times S^1}\abs{(\mathfrak{D}^{\mathbf{k}}\psi{})_0}^2(\tau{},x=1)\dd{}\omega{}\dd{}\theta{}.
\end{split}
\end{equation}
Now \zcref{elliptic-near-bubble-prep-10} and the bootstrap assumptions \zcref{b-boot-1,b-boot-2} give
\begin{equation}\label{final-pointwise-2}
\begin{split}
\sum_{\abs{\mathbf{k}}\le k_{\textnormal{max}}-1}\sup_{\tau{}\in [0,\tau{}_0]}\int _{S^2\times S^1}\abs{(\mathfrak{D}^{\mathbf{k}}\psi{})_0}^2(\tau{},x=1)\dd{}\omega{}\dd{}\theta{}&\lesssim \sum_{\abs{\mathbf{k}}\le k_{\textnormal{max}}-1}E_T[(\partial{}_\tau^{\le 1}\mathfrak{D}^{\mathbf{k}}\psi{})_0](\tau{}) + \int _{\Sigma{}(\tau{})} (P\mathfrak{D}^{\mathbf{k}}\psi{})_0^2\tanh x \dd{}\mu{} \\
&\lesssim A_{\mathbf{k}}\epsilon{}(1+\tau{})^{2\#\mathbf{k}_{\textnormal{max}}},
\end{split}
\end{equation}
since \(\partial_\tau{}\) is a commutator vector field. Combining
\zcref{final-pointwise-1,final-pointwise-2}, we get \zcref{final-pointwise-goal} (with
\(B = 2\#\mathbf{k}_{\textnormal{max}}\)).
\end{proof}
\subsection{Global existence and improved decay for \texorpdfstring{\(\U(1)\)}{U(1)}-symmetric solutions to a semilinear equation}
\label{sec:semilinear-axisymmetric} In this section, we prove
\zcref{semilinear-axisymmetric}, which is a small-data global existence result for
solutions to semilinear equations satisfying the null condition under the
additional assumption that the nonlinearity and initial data are exactly
\(\U(1)\)-symmetric. This assumption allows us to establish improved decay relative to
the setting of \zcref{semilinear}, which does not make symmetry assumptions on the
nonlinearity or the initial data. The reason for this is that \(\U(1)\)-symmetry gives
access to the \(r^p\)-type estimate of \zcref{rp-type-estimate} for the full
solution. To capture the improved decay, we use the hyperboloidal foliation
\(\mathcal{H}(s)\), whose leaves are asymptotically null (see
\zcref{hyperboloidal-foliation}), as opposed to the constant-\(\tau\) foliation
(whose leaves reach spacelike infinity) that we used in \zcref{semilinear}.
\begin{proposition}[Global existence and decay for \(\U(1)\)-symmetric solutions to a semilinear wave equation]
Let \(k_{\textnormal{max}}\ge 8\). Let \(F(\Grad \varphi{})\) be a \(\U(1)\)-symmetric quadratic
nonlinearity satisfying the null condition with time weight \(\ell{}\) (see
\zcref{null-condition,axisymmetric-nonlinearity}) such that \(\abs{\partial_\tau^k\ell{}(\tau{})}\lesssim
1\) for \(0\le k\le k_{\textnormal{max}}\). By \zcref{canonical-double-null},
\(F(\Grad \varphi{})\) can in particular be \(f\cdot g^{\alpha{}\beta{}}\partial{}_\alpha{}\varphi{}\partial{}_\beta{}\varphi{}\) for \(f\) a \(\U(1)\)-symmetric
spacetime function with bounded \(\mathfrak{D}\)-derivatives.

Then there exists \(\epsilon{}_0 > 0\) (independent of \(k_{\textnormal{max}}\)) such that whenever \(\epsilon{}\in [0,\epsilon_0]\), the unique
solution to
\begin{equation}\label{box-semilinear-equation-asym}
\Box{}\varphi{}=F(\Grad \varphi{})
\end{equation}
arising from given smooth \(\U(1)\)-symmetric data that is posed on the hyperboloidal
hypersurface \(\mathcal{H}(0)\) and that satisfies
\begin{equation}\label{axisymmetric-data-small}
\sum_{\abs{\mathbf{k}}\le k_{\textnormal{max}}} \mathcal{E}[\mathring{\mathfrak{D}}^{\mathbf{k}}\psi{}](0) + \sum_{\abs{\mathbf{k}}\le k_{\textnormal{max}}-4} \norm{\partial{}_x\mathring{\mathfrak{D}}^{\mathbf{k}}\psi{}}_{L^\infty(\mathcal{H}(0))}^2\le \epsilon{}
\end{equation}
exists globally to the future of \(\mathcal{H}(0)\). Moreover, there is a
universal constant \(c > 0\) such that the following exponential decay statement
holds for the energy and pointwise norm of the solution:
\begin{equation}\label{axisymmetric-exp-decay}
\sum_{\abs{\mathbf{k}}\le k_{\textnormal{max}}}\mathcal{E}[\mathring{\mathfrak{D}}^{\mathbf{k}}\psi{}](s) + \sum_{\abs{\mathbf{k}}\le k_{\textnormal{max}}-4} \norm{\partial{}_x\mathring{\mathfrak{D}}^{\mathbf{k}}\psi{}}_{L^\infty(\mathcal{H}(s))}^2 + \sum_{\abs{\mathbf{k}}\le k_{\textnormal{max}}-3}\norm{\mathring{\mathfrak{D}}^{\mathbf{k}}\psi{}}_{{L^\infty(\mathcal{H}(s))}}^2\lesssim e^{-cs}\epsilon{}.
\end{equation}
Moreover, the solution \(\psi{}\) has a finite limit at null infinity.
\label{semilinear-axisymmetric}
\end{proposition}
\begin{remark}[Improvement in time decay relative to the setting outside symmetry]
We only require the time weight \(\ell{}(\tau{})\) (and its derivatives) to be bounded,
whereas in \zcref{semilinear} we assume that it decays polynomially in \(\tau{}\).

The exponential decay statement \zcref{axisymmetric-exp-decay} can be contrasted
with the polynomial growth statements \zcref{energy-growth,glob-exist-pointwise} of
\zcref{semilinear}. Since \(s\sim u=\frac{1}{2}(\tau{}-x)\), the quantity \(\psi{}\)
decays exponentially in \(\tau{}\) for fixed \(x\), whereas in \zcref{semilinear}
the solution may grow polynomially in \(\tau{}\) for fixed \(x\). Since the
Minkowskian null coordinate \(U\) satisfies \(\abs{U}\sim e^{-2u}\) (see
\zcref{double-null}), the exponential decay in \(s\) corresponds to a polynomial
decay in \(\abs{U}\).
\end{remark}
\begin{remark}[Prescribing Cauchy data on a hypersurface reaching spacelike infinity]
\Cref{semilinear-axisymmetric} is stated for data posed on a hyperboloidal
hypersurface reaching null infinity. We make this choice to highlight the
improved decay in the interior of a light cone. However, one could extend
our result to the exterior region. That is, one could show using our estimates
that small data on a Cauchy hypersurface reaching spacelike infinity (such as
the surface \(\set{\tau{}=0}\) considered in \zcref{semilinear}) produces a solution
that exists up to the hyperboloidal hypersurface and induces small data there
(in the norm we use). In fact, proving this result is easier than the result in
the interior region. This is because in the exterior region one can use the
weighted multiplier \(w(-u)\partial_\tau{}\) (for example with \(w(q)\sim
e^{\delta{}q}\) when \(q > 0\), where \(\delta{} > 0\)), which produces an
energy estimate analogous to the \(\partial_\tau{}\)-energy estimate of
\zcref{T-estimate-prop}, but with control of better bulk terms.
\end{remark}
\begin{proof}
At the outset, we note that \(\mathring{\mathfrak{D}}\psi{}\) is \(\U(1)\)-symmetric when
\(\psi{}\) is. As discussed in the introduction to the proof of \zcref{semilinear},
the desired results (apart from the improved decay statement
\zcref{axisymmetric-exp-decay}) follow from the ``improvement of the bootstrap
assumptions'' together with a suitable local well-posedness theory, which follows
from our estimates using standard techniques. We will establish the exponential
decay statement \zcref{axisymmetric-exp-decay} after we improve the bootstrap
assumptions. In particular, we only need smallness, and not improved decay, to
prove global existence.

\step{Step 0: Setup and bootstrap assumptions.} We first set up the bootstrap argument.
For \(\abs{\mathbf{k}}\le k_{\textnormal{max}}\), define constants \(A_{\mathbf{k}}\gg 1\) such
that \(A_{\mathbf{k}'}\ll A_{\mathbf{k}}\) whenever \(\mathbf{k}'<\mathbf{k}\).
Write \(A_{<\mathbf{k}} \coloneqq{}
\max_{\mathbf{k}'<\mathbf{k}}A_{\mathbf{k}'}\) and \(A \coloneqq{}\max
_{\abs{\mathbf{k}}\le k_{\textnormal{max}}}A_{\mathbf{k}}\). We assume
\(\epsilon{}_0\) is chosen so that \(A\epsilon{}_0\ll 1\). For
\(\abs{\mathbf{k}}\le k_{\textnormal{max}}\), we suppose that the following
bootstrap assumptions hold for \(s\in{}[0,s_{\textnormal{boot}}]\) for some
\(s_{\textnormal{boot}} > 0\):
\begin{align}
\mathcal{E}[\mathring{\mathfrak{D}}^{\mathbf{k}}\psi{}](s) + \mathcal{F}_T[\mathring{\mathfrak{D}}^{\mathbf{k}}\psi{}](0,s)&\le A_\mathbf{k}\epsilon{}, \label{a-boot-1} \\
\int _{\mathcal{H}(s)} (P\mathring{\mathfrak{D}}^\mathbf{k}\psi{})^2\tanh x\dd{}\mu{} + \sup_{\tau{}\ge 0}\int _{\Sigma{}(\tau{})\cap \mathcal{R}(0,s)} (P\mathring{\mathfrak{D}}^\mathbf{k}\psi{})^2\tanh x\dd{}\mu{}&\le A_{\mathbf{k}}\epsilon{}, \label{a-boot-2}
\end{align}

\step{Step 1: Pointwise estimates.} It follows from the bootstrap assumptions
\zcref{a-boot-1,a-boot-2}, the assumption \zcref{axisymmetric-data-small} on the
smallness of the initial data, and the pointwise estimates
\zcref{Loo-1-asym,Loo-3-asym} of \zcref{global-pointwise-s} that
\begin{equation}\label{pointwise-bootstrap-cons}
\sum_{\abs{\mathbf{k}}\le k_{\textnormal{max}}-3}\norm{\mathring{\mathfrak{D}}^{\mathbf{k}}\psi{}}_{L^\infty(\mathcal{H}(s))}^2 + \sum_{\abs{\mathbf{k}}\le k_{\textnormal{max}}-4}\norm{\underline{L}\mathring{\mathfrak{D}}^{\mathbf{k}}\psi{}}_{L^\infty(\mathcal{H}(s))}^2 + \sum_{\abs{\mathbf{k}}\le k_{\textnormal{max}}-4}\norm{L\mathring{\mathfrak{D}}^{\mathbf{k}}\psi{}}_{L^\infty(\mathcal{H}(s))}^2\lesssim A\epsilon{}
\end{equation}
for \(s\in{}[0,s_{\textnormal{boot}}]\). We have used the fact that \(\partial_\tau{}\) is
among the commutator vector fields \(\mathring{\mathfrak{D}}\).

\step{Step 2: Pointwise estimate for the inhomogeneity in the equation solved by
\(\mathring{\mathfrak{D}}^{\mathbf{k}}\psi{}\).} Using the pointwise estimates of
\zcref{pointwise-bootstrap-cons}, we obtain from \zcref{Dk-P-psi-estimate}, the
commutation formula of \zcref{frak-D-commutation}, the \(\U(1)\)-symmetry of \(\psi{}\), and the
assumptions on \(\ell{}\) the estimate
\begin{equation}\label{asym-commutation-estimate}
\begin{split}
\abs{P\mathring{\mathfrak{D}}^{\mathbf{k}}\psi{}}^2 &\lesssim \abs{\mathring{\mathfrak{D}}^{\mathbf{k}}P\psi{}}^2 + \abs{[P,\mathring{\mathfrak{D}}^{\mathbf{k}}]\psi{}}^2 \\
&\lesssim A\epsilon{}\sum_{\mathbf{k}'\le \mathbf{k}}\Bigl[\frac{1}{\cosh^2 x}((\underline{L}\mathring{\mathfrak{D}}^{\mathbf{k}'}\psi{})^2 + (L\mathring{\mathfrak{D}}^{\mathbf{k}'}\psi{})^2 + (\mathring{\mathfrak{D}}^{\mathbf{k}'}\psi{})^2) + \frac{1}{\cosh ^2\tau{}}\abs{\Grad _{S^2}\mathring{\mathfrak{D}}^{\mathbf{k}'}\psi{}}^2\Bigr] \\
&\qquad + \sum_{\mathbf{k}'<\mathbf{k}} \Bigl[\frac{1}{\cosh ^2x}((\underline{L}\mathring{\mathfrak{D}}^{\mathbf{k}'}\psi{})^2 + (L\mathring{\mathfrak{D}}^{\mathbf{k}'}\psi{})^2 + (\mathring{\mathfrak{D}}^{\mathbf{k}'}\psi{})^2)  + \frac{1}{\cosh ^2\tau{}}\abs{\Grad _{S^2}\mathring{\mathfrak{D}}^{\mathbf{k}'}\psi{}}^2\Bigr]. \end{split}
\end{equation}
To obtain this estimate, we have put terms arising from the estimate of
\zcref{Dk-P-psi-estimate} with fewer derivatives in \(L^\infty\), arguing as in Step
2 of the proof of \zcref{semilinear} (see the discussion after
\zcref{semilinear-I--III}).

\step{Step 3: Improving the bootstrap assumptions and establishing global existence}.
In this step, we improve the bootstrap assumptions \zcref{a-boot-1,a-boot-2}. We
first improve \zcref{a-boot-2} using \zcref{a-boot-1} in Step 3a, and then improve
\zcref{a-boot-1} in Step 3b. This completes the proof of global existence.

\step{Step 3a: Improving the bootstrap assumption \zcref{a-boot-2}.} We first
improve the bootstrap assumption \zcref{a-boot-2}. As a direct consequence
of \zcref{asym-commutation-estimate}, we have
\begin{equation}\label{asym-boot-2}
\begin{split}
\int _{\mathcal{H}(s)} (P\mathring{\mathfrak{D}}^\mathbf{k}\psi{})^2\tanh x\dd{}\mu{} \lesssim A\epsilon{}\sum_{\mathbf{k}'\le \mathbf{k}}\mathcal{E}[\mathring{\mathfrak{D}}^{\mathbf{k}'}\psi{}](s) + \sum_{\mathbf{k}'<\mathbf{k}}\mathcal{E}[\mathring{\mathfrak{D}}^{\mathbf{k}'}\psi{}](s)
\end{split}
\end{equation}
for \(s\in{}[0,s_{\textnormal{boot}}]\). Similarly, \zcref{asym-commutation-estimate} implies (for \(0\le s_1\le s_2\le
s_{\textnormal{boot}}\))

\begin{equation}\label{asym-boot-3}
\sup_{\tau{}\ge 0}\int _{\Sigma{}(\tau{})\cap \mathcal{R}(s_1,s_2)} (P\mathring{\mathfrak{D}}^\mathbf{k}\psi{})^2\tanh x\dd{}\mu{}\lesssim A\epsilon{}\sum_{\mathbf{k}'\le \mathbf{k}}\mathcal{F}_T[\mathring{\mathfrak{D}}^{\mathbf{k}'}\psi{}](s_1,s_2) + \sum_{\mathbf{k}'<\mathbf{k}}\mathcal{F}_T[\mathring{\mathfrak{D}}^{\mathbf{k}'}\psi{}](s_1,s_2).
\end{equation}
Using the bootstrap assumption \zcref{a-boot-1}, we conclude from
\zcref{asym-boot-2,asym-boot-3} (applied with \(s_1 = 0\) and \(s_2=s\) for
\(s\in{}[0,s_{\textnormal{boot}}]\)) that
\begin{equation}
\int _{\mathcal{H}(s)} (P\mathring{\mathfrak{D}}^\mathbf{k}\psi{})^2\tanh x\dd{}\mu{} + \sup_{\tau{}\ge 0}\int _{\Sigma{}(\tau{})\cap \mathcal{R}(0,s)} (P\mathring{\mathfrak{D}}^\mathbf{k}\psi{})^2\tanh x\dd{}\mu{} \lesssim A_{\mathbf{k}}A\epsilon^2 + A_{<\mathbf{k}}\epsilon{}.
\end{equation}
Since \(A_{<\mathbf{k}}\ll A_{\mathbf{k}}\) and \(A\epsilon{}\ll 1\), the right side of the
above estimate is less than \(A_{\mathbf{k}}\epsilon{}/2\) when \(\epsilon{}\)
is small enough. This improves the bootstrap assumptions
\zcref{a-boot-2}.

\step{Step 3b: Improving the bootstrap assumption \zcref{a-boot-1}.} We now improve the
bootstrap assumption \zcref{a-boot-1}. It suffices to show that
\begin{equation}\label{semilinear-energy-decay-prep}
\sum_{\mathbf{k}'\le \mathbf{k}}\mathcal{E}[\mathring{\mathfrak{D}}^{\mathbf{k}'}\psi{}](s_2) + \mathcal{F}_T[\mathring{\mathfrak{D}}^{\mathbf{k}'}\psi{}](s_1,s_2) + \mathcal{B}[\mathring{\mathfrak{D}}^{\mathbf{k}'}\psi{}](s_1,s_2)\lesssim \sum_{\mathbf{k}'\le \mathbf{k}}\mathcal{E}[\mathring{\mathfrak{D}}^{\mathbf{k}'}\psi{}](s_1),
\end{equation}
for \(0\le s_1\le s_2\le s_{\textnormal{boot}}\), where \(\mathcal{B}\) denotes the
bulk term defined in \zcref{rp-type-estimate}. Indeed, applying
\zcref{semilinear-energy-decay-prep} with \(s_1=0\) and \(s_2=s\) for
\(s\in{}[0,s_{\textnormal{boot}}]\), and using the smallness of the initial
data in \zcref{axisymmetric-data-small}, we obtain
\begin{equation}
\sum_{\mathbf{k}'\le \mathbf{k}}\mathcal{E}[\mathring{\mathfrak{D}}^{\mathbf{k}'}\psi{}](s) \lesssim \epsilon{}.
\end{equation}
We have therefore improved the bootstrap assumption
\zcref{a-boot-1} when \(A_0\) is chosen large enough (so that \(1\ll
A_{\mathbf{k}}\) for all \(\mathbf{k}\)). This completes the proof of global
existence.

We now prove \zcref{semilinear-energy-decay-prep}. The key step is to estimate the
inhomogeneous term on the right-hand side of the \(r^p\)-type estimate of
\zcref{rp-type-estimate} applied to \(\mathring{\mathfrak{D}}^{\mathbf{k}}\psi{}\). By
\zcref{asym-commutation-estimate}, we have
\begin{equation}\label{a-semilinear-prep}
\begin{split}
&\sum_{\mathbf{k}'\le \mathbf{k}}\int_{s_1}^{s_2} \int _{\mathcal{H}(s)}\sinh x(P\mathring{\mathfrak{D}}^{\mathbf{k}'}\psi{})^2\dd{}\mu{}\dd{}s \\
&\lesssim A\epsilon{}\sum_{\mathbf{k}'\le \mathbf{k}}\int_{s_1}^{s_2} \int _{\mathcal{H}(s)} \frac{\tanh x}{\cosh x}((\underline{L}\mathring{\mathfrak{D}}^{\mathbf{k}'}\psi{})^2 + (L\mathring{\mathfrak{D}}^{\mathbf{k}'}\psi{})^2 + (\mathring{\mathfrak{D}}^{\mathbf{k}'}\psi{})^2) + \frac{\sinh x}{\cosh ^2\tau{}}\abs{\Grad _{S^2}\mathring{\mathfrak{D}}^{\mathbf{k}'}\psi{}}^2\dd{}\mu{}\dd{}s \\
&\qquad + \sum_{\mathbf{k}'<\mathbf{k}}\int_{s_1}^{s_2} \int _{\mathcal{H}(s)} \frac{\tanh x}{\cosh x}((\underline{L}\mathring{\mathfrak{D}}^{\mathbf{k}'}\psi{})^2 + (L\mathring{\mathfrak{D}}^{\mathbf{k}'}\psi{})^2 + (\mathring{\mathfrak{D}}^{\mathbf{k}'}\psi{})^2) + \frac{\sinh x}{\cosh ^2\tau{}}\abs{\Grad _{S^2}\mathring{\mathfrak{D}}^{\mathbf{k}'}\psi{}}^2\dd{}\mu{}\dd{}s \\
&\lesssim A\epsilon{}\sum_{\mathbf{k}'\le \mathbf{k}}\mathcal{B}[\mathring{\mathfrak{D}}^{\mathbf{k}'}\psi{}](s_1,s_2) + \sum_{\mathbf{k}'<\mathbf{k}}\mathcal{B}[\mathring{\mathfrak{D}}^{\mathbf{k}'}\psi{}](s_1,s_2),
\end{split}
\end{equation}
Using the energy estimate \zcref{rp-equation} of \zcref{rp-type-estimate} to control the
bulk terms on the right-hand side of \zcref{a-semilinear-prep} by energies and
inhomogeneous terms, we obtain
\begin{equation}\label{a-semilinear-prep-2}
\begin{split}
&\sum_{\mathbf{k}'\le \mathbf{k}}\int_{s_1}^{s_2} \int _{\mathcal{H}(s)}\sinh x(P\mathring{\mathfrak{D}}^{\mathbf{k}'}\psi{})^2\dd{}\mu{}\dd{}s \\
&\lesssim A\epsilon{}\sum_{\mathbf{k}'\le \mathbf{k}}\mathcal{E}[\mathring{\mathfrak{D}}^{\mathbf{k}'}\psi{}](s_1) + A\epsilon{}\sum_{\mathbf{k}'\le \mathbf{k}} \int_{s_1}^{s_2}\int _{\mathcal{H}(s)}\sinh x(P\mathring{\mathfrak{D}}^{\mathbf{k}'}\psi{})^2\dd{}\mu{}\dd{}s  + \sum_{\mathbf{k}'<\mathbf{k}}\mathcal{E}[\mathring{\mathfrak{D}}^{\mathbf{k}'}\psi{}](s_1) \\
&\qquad + \sum_{\mathbf{k}'<\mathbf{k}}\int_{s_1}^{s_2} \int _{\mathcal{H}(s)}\sinh x\abs{P\mathring{\mathfrak{D}}^{\mathbf{k}'}\psi{}}^2\dd{}\mu{}\dd{}s.
\end{split}
\end{equation}
The second term on the right-hand side of \zcref{a-semilinear-prep-2} can be
absorbed to the left-hand side using the smallness of \(\epsilon{}\). The final
term on the right-hand side of \zcref{a-semilinear-prep-2} is present only when
\(\mathbf{k}\neq{}0\), and so an induction argument shows that
\begin{equation}\label{asym-sinh-estimate}
\begin{split}
\sum_{\mathbf{k}'\le \mathbf{k}}\int_{s_1}^{s_2} \int _{\mathcal{H}(s)}\sinh x(P\mathring{\mathfrak{D}}^{\mathbf{k}'}\psi{})^2\dd{}\mu{}\dd{}s \lesssim A\epsilon{}\sum_{\mathbf{k}'\le \mathbf{k}}\mathcal{E}[\mathring{\mathfrak{D}}^{\mathbf{k}'}\psi{}](s_1)  + \sum_{\mathbf{k}'<\mathbf{k}}\mathcal{E}[\mathring{\mathfrak{D}}^{\mathbf{k}'}\psi{}](s_1).
\end{split}
\end{equation}
Since \(A\epsilon{}\lesssim 1\), the energy estimate \zcref{rp-equation} of
\zcref{rp-type-estimate} combined with \zcref{asym-sinh-estimate} gives
\begin{equation}
\begin{split}
&\sum_{\mathbf{k}'\le \mathbf{k}}\mathcal{E}[\mathring{\mathfrak{D}}^{\mathbf{k}'}\psi{}](s_2) + \mathcal{F}_T[\mathring{\mathfrak{D}}^{\mathbf{k}'}\psi{}](s_1,s_2) + \mathcal{B}[\mathring{\mathfrak{D}}^{\mathbf{k}'}\psi{}](s_1,s_2) \\
&\lesssim \sum_{\mathbf{k}'\le \mathbf{k}}\mathcal{E}[\mathring{\mathfrak{D}}^{\mathbf{k}'}\psi{}](s_1) + \int_{s_1}^{s_2} \int _{\mathcal{H}(s)}\sinh x\abs{P\mathring{\mathfrak{D}}^{\mathbf{k}'}\psi{}}^2\dd{}\mu{}\dd{}s  \\
&\lesssim \sum_{\mathbf{k}'\le \mathbf{k}}\mathcal{E}[\mathring{\mathfrak{D}}^{\mathbf{k}'}\psi{}](s_1),
\end{split}
\end{equation}
which is \zcref{semilinear-energy-decay-prep}.

\step{Step 4: Proof of the exponential decay statement \zcref{axisymmetric-exp-decay}.}
Finally, we establish the exponential decay statement
\zcref{axisymmetric-exp-decay}. Let \(0\le s_1\le s_2\le s_{\textnormal{boot}}\). By
\zcref{semilinear-energy-decay-prep} (summed over all \(\abs{\mathbf{k}}\le
k_{\textnormal{max}}\)) and \zcref{rp-K-bound}, we have
\begin{equation}
\sum_{\abs{\mathbf{k}}\le k_{\textnormal{max}}}\mathcal{E}[\mathring{\mathfrak{D}}^{\mathbf{k}}\psi{}](s_2) + \int_{s_1}^{s_2}\sum_{\abs{\mathbf{k}}\le k_{\textnormal{max}}} \mathcal{E}[\mathring{\mathfrak{D}}^{\mathbf{k}}\psi{}](s)\dd{}s\lesssim \sum_{\abs{\mathbf{k}}\le k_{\textnormal{max}}}\mathcal{E}[\mathring{\mathfrak{D}}^{\mathbf{k}}\psi{}](s_1).
\end{equation}
By \zcref{exponential-decay}, this implies the exponential decay statement
\begin{equation}\label{asym-energy-decay}
\sum_{\abs{\mathbf{k}}\le k_{\textnormal{max}}}\mathcal{E}[\mathring{\mathfrak{D}}^{\mathbf{k}}\psi{}](s)\lesssim e^{-cs}\sum_{\abs{\mathbf{k}}\le k_{\textnormal{max}}}\mathcal{E}[\mathring{\mathfrak{D}}^{\mathbf{k}}\psi{}](0)\lesssim e^{-cs}\epsilon{}
\end{equation}
for some universal constant \(c > 0\). It now follows from
\zcref{asym-sinh-estimate} that
\begin{equation}\label{asym-sinh-decay}
\sum_{\abs{\mathbf{k}}\le k_{\textnormal{max}}}\int_{s_1}^{s_2} \int _{\mathcal{H}(s)}\sinh x(P\mathring{\mathfrak{D}}^{\mathbf{k}}\psi{})^2\dd{}\mu{}\dd{}s\lesssim e^{-cs_1}\epsilon{}.
\end{equation}
Substituting \zcref{asym-energy-decay,asym-sinh-decay} into the energy estimate
\zcref{rp-equation} of \zcref{rp-type-estimate}, we conclude also that
\begin{equation}\label{asym-F-decay}
\sum_{\abs{\mathbf{k}}\le k_{\textnormal{max}}}\sup_{\tau{}\ge 0}\mathcal{F}_T[\mathring{\mathfrak{D}}^{\mathbf{k}}\psi{}](s_1,s_2)\lesssim e^{-cs_1}\epsilon{}.
\end{equation}
Combining \zcref{asym-boot-2,asym-boot-3,asym-energy-decay,asym-F-decay}, we get
\begin{equation}\label{asym-P-decay}
\sum_{\abs{\mathbf{k}}\le k_{\textnormal{max}}}\int _{\mathcal{H}(s_1)} (P\mathring{\mathfrak{D}}^\mathbf{k}\psi{})^2\tanh x\dd{}\mu{} + \sup_{\tau{}\ge 0}\int _{\Sigma{}(\tau{})\cap \mathcal{R}(s_1,s_2)} (P\mathring{\mathfrak{D}}^\mathbf{k}\psi{})^2\tanh x\dd{}\mu{}\lesssim e^{-cs_1}\epsilon{}.
\end{equation}
Finally, the pointwise estimates \zcref{Loo-1-asym,Loo-3-asym} of
\zcref{global-pointwise-s} together with
\zcref{asym-energy-decay,asym-F-decay,asym-P-decay} now give
\begin{equation}\label{asym-pointwise-decay}
\sum_{\abs{\mathbf{k}}\le k_{\textnormal{max}}-3}\norm{\mathring{\mathfrak{D}}^{\mathbf{k}}\psi{}}_{L^\infty(\mathcal{H}(s))}^2 + \sum_{\abs{\mathbf{k}}\le k_{\textnormal{max}}-4}\norm{\partial{}_x\mathring{\mathfrak{D}}^{\mathbf{k}}\psi{}}_{L^\infty(\mathcal{H}(s))}^2 \lesssim e^{-cs}\epsilon{} + e^{-c\max (0,s-1)}\epsilon{} + \mathbf{1}_{s<1}\epsilon{}\lesssim e^{-cs}\epsilon{}.
\end{equation}
The desired statement \zcref{axisymmetric-exp-decay} now follows from
\zcref{asym-energy-decay,asym-pointwise-decay}.

\step{Step 5: Existence of the radiation field.} Since the \(\mathcal{E}\)-energy of
\(\psi{}\) is finite, the fact that \(\psi{}\) has a limit at null infinity can be
established as in \zcref{linear-radiation-field}.
\end{proof}
\appendix
\section{Bounce solutions and the discovery of the Witten bubble spacetime}
\label{appendix} How could one discover the
Witten bubble spacetime? As explained in \cite{PhysRevD.15.2929,Witten1981-gn},
the semiclassical decay of unstable ground states can be studied by searching
for a ``bounce'' solution. A bounce solution is a Riemannian \(5\)-manifold that
solves the Einstein equations in Riemannian signature (that is, it has vanishing
Ricci curvature) and approaches a Riemannian signature version of the ground
state (for us \(\R^4\times S^1\)) near infinity. The unstable ground state will
decay into a solution to the Einstein vacuum equations that agrees with the
bounce solution on a \(4\)-dimensional spacelike hypersurface which can ``be
regarded as \(t = 0\).'' Thus a bounce solution which remains real-valued after
``analytic continuation'' to Lorentzian signature (a Wick rotation of the form
\(t\mapsto it\)) represents an instability of the ground state.

We now reproduce how Witten \cite{WITTEN1982481} finds a bounce solution leading
to an instability of the Kaluza--Klein spacetime. The Euclidean Kaluza--Klein
space \(\R^4\times S^1\) with circle of radius \(R\) has metric
\begin{equation}
g_{\textnormal{KKE}} = \dd{}t^2 + \dd{}x^2 + \dd{}y^2 + \dd{}z^2 +R^2\dd{}\theta^2,
\end{equation}
where \(x\), \(y\), \(z\), and \(t\) range over \(\R\) and \(\theta{}\) is a periodic
coordinate with period \(2\pi{}\). Setting \(\rho{} \coloneqq{} \sqrt{t^2 + x^2 + y^2 +
z^2}\) and using polar coordinates, this metric takes the form
\begin{equation}\label{KKE-metric}
g_{\textnormal{KKE}} = \dd{}\rho{}^2 + R^2\dd{}\theta^2 + \rho{}^2g_{S^3},
\end{equation}
where \(g_{S^3}\) is the metric on the unit round \(3\)-sphere. To look for a
bounce solution, Witten \cite{WITTEN1982481} introduces the metric
\begin{equation}\label{gb-alpha}
g_{\textnormal{bounce},\alpha{}} = \Bigl(1-\frac{\alpha{}}{\rho^2}\Bigr)^{-1}\dd{}\rho^2 + R^2\Bigl(1 - \frac{\alpha{}}{\rho^2}\Bigr)\dd{}\theta^2 + \rho{}^2g_{S^3},
\end{equation}
which is an Einstein metric that approaches the metric \zcref{KKE-metric} near
infinity. One may recognize this as a Wick rotation of the five-dimensional
Schwarzschild solution.
\begin{remark}[Range of the coordinate \(\rho\) and value of the parameter \(\alpha\) in the metric \cref{gb-alpha}]
The metric \zcref{gb-alpha} is defined for \(\rho{}\in (\sqrt{\alpha{}},\infty)\). At
\(\rho{}=\sqrt{\alpha{}}\) there is a coordinate singularity akin to that of
polar coordinates in the plane. Indeed, the change of variables \(\rho{} =
\sqrt{\alpha{}} + \lambda^2\) transforms \((1-\alpha{}/\rho^2)^{-1}\dd{}\rho^2 +
R^2(1-\alpha{}/\rho^2)\dd{}\theta^2\) into \(2\sqrt{\alpha{}}(\dd{}\lambda^2 +
R^2(\lambda^2/\alpha{})\dd{}\theta^2)\), to first order in \(\lambda^2\). This
``polar coordinate'' expression extends regularly to \(\lambda{} = 0\) if and only
if \(\theta{}\) has period \(2\pi{}R^{-1}\sqrt{\alpha{}}\). We conclude
that the metric \zcref{gb-alpha} is regular at \(\rho{} = \sqrt{\alpha{}}\) exactly
when \(\alpha{} = R^2\) (since the periodic coordinate \(\theta{}\) has period \(2\pi{}\)).
\end{remark}
Now let \(\vartheta{}\) be a spherical coordinate on \(S^3\), so that
\begin{equation}\label{gS3}
g_{S^3} = \dd{}\vartheta^2 + \sin ^2\vartheta{}g_{S^2},
\end{equation}
where \(g_{S^2}\) is the metric on the unit round \(2\)-sphere. Since \(t =
r\cos \vartheta{}\) in flat Euclidean space, the hypersurface \(\set{\vartheta{}
= \pi{}/2}\) can be regarded as \(t=0\). The Wick rotation \(t\mapsto it\) then becomes \(\vartheta{}\mapsto \pi{}/2 + i\tau{}\), where \(\tau{}\in
\R\). In this way one obtains from \zcref{gS3} the expression \(-\dd{}\tau^2 + \cosh
^2\tau{} g_{S^2}\), and so one obtains the Witten bubble metric
\zcref{intro-WB-metric} from the metric \zcref{gb-alpha} with \(\alpha{} = R^2\).
\begin{remark}[The Witten bubble as a double Wick rotation of Schwarzschild]
The Witten bubble metric can be viewed as a double Wick rotation of the
Schwarzschild metric of \((4 + 1)\) dimensions, which can be written as
\begin{equation}
g_{\textnormal{Sch}} = -\Bigl(1-\frac{2M}{r^2}\Bigr)\dd{}t^2 + \Bigl(1-\frac{2M}{r^2}\Bigr)^{-1}\dd{}r^2 + r^2g_{S^3}.
\end{equation}
Under Wick rotation (namely the formal replacement of a coordinate \(x\) with
\(ix\)), the timelike \(t\)-direction in Schwarzschild becomes the compact
\(S^1\) dimension with coordinate \(\theta{}\), and the equatorial \(\vartheta{}\)-direction in
the \(S^3\) factor becomes the time coordinate \(\tau{}\) of the Witten bubble
spacetime.

This double Wick rotation turns the time-translation symmetry of Schwarzschild
into the \(\textnormal{U}(1)\)-symmetry of the Witten bubble, namely the
translation of the coordinate \(\theta{}\) in the compact \(S^1\) dimension.
Likewise, the \(\SO(4)\)-symmetry of Schwarzschild becomes the
\(\textnormal{SO}(3,1)\)-symmetry of the Witten bubble.
\label{wick-rotation}
\end{remark}
\begin{remark}[Relation to the positive mass theorem]
Since signals cannot travel faster than the speed of light, a spacetime into
which the Kaluza--Klein spacetime decays must have Kaluza--Klein asymptotics
towards spatial infinity. Since the Kaluza--Klein spacetime has zero energy
(more precisely, ADM mass), it can only decay into another spacetime with zero
energy.

By the positive mass theorem \cite{Schoen1979-gj,Schoen1981-yp,Witten1981-gn},
Minkowski space \(\R^{3 + 1}\) is the unique spacetime of zero energy in
classical \((3 + 1)\)-dimensional gravity. On the other hand, in \((4 +
1)\)-dimensional gravity, both the Kaluza--Klein spacetime and the Witten bubble
spacetime have zero energy. Indeed, in the computation of the integral defining
the energy (taken over a Cauchy hypersurface), only terms of order \(1/\rho{}\)
are relevant. On the hypersurface \(\set{\tau{}=0}\) of the Witten bubble
spacetime (which is a Cauchy hypersurface), the Witten bubble metric deviates
from the Kaluza--Klein metric only at order \(1/\rho^2\), and so the Witten
bubble spacetime has zero energy. The Witten bubble is therefore a
counterexample to a version of the positive mass theorem formulated for
asymptotically Kaluza--Klein initial data.

In fact, Witten \cite{WITTEN1982481} says that the techniques of Brill and Deser
\cite{Brill1968-iu} can produce solutions of the Einstein vacuum equations with
Kaluza--Klein asymptotics that have \emph{negative} energy. This dramatic failure of
the positive mass theorem is due to the change in topology (from Kaluza--Klein
initial data on \(\R^3\times S^1\) to initial data on \(\R^2\times S^2\) as in
the Witten bubble spacetime).
\end{remark}
\section{The causal geodesic flow on the Witten bubble spacetime}
\label{geodesic-motion} In this section, we analyze the \emph{(causal) geodesic flow} of
the Witten bubble spacetime \((\mathcal{M},g)\), which describes the motion of
free-falling particles. Our main result is \zcref{causal-geodesics}. The motion of
geodesics has also been studied in
\cite{PhysRevD.39.3151,Ofer_Aharony_2002,Bachelot_2016}.

We analyze the geodesic flow using canonical coordinates on the tangent bundle
\(T\mathcal{M}\). By canonical coordinates, we mean a local coordinate system
\((\mathbf{x}^\alpha{},\mathbf{p}^\alpha{})\) on \(T\mathcal{M}\), where
\((\mathbf{x}^\alpha{})\) are local coordinates on \(\mathcal{M}\) and
\((\mathbf{p}^\alpha{})\) are momentum coordinates on the fibers of
\(T\mathcal{M}\) such that \(\mathbf{p}^\alpha{} = \dd{}\mathbf{x}^\alpha{}(\mathbf{p})\)
for all \(\mathbf{p}\in T_x\mathcal{M}\). In canonical coordinates, the geodesic
flow is defined by the geodesic equations
\begin{equation}
\dot{\mathbf{x}}^\alpha{} = \mathbf{p}^\alpha{},\qquad \dot{\mathbf{p}}^\alpha{}= -\Gamma^\alpha_{\beta{}\gamma{}}\mathbf{p}^\beta{}\mathbf{p}^\gamma{}.
\end{equation}
Here \(\Gamma^\alpha_{\beta{}\gamma{}}\) are the Christoffel symbols of the metric \(g\) in the
coordinates \((\mathbf{x}^\alpha{})\), and the dot denotes differentiation with
respect to the affine parameter of the flow.

Since the Witten bubble spacetime is highly symmetric, its geodesic flow is
\emph{completely integrable}. Associated to the symmetries of the spacetime are various
conserved quantities: if \(X\) is a Killing vector, then \(g(\mathbf{p},X)\) is
conserved along the geodesic flow.
\begin{definition}[Conserved quantities along the geodesic flow]
Consider the \emph{rest mass} \(m\ge 0\) associated to a causal geodesic, defined by
\begin{equation}\label{mass-relation}
-m^2(\mathbf{x},\mathbf{p}) \coloneqq{} R^{-2}g(\mathbf{p},\mathbf{p}).
\end{equation}
By the geodesic equation \(\Grad _{\mathbf{p}}\mathbf{p} = 0\), the quantity \(m\) is conserved along
the flow.

Next, we introduce the \emph{\(S^1\)-angular momentum} \(\ell_\theta{} : T\mathcal{M}\to \R\):
\begin{equation}\label{S1-ang-momentum}
\ell{}_\theta{}(\mathbf{x},\mathbf{p}) \coloneqq{}R^{-2}g(\mathbf{p},\partial{}_\theta{})
\end{equation}
Since \(\partial_\theta{}\) is Killing, \(\ell_\theta{}\) is conserved along the flow.

Finally, we define the \emph{(squared) hyperbolic angular momentum}
\(\ell_{\textnormal{hyp}}^2 : T\mathcal{M}\to \R\) by
\begin{equation}\label{hyp-ang-momentum}
\ell{}_{\textnormal{hyp}}^2(\mathbf{x},\mathbf{p}) \coloneqq{}-\sum_{i=1}^3R^{-4}g(\mathbf{p},K_i)^2 + \sum_{i=1}^3 R^{-4}g(\mathbf{p},\Omega{}_i)^2
\end{equation}
Since the boosts \(K_i\) and the rotations \(\Omega_i\) are Killing vector fields
(see \zcref{killing}), \(\ell_{\textnormal{hyp}}^2\) is conserved along the flow.
\label{conserved-quantities}
\end{definition}
\begin{lemma}
We have \(\ell_{\textnormal{hyp}}^2\ge 0\).
\label{angular-momentum-well-defined}
\end{lemma}
\begin{proof}
We use the (global) Cartesian coordinates \(\mathbf{x} =
(\tau{},\hat{y},\hat{z},\vartheta{},\phi{})\) (see \zcref{sec:coordinates}),
in terms of which the function \(x\) is determined implicitly by \zcref{cart-def}. A
straightforward computation (which we omit) reveals that
\begin{equation}\label{l-hyp-expression}
\ell{}_{\textnormal{hyp}}^2 = \cosh^2 x(\cosh^2x(\mathbf{p}^\tau{})^2 - \cosh ^2x\cosh ^2\tau{}((\mathbf{p}^\vartheta{})^2 + \sin ^2\vartheta{}(\mathbf{p}^\phi{})^2))
\end{equation}
The mass relation \zcref{mass-relation} then reads
\begin{equation}\label{mass-relation-cons}
\begin{split}
-m^2 &= -\cosh^2 x (\mathbf{p}^\tau{})^2 + \mathfrak{F}(x)^2((\mathbf{p}^{\hat{y}})^2 + (\mathbf{p}^{\hat{z}})^2) + \cosh ^2x \cosh ^2\tau{} ((\mathbf{p}^\vartheta{})^2 + \sin ^2\vartheta{} (\mathbf{p}^\phi{})^2) \\
&= \mathfrak{F}(x)^2((\mathbf{p}^{\hat{y}})^2 + (\mathbf{p}^{\hat{z}})^2) - (\cosh ^2x)^{-2}\ell{}_{\textnormal{hyp}}^2
\end{split}
\end{equation}
The result follows.
\end{proof}
We can interpret the quantity \(\ell_{\textnormal{hyp}}\) as the rest mass of a
reparametrization of the projection of a geodesic in \(\mathcal{M}\) to the
\(\textnormal{dS}_{2 + 1}\) factor (given by level sets of \(x\)), which is itself a geodesic in
\(\textnormal{dS}_{2 + 1}\).
\begin{lemma}[Interpretation of \(\ell_{\textnormal{hyp}}\)]
Let \(\gamma{}\) be a geodesic in \(\mathcal{M}\), and let \(g_{\textnormal{dS}_{2 +
1}} = -\dd{}\tau^2 + \cosh ^2\tau{} g_{S^2}\) be the de Sitter metric. On any
interval where \(x(\gamma{})\neq{}0\), write \(\gamma{} =
(x,\tilde{\gamma{}},\theta{})\), where \(\tilde{\gamma{}} =
(\tau{},\vartheta{},\phi{})\). Then \(\tilde{\gamma{}}\) can be reparametrized
to a geodesic in \(\textnormal{dS}_{2 + 1}\). Indeed, if \(\gamma{}\) is
parameterized by \(\lambda{}\), then \(\tilde{\lambda{}}\) is an affine
parameter for \(\tilde{\gamma{}}\), where \(\tilde{\lambda{}}\) is defined by
\begin{equation}\label{lambda-til-def}
\dv{\tilde{\lambda{}}}{\lambda{}} = (\cosh x(\lambda{}))^{-2}.
\end{equation}
Moreover, the rest mass of \(\tilde{\gamma{}}\) (with respect to the metric
\(g_{\textnormal{dS}_{2 + 1}}\)) is \(\ell_{\textnormal{hyp}}\ge 0\).
\label{lhyp-interpretation}
\end{lemma}
\begin{proof}
Let \(s^A\) be local coordinates on \(\textnormal{dS}_{2 + 1}\). Write \(h = g_{\textnormal{dS}_{2 + 1}}\). The Lagrangian
for a geodesic on \(\mathcal{M}\) is
\begin{equation}
\frac{1}{2}[\cosh ^2x \dot{x}^2  + \cosh ^2 x h_{AB}(s)\dot{s}^A\dot{s}^B + \tanh^2 x\dot{\theta{}}^2],
\end{equation}
where the dot means differentiation with respect to \(\lambda{}\). The corresponding
Euler--Lagrange equation for \(s^A\) is
\begin{equation}
\frac{\dd{}}{\dd{}\lambda{}}(\cosh ^2 x h_{AB}\dot{s}^B) = \frac{1}{2} \cosh^2x \partial{}_A(h_{BC}\dot{s}^B\dot{s}^C).
\end{equation}
After expanding, dividing by \(\cosh ^2x\), and raising the \(A\)-index, we find
that this is equivalent to
\begin{equation}
\Grad ^h_{\dot{s}}\dot{s} + 2\tanh x \dot{x} \dot{s} = 0.
\end{equation}
Writing \('\) for differentiation with respect to \(\tilde{\lambda{}}\) defined by
\zcref{lambda-til-def}, so that \(s' = \cosh ^2x \dot{s}\).
\begin{equation}
\Grad ^h_{s'}s' = 0,
\end{equation}
and so \(s(\tilde{\lambda{}})\) is a geodesic of de Sitter. Next, we compute in
\((\tau{},\omega{})\) coordinates on \(\textnormal{dS}_{2 + 1}\) that
\begin{equation}
h(s',s') = \cosh ^4 x h(\dot{s},\dot{s}) = \cosh ^4 x[-(\mathbf{p}^\tau{})^2 +  \cosh ^2\tau{}((\mathbf{p}^\vartheta)^2 + \sin ^2\vartheta{}(\mathbf{p}^\phi{})^2)] = -\ell{}_{\textnormal{hyp}}^2,
\end{equation}
where the last equality is \zcref{l-hyp-expression}.

If a geodesic \(\gamma{}\) on \(\mathcal{M}\) crosses the bubble, then \(\ell_\theta{} = 0\)
\end{proof}
\begin{definition}[Weak trapping]
We say a curve \(\gamma{}\) in \(\mathcal{M}\) is \emph{weakly trapped} if its \(x\)-value
remains bounded. That is, there exist \(x_{\textnormal{min}}\) and
\(x_{\textnormal{max}}\) satisfying \(0\le x_{\textnormal{min}}\le
x_{\textnormal{max}}<\infty\) such that \(x(\gamma{})\in
[x_{\textnormal{min}},x_{\textnormal{max}}]\).
\label{weakly-trapped}
\end{definition}
\begin{proposition}[Causal geodesics in the Witten bubble spacetime]
Let \(\gamma{}\) be an inextendible causal geodesic in \(\mathcal{M}\). Recall from
\zcref{conserved-quantities} the following conserved quantities associated to
\(\gamma{}\): the mass \(m\ge 0\), the \(S^1\)-angular momentum
\(\ell_\theta{}\in \R\), and the hyperbolic angular momentum
\(\ell_{\textnormal{hyp}}^2\ge 0\). Then:
\begin{enumerate}
\item \label{degenerate-geodesics} If \(\ell{}_{\textnormal{hyp}}^2 = 0\), then \(\ell_\theta{} = 0\),
and \(\gamma{}\) is null (\(m=0\)) and lies in a level set of \(x\). In
particular, \(\gamma{}\) is weakly trapped.
\item \label{hypyes-S1no-geodesic} If \(\ell_{\textnormal{hyp}}^2 > 0\) and \(\ell_\theta{} = 0\), then:
\begin{enumerate}
\item \label{hypyes-S1no-geodesic-timelike} If \(\gamma{}\) is timelike \((m > 0)\), then it reaches the bubble and satisfies
\(\ell_{\textnormal{hyp}}^2\ge m^2\). The \(x\)-value of \(\gamma{}\) is
periodic in the affine parameter and is constant exactly when
\(\ell_{\textnormal{hyp}}^2 = m^2\). In particular, \(\gamma{}\) is weakly
trapped.
\item \label{hypyes-S1no-geodesic-null} If \(\gamma{}\) is null \((m = 0)\), then it meets the bubble at exactly one
value of the affine parameter, and it is not weakly trapped.
\end{enumerate}
\item \label{S1yes-geodesic} If \(\ell_\theta{} \neq{} 0\), then \(\gamma{}\) does not reach the bubble, and
its \(x\)-value is periodic in the affine parameter. In particular, \(\gamma{}\) is
weakly trapped.
\end{enumerate}
In all cases, \(\gamma{}\) is complete, and so the Witten bubble spacetime is causally
geodesically complete.
\label{causal-geodesics}
\end{proposition}
\begin{remark}[Geometric consequences]
Observe the following immediate consequences of \zcref{causal-geodesics}:
\begin{itemize}
\item Every geodesic with non-zero \(S^1\)-angular momentum is weakly trapped and
does not reach the bubble.
\item Timelike geodesics with zero \(S^1\)-angular momentum are weakly trapped and
either lie in the bubble or oscillate between the bubble and another level set
of \(x\).
\item Null geodesics with non-zero hyperbolic angular momentum and zero
\(S^1\)-angular momentum are not weakly trapped: they reach the bubble exactly
once and escape to \(\set{x=\infty}\).
\item Null geodesics with zero hyperbolic angular momentum lie a level set of \(x\).
\end{itemize}
\end{remark}
\begin{proof}
We will parametrize \(\gamma{}\) by the affine parameter \(\lambda{}\). Unless
stated otherwise, we will use the (global) Cartesian coordinates \(\mathbf{x} =
(\tau{},\hat{y},\hat{z},\vartheta{},\phi{})\).

\step{Step 1: Proof of part \zcref{degenerate-geodesics}.} We first consider the case
\(\ell_{\textnormal{hyp}}^2 = 0\). It follows from
\zcref{mass-relation-cons} that if \(\ell_{\textnormal{hyp}}^2 = 0\), then \(m = 0\)
and \(\mathbf{p}^{\hat{y}}=\mathbf{p}^{\hat{z}}=0\). Since \(x\) is
determined by \(\hat{y}\) and \(\hat{z}\), this means that \(\gamma{}\)
is constrained to a level set of \(x\). Moreover, the computation
\begin{equation}\label{ltheta-comp}
\ell{}_\theta{} = \mathfrak{F}(x)^2(\hat{y}\mathbf{p}^{\hat{z}} - \hat{z}\mathbf{p}^{\hat{y}}).
\end{equation}
shows that \(\ell_\theta{} = 0\).

\step{Step 2: Proof of part \zcref{hypyes-S1no-geodesic}.} Next, we consider the case
\(\ell_{\textnormal{hyp}}^2 > 0\) and \(\ell_\theta{} = 0\). It follows from the computation
\zcref{ltheta-comp} that geodesics with \(\ell_\theta{} = 0\) lie on a line through the origin
in the \((\hat{y},\hat{z})\)-plane. By \(\U(1)\)-symmetry, we may assume without
loss of generality that this line is \(\set{\hat{z}=0}\). Since the function
\(x\) is not smooth at the bubble, we introduce the signed radial coordinate
\(\tilde{x}\): we extend the function
\begin{equation}
\mathfrak{f}(x) = e^{\cosh x}\tanh (x/2)
\end{equation}
to an odd smooth function of \(\tilde{x}\in \R\), and define \(\tilde{x}\) as a
function along \(\gamma{}\) by
\begin{equation}
\hat{y}=\mathfrak{f}(\tilde{x}),\qquad \hat{z}=0.
\end{equation}
Since \(\mathfrak{f}'(0)\neq{}0\), the function \(\tilde{x}\) is smooth across the bubble, and
by construction \(x = \abs{\tilde{x}}\). The geodesic equation for \(\tilde{x}\) is
\begin{equation}\label{xtil-second-order}
\tilde{x}'' = -m^2\sinh \tilde{x}\cosh \tilde{x},
\end{equation}
where \('\) denotes differentiation with respect to \(\tilde{\lambda{}}\). We now derive
the corresponding first integral. Whenever \(x(\gamma{})\neq{}0\), the
expression \zcref{l-hyp-expression} for \(\ell_{\textnormal{hyp}}^2\) turns the mass
relation \zcref{mass-relation} into the radial equation
\begin{equation}\label{px-sq-expr}
(\mathbf{p}^x)^2 = \frac{1}{\cosh^2 x}\Bigl(\frac{\ell{}_{\textnormal{hyp}}^2}{\cosh ^2x} - \frac{\ell{}_\theta^2}{\tanh ^2x} - m^2\Bigr).
\end{equation}
Changing from \(\lambda{}\) to \(\tilde{\lambda{}}\), we obtain the effective radial equation
\begin{equation}\label{xtil-first-integral}
(\tilde{x}')^2 = \ell^2_{\textnormal{hyp}} - m^2\cosh^2 \tilde{x}.
\end{equation}
In particular, we have
\begin{equation}\label{xtil-first-integral-cons}
\ell_{\textnormal{hyp}}^2\ge m^2.
\end{equation}

\step{Step 2a: \(\gamma{}\) is timelike.} We first suppose that \(m > 0\). If
\(\ell_{\textnormal{hyp}}^2 = m^2\), then \zcref{xtil-first-integral} forces
\(\tilde{x}' = \tilde{x} = 0\). Otherwise, by \zcref{xtil-first-integral-cons} we
have \(\ell_{\textnormal{hyp}}^2 > m^2\). In this case, define
\(\tilde{x}_{\textnormal{max}} \coloneqq{}\cosh ^{-1}(\ell{}_{\textnormal{hyp}}/m)\), so that
\(\abs{\tilde{x}}\le \tilde{x}_{\textnormal{max}}\) (by \zcref{xtil-first-integral}). From
\zcref{xtil-second-order} we see that if \(\tilde{x} =
\tilde{x}_{\textnormal{max}}\), then \(\tilde{x}' = 0\) and \(\tilde{x}'' < 0\),
and if \(\tilde{x} = -\tilde{x}_{\textnormal{max}}\), then \(\tilde{x}' = 0\)
and \(\tilde{x}'' > 0\). Moreover, we have
\begin{equation}
\int_{-\tilde{x}_{\textnormal{max}}}^{\tilde{x}_{\textnormal{max}}} \frac{\dd{}\tilde{x}}{\sqrt{\ell{}_{\textnormal{hyp}}^2 - m^2\cosh ^2\tilde{x}}} < \infty.
\end{equation}
It follows from this and an argument invoking uniqueness of solutions to the ODE
\zcref{xtil-second-order} that \(\tilde{x}\) is periodic in \(\tilde{\lambda{}}\), and so
\(x = \abs{\tilde{x}}\) is periodic in \(\lambda{}\).

\step{Step 2b: \(\gamma{}\) is null.} If \(m = 0\), then the first integral
\zcref{xtil-first-integral} becomes
\begin{equation}
(\tilde{x}')^2 = \ell^2_{\textnormal{hyp}},
\end{equation}
and so
\begin{equation}
\tilde{x}(\tilde{\lambda{}}) = \pm \ell_{\textnormal{hyp}}\cdot (\tilde{\lambda{}}-\tilde{\lambda{}}_0)
\end{equation}
for some \(\tilde{\lambda{}}_0\in \R\). It follows that \(\gamma{}\) crosses the bubble exactly
once and \(x(\gamma{})\to \infty\) as \(\lambda{}\to \pm \infty\).

\step{Step 3: Proof of part \zcref{S1yes-geodesic}.} By \zcref{ltheta-comp}, if \(\ell_\theta{} \neq{} 0\),
then \(x(\gamma{}) > 0\). Then \zcref{px-sq-expr} applies globally, and we deduce that \(x\ge
x_{\textnormal{min}} > 0\). Reparametrizing by \(\tilde{\lambda{}}\) in place of \(\lambda{}\)
as in Step 2, we obtain from \zcref{px-sq-expr} the equation
\begin{equation}\label{x'-equation}
(x')^2 + V(x) = \ell{}_{\textnormal{hyp}}^2,\qquad V(x) = m^2\cosh ^2x + \ell{}_\theta^2 \frac{\cosh ^4x}{\sinh ^2x}.
\end{equation}
Moreover, the \(x\)-component of the geodesic equation is
\begin{equation}\label{x''-equation}
x'' = -\frac{1}{2}V'(x).
\end{equation}
As a function of \(s = \sinh ^2x\), we have (using \(\cosh ^2x = 1 + s\))
\begin{equation}
V(s) = m^2 + 2\ell{}_\theta^2 +  (m^2 + \ell{}_\theta^2)s + \frac{\ell{}_\theta^2}{s}.
\end{equation}
Since \(V\) is strictly convex as a function of \(s\) and tends to \(\infty\) as
\(s\to{}0\) or \(s\to{}\infty\), it has a unique minimum. Since \(x\mapsto  s(x) = \sinh^2 x\) is
strictly increasing for \(x > 0\), we conclude that \(V\) also has a unique
minimum as a function of \(x\). It follows from this and \zcref{x'-equation} that
the possible values of \(x(\gamma{})\), which satisfy \(V(x)\le \ell_{\textnormal{hyp}}^2\),
are exactly the interval \([x_{\textnormal{min}}, x_{\textnormal{max}}]\), where
\(V(x_{\textnormal{min}}) = V(x_{\textnormal{max}}) = \ell_{\textnormal{hyp}}^2\). If
\(x_{\textnormal{min}} = x_{\textnormal{max}}\), then \zcref{x''-equation} shows
that \(x(\gamma{})\) is constant. Otherwise, the two endpoints are distinct simple
roots of \(V(x) - \ell_{\textnormal{hyp}}^2\), and \(V'(x_{\textnormal{min}}) < 0\)
and \(V'(x_{\textnormal{max}}) > 0\). It follows from \zcref{x''-equation} that
\(x'' > 0\) if \(x = x_{\textnormal{min}}\) and \(x'' < 0\) if \(x =
x_{\textnormal{max}}\). Moreover,
\begin{equation}
\int_{x_{\textnormal{min}}}^{x_{\textnormal{max}}} \frac{\dd{}x}{\sqrt{\ell{}_{\textnormal{hyp}}^2 - V(x)}} < \infty,
\end{equation}
by the simplicity of \(x{_{\textnormal{min}}}\) and \(x_{\textnormal{max}}\) as
roots of \(V(x) - \ell_{\textnormal{hyp}}^2\). By an argument invoking
uniqueness of solutions to \zcref{x''-equation}, we conclude that \(x\) is periodic
in \(\tilde{\lambda{}}\). Since \(\frac{\dd{}\tilde{\lambda{}}}{\dd{}\lambda{}}
= (\cosh x(\lambda{}))^{-2}\) is positive and periodic, \(x\) is also periodic
in \(\lambda{}\).

\step{Step 4: Proof of geodesic completeness.} Since \(\frac{\dd{}\lambda{}}{\dd{}\tilde{\lambda{}}} \ge 1\),
it is enough to prove that \(\gamma{}\) is complete with respect to the
non-affine parameter \(\tilde{\lambda{}}\). First, we note that if \(\gamma{}\)
crosses the bubble, then the proof of \zcref{lhyp-interpretation} goes through with
\(x\) replaced by the smooth signed radial coordinate \(\tilde{x}\). By this
minor extension of \zcref{lhyp-interpretation} and the geodesic completeness of de
Sitter, it is enough to prove that the \((\hat{y},\hat{z})\)-motion is complete.
If \(\ell_\theta{} = 0\), as in Step 2, this is the same as proving that the
\(\tilde{x}\)-motion is complete. If \(\ell_\theta{} \neq{} 0\), as in Step 3, then
this is the same as proving that the \((x,\theta{})\)-motion is complete.
\begin{itemize}
\item If \(\ell_{\textnormal{hyp}}^2 = 0\), then \(\hat{y}\) and \(\hat{z}\) are
constant, as shown in Step 1.
\item If \(\ell_{\textnormal{hyp}} > 0\) and \(\ell_\theta{} = 0\), then:
\begin{itemize}
\item If \(m > 0\), then \(\tilde{x}\) is periodic, as shown in Step 2a.
\item If \(m = 0\), then \(\tilde{x}(\tilde{\lambda{}})= \pm \ell_{\textnormal{hyp}}\tilde{\lambda{}} +
    C\) for some constant \(C\), as shown in Step 2b, and so the
\(\tilde{x}\)-motion is complete.
\end{itemize}
\item If \(\ell_\theta{} \neq{} 0\), then, as shown in Step 3, \(x\) oscillates periodically
between \(x_{\textnormal{min}} > 0\) and \(x_{\textnormal{max}}\ge
  x_{\textnormal{min}}\), and so the \(x\)-motion is complete. Moreover,
\(\theta{}' = \ell_\theta{} (\cosh x)^2(\tanh x)^{-2}\), so that \(c\le
  \abs{\theta{}'}\le C\) for constants \(c>0\) and \(C > 0\), and so the
\(\theta{}\)-motion is complete.
\end{itemize}
\end{proof}
\hypersetup{urlcolor=Black}
\printbibliography[heading=bibintoc, title={References}]
\end{document}